\documentclass[11pt,a4paper]{article}
\usepackage[left=2.5cm,right=2.5cm,bottom=3cm,top=3cm,bindingoffset=0.5cm]{geometry}

\usepackage[utf8]{inputenc}
\usepackage{color}
\usepackage{amsmath}
\usepackage{mathtools}
\usepackage{amssymb}
\usepackage{amsthm}
\usepackage{fancyhdr}
\usepackage{csquotes}
\usepackage{authblk}
\usepackage{bbm}
\usepackage{comment}

\usepackage{tikz}

\usepackage[T1]{fontenc}
\usepackage{amsfonts}
\usepackage{vmargin}
\usepackage{setspace}
\usepackage{xcolor}
\usepackage{esint}

\definecolor{vertfonce}{rgb}{0.20, 0.46, 0.25}
\definecolor{rougefonce}{rgb}{0.64, 0.09, 0.20}
\usepackage[breaklinks=true,
            colorlinks=true,
            linkcolor=rougefonce,
            citecolor=vertfonce]{hyperref}
\usepackage{cleveref}

\usepackage{mathrsfs, enumerate}

\theoremstyle{definition} \newtheorem{definition}{Definition}[section]
\theoremstyle{plain} \newtheorem{theorem}[definition]{Theorem}
\theoremstyle{plain} 
\theoremstyle{plain} \newtheorem{proposition}[definition]{Proposition}
\theoremstyle{plain} \newtheorem{lemma}[definition]{Lemma}
\theoremstyle{plain} \newtheorem{corollary}[definition]{Corollary}
\theoremstyle{plain} \newtheorem{remark}[definition]{Remark}
\theoremstyle{definition} 

\numberwithin{equation}{section}
\fancypagestyle{plain}{%
  \fancyhf{}%
}

\newcommand{\Z}{\mathbb{Z}}

\newcommand{\R}{\mathbb{R}}
\newcommand{\C}{\mathbb{C}}

\DeclareMathOperator{\supp}{\mathrm{supp}}

\newcommand{\ldelta}{\ell_{\delta}}

\newcommand{\one}{\mathbbm{1}}
\newcommand{\rmu}{\rho_{\mu}}

\newcommand{\dd}{\mathrm{d}}
\newcommand{\grandO}{\mathcal{O}}

\newcommand{\V}{\widehat{v}}

\newcommand{\gt}{g}

\DeclarePairedDelimiterX\innerp[2]{\langle}{\rangle}{
	#1, #2
}

\usepackage[explicit]{titlesec}
\titleformat{\section}{\centering\Large\bfseries}{\thesection \ --}{0.7em}{\Large\bfseries #1}
\titleformat{\subsection}{\centering\large\bfseries}{\thesubsection \ --}{0.4em}{\large\bfseries #1}
\titleformat{\subsubsection}{\centering\bfseries}{\thesubsubsection \ --}{0.4em}{\bfseries #1}

\makeatletter

\@addtoreset{equation}{section}
\makeatother

\allowdisplaybreaks[1]

\usepackage{hyperref}

\begin{document}

\title{The Ground State Energy of a Two-Dimensional Bose Gas}
\author{S. Fournais\thanks{fournais@math.au.dk}}
\author{T. Girardot\thanks{theotime.girardot@math.au.dk}}
\author{L. Junge\thanks{junge@math.au.dk}}
\author{L. Morin\thanks{leo.morin@math.au.dk}}
\author{M. Olivieri\thanks{marco.olivieri@math.au.dk}}
\affil{\small{Department of Mathematics, Aarhus University,\\  Ny Munkegade 118,  DK-8000 Aarhus C, Denmark}}
\date{}

\maketitle
\begin{abstract}
We prove the following formula for the ground state energy density of a dilute Bose gas with density $\rho$ in $2$ dimensions in the thermodynamic limit 
\begin{align*}
e^{\rm{2D}}(\rho) = 4\pi \rho^2 Y\Big(1 - Y \vert \log Y \vert + \Big( 2\Gamma + \frac{1}{2} + \log(\pi) \Big) Y \Big) + o(\rho^2 Y^{2}),
\end{align*}
as $\rho a^2 \rightarrow 0$.
Here $Y= |\log(\rho a^2)|^{-1}$ and $a$ is the scattering length of the two-body potential. This result in $2$ dimensions corresponds to the famous Lee-Huang-Yang formula in $3$ dimensions. The proof is valid for essentially all positive potentials with finite scattering length, in particular, it covers the crucial case of the hard core potential.
\end{abstract}

\tableofcontents

\section{Introduction}

The calculation of the ground state energy of a dilute gas of bosons is of fundamental importance and has been the focus of much attention in recent years. 
This question can be posed in all dimensions of the ambient space, but
of course, the most important case from the point of view of Physics is the $3$-dimensional situation. However, also $1$ and $2$ dimensions are experimentally realizable. In this paper we study the $2$-dimensional setting and prove an asymptotic formula analogous to the famous Lee-Huang-Yang formula in $3$-dimensions.

Let us be more precise about the setting of the result. We consider positive, measurable potentials $v: {\mathbb R}^2 \rightarrow [0,+\infty]$ that are radial. Given such a potential, we will let $a=a(v)$ be its scattering length (for details on the scattering length see Section~\ref{section:scattering}) and define the Hamiltonian
\begin{align}
H(N,L) = \sum_{j=1}^N -\Delta_j + \sum_{j<k} v(x_j-x_k),
\end{align}
on $L^2(\Omega^N)$, with $\Omega = [-\frac{L}{2}, \frac{L}{2}]^2$.
The ground state energy density in the thermodynamic limit $e^{\rm{2D}}(\rho)$ is then defined by
\begin{align}\label{eq:erho_def}
e^{\rm{2D}}(\rho) := \lim_{\substack{L \rightarrow \infty\\ N/L^2 \rightarrow \rho}} L^{-2} \inf_{\Psi \in C_0^{\infty}(\Omega^N)} \frac{ \langle \Psi, H(N,L) \Psi \rangle}{\|\Psi\|^2}.
\end{align}
It is a standard result that the limit exists, and actually our analysis of $e^{\rm{2D}}(\rho)$ proceeds by giving upper bounds on the $\limsup$ and lower bounds on the $\liminf$.
It is also well-known that the limit is independent of the boundary conditions. The fact that we consider $\Psi \in C_0^{\infty}$ in the formula above, corresponds to the choice of Dirichlet boundary conditions for concreteness.

\begin{theorem}[Main result]\label{thm:LHY_2D}
For any constants $C_0, \eta_0 >0$, there exist $C, \eta >0$ (depending only on $C_0$ and $\eta_0$) such that the following holds. If the potential $v: {\mathbb R}^2 \rightarrow [0,+\infty]$ is non-negative, measurable and radial with scattering length $a$ and $\rho a^2 < C^{-1}$,
and, furthermore,
\begin{align}\label{eq:decay}
 v(x) \leq \frac{C_0}{\vert x \vert^{2}} \Big( \frac{a}{\vert x \vert} \Big)^{\eta_0}, \qquad \text{for all } \vert x \vert \geq C_0 a .
\end{align}
Then 
\begin{align}\label{eq:MainHC}
\Big| e^{\rm{2D}}(\rho) - 4\pi \rho^2 \delta_0 \Big(1 + \Big(2\Gamma + \frac{1}{2} + \log(\pi) \Big) \delta_0 \Big)
\Big| \leq C \rho^2 \delta^{2+\eta}_0,
\end{align}
with
\begin{align}\label{eq:def_delta}
\delta_0 := |\log(\rho a^2 |\log(\rho a^2)|^{-1})|^{-1},
\end{align}
where $\Gamma=0.577\ldots$ is the Euler-Mascheroni constant.
\end{theorem}

In terms of the simpler parameter $Y = |\log(\rho a^2)|^{-1}$, we get from \eqref{eq:MainHC}, expanding $\delta_0$ in terms of $Y$, the three-term asymptotics
\begin{align}\label{eq:simpleLHY}
e^{\rm{2D}}(\rho) = 4\pi \rho^2 Y\Big(1 - Y \vert \log Y \vert + \Big( 2\Gamma + \frac{1}{2} + \log(\pi) \Big) Y \Big) + {\mathcal O}(\rho^2 Y^{2+\eta}).
\end{align}
Here the third term in the asymptotics is analogous to the famous Lee-Huang-Yang term in the $3$-dimensional situation.

Notice, in particular, that the decay assumption \eqref{eq:decay} is valid for potentials with compact support. So Theorem~\ref{thm:LHY_2D} applies to the very important special case of the hard core potential of radius $a$:
\begin{align}
v_{\rm hc}(x) = \begin{cases} 0, & |x| > a, \\ +\infty, & |x| \leq a.
\end{cases}
\end{align}
For this potential the radius of the support is equal to its scattering length.

The proof of Theorem~\ref{thm:LHY_2D} will proceed by establishing upper and lower bounds. In Theorems~\ref{thermo dynamic limit theorem} and~\ref{thm:main_lower} below, we will state more precisely the estimates for the upper and lower bounds, respectively, and the assumptions necessary for each of these.
In Section~\ref{sec:strat} below, we will give an outline of the paper as well as these precise statements.

The first term $4\pi \rho^2 Y$ in \eqref{eq:simpleLHY} was understood in \cite{PhysRevA.3.1067} but a full proof was only given in 2001 in the paper \cite{MR1827922}.
Calculations beyond leading order were given in \cite{PhysRevE.64.027105,HINES197812,PhysRevLett.102.180404,Yang_2008}, but have so far not been rigorously proven. The recent papers  \cite{MR4261710,CCS2D} give an analogous expansion of the ground state energy in the setting of the Gross-Pitaevskii regime, giving furthermore information about the excitation spectrum. The constant in the second order term was also found in \cite{MR3986933} by restricting to quasi-free states in a special scaling regime.

In the $3$-dimensional case, the asymptotic formula for the energy density (with $e^{\rm{3D}}(\rho)$ defined analogously to \eqref{eq:erho_def} and $a$ being here the $3$-dimensional scattering length) is
\begin{align}\label{eq:LHY3D}
e^{\rm{3D}}(\rho) = 4\pi a \rho^2 \Big( 1 + \frac{128}{15\sqrt{\pi}} \sqrt{\rho a^3} \Big) + o( a \rho^2\sqrt{\rho a^3} ).
\end{align}
This is the famous Lee-Huang-Yang formula. The leading order term goes back to \cite{Lenz}, and the second term---the Lee-Huang-Yang (LHY) term---were given in \cite{Bog,LHY}.
Mathematically rigorous proofs of the leading order term were given in \cite{dyson} (upper bound) and \cite{LY} (matching lower bound).
Upper bounds for sufficiently regular potentials to the precision of the LHY-term were given in \cite{ESY} (correct order only), \cite{YY} (first upper bound with correct constant on the LHY-term) with recent improvements in \cite{BCS}.
Lower bounds of second order were given in \cite{FS} (potentials in $L^1$) and \cite{FS2} (general case including the hard core potential).
The upper bound in $3$-dimensions in the case of potentials with large $L^1$-norm, in particular the key example of hard core potentials, is still open.

As can be understood from this overview of results from the analysis of the $3D$ case, it is difficult to prove precise results on the energy when $\int_{{\mathbb R}^3} v$ is much larger than the scattering length $a(v)$, \textit{i.e.}, the hard core case. In $2$-dimensions the analogous comparison is between $\int_{{\mathbb R}^2} v$ and $\delta_0$, which always satisfy $\int_{{\mathbb R}^2} v \gg \delta_0$. So in $2$-dimensions we face similar challenges as in the $3$D hard core case, even for regular potentials.
This is one of the reasons why progress on the $2$D problem has been slower. It is therefore remarkable that Theorem~\ref{thm:LHY_2D} can be established, including both upper and lower bounds, without any extra assumptions on the potentials. 

Throughout the paper we will use the standard convention that $C>0$ will denote an arbitrarily large universal constant whose value can change from one line to the other.

\paragraph{Notation}
We will use the following notation for Fourier transforms,
$$
\widehat{f}(p) = \widehat{f}_p = \int e^{-ixp} f(x)\,\dd x.
$$
In the paper we will use the notation $A \ll B$ in a precise sense given by \eqref{eq:order_comparison}. 

\paragraph{ACKNOWLEDGEMENTS}
This research was partially supported by the grant 0135-00166B from Independent Research Fund Denmark.

\section{Strategy of the proof}\label{sec:strat}

\subsection{Upper bound }

As upper bound we prove the following theorem.

\begin{theorem}\label{thermo dynamic limit theorem}For any constants $C_0$, $\eta_0 >0$,  there exists $C$ (that depends only on $C_0$ and $\eta_0$) such that the following holds. Let $v : \R^2 \rightarrow [0,\infty]$ be a non-negative, measurable and radial potential with scattering length $a<\infty$, and satisfying the following decay property, 
\begin{equation}\label{eq:v_tail}
 v(x) \leq \frac{C_0}{\vert x \vert^{2}} \Big( \frac{a}{\vert x \vert} \Big)^{\eta_0} \quad \text{for} \quad \vert x \vert \geq C_0 a .
 \end{equation}
Then, if $\rho a^2 < C^{-1}$,  
\begin{equation*}
e^{\rm{2D}}(\rho) \leq  4\pi \rho^2 \delta_0 \Big(1 + \Big(2\Gamma + \frac{1}{2} + \log(\pi) \Big) \delta_0 \Big)
 + C \rho^2 \delta^{3}_0\vert\log(\delta_0)\vert,
\end{equation*}
with $\delta_0$ given by \eqref{eq:def_delta}. 
\end{theorem}

 In order to prove Theorem~\ref{thermo dynamic limit theorem}, we will reduce the analysis to the case of compactly supported potentials on a smaller periodic box $\Lambda = \Lambda_\beta = [- \frac{L_{\beta}}{2},\frac{L_{\beta}}{2} ]^2$ with length 
\begin{equation}
L_\beta = \rho^{-1/2} Y^{-\beta}, \qquad \beta >0.
\end{equation}
In this box, if the density is $\rho$, the number of particles is $N = \rho L_{\beta}^2 = Y^{-2\beta }\gg 1$. Throughout the paper we find conditions on $\beta$ over which we will optimize. For a potential $v$ with $\supp v \subseteq B(0,\frac{L_{\beta}}{2})$, we consider the following Hamiltonian acting on the Fock space $\mathscr{F}_s(L^2(\Lambda_\beta))$,
\begin{equation}\label{eq:upperbound.grandcanonical}
 \mathcal H_v = \bigoplus_{n \geq 0} \bigg( \sum_{i=1}^n - \Delta_{x_i}^{\rm{per}} +  \sum_{1 \leq i<j \leq n} v^{\rm{per}}(x_i-x_j) \bigg) .
\end{equation} 
Here $\Delta^{\rm{per}}$ is the periodic Laplacian, and $v^{\rm{per}}(x) = \sum_{m \in {\mathbb Z}^2}v(x+L_{\beta}m)$ is the periodic version of $v$. Note that for any $p \in \frac{2\pi}{L_\beta} \Z^2$, the Fourier coefficient of $v^{\rm{per}}$ is equal to the Fourier transform $\widehat v (p)$, because the radius of the support of $v$ is smaller than $L_{\beta}$. In this setting we prove the following result.

\begin{theorem}\label{thm.upperbound.grandcanonical}
For any $\beta \geq\frac{3}{2}$, there exists $C >0$, depending only on $\beta$ such that the following holds. Let $\rho >0$ and $v : \R^2 \rightarrow [0,\infty]$ be a non-negative, measurable and radial potential with scattering length $a$ and $\supp v \subset B(0,R)$ for some $R >0$. If $\rho R^2 \leq Y^{2 \beta + 2}$ and $\rho a^2 \leq C^{-1}$, then there exists a normalized trial state $\Psi \in \mathscr{F}_s(L^2(\Lambda_\beta))$, such that,
\begin{align*}
\langle \mathcal H_v \rangle_\Psi \leq 4\pi L_\beta^2 \rho^2 \delta_0 \Big(1 + \Big(2\Gamma + \frac{1}{2} + \log(\pi) \Big) \delta_0 \Big)
+ CL_{\beta}^2 \rho^2 \delta^{3}_0\vert\log(\delta_0)\vert.
\end{align*}
Moreover $\Psi$ satisfies $\langle \mathcal N \rangle_\Psi \geq N (1- C Y^2) ,$ and $\langle \mathcal N^2 \rangle_{\Psi} \leq 9 N^2$, where $\mathcal{N}$ is the number operator on $\mathscr{F}_s(L^2(\Lambda_{\beta}))$ and $N=\rho L_\beta^2= Y^{- 2 \beta}$.
\end{theorem}

We will show below, using standard techniques, how Theorem~\ref{thermo dynamic limit theorem} follows from Theorem~\ref{thm.upperbound.grandcanonical}. The proof of Theorem~\ref{thm.upperbound.grandcanonical} is divided into two parts, Sections~\ref{sec.soft} and~\ref{sec.upper.general}. We first prove in Section~\ref{sec.soft} a weaker upper bound when the potential is regular enough (we call it a \textit{soft} potential), using a quasi-free state $\Phi$. Indeed, this is an adaptation of the method of \cite{ESY,Napio1,Napio2} to the 2D case which amounts to a minimization of the energy over quasi-free states. We show in Theorem~\ref{thm.upperbound.soft.pot} with a good choice of $\Phi$ to our level of precision gives 
\begin{align}
		\langle  \mathcal H_v  \rangle_\Phi \leq 4\pi L_\beta^2 \rho^2 \delta_0 \Big(1 + \Big(2\Gamma + \frac{1}{2} + \log(\pi) \Big) \delta_0 \Big)
		+ CL_{\beta}^2 \rho^2 \delta_0(\widehat{v}_0-\widehat{g}_0)+CL_{\beta}^2 \rho^2 \delta^{2}_0\widehat{v}_0.\label{eq:strat01}
\end{align}
Here $g =  \varphi v$ and $\varphi$ is the scattering solution associated to $v$ (see Section~\ref{section:scattering} for the precise definition of $\varphi$ and with parameter $\delta_0$). This provides a first upper bound, but it is not enough to prove Theorem~\ref{thm.upperbound.grandcanonical}, unless $v$ admits a Fourier transform and $\widehat{v}_{0}$ is of order $\widehat{g}_{0}$. In Section~\ref{sec.upper.general} we explain how to reduce from any $v$ to a soft potential. To this end, we take care of the influence of the potential on a much shorter length scale by introducing $\varphi_b$ as the scattering solution normalized at
\begin{equation}
	b=\rho^{-\frac{1}{2}}Y^{\beta+\frac{1}{2}}.
\end{equation}
 Then our state will be the following product state
\begin{equation}\label{eq:strat02}
\Psi = \sum_{n \geq 1} F_n \Phi_n  , \quad F_n , \Phi_n \in L^2_s(\Lambda_\beta^n),
\end{equation} 
where $\Phi = \sum \Phi_n$ is a quasi-free state and $F_n$ is the so-called \textit{Jastrow function},
\begin{equation}\label{eq:jastrow}
F_n(x_1, \ldots, x_n) = \prod_{1 \leq i< j \leq n} f(x_i-x_j),
\end{equation}
with $f= \min(1,\varphi_b)$. When we compute the energy of such a state $\Psi$ we get 
\begin{equation}\label{eq:strat03}
\langle \mathcal H_v \rangle_\Psi \leq \langle \mathcal H_{\widetilde v} \rangle_\Phi + \langle \mathcal R \rangle _\Phi ,
\end{equation}
where $\mathcal R$ is an error term and $\widetilde v$ is the following soft potential,
\begin{equation}\label{eq:strat04}
\widetilde v = 2 f'(b) \delta_{\{\vert x \vert = b\}}.
\end{equation}
The power of $Y$ in $b$ is chosen minimal such that $\Vert \Psi\Vert=\Vert \Phi\Vert+O(Y^2)$, see Lemma~\ref{lemma denominator}. The same can be said about the assumption on $R$ (the support of $v$) in Theorem~\ref{thm.upperbound.grandcanonical}. Thus for the result to apply for the widest range of potentials we will want to choose $\beta$ as small as possible.

The potential in \eqref{eq:strat04} is soft in the sense that it has a decaying Fourier transform and $\widehat{\widetilde{v}}_0 \simeq \widehat{\widetilde g}_0$ (see Lemma~\ref{lemma.soft.pot.properties} for precise estimates). Then we can take for $\Phi$ the optimal quasi-free state satisfying \eqref{eq:strat01} for $\widetilde v$ and this turns out to be enough to prove Theorem~\ref{thm.upperbound.grandcanonical}.

Since the Jastrow factor \eqref{eq:jastrow} encodes all 2-particle interactions---at least on short scales---it is a natural trial state for getting upper bounds on the energy. In particular, it has been used to get the correct first order upper bound, both in 3D \cite{dyson} and 2D \cite{MR1827922}. In the product state $\Psi$, the Jastrow factor deals with short distance correlations between particles (when $\vert x_i - x_j \vert \leq b$), while long range effects are dealt with by the quasi-free state $\Phi$. In the case of hard core potentials, the Jastrow factor also imposes the necessary condition that our state vanishes whenever two particles are too close. 

We emphasize the following major differences between 2D and 3D. To be able to reduce to the quasi-free state $\Phi$, we need to bound $\grandO(N^2)$ terms of the form $f(x_i-x_j)$ by $1$. The number of particles $N$ in our box is not too large (powers of $|\log(\rho a^2)|$) thus making this error controllable. This is not at all the case in dimension 3, because the number of particles in the box that would give the thermodynamic limit is of order $(\rho a^3)^{-2}$. However, a similar state as ours was successfully used in the 3D Gross-Pitaevskii regime \cite{BCOPS}. In this regime the number of particles is $(\rho a^3)^{-\frac{1}{2}}$, which allows the authors, with substantially more work, to get through to a good upper bound.

Finally, one should notice that $\Phi$ is a quasi-free state, and does not include the soft pair interactions that were necessary in \cite{YY,BCS} to get the correct upper bound in 3D. Indeed, for a quasi-free state $\Phi$ the second order energy bounds are in terms of $\widehat{v}_0$ and to get the correct constant one needs to change $\widehat{v}_0$'s into $\widehat{g}_0$'s. This is the role of soft pairs. However, our potential $\widetilde v$ from \eqref{eq:strat04} already satisfies $\widehat{\widetilde{v}}_0 - \widehat{\widetilde g}_0 = {\mathcal O}(Y^2 \vert \log Y \vert)$ (see Lemma~\ref{lemma.soft.pot.properties}) and this replacement only gives errors of order $\rho^2 Y^3 \vert \log Y \vert$. It is possible that we could add the soft pair interactions into $\Phi$ to reduce this error at the expense of a much longer and more technical proof.

We conclude this section by proving Theorem~\ref{thermo dynamic limit theorem} using Theorem~\ref{thm.upperbound.grandcanonical} and the classical theory of localization to smaller boxes which is added for convenience in Appendix~\ref{appendix b}. 

\begin{proof}[Proof of Theorem~\ref{thermo dynamic limit theorem}]
We first cut our potential in order to apply Theorem~\ref{thm.upperbound.grandcanonical}. We write $v=v\one_{B(0,R)}+v\one_{B(0,R)^c}=v_{R}+v_{tail},$ where $R=\rho^{-\frac{1}{2}}Y^{\beta + 2}$. We denote by $a_{R}$ the scattering length of $v_{R}$. To get estimates on the energy density $e(\rho)$ we use the standard theory developed in Appendix~\ref{appendix b}. The idea is to extend $L_{\beta}$ with $R_0$ and force the trial function to have Dirichlet boundary conditions on the box of sidelength $L_{\beta}+R_0$. Thereafter one glues together these small Dirichlet boxes, separated by corridors of size  $R$. Since this process will slightly change the density, we choose for a given $\rho>0$, the larger density $\widetilde{\rho}$ satisfying
\[\rho=\widetilde{\rho}(1-2C\widetilde{Y}^2)\Big(1+\frac{2R_0}{L_\beta}+\frac{R}{L_\beta}\Big)^{-2},\]
where $C$ is the same as in Theorem~\ref{thm.upperbound.grandcanonical}, and $R_0=\rho^{-\frac 1 2} Y^{-\frac{1}{2}}$. This choice of $R_0$ is in fact optimal as one can see from the error term $C \frac{\rho}{L_{\beta}R_0}$ coming from the glueing process in \eqref{eq:erhoHcut}. Here we use the notation $\widetilde{Y} = \vert \log( \widetilde \rho a_{R}^2) \vert^{-1}$ and $\widetilde{\delta_0} = \vert \log( \widetilde \rho a_{R}^2\widetilde{Y}) \vert^{-1}$. If $\rho a^2$ is small enough then $\widetilde{\rho} a_{R}^2 \leq C^{-1}$, and we may use Theorem~\ref{thm.upperbound.grandcanonical} to find a periodic trial state $\Psi$ for the density $\widetilde{\rho}$ and potential $v_{R}$ satisfying
\begin{align}
\langle \mathcal H_{v_R} \rangle_\Psi &\leq 4\pi L_\beta^2 \widetilde{\rho}^2 \widetilde{\delta_0} \Big(1 + \Big(2\Gamma + \frac{1}{2} + \log(\pi) \Big) \widetilde{\delta_0} \Big)
+ CL_{\beta}^2 \widetilde{\rho}^2 \widetilde{\delta_0}^3\vert\log(\widetilde{\delta_0})\vert , \label{eq:Hvcut1}
\end{align}
with $\langle \mathcal N \rangle_\Psi \geq \widetilde{\rho}L_\beta^2(1-C\widetilde{Y}^2),$ and $\langle \mathcal N^2 \rangle_{\Psi} \leq C \widetilde{\rho}^2 L_\beta^4$ .

Since $\frac{R}{L_\beta} \ll \frac{R_0}{L_\beta} = Y^{-\frac{1}{2}+\beta} \ll 1$, we have $\vert \rho - \widetilde \rho \vert \leq C \rho Y^{-\frac{1}{2}+\beta}$, and we can change $\widetilde \rho$ into $\rho$ in \eqref{eq:Hvcut1} up to smaller errors if
\begin{equation}\label{condition on beta}
\beta \geq \frac{5}{2}.
\end{equation}
One can show that the $C$ appearing in \eqref{eq:Hvcut1} only increases in $\beta$ (see \eqref{eq. how C depends on beta}). Thus we find  $\beta=5/2$ to be optimal. We can also change $a_{R}$ into $a$ because the right-hand side of \eqref{eq:Hvcut1} is an increasing function of the scattering length and $a_{R} \leq a$. Thus,
\begin{align}
\langle \mathcal H_{v_{R}} \rangle_\Psi & \leq 4\pi L_\beta^2 \rho^2 \delta_0 \Big(1 + \Big(2\Gamma + \frac{1}{2} + \log(\pi) \Big) \delta_0 \Big)
+ CL_{\beta}^2 \rho^2 \delta^{3}_0\vert\log(\delta_0)\vert , \label{eq:Hvcutbound}
\end{align}
and the bounds on the number of particles become
\begin{equation}\label{eq:NPsirho}
\langle \mathcal N \rangle_\Psi \geq (\rho+c\rho^2)  (L_\beta+2R_0+R)^2 , \quad \text{and} \quad \langle \mathcal N^2 \rangle_{\Psi} \leq  C(\rho L_\beta^2)^2,
\end{equation}
for some $c> 0$. 

Now we can use Theorems~\ref{lem.B1} and~\ref{thm.B2} to glue small boxes together. We get a sequence $\Psi_{k(L_\beta+2R_0+R)}\in \mathscr{F}_s\big(L^{2}\big(\Lambda_{k(L_{\beta}+2R_0+R)}\big)\big)$ with Dirichlet boundary conditions, for $k\in \mathbb{N}$, on the box \[\Lambda_{k(L_\beta + 2 R_0 + R)} = \Big[ - \frac 1 2 k(L_\beta + 2 R_0 + R) , \frac 1 2 k(L_\beta + 2 R_0 + R) \Big]^2,\] satisfying 
\begin{equation}\label{eq:NPsik1}
\langle \mathcal N \rangle_{\Psi_{k(L_{\beta}+2R_0+R)}} = k^2\langle \mathcal N \rangle_\Psi,\qquad \langle \mathcal N^2 \rangle_{\Psi_{k(L_{\beta}+2R_0+R)}} = k^4\langle \mathcal N^2 \rangle_\Psi,
\end{equation}
and
\begin{equation}\label{eq:HvtoHvcut}
\langle \mathcal{H}_v \rangle_{\Psi_{k(L_\beta+2R_0+R)}}\leq k^2\Big(\langle \mathcal{H}_{v_{R}} \rangle_{\Psi}+\langle \mathcal{N} \rangle_{\Psi}\frac{C}{L_\beta R_0}+\langle \mathcal{N}^2 \rangle_{\Psi}\frac{Ca^{\eta_0}}{R^{2+\eta_0}}\Big),
\end{equation}
where $C$ only depends on $\eta_0$ and $C_0$.
Note that we have the original potential in the left-hand side of \eqref{eq:HvtoHvcut} because by the proof of Theorem~\ref{thm.B2}, $v_{tail}$ only produces an error term.
By construction this sequence satisfies the conditions on the number of particles for Theorem~\ref{thm.B3}, and we conclude
\begin{equation}\label{eq:erhoHcut}
e^{2D}(\rho)\leq \lim_{k\rightarrow \infty}\frac{\langle \mathcal{H}_v \rangle_{\Psi_{k(L_\beta+2R_0+R)}}}{
k^2L_\beta^2}\leq \frac{\langle \mathcal{H}_{v_{R}} \rangle_{\Psi}}{L^2_\beta}+ C \frac{\rho}{L_\beta R_0}+ Ca^{\eta_0}\frac{\rho^2 L_\beta^2}{R^{2+\eta_0}},
\end{equation}
where in the last inequality we used \eqref{eq:HvtoHvcut}, \eqref{eq:NPsirho} and that $\langle \mathcal{N}\rangle^2_{\Psi} \leq \langle \mathcal{N}^2\rangle_{\Psi}$. With our choice of parameters including \eqref{condition on beta}, the two last terms in \eqref{eq:erhoHcut} are errors. Theorem~\ref{thermo dynamic limit theorem} follows from \eqref{eq:erhoHcut} and \eqref{eq:Hvcutbound}.
\end{proof}

\subsection{Lower bound}

In this section we provide the strategy of proof for the theorem below. 

\begin{theorem}\label{thm:main_lower}

For any constant $\eta_1 >0$ there exist $C, \eta >0$ (depending only on $\eta_1$) such that the following holds. Let $\rho >0$ and $v: {\mathbb R}^2 \rightarrow [0,+\infty]$ be a non-negative, measurable and radial potential with scattering length $a < \infty$. If $\rho a^2 < C^{-1}$ and

\begin{equation}\label{eq:assumption_v_decay}
\int_{ \{\vert x \vert > \rho^{-1/2} \}} v(x) \log \Big( \frac{\vert x \vert}{a} \Big) ^2 \dd x \leq Y^{\eta_1},
\end{equation}
then
\begin{align}
 e^{\rm{2D}}(\rho) \geq  4\pi \rho^2 \delta_0 \Big(1 + \Big(2\Gamma + \frac{1}{2} + \log(\pi) \Big) \delta_0 \Big)
+  C \rho^2 \delta^{2+\eta}_0,
\end{align}
with $\delta_0$ as defined in \eqref{eq:def_delta}.
\end{theorem}

The proof has the same structure as in the $3D$ hard core case analyzed in \cite{FS2}. However, the $2D$ case comes with its own challenges due to the logaritmic divergences and changes of the lengthscales.

We introduce the lengths $\ell, \ell_{\delta}$ which a posteriori are going to correspond to

\begin{equation}
\ell = \rho^{-1/2} Y^{-1/2-\alpha}, \qquad \ell_{\delta} = \frac 1 2 e^{\Gamma} \rho^{-1/2} Y^{-1/2},
\end{equation}
for a certain $\alpha \in (0,1)$, the second of which being called the \textit{healing length}.
The proof will depend on a precise choice of a number of parameters. For convenience these and the relations between them have been collected in Appendix~\ref{app:parameters}.

We work at three different lengthscales:

\begin{itemize}

\item the thermodynamical scale, in the box $\Omega = [-L/2, L/2]^2$, where we state the main result in the limit $L \rightarrow + \infty$;

\item the large box scale $\Lambda = [-\ell/2, \ell/2]^2$, where we prove most of the results and by the sliding localization techniques we integrate over all these boxes to prove the lower bound in the whole thermodynamical box;

\item the small box scale $B = [-d \ell/2, d \ell/2]^2$, with $d \ll 1$, where we derive a bound for the number of particles excited out from the condensate, fundamental for the general strategy, obtaining the Bose-Einstein condensation (BEC).

\end{itemize}

The relations

\begin{equation}
d \ell \ll \ell_{\delta}\ll \ell \ll L,
\end{equation}

guarantee that the boxes are in a chain of inclusions.

Here below there is a summary of the strategy of the proof, collecting all the key points.

\begin{enumerate}

\item In Section~\ref{subsec:GCE} we reformulate the problem in a grand canonical setting, adding a chemical potential to the Hamiltonian, in order to control the distribution of particles. We also reduce to potentials which are compactly supported and with norm\footnote{The power $\frac{1}{8}$ is not optimal but chosen for convenience, in particular to be in agreement with \eqref{eq:M_large_matrices}.} $\|v\|_1 \leq Y^{-1/8}$, using the analysis of the scattering equation from Section~\ref{section:scattering}. Theorem~\ref{thm:main_lower} is shown to be a consequence of Theorem~\ref{thm:grancanonicaltherm}.

\item In order to prove Theorem~\ref{thm:grancanonicaltherm}, in Section~\ref{subsec:localization_large_boxes} we use a sliding localization technique to reduce the problem from the thermodynamical box $\Omega$ to the large box $\Lambda$. We derive a lower bound where the original Hamiltonian is bounded by the integral of large box Hamiltonians with varying position on the full space. The main result is then reduced to the proof of an analogous lower bound for a Hamiltonian on $\Lambda$, namely Theorem~\ref{thm:largebox_lower}. The next sections focus on this proof.

\item We split the potential energy on the large box in Section~\ref{subsec:splitpot} by means of projectors $P$ and $Q$ onto and outside the condensate, respectively. 
This idea has its roots in \cite{MR1301362}.
The splitting produces terms involving from $0$ to $4$ $Q$ projectors which can be interpreted as follows. The two body interaction given by $v$ can be seen as producing two particles with new outcoming momenta from two particles with given incoming momenta. The terms with $j$ $Q$'s in the Hamiltonian correspond to the case when $j$ non-zero momenta are involved in this interaction.

We replace some of the terms in $v$ by $g$, the error being absorbed by a positive term $\mathcal Q_4^{\text{ren}}$. This overcomes the difficulty, well known in the literature, to deal with the error made by substituting $\widehat v_0$ by $\widehat g_0$  in the leading term of the expansion. Since $\widehat g_0$ is much smaller than $\widehat v_0$, this can be interpreted as a renormalization procedure.

\item In \cite{YY} it is underlined how the interaction of the so-called \textit{soft pairs} contributes significantly to the energy. These correspond to two interacting high-momenta producing one $0$-momentum and one low-momentum. This is encoded in the action of the $3Q$ term, from which the soft pairs contribution emerges after two localizations. The first one (low outcoming momentum) is proved in Section~\ref{sec:Q3loc} and the second one (high incoming momenta) in Lemma~\ref{lem:Q3highloc}, the latter being easier treated in second quantization.

\item In Section~\ref{sec:apriori} we show that we can restrict the action of the Hamiltonian to vectors with a bounded number of low excitations (Theorem~\ref{thm:excitationrestriction}). This localization produces errors which are estimated in Appendix~\ref{app:proofd1d2}. It is based on a priori bounds on the number of excitations (Theorem~\ref{thm:excitationsbound}) which are proven in Appendix~\ref{app:low-smallbox}. In this Appendix we use a localization to small boxes $B$ (of lengthscale $d \ell \ll \rho^{-1/2} Y^{-1/2}$), in which we prove  Bose-Einstein condensation.

\item Section~\ref{sec:c-number} contains lower bounds that use a second quantization formalism in momentum space. We first write the Hamiltonian in this formalism in Section~\ref{subsec:sqh}. Then we use the $c-$number substitution in Section~\ref{subsec:c-number}, thus reducing to a problem of minimization for particles outside the condensate. The operators related to the condensate act as numbers over the class of coherent states over which we minimize. 
%We call the resulting operator  the Bogoliubov Hamiltonian. 
After this procedure we arrive at an operator which contains a constant term, a quadratic operator which we call the \textit{Bogoliubov Hamiltonian}, as well as other terms up to order $3$, in creation and annihilation operators of non-zero momenta, that have to be controlled.

\item In Section~\ref{sec:bogestimates} we distinguish the two cases where the density of particles in the condensate $\rho_z$ is far or close to $\rho_{\mu}$. In the first case, analyzed in Section~\ref{subsec:rhofar}, the associated energy is high enough to be bounded from below and to compensate the error in the leading term. In the second case a more careful analysis is needed. We use standard techniques, collected in Appendix~\ref{sec:bogdiag}, to diagonalize the Bogoliubov Hamiltonian and obtain a diagonalized quadratic Hamiltonian, useful for lower bounds, summed to an integral. This last integral is calculated in Appendix~\ref{app:bogintegral}, and we show how it gives both the compensation of the error in the leading term and the desired second order contribution. What remains at this point is to show how the remainders, including the localized $3Q$ term, are error terms, and this is the content of the technical Section~\ref{subsec:Q3_technical}. There we show how the contribution of the soft pairs is compensated by the positive Hamiltonian obtained by the diagonalization procedure.

\item Finally, in Section~\ref{subsec:proof_lower_final} we use all the previous results to give a proof of Theorem~\ref{thm:largebox_lower}, with the choices of the parameters in Appendix~\ref{app:parameters}, where all the conditions used to prove the lower bound are collected.

\item In the proof we need two technical estimates, namely \eqref{eq:kinetic_momentum_comparison}, \eqref{eq:kinetic_momentum_comparison2} and \eqref{eq:corF6}, which are taken from the 3D case and are independent of dimension. They are only stated and we refer to \cite{FS2} for the proof.

\end{enumerate}

 \section{The Scattering Solution in \texorpdfstring{$2$}{2} dimensions}\label{section:scattering}
\subsection{Basic theory}
In this section we establish the notation and results surrounding the two dimensional two body scattering problem. The standard properties of the scattering solutions stated below are well known and can be found in \cite[Appendix A]{greenbook}. We will only consider radial and positive potentials $v:\mathbb{R}^2\rightarrow [0,\infty]$, furthermore if $v$ is compactly supported we denote by $R$ the radius of the support of $v$, \textit{i.e.}, $v(x)=0$ if $\vert x\vert \geq R$. 
\begin{definition}
	For a compactly supported $v$ its scattering length $a=a(v)$ is defined as
	\begin{equation}\label{variational definition of the scattering length}	\frac{2\pi}{\log(\frac{\widetilde{R}}{a})}=\inf\Big\{\int_{B(0,\widetilde{R})}\Big(\vert \nabla \varphi\vert^2+\frac{1}{2}v\varphi^2 \Big)\dd x\quad \Big\vert \quad \varphi\in H^1(B(0,\widetilde{R})),\quad \varphi\vert_{\partial B(0,\widetilde{R})}=1\Big\},
	\end{equation}
	where $\widetilde{R}> R$ is arbitrary.
\end{definition}
By the positivity of the right hand side we find $a\leq R$. It is also easy to verify that $a$ is an increasing function of $v$ and is independent of $\widetilde{R}>R$. Furthermore for any $\widetilde{R}$ the above functional has a unique minimizer $\varphi_{v,\widetilde{R}}=\log(\frac{\widetilde{R}}{a(v)})^{-1}\varphi_v^{(0)}(x)$, where, for $v\in L^1(\mathbb{R}^2)$, we have
\begin{equation} \label{scattering_equation}
	-\Delta \varphi^{(0)}_v+\frac{1}{2}v\varphi^{(0)}_v =0 \quad \text{on $\R^2$},
\end{equation}
in the distributional sense. Furthermore,
$$\varphi^{(0)}_v(r)=\log\Big(\frac{r}{a(v)}\Big), \quad \text{for $r\geq R$},$$
and $\varphi^{(0)}_v$ is a monotone, non-decreasing and non-negative, radial function. We will omit the $v$ in the notation of the scattering length if the potential is clear from the context.

The logarithm in the $2D$-scattering solution is clearly unbounded for large values of $r$. This is a major difference to the $3D$ behaviour (where the scattering solution behaves as $1-\frac{a}{r}$ at infinity).
Therefore the scattering solution normalized to $1$ at a certain length $\widetilde{R}$ is of much greater importance. Using the parameter 
\begin{equation}\label{eq:delta}
	\delta=\frac{1}{2}\log\Big(\frac{\widetilde{R}}{a}\Big)^{-1},\qquad \text{\textit{i.e.}} \qquad \widetilde{R}=ae^{\frac{1}{2\delta}},
\end{equation}we define on $\mathbb{R}^2$
\begin{align}\label{eq:scat_defs}
	\varphi=\varphi_{v,\delta} = 2 \delta \varphi^{(0)}, \qquad \omega = 1- \varphi,
	\qquad g = v \varphi=v(1-\omega).
\end{align}
Clearly,
\begin{align}\label{eq:ScatOmega}
	-\Delta \omega = \frac{1}{2} g,
\end{align}
and, using the divergence theorem,
\begin{align}\label{eq:IntegralScattLength}
	\int g \,\dd x = 8\pi \delta.
\end{align}
We remark here again a difference between the $2D$ and $3D$ case: in $3D$, $\varphi$ would be normalized to $1$ at infinity and \eqref{eq:IntegralScattLength} would have an $a$ instead of $\delta$. 
\begin{remark}[On the parameter $\delta$] \label{remark.delta}
We clearly have some freedom in the choice of $\delta$, which amounts to determine a normalization lengthscale $\widetilde R$ for $\varphi$. Throughout the paper, we will need $\delta$ to be of the same order as $Y = \vert \log (\rho a^2) \vert^{-1}$, namely
\begin{equation}\label{eq. standing assumption on delta}
	\frac{Y}{2}\leq \delta\leq 2 Y, \qquad \text{or, equivalently, } \quad  (\rho a^2)^{-1/4} \leq \frac{\widetilde R}{a} \leq (\rho a^2)^{-1}.
\end{equation}
With this condition we can always exchange $Y$ and $\delta$ when estimating errors. We thus get upper and lower bounds on the energy depending on the parameter $\delta$. In both cases, it turns out that the optimal choice is $\delta = \delta_0$, which corresponds to $\frac{\widetilde R}{a} = (\rho a^2 Y)^{-1/2}$. See also Remarks~\ref{remark.delta.up} and~\ref{remark.delta.low}.
\end{remark}
\subsection{Potentials without compact support}
\begin{definition}
	For a potential $v$ without compact support the scattering length is defined as
	\[a(v)=\lim_{n\rightarrow \infty}a(v\one_{B(0,n)}).\]
\end{definition}
Since $a$ is an increasing function of $v$ the limit exits if and only if $\{a(v\one_{B(0,n)})\}_n$ is bounded, which by \cite[Lemma 1]{Landon_2012} is true if and only if there exists a $\widetilde{b}>0$ such that
\[\int_{\{\vert x\vert >\widetilde{b}\}}v(x)\log\Big(\frac{\vert x\vert }{\widetilde{b}}\Big)^2 \dd x<\infty.\]

We need to localize our potentials to have compact support. The next result estimates the change this localization induces in the scattering length.

\begin{lemma}\label{lemma: cutting potential so have compact support}
	For a potential $v$ with finite scattering length $a$ and $R>a$, let $v_R=\one_{B(0,R)}v$ and $a_R$ be its associated scattering length. Then,
	\begin{equation}\label{eq. scattering length of non-compact potential}
		0\leq \frac{2\pi}{\log(\frac{R}{a})}-\frac{2\pi}{\log(\frac{{R}}{a_R})}\leq \frac{1}{2} \int_{\{\vert x\vert > R\}}v(x)\frac{\log(\frac{\vert x\vert}{a})^2}{\log(\frac{R}{a})^2}\dd x.
	\end{equation}
\end{lemma}
\begin{proof}
	Let $\varphi_1$ be the scattering solution for $v_R$ normalized at $R$, and let
	\begin{equation}
		\varphi_n(x):=\begin{cases}
			\varphi_1(x)\frac{\log(\frac{R}{a})}{\log(n)}, &\vert x\vert \leq R,\\
			\frac{\log(\frac{\vert x\vert}{a})}{\log(n)}, &\vert x\vert \geq R.
		\end{cases}
	\end{equation}
	Notice that $\varphi_n$ is normalized at $a\cdot n$ and continuous. We use it as a trial function in the variational problem of $v_n=\one_{B(0,a\cdot n)}v$, with $n\cdot a>R$, to get (with $a_n:=a(v_n)$)
	\begin{equation}\label{eq:proof_scatt_prop_logn2}
		\frac{2\pi}{\log(\frac{a\cdot n}{a_n})}\leq\int_{\{\vert x\vert < R\}}\Big( \vert \nabla \varphi_n\vert^2+\frac{1}{2}v_n\varphi_n^2\Big) \dd x+\int_{\{R<\vert x\vert < a\cdot n\}}\Big( \vert \nabla \varphi_n\vert^2+\frac{1}{2}v_n\varphi_n^2\Big) \dd x.
	\end{equation}
	Since $\varphi_n$ is just a multiple of the scattering solution of $\varphi_1$ inside $R$ the first integral gives
\begin{equation}\label{eq:proof_scattpropr_logn}
	\int_{\{\vert x\vert < R\}} \Big(\vert \nabla \varphi_n\vert^2+\frac{1}{2}v_n\varphi_n^2 \Big)\dd x=\frac{2\pi}{\log(\frac{R}{a_R})}\frac{\log(\frac{R}{a})^2}{\log(n)^2}.
\end{equation}
	The second term is directly calculated using the explicit formula for $\varphi_n$,
	\begin{equation}
		\int_{\{R<\vert x\vert < a\cdot n\}}\Big(\vert \nabla \varphi_n\vert^2+\frac{1}{2}v_n\varphi_n^2\Big) \dd x\leq \frac{2\pi\log(\frac{a\cdot n}{R})}{\log(n)^2}+\frac{1}{2\log(n)^2}\int_{\{\vert x\vert >R\}}v\log\Big(\frac{\vert x\vert}{a}\Big)^2.
	\end{equation}
	By \eqref{eq:proof_scattpropr_logn}, multiplying \eqref{eq:proof_scatt_prop_logn2} through with $\log(n)^2$ and letting $n\rightarrow \infty$, whereby $a_n\rightarrow a$, yields 
	\[2\pi\log\Big(\frac{R}{a}\Big)\leq \frac{2\pi\log(\frac{R}{a})^2}{\log(\frac{R}{a_R})}+\frac{1}{2}\int_{\{\vert x\vert >R\}}v\log\Big(\frac{\vert x\vert}{a}\Big)^2\dd x.\]
	The result then follows by dividing through with $\log(\frac{R}{a})^2$. 
\end{proof}
\subsection{Compactly supported potentials with large integrals}
We state and prove here in the $2D$ setting a similar approximation result as the one found in \cite[Theorem 1.6]{FS2} for the scattering length in $3D$.
\begin{lemma}\label{lemma. large integral potential} 
	For a radial, positive $v\in L^1(\mathbb{R}^2)$ with support contained in $B(0,R)$ there exists, for any $T>0$, a $v_T:\mathbb{R}^2\rightarrow [0,+\infty]$ satisfying 
	\begin{equation}
		0\leq v_T(x)\leq v(x), \quad \text{ for all } x\in\mathbb{R}^2, \text{ and }\qquad \int v_T\leq 4 \pi T,\
	\end{equation}
	and such that
	\begin{equation}		
			\frac{2\pi}{\log(\frac{R}{a})}-\frac{2\pi}{\log(\frac{R}{a_T})}\leq \frac{2\pi}{\log(\frac{R}{a})^2T}.
	\end{equation}
\end{lemma}
\begin{proof}
	Due to the integrability assumption on $v$ we may define $R_T=\inf\{R'>0: \int_{\{\vert x\vert\geq R'\}}v\dd x< 4 \pi T\}$ and
	\begin{equation}
		v_T:=v\one_{\{\vert x\vert >R_T\}}.
	\end{equation}
	Clearly,
	\begin{equation}\label{eq. integration of vt}
		\int v_T=4 \pi T.
	\end{equation}
	Also, we may assume $R_T>0$. Otherwise there is nothing to prove. 
	
	Let $\varphi$ be the scattering solution of $v$ and $\varphi_T$ the scattering solution of $v_T$ both normalized at $\widetilde{R}>R$. We have from \eqref{eq:IntegralScattLength}, using that $\varphi_T$ is a non-decreasing function,
	\[\frac{4\pi}{\log(\frac{\widetilde{R}}{a_T})}=\int v_T\varphi_T\dd x\geq \varphi_T(R_T)\int v_T=4\pi  \varphi(R_T)T,\]
	and hence
	\begin{equation}\label{eq. bound on phi(Rt)}
		\varphi(R_T)\leq \frac{1}{\log(\frac{\widetilde{R}}{a_T})T}.
	\end{equation}
	Next we define
	\[u=\one_{\{\vert x\vert >R_T\}}(\varphi_T-\omega_T\varphi_T(R_T))\qquad \text{where}\quad \omega_T(x)=1-\frac{\log(\frac{\vert x\vert}{R_T})}{\log(\frac{\widetilde{R}}{R_T})}.\]
	Observe that $u(\widetilde{R})=1$ and we may therefore apply it as a trial function in the functional for $a$ to get 
	\begin{equation}\label{eq. integration of u}
		\frac{2\pi}{\log(\frac{\widetilde{R}}{a})}\leq \int_{\{\vert x\vert <\widetilde{R}\}}\Big(\vert\nabla u\vert^2+\frac{1}{2}vu^2\Big)\dd x:=E_1+E_2+E_3,
	\end{equation}
	with
	\begin{equation*}
		\begin{aligned}
			&E_1=\int_{\{R_T <\vert x\vert <\widetilde{R}\} } \Big(\vert \nabla \varphi_T\vert^2+ \frac{1}{2}v_T\varphi_T^2 \Big) \dd x=\frac{2\pi}{\log(\frac{\widetilde{R}}{a_T})},\\
			&E_2=-2\varphi_T(R_T)\int_{\{R_T <\vert x\vert <\widetilde{R}\}} \Big(\nabla \varphi_T\nabla \omega_T+ \frac{1}{2}v_T\varphi_T\omega_T \Big) \dd x=0,\\
			&E_3=\varphi_T(R_T)^2\int_{\{R_T <\vert x\vert <\widetilde{R}\}} \Big(\vert \nabla\omega_T \vert^2+ \frac{1}{2}v_T\omega_T^2 \Big)\dd x.
		\end{aligned}
	\end{equation*}
	For $E_2$ we integrated by parts and used that $\varphi_T$ is harmonic inside $B(0,R_T)$, thus constant, which makes the boundary term vanish. For $E_3$ we use that $\omega_T\leq1 $ on the given interval, so
	combining \eqref{eq. integration of vt}, \eqref{eq. bound on phi(Rt)} and  \eqref{eq. integration of u}  yields
	\begin{equation}\label{eq. E_3}
		\frac{2\pi}{\log(\frac{\widetilde{R}}{a})}-\frac{2\pi}{\log(\frac{\widetilde{R}}{a_T})}\leq E_3\leq\frac{2\pi}{\log(\frac{\widetilde{R}}{a_T})^2}\bigg( \frac{1}{\log(\frac{\widetilde{R}}{R_T})T^2}+\frac{1}{T}\bigg).
	\end{equation}
Using that $a\geq a_T$ we may replace $a_T$ with $a$ on the right hand side. Secondly, we observe that the function  
	\[\bigg(\frac{1}{\log(\frac{\widetilde{R}}{a})}-\frac{1}{\log(\frac{\widetilde{R}}{a_T})}\bigg)\log\Big(\frac{\widetilde{R}}{a}\Big)^2,\]
	is increasing in $\widetilde{R}$ so we may replace $\widetilde{R}$ with $R$ in the above expression and use \eqref{eq. E_3} to get
		\begin{equation*}
		\frac{2\pi}{\log(\frac{R}{a})}-\frac{2\pi}{\log(\frac{R}{a_T})}\leq \frac{2\pi}{\log(\frac{R}{a})^2}\bigg( \frac{1}{\log(\frac{\widetilde{R}}{R_T})T^2}+\frac{1}{T}\bigg).
	\end{equation*}
Now the result follows by letting $\widetilde{R}$ go to infinity. 
\end{proof}
We are ready to prove the main theorem of this section which gives us the ability to deal with a wide range of potentials including, most notably, the hard core.
\begin{theorem}\label{theorem. main potential reduction}
	For a radial, positive potential $v:\mathbb{R}^2\rightarrow [0,\infty]$ with finite scattering length $a$ there exists, for any $R>a$ and $T,\epsilon>0$, a potential $v_{T,R,\epsilon}$ such that
	\begin{equation}\label{eq. conditions on new potential}
		\supp(v_{T,R,\epsilon})\subset B(0,R),\qquad 0\leq v_{T,R,\epsilon}(x)\leq v(x),\qquad 
		\int v_{T,R,\epsilon}\leq4\pi T,
	\end{equation}
	and its scattering length $a_{T,R,\epsilon}$ satisfies
	\begin{equation}\label{eq. scattering condtions}
		\frac{2\pi}{\log(\frac{R}{a})}-\frac{2\pi}{\log(\frac{R}{a_{T,R,\epsilon}})}\leq \frac{1}{\log(\frac{R}{a})^2}\bigg(\frac{2\pi}{T}(1+\epsilon)+
		\frac{1}{2} \int_{\{\vert x\vert > R\}}v(x)\log\Big(\frac{\vert x\vert}{a}\Big)^2\bigg) \dd x.
	\end{equation}
\end{theorem}
\begin{proof}
	 Lemma~\ref{lemma. large integral potential} applied to $v_R^n=\one_{B(0,R)}\min(n,v)$ yields a $v_{R,T}^n$ satisfying all three conditions of \eqref{eq. conditions on new potential} and 
	\begin{equation}
		\frac{2\pi}{\log(\frac{R}{a_R^n})}-\frac{2\pi}{\log(\frac{R}{a_{T,R}^n})}\leq \frac{2\pi}{\log(\frac{R}{a})^2T},
	\end{equation}
	for all $n\in \mathbb{N}$. In the above we used that $a_{T,R}^n\leq a_R^n\leq a_R\leq a$ (where $a_{T,R}^n, a_R^n, a_R$ are the scattering lengths of $v_{T,R}^n, v_{R}^n, v_R$, respectively). Choosing $n_0$ large enough such that $a_R^{n_0}$ is close enough to $a_R$ gives an $a_{T,R,\epsilon}:=a_{T,R}^{n_0}$ satisfying
	\begin{equation}
		\frac{2\pi}{\log(\frac{R}{a_R})}-\frac{2\pi}{\log(\frac{R}{a_{T,R,\epsilon}})}\leq \frac{2\pi}{\log(\frac{R}{a})^2T}(1+\epsilon).
	\end{equation}
	We conclude using \eqref{eq. scattering length of non-compact potential} which gives the integral term of \eqref{eq. scattering condtions}.
\end{proof}

\subsection{Fourier analysis on the scattering equation}
Due to Theorem~\ref{theorem. main potential reduction} we may assume our potentials to be compactly supported and $L^1$, thus making the Fourier transform well defined. The scattering solution $\varphi$ will be the one defined in \eqref{eq:scat_defs} which is normalized to $1$ outside the support of $v$. In order to discuss the Fourier transform of the scattering solution, we recall some standard results surrounding the Fourier transform of the logarithm. We denote by $\mathcal{S}$ and $\mathcal{S}'$ the Schwartz space and the space of tempered distribution on $\mathbb{R}^2$, respectively.
\begin{lemma}\label{lem:FouLog}
	For $D>0$,
	let $L_D$ denote  the tempered distribution given by the function $\log (|x|/D)$ in ${\mathbb R}^2$. The Fourier transform of $L_D$ satisfies
	\begin{align}\label{eq:InclGamma}
		\langle \widehat{L_D} , \varphi \rangle_{{\mathcal S}',  {\mathcal S}} =-(2\pi)  \int_{{\mathbb R}^2} \frac{\varphi(p)-\varphi(0) \one_{\{|p| \leq 2 e^{-\Gamma}D^{-1}\}}}{p^2}\,\dd p,
	\end{align}
	where $\Gamma $ denotes the Euler-Mascheroni constant,
	\begin{align}
		\Gamma:= - \int_0^{\infty} e^{-x} \log x\, \dd x \approx 0.5772.
	\end{align}
\end{lemma} 
It follows from \eqref{eq:InclGamma} that, for any $f \in \mathcal S$,
\begin{align}\label{eq:convolutionLog}
	\iint_{{\mathbb R}^2 \times {\mathbb R}^2} \overline{f(x)} f(y) \log\Big(\frac{|x-y|}{D}\Big)\, \dd x \dd y = -2\pi
	\int_{{\mathbb R}^2} \frac{|\widehat{f}(p)|^2-|\widehat{f}(0)|^2 \one_{\{|p| \leq 2 e^{-\Gamma}D^{-1}\}}}{p^2}\, \dd p.
\end{align}

Using the notation from \eqref{eq:scat_defs} and \eqref{eq:ScatOmega}, we may compute the Fourier transform of $\omega$. 
In the $3D$ case one gets that $\widehat{\omega}(p) = \frac{\widehat{g}(p)}{2p^2}$, but in $2D$ this formula has to be corrected by a distribution supported at the origin according to Lemma~\ref{lem:FouLog}, see Lemma~\ref{lem:FT-w} below.

\begin{lemma}\label{lem:FT-w}
	Let $\widehat{\omega}$ denote the Fourier transform of $\omega$. Then $\widehat{\omega}$
	is the tempered distribution given by
	\begin{equation}\label{eq:Fou-w}
		\langle \widehat{\omega},u\rangle_{{\mathcal S}',  {\mathcal S}} = \int \frac{\widehat{g}(p) u(p) - \widehat{g}(0) u(0) \one_{\{|p|\leq \ldelta^{-1}\}} }{2p^2} \dd p,\end{equation}
	where, recalling the definition of $\widetilde{R}$ in \eqref{eq:delta},
	\begin{equation}\label{eq:Defldelta}
		\ldelta = \frac{a e^{\Gamma}}{2} e^{\frac{1}{2\delta}} =\frac{1}{2}e^\Gamma\widetilde{R}.
	\end{equation}
\end{lemma}

\begin{proof}
	By taking the Fourier transform (in the sense of tempered distributions) on \eqref{eq:ScatOmega} we get $\widehat{\omega}(p) = \frac{\widehat{g}(p)}{2p^2}$ on ${\mathbb R}^2 \setminus \{0\}$. To determine the singular term at $0$, we write, with $\widetilde{\omega} \in L^1\cap L^2$,
	\begin{equation}
		\omega = 1 - 2\delta \log(r/a) + \widetilde{\omega} = - 2\delta \log\Big(\frac{r}{a e^{\frac{1}{2\delta}}}\Big) + \widetilde{\omega}.
	\end{equation}
	So, using the Fourier transform of the logarithm as recalled in \eqref{eq:InclGamma}
	\begin{equation}
		\langle \widehat{\omega},u\rangle_{{\mathcal S}',  {\mathcal S}} 
		= 4\pi \delta \int_{{\mathbb R}^2} \frac{u(p)-u(0) \one_{\{|p| \leq 2 a^{-1} e^{-\Gamma-\frac{1}{2\delta}}\}}}{p^2}\, \dd p + \int \widehat{\widetilde{\omega}}(p) u(p)\, \dd p
	\end{equation}
	where $\widehat{\widetilde{\omega}} \in L^2\cap L^{\infty} ({\mathbb R}^2)$.
	We conclude upon recalling from \eqref{eq:IntegralScattLength} that $8\pi \delta = \widehat{g}_0$.
\end{proof}

Thanks to the previous lemma we are able to prove some important properties of $\widehat{g\omega}(0)$ which are going to be key through all the paper.

\begin{lemma}\label{eq:gdecay}
	The following identity holds
	\begin{align}\label{eq:gomega_equivalence}
		\widehat{g\omega}(0) = \int_{\mathbb{R}^2} \frac{\widehat{g}_k^2 - \widehat{g}_0^2 \one_{\{|k|\leq \ell^{-1}_{\delta}\}}}{2 k^2}\, \dd k
	\end{align}
	and, furthermore, the following bounds hold
	\begin{align}
		\;|\widehat{g\omega}(0)|&\leq C \delta,\label{eq. bound on gw_0}\\
		\bigg\vert\int_{\{|k|\leq \ell_{\delta}^{-1}\}} \frac{\widehat{g}_k^2 - \widehat{g}_0^2}{2 k^2}\, \dd k\bigg\vert &\leq C R^2 \delta^2 \ell^{-2}_{\delta},\label{eq:332}\\
		\bigg\vert\int_{\{|k|\geq \ell_{\delta}^{-1}\}} \frac{\widehat{g}_k^2}{2 k^2}\,\dd k\bigg\vert &\leq C \delta.\label{eq. bound on intg2}
	\end{align}
\end{lemma}

\begin{proof}
	Formula \eqref{eq:gomega_equivalence} is given by a direct application of Lemma~\ref{lem:FT-w} choosing $u = \widehat{g}$. The first bound \eqref{eq. bound on gw_0} follows because in the support of $g$, $\omega \leq 1$ and $\widehat{g}_0 = 8\pi\delta$. The last bound \eqref{eq. bound on intg2} follows once we have proved the second one. In order to do that, we consider a Taylor expansion to the second order, due to the symmetry of $g$, which gives 
	\begin{equation}
		\bigg\vert\int_{\{|k|\leq \ell_{\delta}^{-1}\}} \frac{\widehat{g}_k^2 - \widehat{g}_0^2}{2 k^2}\,\dd k\bigg\vert \leq R^2 |\widehat{g}_0|^2 \int_{\{|k|\leq \ell_{\delta}^{-1}\}} \dd k = CR^2 \delta^2 \ell^{-2}_{\delta}.
	\end{equation}
\end{proof}

\subsection{Spherical measure potentials}

For the upper bound, we will change the potential in order to ensure small $L^1$ norm. For a potential $v$ supported in $B(0,R)$ and $b>R$, let \[f(x):=\min\Big(1,\varphi^{(0)}(x)\log\Big(\frac{b}{a}\Big)^{-1} \Big).\]
Thus, $f$ is the scattering solution in $B(0,b)$ normalized at $b$ and extended by one. The new potential $\widetilde{v}$ will then be described by the deviation of $f$ being the actual scattering solution, \textit{i.e.},
\begin{equation}
	\widetilde{v}=2\Big(-\Delta f+\frac{1}{2}vf\Big),
\end{equation}
where the above equality is to be thought of in a distributional sense. The factor $2$ is important and should be thought of as the number of particles involved in the scattering process. A quick calculation shows that 
\begin{equation}\label{eq:deltaPot}
\widetilde{v}=2f'(b)\delta_{\{\vert x\vert=b\}}=2\frac{1}{b\log(\frac{b}{a})}\delta_{\{\vert x\vert=b\}},
\end{equation}
where $f'(b)$ is to be understood as the outgoing radial derivative of $f$ at length $b$. We show in Section~\ref{sec.upper.general} how we reduce to this potential. The simple, but essential properties of $\widetilde{v}$ are stated in the lemma below.

\begin{lemma} \label{lemma.soft.pot.properties} Let $v$ and $\widetilde{v}$ be given as above. We use the notation $\widetilde{a}=a(\widetilde{v})$ and $a=a(v)$. Furthermore, let $\widetilde{\varphi}$ be the scattering solution of $\widetilde{v}$ normalized at $\widetilde{R}>b$. Then 
	\begin{enumerate}
		
		\item \label{item. scattering lengths agree} The scattering lengths agree, \textit{i.e.}, $\widetilde{a}=a$.
		
		\item \label{item. scattering bessel function}  $\widehat{ \widetilde v}(p)=2f'(b) b J_0(b\vert p\vert)$, where $J_0$ is the zeroth spherical Bessel function. In particular there exists a universal constant $C>0$ such that 
		\begin{equation}\label{eq:vtildehat.bound}
			\vert \widehat{ \widetilde{v} }_p\vert \leq C \frac{\widehat{\widetilde v}_0}{\sqrt{b\vert p\vert}}.
		\end{equation}
		
		\item \label{item. size of v and g}  $\displaystyle  \widehat{ \widetilde v}_0:=\langle v,1\rangle = \frac{4 \pi}{\log(b / a)}$, and $\displaystyle \widehat{\widetilde{g}}_0=(\widehat{ \widetilde v \widetilde \varphi})_0 = \frac{4\pi}{\log(\widetilde{R}/a)}.$		
	\end{enumerate}
\end{lemma}

\begin{proof}
The potential $\widetilde{v}$ is a spherical measure on the sphere $\{\vert x\vert =b\}$ and thus $\widetilde{\varphi}$ is harmonic both inside and outside this sphere. We may therefore conclude from the continuity of $\widetilde{\varphi}$ that
	\begin{equation}\label{eq. behavior of phi tilde}
		\log (\widetilde{R} / \widetilde a) \widetilde{\varphi}\left (r\right ) =
		\begin{cases}
			\log (b/\widetilde a),  & \mbox{if } r\leq b, \\
			\log (r/ \widetilde a),  & \mbox{if } r > b.
		\end{cases}
	\end{equation}
	From the scattering equation
	\[-\Delta \widetilde{\varphi}+\frac{1}{2}\widetilde{v}\widetilde{\varphi}=0,\]
	we find
	\begin{equation}\label{eq. scattering equation of phi tilde}
		\widetilde{\varphi}'(b)=f'(b)\widetilde{\varphi}(b)
	\end{equation}
	where $\widetilde{\varphi}'(b)$ means the ingoing radial derivative of $\widetilde{\varphi}$ at length $b$. Combining \eqref{eq. behavior of phi tilde} and \eqref{eq. scattering equation of phi tilde} yields~\ref{item. scattering lengths agree}.
	Property~\ref{item. scattering bessel function}. is a direct consequence of $\widetilde{v}$ being a uniform measure on the sphere $\{\vert x\vert\} =b$ and the behaviour of $J_0$ at infinity.
	Finally, the identities in~\ref{item. size of v and g}. follow immediately after realizing that
	\[g=\widetilde{\varphi}(b)v.\]
\end{proof}

\section{Upper bound for a soft potential} \label{sec.soft}
We denote by $\mathscr{M}_c$ the set of compactly supported, finite, regular, positive, radial Borel measures on $\R^2$. If $v \in \mathscr{M}_c$, it admits a Fourier transform $\widehat v$. The aim of this section is to prove an upper bound on the ground state energy of $\mathcal H_v$ on the box $\Lambda_\beta = [ - \frac{L_\beta}{2} , \frac{L_\beta}{2} ]^2$ for potentials $v \in \mathscr{M}_c$, under some additional decay assumption on the Fourier transform of $v$. We recall that $L_\beta = \rho^{- \frac 1 2} Y^{-\beta}$.

In this section, we will denote by $\varphi$ the scattering solution of the given $v$, normalized at length $\widetilde{R}$, and $g=\varphi v$, see \eqref{eq:scat_defs}.
Notice here that the theory of Section~\ref{section:scattering} clearly extends to potentials $v \in \mathscr{M}_c$. In particular, the scattering equation \eqref{scattering_equation} is valid.
We recall that $0 \leq \widehat{g}_0=8\pi \delta \leq C Y$ by \eqref{eq:IntegralScattLength} and \eqref{eq:delta}. We prove the following upper bound, which is very similar in spirit to the upper bound of \cite{ESY} in the $3D$ case.

\begin{theorem}\label{thm.upperbound.soft.pot}
	For any given $c_0>0$ and $\beta\geq \frac{3}{2}$, there exists $C_\beta>0$ (only depending on $c_0$ and $\beta$) such that the following holds. 
Let $\rho >0$ and $v \in \mathscr{M}_c$ be a radial positive measure with scattering length $a$ and $\supp v \subset B(0,R)$, for some $R>0$. 
Let ${\mathcal H}_v$ be as defined in \eqref{eq:upperbound.grandcanonical}.
Assume that
	\begin{equation}\label{eq:vhatp}
		\vert \widehat{g}_p \vert \leq c_0 \frac{\widehat{g}_0}{\sqrt{R \vert p \vert}}, \qquad \forall \vert p \vert \geq a^{-1}.
	\end{equation}
	Then, if $\rho R^2 \leq  Y$ and $\rho a^2 \leq C_{\beta}^{-1}$, one can find a normalized trial state $\Phi \in \mathscr{F}_s(L^2(\Lambda_\beta))$ satisfying
	\begin{align*}
		\langle  \mathcal H_v  \rangle_\Phi &\leq   4\pi L_\beta^2 \rho^2 \delta_0 \Big(1 + \Big(2\Gamma + \frac{1}{2} + \log(\pi) \Big) \delta_0 \Big)
		\\&\quad + CL_{\beta}^2 \rho^2 \delta_0(\widehat{v}_0-\widehat{g}_0)+CL_{\beta}^2 \rho^2 \delta^{2}_0\widehat{v}_0.
	\end{align*}
Here $\Phi$ is a quasi-free state such that $\langle \mathcal N \rangle_\Phi = N ,$ and $\langle \mathcal N^2 \rangle_\Phi \leq 9 N^2$, with $N = \rho L_\beta^2 = Y^{-2\beta}$.
\end{theorem}

\begin{remark}
	Note that this result is much weaker than Theorem~\ref{thm.upperbound.grandcanonical}. Indeed, the remainders are only of order $\rho^2 L_\beta^2 \delta_0^2$ and $\rho^2 L_\beta^2 \delta_0$ and thus much larger than the 2D-LHY term, unless $\widehat{v}_0 = \widehat{g}_0 + o(\delta_0)$. Moreover, Theorem~\ref{thm.upperbound.soft.pot} only holds for potentials with finite integral and, in particular, it does not allow for a hard core. However, in the proof of Theorem~\ref{thm.upperbound.grandcanonical} in Section~\ref{sec.upper.general} we will show how to reduce to such potentials. More precisely, we will apply Theorem~\ref{thm.upperbound.soft.pot} to a surface potential of the form \eqref{eq:deltaPot}.
\end{remark}

The rest of Section~\ref{sec.soft} is dedicated to the proof of Theorem~\ref{thm.upperbound.soft.pot}. We will give an explicit trial state and state several technical calculations as lemmas. In the end we collect the pieces and finish the proof.

\subsection{A quasi-free state} \label{subsec:quasifree}

We will define our trial state $\Phi$ in second quantization formalism. On the bosonic Fock space $\mathscr{F}( L^2(\Lambda_\beta))$, we will denote by $a_p^\dagger$ and $a_p$ the creation and annihilation operators associated to the function $x \mapsto \vert \Lambda_\beta \vert^{-\frac 1 2} \exp (ipx)$, for $p \in \Lambda^*_\beta = (\frac{2\pi}{L_\beta} \Z )^2$.
Our quasi-free state is $\Phi = T_{\nu} W_{N_0} \Omega$ where $\Omega$ is the vacuum, $W_{N_0}$ creates the condensate and $T_\nu$ the excitations:
\begin{equation}\label{def:WT}
	W_{N_0} = \exp \Big( \sqrt{N_0} (a_0^\dagger - a_0) \Big), \quad T_{\nu} = \exp \Big(\frac{1}{2} \sum_{p\neq 0} \nu_p (a_p^\dagger a_{-p}^\dagger - a_p a_{-p}) \Big),
\end{equation}
for a given $N_0\leq N$ associated with $\rho_0 :=  N_0/L_{\beta}^2$.
These operators have the nice properties that 
\begin{equation}\label{eq:conj}
	W_{N_0}^* a_0 W_{N_0} = a_0 + \sqrt{ N_0 }, \quad \text{ and } \quad T_\nu^* a_p T_\nu = \cosh (\nu_p) a_p + \sinh (\nu _p) a_{-p}^\dagger.
\end{equation}
In particular, for any  $p$, $q \in \Lambda^*_\beta$,
\begin{equation} \label{eq.Phi.apaq}
	\langle a_q^\dagger a_p \rangle_\Phi =
	\begin{cases}
		N_0, & \text{ if }  p=q=0, \\
		0, & \text{ if }  p\neq q, \\
		\gamma_q, & \text{ if } p=q \neq 0,
	\end{cases}
	\quad \text{and} \quad
	\langle a_q a_p \rangle_\Phi=\langle a_q^\dagger a_p^\dagger \rangle_\Phi =
	\begin{cases}
		N_0, & \text{ if }  p=q=0, \\
		0, & \text{ if }  p\neq -q, \\
		\alpha_q, & \text{ if } p=-q \neq 0,
	\end{cases}
\end{equation}
where $\alpha_p = \cosh (\nu_p) \sinh (\nu_p)$ and $\gamma_p = \sinh (\nu_p)^2$. We choose the coefficient $\nu_p$ such that
\begin{equation}\label{def:alphagamma}
	\alpha_p = \frac{-\rho_0 \widehat{g}_p}{2 \sqrt{p^4 + 2 \rho_0 \widehat{g}_p p^2}} , \qquad \gamma_p = \frac{p^2 + \rho_0 \widehat{g}_p - \sqrt{p^4 + 2 \rho_0 \widehat{g}_p p^2}}{2 \sqrt{p^4 + 2 \rho_0 \widehat{g}_p p^2}}\geq 0 ,
\end{equation}
this specific choice coming from an optimization of the energy, as we will see later. Note that $\alpha_p^2 = \gamma_p (\gamma_p + 1)$, making it a possible choice. These coefficients satisfy the following estimates.

\begin{lemma}\label{lem:ag}We estimate  the sum (over $\Lambda_{\beta}^{*}$) of $\alpha_{p}$ and $\gamma_{p}$:
	\begin{equation}
		\sum_{p\neq 0}\left \vert\alpha_{p}\right \vert \leq C N, \quad \text{and}\quad \sum_{p\neq 0}\gamma_{p}\leq  C N \delta.
	\end{equation}
\end{lemma}
\begin{proof}
	We start from the expression of $\alpha_{p}$ \eqref{def:alphagamma} and split the sum between $\vert p \vert \leq \sqrt{\rho_0 \widehat{g}_0}$ and  $\vert p \vert \geq \sqrt{\rho_0 \widehat {g}_0 }$:
	\begin{align*}
		\sum_{p\neq 0}\left \vert\alpha_{p}\right \vert &\leq C\sqrt{\rho_0 \widehat{g}_0}\sum_{0<\vert p\vert\leq \sqrt{\rho_0 \widehat{g}_0}}\frac{1}{\vert p\vert}+ C\rho_0\sum_{\vert p\vert\geq \sqrt{\rho_0 \widehat{g}_0}}\frac{ \vert \widehat{g}_{p}\vert}{\vert p\vert^{2}}\\
		&\leq CL_\beta^{2} \sqrt{\rho_0 \widehat{g}_0}\int_{0}^{\sqrt{\rho_0 \widehat{g}_0}}\dd u + C L_\beta^{2} \rho_0 \widehat{g}_0\int_{\sqrt{\rho_0 \widehat{g}_0}}^{a^{-1}}\frac{\dd u}{u}+ C L_\beta^{2}\rho_0 \int_{ a^{-1}}^{+\infty}\frac{ \widehat{g}_0}{R^{1/2}u^{3/2}} \dd u\\
		&\leq CL_\beta^{2}\rho_0 \widehat{g}_0 (1 + \vert \log(a^2 \rho_0 \widehat{g}_0) \vert) \leq C N,
	\end{align*}
	where we used the decay of $\widehat{g}_{p}$ at infinity \eqref{eq:vhatp} and the bound $a \leq R$. \newline
	For $\gamma_{p}$ we also split the sum this way. For $p\leq \sqrt{\rho_0 \widehat{g}_0}$ we obtain that 
	\begin{equation*}
		\sum_{\vert p\vert\leq \sqrt{\rho_0 \widehat{g}_{0}}}\vert \gamma_{p}\vert \leq C\sum_{\vert p\vert\leq \sqrt{\rho_0 \widehat{g}_{0}}}\frac{\sqrt{\rho_0 \widehat{g}_{0}}}{\vert p\vert } \leq  CL_\beta^{2} \rho_0 \widehat{g}_{0} \leq C N_0 \delta.
	\end{equation*}
	For $p\geq \sqrt{\rho_0 \widehat{g}_0}$ we expand the square root and find
	\begin{equation*}
		\sum_{\vert p\vert\geq \sqrt{\rho_0 \widehat{g}_{0}}}\vert \gamma_{p}\vert \leq C\sum_{\vert p\vert\geq \sqrt{\rho_0 \widehat{g}_{0}}}\frac{\left (\rho_0 \widehat{g}_{0}\right )^{2}}{\vert p\vert^{4} } \leq  CL^{2}_\beta \rho_0 \widehat{g}_{0} \leq C N_0 \delta,
	\end{equation*}
	which concludes the proof.
\end{proof}

Finally  choose $N_0$ such that 
\begin{equation} \rho L_{\beta}^2 = N=N_0 + \sum_{p \neq 0} \gamma_p .
\end{equation}
Note that with this choice $\Phi$ has the expected average number of particles as stated in the next lemma.
\begin{lemma}\label{lem:numberPhi}
	The state $\Phi=  T_{\nu} W_{N_0} \Omega$ satisfies
	\[ \langle \mathcal{N} \rangle_ \Phi = N , \qquad \langle {\mathcal N}^2 \rangle_\Phi \leq C N^2 , \]
	where $\mathcal N = \sum_{p \in \Lambda_\beta^*} a_p^\dagger a_p$ is the number operator.
\end{lemma}

\begin{proof}
	First we can use the property \eqref{eq.Phi.apaq} to find
	\begin{equation}
		\langle \mathcal N \rangle_\Phi = N_0 + \sum_{p\neq 0} \gamma_p = N.
		\label{eq:N}
	\end{equation}  
	For $\mathcal N^2$ we split the sums according to zero and non-zero momenta, and then conjugate by $W_{N_{0}}$,
	\begin{align*}
		\langle \mathcal N^2 \rangle_\Phi &= \sum_{q,p} \langle a_p^\dagger a_p a_q^\dagger a_q \rangle_\Phi =N_{0}^{2}+N_{0}+\sum_{q \neq 0}\langle a_0^\dagger a_0 a_q^\dagger a_q+\text{h.c} \rangle_\Phi+ \sum_{q\neq 0,p\neq 0} \langle a_p^\dagger a_p a_q^\dagger a_q \rangle_\Phi\\
		&=N_{0}^{2}+N_{0}\Big( 1+2\sum_{p\neq 0} \gamma_p \Big)+\sum_{q\neq 0, p\neq 0} \langle a_p^\dagger a_p a_q^\dagger a_q \rangle_\Phi . 
	\end{align*}
	Now we use Lemma~\ref{lem:ag} and apply Wick's Theorem \cite[Theorem 10.2]{Sol_07} to the state $T_\nu \Omega$ to find
	\begin{align*}
		\langle \mathcal N^2 \rangle_\Phi &\leq 4N^{2}+\sum_{q\neq 0,p\neq 0} \Big( \langle a_p^\dagger a_p \rangle_\Phi \langle a_q^\dagger a_q \rangle_\Phi + \langle a_p^\dagger a_q^\dagger \rangle_\Phi \langle a_p a_q \rangle_\Phi + \langle a_p^\dagger a_q \rangle_\Phi \langle a_p a_q^\dagger \rangle_\Phi \Big)\\
		&\leq  4N^{2}+ \Big( \sum_{p\neq 0} \gamma_p\Big)^{2} + \sum_{p\neq 0} \alpha_p^{2}+ \sum_{p\neq 0} (\gamma_{p}^{2}+\gamma_{p})\leq C N^{2} \,
	\end{align*}
	using Lemma~\ref{lem:ag}.
\end{proof}

\subsection{Energy of \texorpdfstring{$\Phi$}{Phi}} \label{subsec:energyPhi}

In order to get an upper bound on the energy of $\Phi$ we first introduce some notations. Our convention for the convolution will be
\begin{equation}
	\V * A (p) = \int \V(p-q) A(q) \dd q.
\end{equation}
%If $B$ is another function in $L^{\infty}(\mathbb{R}^2)^{*}$ or $B=\widehat \omega $,
We also define the quantity $D(A,B)$ by
\begin{equation}
	D(A,B) = \frac{1}{(4\pi^2)^2} \int \V * A (p) B(p) \dd p = \frac{1}{(4\pi^2)^2} \langle B , \V * A \rangle,
\end{equation}
and observe that it is symmetric in the entries. Then we prove the following result.

\begin{lemma}\label{lem:up1} Under the assumptions of Theorem~\ref{thm.upperbound.soft.pot}, there exists a constant $C>0$, independent of $v$ and $\rho$, such that
	\begin{equation*} 
		\vert\Lambda_\beta \vert ^{-1} \langle \mathcal{H}_v \rangle_\Phi \leq \frac{\rho^2}{2 } \V_0 + \int \Big((p^2 + \rho_0 \V_p) \gamma_p + \rho_0 \V_p \alpha_p \Big)\frac{\dd p}{4\pi^2} + \frac{1}{2} D(\alpha, \alpha) + C \V_0 \rho^2 Y^3.
	\end{equation*}
\end{lemma}

\begin{proof}
	One can write $\mathcal H_v$ in second quantization in momentum variable,
	\[  \mathcal H_v = \sum_{p \in \Lambda_\beta^*} p^2 a_p^\dagger a_p + \frac{1}{2\vert \Lambda_\beta \vert} \sum_{p,q,r} \V_r a^\dagger_{p+r} a^\dagger_q a_{q+r} a_p , \]
	and express the energy of $\Phi$ in terms of $\alpha_{p}$ and $\gamma_{p}$ as follows. 
	\iffalse
	\begin{align*}
		\mathcal H_v = \sum_{p\neq 0} p^2 a_p^\dagger a_p + \frac{1}{2 \vert \Lambda_\beta \vert} \Bigg( &\widehat v_0 a_0^\dagger a_0^\dagger a_0 a_0 + 2 \sum_{p\neq 0} ( \widehat v_0 + \widehat v_p)  a_0^\dagger a_0 a_p^\dagger a_p + \sum_{r\neq 0} \widehat v_r a_0^\dagger a_0^\dagger a_r a_{-r} + h. c.  \\ 
		& + \sum_{\substack{q,r\neq 0 \\ q+r \neq 0}} \widehat v_r a_r^\dagger a_q^\dagger a_{q+r} a_0 + h.c. + \sum_{\substack{p,r\neq 0 \\ p+r \neq 0}} \widehat v_{r} a_{p+r}^\dagger a_{-r}^\dagger a_{p} a_0 + h.c.  \nonumber \\
		&+ \sum_{\substack{ p,q \neq 0 \\  p+r, q+r  \neq 0 }} \widehat v_r a_{p+r}^\dagger a_q^\dagger a_{q+r} a_p \Bigg).
	\end{align*}
\fi
	We conjugate by $W_{N_0}$ using \eqref{eq:conj}, which amounts to change the $a_{0}$'s in $\sqrt{N_{0}}$. Since $\Phi = T_{\nu} W_{N_0} \Omega$ with no $a_{0}$ in $T_{\nu}$ (see \eqref{def:WT}), when we apply $\Phi$ we find
	\begin{align}
		\langle \mathcal H_v \rangle_\Phi &=  \sum_{p \neq 0} p^2 \langle a_p^\dagger a_p \rangle_\Phi + \frac{N_0^2}{2\vert \Lambda_\beta \vert} \V_0 + \frac{N_0}{ \vert \Lambda_\beta \vert}  \sum_{p \neq 0} ( \widehat v_0 + \widehat v_p) \langle a_p^\dagger a_p \rangle_{\Phi} + \frac{N_0}{\vert \Lambda_\beta \vert} \sum_{r\neq 0} \widehat v_r \langle a_r a_{-r} \rangle_\Phi  \nonumber \\ &+ \frac{\sqrt{N_0}}{\vert \Lambda_\beta \vert}  \sum_{\substack{q,r\neq 0 \\ q+r \neq 0}} \widehat v_r  \langle a_r^\dagger a_q^\dagger a_{q+r} \rangle_{\Phi}   + \frac{\sqrt{N_0}}{\vert \Lambda_\beta \vert} \sum_{\substack{p,r\neq 0 \\ p+r \neq 0}} \widehat v_{r} \langle a_{p+r}^\dagger a_{-r}^\dagger a_{p} \rangle_\Phi \nonumber \\ 
&+ \frac{1}{2 \vert \Lambda_\beta \vert} \sum_{\substack{ p,q \neq 0 \\  p+r, q+r  \neq 0 }} \widehat v_r \langle a_{p+r}^\dagger a_q^\dagger a_{q+r} a_p \rangle_\Phi. \label{eq:conjW0}
	\end{align}
	We can use Wick's Theorem \cite[Theorem 10.2]{Sol_07} to the state $T_\nu \Omega$. \iffalse
	\begin{align}\langle a^\dagger_{p+r} a^\dagger_q a_{q+r} a_p\rangle_\Phi=  &\langle a_{p+r}^\dagger a_q^\dagger \rangle_\Phi \langle a_{q+r} a_p \rangle_\Phi  + \langle a_{p+r}^\dagger a_{q+r} \rangle_\Phi \langle a_q^\dagger a_p \rangle_\Phi 
	\nonumber \\ &
	+ \langle a_{p+r}^\dagger a_p \rangle _\Phi \langle a_q^\dagger a_{q+r} \rangle_\Phi  .
		\label{eq:Wick}
	\end{align}
\fi By definition of $\alpha_p$ and $\gamma_p$ in \eqref{eq.Phi.apaq} \iffalse
	\begin{align}
		\langle \mathcal H _v \rangle_\Phi &= \frac{N_0^2}{2\vert \Lambda_\beta \vert} \V_0 + \sum_{p \in \Lambda_\beta^*} p^2 \gamma_p + \frac{N_0}{ \vert \Lambda_\beta \vert}\sum_{p \neq 0} \Big((\V_0 + \V_p) \gamma_p + \V_p \alpha_p\Big) + \frac{1}{2 \vert \Lambda_\beta \vert }\sum_{\substack{q\neq 0\\ q+r \neq 0}} \V_r \alpha_q \alpha_{q+r} \nonumber \\
		&\quad + \frac{1}{2 \vert \Lambda_\beta \vert }\sum_{\substack{q\neq 0\\ q+r \neq 0}} \V_r \gamma_q \gamma_{q+r} + \frac{1}{2 \vert \Lambda_\beta \vert }\sum_{p,q \neq 0} \V_0 \gamma_q \gamma_p. \label{eq:HPhi01}
	\end{align} \fi
	together with $N^2 = (N_0 + \sum_{p\neq 0} \gamma_p)^2$ we deduce
	\begin{align}
		\langle \mathcal H_v \rangle_\Phi &= \frac{N^2}{2 \vert \Lambda_\beta \vert} \V_0 + \sum_{p\neq 0} p^2 \gamma_p +  \frac{N_0}{\vert \Lambda_\beta \vert} \sum_{p\neq 0} \Big( \V_p \gamma_p + \V_p \alpha_p  \Big)\nonumber \\
		&\quad  + \frac{1}{2 \vert \Lambda_\beta \vert }\sum_{\substack{q\neq 0\\ q+r \neq 0}} \V_r \alpha_q \alpha_{q+r} + \frac{1}{2 \vert \Lambda_\beta \vert }\sum_{\substack{q\neq 0\\ q+r \neq 0}} \V_r \gamma_q \gamma_{q+r} .
	\end{align}
	We bound the last term in the above using Lemma~\ref{lem:ag}. With $\rho = N \vert \Lambda_\beta \vert^{-1}$ and $\rho_{0}=N_{0} \vert \Lambda_\beta \vert^{-1}$ we deduce
	\begin{align}\label{eq. sum just before the int}
		\vert \Lambda_\beta \vert^{-1} \langle \mathcal H_v \rangle_\Phi &\leq  \frac 1 2 \rho^2 \V_0 + \frac{1}{\vert \Lambda_\beta \vert} \sum_{p \neq 0} \Big((p^2 + \rho_0 \V_p) \gamma_p + \rho_0 \V_p \alpha_p \Big)\nonumber \\ &\quad+ \frac{1}{2\vert \Lambda_\beta \vert^2} \sum_{\substack{q\neq 0\\ q+r \neq 0}} \V_r \alpha_q \alpha_{q+r}  + C \V_0 \rho^2 Y^2 .
	\end{align}
	Up to errors $\mathcal{E} \leq C \V_0 \rho^2 Y^{\frac{1}{2}+\beta}$, we can approximate these Riemann sums by integrals (see Lemma~\ref{lem:RiemannSums2}) and the lemma follows. In fact, the requirement $\beta\geq 3/2$ in Theorem~\ref{thm.upperbound.soft.pot} comes from here.
\end{proof}

\begin{lemma}\label{lem:up2}
	Under the assumptions of Theorem~\ref{thm.upperbound.soft.pot}, there exists a constant $C>0,$ independent of $v$ and $\rho$, such that
	\begin{align*}
		\vert \Lambda_\beta \vert ^{-1} \langle \mathcal{H} \rangle _\Phi &\leq  \frac{\rho^2}{2} \widehat \gt_0 + \frac 1 2 \int \Big(  \sqrt{p^4 + 2 \rho_0 p^2 \widehat \gt_p} -p^2 - \rho_0 \widehat \gt_p+ \rho_0^2\frac{\widehat \gt^2_p - \widehat \gt_0^2 \one_{\{p \leq  2 e^{-\Gamma} \widetilde{R}^{-1}\}}}{2p^2} \Big)\frac{\dd p }{4 \pi^2} \\ &\quad+ \frac{1}{2} D(\alpha + \rho_0 \widehat \omega, \alpha + \rho_0 \widehat \omega)  + C \rho^2 Y  ( \V_0 - \widehat g_0 ) + C \V_0 \rho^2 Y^2 .
	\end{align*}
\end{lemma}

\begin{proof}
	We recall the definition \eqref{eq:scat_defs} of $\omega$, and we insert $\rho_0 \widehat{\omega}$ into $D(\alpha,\alpha)$,
	\[ D(\alpha,\alpha) = - \rho_0^2 D(\widehat \omega, \widehat \omega) - 2 \rho_0 D(\alpha, \widehat \omega) + D(\alpha + \rho_0 \widehat \omega, \alpha + \rho_0 \widehat \omega) . \]
	Inserting this into Lemma~\ref{lem:up1} we find
	\begin{align} \frac{\langle \mathcal{H}_v \rangle _\Phi}{\vert \Lambda_\beta \vert} &\leq \frac{\rho^2}{2} \V_0 -  \frac{\rho_0^2}{2} D(\widehat \omega, \widehat \omega) +  \int \Big((p^2 + \rho_0 \V_p) \gamma_p + \rho_0 (\V_p - (\V * \widehat \omega)_p) \alpha_p\Big) \frac{\dd p}{4\pi^2} \nonumber \\ &\quad+ \frac{1}{2} D(\alpha + \rho_0 \widehat \omega, \alpha + \rho_0 \widehat \omega) + C \V_0 \rho^2 Y^2 . \label{eq:HPhi04}
	\end{align}
	Now note that $\widehat \gt _p = \V_p - ( \V * \widehat \omega)_p$ and,
	\begin{align*}
		\frac{\rho^2}{2} \V_0 &= \frac{\rho^2}{2} \widehat \gt _0 + \frac{\rho_0^2}{2} (\widehat{v \omega})_0 + \frac{\rho^2 -\rho_0^2 }{2} (\widehat{v \omega})_0 \\
		&= \frac{\rho^2}{2} \widehat \gt_0 + \frac{\rho_0^2}{2} (\widehat{\gt \omega})_0 + \frac{\rho_0^2}{2} (\widehat{v \omega^2})_0 + \frac{\rho^2 -\rho_0^2 }{2} (\widehat{v \omega})_0.
	\end{align*}
	This equality inserted in \eqref{eq:HPhi04}, together with $D(\widehat \omega, \widehat \omega) = (\widehat{v  \omega^2})_0$ implies
	\begin{align} \frac{\langle \mathcal{H}_v \rangle _\Phi}{\vert \Lambda \vert} &= \frac{\rho^2}{2} \widehat \gt_0 +  \frac{\rho_0^2}{2} (\widehat{\gt \omega})_0 +  \int \Big((p^2 + \rho_0 \V_p) \gamma_p + \rho_0 \widehat \gt_p \alpha_p \Big)\frac{\dd p}{4\pi^2} \nonumber \\ &\quad+ \frac{1}{2} D(\alpha + \rho_0 \widehat \omega, \alpha + \rho_0 \widehat \omega)  + \frac{\rho^2 - \rho_0^2}{2} (\widehat{v\omega})_0  + C \V_0 \rho^2 Y^2 \label{eq:HPhi05}.
	\end{align}
	Our choice of $\gamma$ and $\alpha$ minimizes the integral where we replaced $\V_p$ by $\widehat \gt_p$, and by explicit computation using the definition \eqref{def:alphagamma} of $\alpha$ and $\gamma$ we find 
	\begin{equation}\label{eq.int.bog0}
		\int (p^2 + \rho_0 \widehat \gt_p) \gamma_p + \rho_0 \widehat \gt_p \alpha_p \frac{\dd p}{2 \pi^2} = \frac 1 2 \int  \Big(\sqrt{p^4 + 2 \rho_0 p^2 \widehat \gt_p}-p^2 - \rho_0 \widehat \gt_p \Big)\frac{\dd p }{4 \pi^2} .
	\end{equation}
	Moreover the formula for $\widehat{g\omega}$ from Lemma~\ref{eq:gdecay} yields
	\[ (\widehat{\gt \omega})_0 = \langle \widehat \omega, \widehat \gt \rangle = \int \frac{\widehat \gt^2_p - \widehat \gt_0^2 \one_{\{p \leq \ell_\delta^{-1}\}}}{2p^2} \frac{\dd p}{4 \pi^2}, \]
	where $\ell_\delta = \frac 1 2 e^{\Gamma} \widetilde{R}$. Inserting this and \eqref{eq.int.bog0} into \eqref{eq:HPhi05} we find
	\begin{align*} \frac{\langle \mathcal{H} \rangle _\Phi}{\vert \Lambda \vert} &= \frac{\rho^2}{2} \widehat \gt_0 + \frac 1 2 \int \Big(  \sqrt{p^4 + 2 \rho_0 p^2 \widehat \gt_p} -p^2 - \rho_0 \widehat \gt_p+ \rho_0^2\frac{\widehat \gt^2_p - \widehat \gt_0^2 \one_{\{p \leq \ell^{-1}_{\delta}\}}}{2p^2} \Big)\frac{\dd p }{4 \pi^2} \\ &\quad+ \frac{1}{2} D(\alpha + \rho_0 \widehat \omega, \alpha + \rho_0 \widehat \omega)  + \frac{\rho^2 - \rho_0^2}{2} (\widehat{v\omega})_0 + \rho_0 \int (\V_p - \widehat \gt_p) \gamma_p \frac{\dd p}{4 \pi^2}+ C \V_0 \rho^2 Y^2,
	\end{align*}
	where the last integral comes from the replacement of $\V_p$ by $\widehat g_p$ in the first term of the integral in \eqref{eq:HPhi05}. Since $ \rho - \rho_0 \leq C\rho Y$ (Lemma~\ref{lem:ag} and Lemma~\ref{lem:RiemannSums2}) and $\vert \V_p - \widehat g_p \vert \leq (\widehat{v\omega})_0 = \V_0 - \widehat g_0$, we can bound
	\[ \frac{\rho^2 - \rho_0^2}{2} (\widehat{v\omega})_0 + \rho_0 \int (\V_p - \widehat \gt_p) \gamma_p \frac{\dd p}{4 \pi^2} \leq C \rho^2 Y (\V_0 - \widehat g_0), \]
	and the lemma follows.
\end{proof}
In the following lemma we estimate the remainder term from Lemma~\ref{lem:up2}.
\begin{lemma}\label{lem:up3}
	There is a $C>0$ independent of $v$ and $\rho$ such that:
	\[ D(\alpha + \rho_0 \widehat \omega, \alpha + \rho_0 \widehat \omega) \leq\rho_{0}^{2}\delta^{2}\V_{0} \Big\vert\frac{1}{\delta}-\frac{1}{Y}+\log\delta \Big\vert+\rho_{0}^{2}\delta^{2}\V_{0} \Big(\frac{1}{\delta}-\frac{1}{Y}+\log\delta \Big)^{2}+ C  \V_0 \rho_0^2 \widehat{g}_0^2 .\]
In particular, with $\delta = \delta_0$ we deduce
\[ D(\alpha + \rho_0 \widehat \omega, \alpha + \rho_0 \widehat \omega)  \leq C \widehat v_0 \rho^2 \delta_0^2. \]
\end{lemma}

\begin{proof}
	We first estimate $\varphi_p := \langle \alpha + \rho_0 \widehat \omega, \V_{p-\cdot} \rangle$, using Lemma~\ref{lem:FT-w} as
	\begin{align*}
		\varphi_p &=  \int \Big(\frac{ \rho_0 \widehat \gt_q \V_{p-q} - \rho_0 \widehat \gt_0 \V_p \one_{\{\vert q\vert < \sqrt{\rho_{0} \widehat \gt_0}\}}}{2 q^2}+\frac{ \rho_0 \widehat \gt_0 \V_p \one_{\{ \ell^{-1}_{\delta}< \vert q\vert<\sqrt{\rho_{0} \widehat \gt_0} \}}}{2 q^2} \\ & \qquad \qquad  - \frac{\rho_0 \widehat \gt_q \V_{p-q}}{2 \sqrt{ q^4 + 2 \rho_0 q^2 \widehat \gt_q}}\Big) \frac{\dd q}{4 \pi^2}\\ &= \varphi_p^{(1)} + \varphi_p^{(2)} +\varphi_p^{(3)},
	\end{align*}
	with
	\begin{align}
		\vert \varphi_p^{(1)} \vert &=  \Big\vert \int_{\{\vert q \vert > \sqrt{\rho_0 \widehat \gt_0}\}} \frac{\rho_0 \widehat \gt_q \V_{p-q}}{2q^2} \Big( 1 - \frac{1}{\sqrt{1+ \frac{2 \rho_0 \widehat \gt_q}{q^2}}} \Big) \frac{\dd q}{4 \pi^2} \Big\vert \nonumber \\ &\leq C \V_0 \int_{\{\vert q \vert > \sqrt{\rho_0 \widehat \gt_0}\}} \frac{(\rho_0 \widehat \gt_0)^2}{q^4} \dd q \leq C \V_0 \rho_0 \widehat{g}_0 . \label{eq:boundphi1p}
	\end{align}
We also calculate
\begin{align}
	\vert \varphi_p^{(2)} \vert &=\Big\vert \int_{\{ \ell^{-1}_{\delta}<\vert q \vert < \sqrt{\rho_0 \widehat \gt_0}\}} \frac{\rho_0 \widehat \gt_0 \V_{p}}{2q^2} \frac{\dd q}{4 \pi^2} \Big\vert \nonumber\\
	&\leq \frac{\rho_{0}\widehat \gt_0\V_0}{4\pi}\log \Big(\sqrt{\rho_{0}a^{2}}\sqrt{\widehat \gt_0}\frac{e^{\Gamma}}{2}e^{\frac{1}{2\delta}} \Big)\nonumber\\
	&= \rho_{0}\delta \V_0 \Big( \frac{1}{\delta}-\frac{1}{Y}+\log\delta +C \Big).
	\end{align}
Using \eqref{eq:Defldelta} we have,
	\begin{align*}
		\vert \varphi_p^{(3)} \vert &= \Big\vert \int_{\{\vert q \vert <\sqrt{\rho_{0} \widehat \gt_0}\}} \frac{\rho_0 (\widehat \gt_q \V_{p-q} - \widehat \gt_0 \V_p ) }{2 q^2} \frac{\dd q}{4\pi^2} - \int_{\{\vert q \vert < \sqrt{\rho_0 \widehat{g}_0}\}} \frac{\rho_0 \widehat \gt_q \V_{p-q}}{2 \sqrt{q^4 + 2 \rho_0 q^2 \widehat \gt_q}} \frac{\dd q }{4 \pi^2} \Big\vert \\
		& \leq \int_{\{\vert q \vert < \sqrt{\rho_{0} \widehat \gt_0}\}} \frac{\rho_0 \widehat \gt_0 \vert \V_{p-q} - \V_p \vert +\rho_0   \vert\widehat{g}_q  - \widehat{\gt}_0   \vert\widehat{v}_0}{2 q^2} \frac{\dd q}{4\pi^2} + C \V_0 \sqrt{\rho_0 \widehat \gt_0}  \int_{\{\vert q \vert < \sqrt{\rho_{0} \widehat \gt_0}\}} \frac{1}{\vert q \vert} \dd q \\
		& \leq C \Vert \nabla \V \Vert_{\infty} \int_{\{\vert q \vert < \sqrt{\rho_{0} \widehat \gt_0}\}} \frac{\rho_0 \widehat \gt_0}{\vert q \vert} \dd q + C \V_0 \rho_0 \widehat{g}_0 \leq C\V_0 \rho_0 \widehat{g}_0,
	\end{align*}
where we used $\ell^{-1}_{\delta}\sim \sqrt{\rho\widehat{g}_0}\leq R$ and $\Vert\nabla \widehat{v}_0\Vert_{\infty}\leq R\V_0$. In the end we obtain
\begin{align*}
	\vert \varphi_p \vert &\leq C\V_0 \rho_0 \widehat{g}_0 +\rho_{0}\delta \V_0 \Big(\frac{1}{\delta}-\frac{1}{Y}+\log\delta +C \Big).
\end{align*}
Similarly we have bounds on the gradient of $\varphi$, namely
	\begin{equation}
		\vert \nabla \varphi_p \vert \leq CR\V_0 \rho_0\widehat{g}_0.
	\end{equation}
	Now we turn to
	\begin{align*}
		D(\alpha + \rho_0 \widehat \omega, \alpha + \rho_0 \widehat \omega) &= \langle \alpha + \rho_0 \widehat \omega, \varphi \rangle \\  &= \int \Big(\frac{ \rho_0 \widehat \gt_q \varphi_{q} - \rho_0 \widehat \gt_0 \varphi_0 \one_{\{\vert q \vert > \ell^{-1}_{\delta}\}}}{2 q^2} - \frac{\rho_0 \widehat \gt_q \varphi_{q}}{2 \sqrt{ q^4 + 2 \rho_0 q^2 \widehat \gt_q}} \Big)\frac{\dd q}{4 \pi^2},
	\end{align*}
	which we in the same way write as $D_1 + D_2+ D_{3}$ with
	\begin{align*}
		\vert D_1\vert &= \Big\vert  \int_{\{\vert q \vert > \sqrt{\rho_0\widehat{g}_0}\}} \frac{\rho_0 \widehat \gt_q \varphi_q}{2 q^2} \Big( 1 - \frac{1}{\sqrt{1+ \frac{2 \rho_0 \widehat \gt_q}{q^2}}} \Big)  \frac{\dd q}{4 \pi^2} \Big\vert \\
		&\leq \frac{\Vert \varphi \Vert_{\infty}}{8\pi^{2}}  \int_{\{\vert q \vert > \sqrt{\rho_0\widehat{g}_0}\}} \frac{(\rho_0 \widehat \gt_0)^2}{q^4} \dd q  \\
		&\leq C\V_0 \rho_0^{2} \widehat{g}_0^{2}+\rho_{0}^{2}\delta^{2}\V_{0} \Big(\frac{1}{\delta}-\frac{1}{Y}+\log\delta +C \Big),
	\end{align*}
	and using the bounds on $\varphi$, we find $\vert D_1 \vert \leq  C  \V_0 \rho_0^2 \widehat{g}_0^2$. The technique to bound $D_{2}$ is the same as for $\varphi^{(2)}$ and its provides
	\begin{equation*}
		\vert D_2 \vert  = \Big\vert\int_{\{\ell^{-1}_{\delta} < \vert q \vert < \sqrt{\rho_0 \widehat{g}_0} \}} \frac{ \rho_0 \widehat \gt_0 \varphi_0 }{2 q^2} \frac{\dd q}{4 \pi^2}\Big\vert \leq  \V_0 \rho_0^2\delta^{2} \Big(\frac{1}{\delta}-\frac{1}{Y}+\log\delta +C \Big)^{2}.
	\end{equation*}
Lastly $D_3$ is bounded just as $\varphi^{(3)}$,
	\begin{align*}
		\vert D_3 \vert &= \Big\vert \int_{\{\vert q \vert < \sqrt{\rho_0\widehat{g}_0}\}} \frac{\rho_0 (\widehat \gt_q \varphi_{q} - \widehat \gt_0 \varphi_0 ) }{2 q^2} \frac{\dd q}{4\pi^2} - \int_{\{\vert q \vert < \sqrt{\rho_0\widehat{g}_0}\}} \frac{\rho_0 \widehat \gt_q \varphi_{q}}{2 \sqrt{q^4 + 2 \rho_0 q^2 \widehat \gt_q}} \frac{\dd q }{4 \pi^2} \Big\vert \\
		&\leq C  \V_0 \rho_0^2 \widehat{g}_0^2,
	\end{align*}
	from which the result follows. 
\end{proof}

Now we have all necessary ingredients to conclude the proof of Theorem~\ref{thm.upperbound.soft.pot}.

\begin{proof}[Proof of Theorem~\ref{thm.upperbound.soft.pot}]
	We take the trial state $\Phi$ defined in Section~\ref{subsec:quasifree}, which has the expected bounds on number of particles from Lemma~\ref{lem:numberPhi}. The energy of $\Phi$ is bounded by Lemma~\ref{lem:up2} together with Lemma~\ref{lem:up3}, and using $\delta_0\geq \frac{1}{2}Y$ we find
	\begin{align*}
		\frac{\langle \mathcal{H} \rangle _\Phi}{\vert \Lambda_\beta \vert}   &\leq \frac{\rho^2}{2} \widehat \gt_0 + \frac 1 2 \int \Big(  \sqrt{p^4 + 2 \rho_0 p^2 \widehat \gt_p} -p^2 - \rho_0 \widehat \gt_p+ \rho_0^2\frac{\widehat \gt^2_p - \widehat \gt_0^2 \one_{\{\vert p\vert \leq  \ell_{\delta}^{-1}\}}}{2p^2} \Big)\frac{\dd p }{4 \pi^2} \nonumber \\ &\quad+ \rho_{0}^{2}\delta^{2}\V_{0} \Big(\frac{1}{\delta}-\frac{1}{Y}+\log\delta \Big)+\rho_{0}^{2}\delta^{2}\V_{0} \Big(\frac{1}{\delta}-\frac{1}{Y}+\log\delta \Big)^{2}\\ &\quad +C \rho^2 \delta (\V_0 - \widehat g_0)+ C  \V_0 \rho_0^2 \widehat{g}_0^2.
	\end{align*}
	Now this integral can be estimated by Proposition~\ref{prop:integralapprox2} and using $\rho-\rho_0\leq C\rho Y$ we find
	\begin{align}
		\frac{\langle \mathcal{H} \rangle _\Phi}{\vert \Lambda_\beta \vert}   &\leq  \frac{\rho^2}{2} \widehat \gt_0 + 4 \pi \rho^2 \delta^2 \Big( \frac{1}{\delta} - \frac{1}{Y} + \log \delta + \Big( \frac 1 2 + 2 \Gamma + \log \pi \Big) \Big) +\frac{1}{2} \rho^{2}\delta^{2}\V_{0} \Big\vert\frac{1}{\delta}-\frac{1}{Y}+\log\delta \Big\vert \nonumber \\ &\quad +\frac{1}{2}\rho^{2}\delta^{2}\V_{0} \Big(\frac{1}{\delta}-\frac{1}{Y}+\log\delta \Big)^{2}+C \rho^2 \delta (\V_0 - \widehat g_0)+ C  \V_0 \rho^2 \widehat{g}_0^2. \label{eq.upperbound.delta}
	\end{align}
Finally, with $\widehat g_0 = 8 \pi \delta$ and the specific choice $\delta = \delta_0$ we deduce
		\begin{align}
		\frac{\langle \mathcal{H} \rangle _\Phi}{\vert \Lambda_\beta \vert} &\leq  4\pi \rho^2 \delta_0 \Big(1 + \Big(2\Gamma + \frac{1}{2} + \log(\pi) \Big) \delta_0 \Big)+ C \rho^2 \delta_0(\widehat{v}_0-\widehat{g}_0)+C \rho^2 \delta^{2}_0\widehat{v}_0.
	\end{align}
\end{proof}
%%%%%%%%%%%%%%%%%%%%%%%%%%%%%%%%%%%

\begin{remark}\label{remark.delta.up}
In the case of the spherical measure potential \eqref{eq:deltaPot}, one can see that the upper bound \eqref{eq.upperbound.delta} is minimized (to the available energy precision) by the choice $\delta = \delta_0$. Indeed, even though the first four terms suggest to choose the smallest $\delta$ possible, the contribution of $\widehat v_0 - \widehat g_0$ yields a minimizer of the form
\begin{equation}
\delta = Y (1 + cY\vert \log Y \vert ).
\end{equation}
The constant $c$ will only affect the energy to a precision higher than the one we reach. We pick for simplicity $c=-1$ to obtain our $\delta_{0}$ providing useful cancelations.
\end{remark}

%%%%%%%%%%%%%%%%%%%%%%%%%%%%%%%%%%%

\section{General upper bound}\label{sec.upper.general}

In this section we prove Theorem~\ref{thm.upperbound.grandcanonical}, using the results of Section~\ref{sec.soft}. We let $\beta \geq \frac{3}{2} $ be given and we work on the box $\Lambda_\beta = [ - \frac{L_\beta}{2} , \frac{L_\beta}{2} ]^2$ of size $L_\beta = \rho^{- \frac 1 2} Y^{- \beta}$. Moreover, the number of particles at density $\rho$ is $N= Y^{- 2 \beta}$.

\subsection{Trial state}

Let $v$ be a non-negative measurable and radial potential with scattering length $a$ and $\supp(v) \subset B(0,R)$, with $\rho R^2 \leq Y^{2 \beta  +2}$. We consider $\varphi_b$ the associated scattering solution normalized at length $b = \rho^{-1/2} Y^{\beta + \frac 1 2}$. In other words $\varphi = 2 \delta_\beta \varphi^{(0)}$ with $\delta_\beta = \frac 1 2 \log \left( b / a \right)^{-1}$, see \eqref{eq:delta}. We also define $g=v\varphi$. Note that $R \ll b$. Let $f=\min\left (1,\varphi\right )$ be the truncated scattering solution. It satisfies
\begin{equation}\label{eq:scatt2}
 -\Delta f (x) + \frac{1}{2} v(x) f(x) = 0 \quad \text{on } B(0,b) ,
\end{equation} 
and is normalized such that $f(x) = 1$ for $\vert x \vert \geq b$.
We define a grand canonical trial state as 
\begin{equation}\label{eq:defpsi}
\Psi = \sum_{n \geq 0} \Phi_n F_n \in \mathscr{F}_s\big (L^{2}\left (\Lambda_{\beta}\right ) \big)
\end{equation}
 where $\Phi = \sum_n \Phi_n \in \mathscr{F}_s\big(L^{2}\left (\Lambda_{\beta}\right )\big )$ is a quasi-free state defined in \eqref{def:WT} and $F_n$ is the Jastrow factor
\begin{equation} 
F_n(x_1, \ldots, x_n) = \prod_{1 \leq i < j \leq n} f(x_i -x_j) .
\end{equation}
We will use the notation $f(x_i -x_j)=f_{ij}$ and $\nabla f (x_i - x_j) = \nabla f_{ij}$. Finally note that
\begin{equation}
\nabla_{i}F_n(x_1, \ldots, x_n)=\sum_{\substack{j=1\\ j\neq i}}^{n}\nabla_{i}f_{ij}\frac{F_n}{f_{ij}}.\label{eq:nabF}
\end{equation}

\subsection{Reduction to a soft potential}

In this section we prove that the energy $\langle \Psi , \mathcal H_v \Psi \rangle$ can be bounded by $\langle \Phi , \mathcal H_{\widetilde v} \Phi \rangle$ where $\widetilde v$ is a nicer potential. This is the effect of the Jastrow factor $F_n$, and we are thus reduced to optimizing the choice of the quasi-free state $\Phi$ according to the potential $\widetilde v$.

\begin{lemma}\label{lemma reduction to soft}
Consider the radial potential $\widetilde{v}(x) = 2 f'(b) \delta_{\{\vert x \vert = b\}}$ (with $f'$ being understood as the radial derivative). Then the state $\Psi$ defined in \eqref{eq:defpsi} satisfies
\[ \langle \Psi, \mathcal H_v \Psi \rangle \leq \langle \Phi, \mathcal H_{\widetilde v} \Phi \rangle - \langle \Phi, \mathcal R \Phi \rangle ,\]
where $\mathcal R= \oplus_n \mathcal R_n$ with
\[ \mathcal  R_n = \sum_{\{i,j,k\}} \frac{\nabla f_{ij}}{f_{ij}} \cdot \frac{\nabla f_{ik}}{f_{ik}} F_n^2, \]
where we introduced the notation $$\{i,j,k\}=\{\text{set of pairwise distinct indices}\,\, i, j, k \,\,\text{running from}\,\, 1\,\, \text{to}\,\, n\}. $$
\end{lemma}

\begin{proof}
The energy of the $n$-th sector state is
\begin{align}\label{eq:squared}
\langle \Psi_n , \mathcal{H}_n \Psi_n \rangle =  \sum_{i=1}^n \int_{\Lambda^n} \big(F_n^2 \vert \nabla_i \Phi_n \vert^2 &+ \vert \nabla_i F_n \vert^2 \Phi_n^2 + 2 F_n \nabla_i F_n \cdot \Phi_n \nabla_i \Phi_n \big)\dd x \nonumber \\
& + \sum_{1\leq i<j\leq n} \int_{\Lambda^n} v(x_i-x_j) F_n^2 \Phi_n^2 \dd x.
\end{align}
 The second term in \eqref{eq:squared} can be written via \eqref{eq:nabF} as
\begin{align}\label{eq:squared2}
\sum_{i=1}^n \int_{\Lambda^n} \vert \nabla_i F_n \vert^2 \Phi_n^2 \dd x = \sum_{i\neq j} \int_{\Lambda^n} \vert \nabla f_{ij} \vert^2 \frac{F_n^2}{f_{ij}^2} \Phi_n^2 \dd x + \sum_{\{i,j,k\}} \int_{\Lambda^n} \frac{\nabla f_{ij}}{f_{ij}} \cdot \frac{\nabla f_{ik}}{f_{ik}} F_n^2 \Phi_n^2 \dd x .
\end{align}
Note that, in the first part of \eqref{eq:squared2} the integration in $x_i$ is only supported on the ball $\vert x_i - x_j \vert \leq b$, because $f_{ij} =1$ outside this ball. We integrate by parts on this ball to find
\begin{align}
\sum_{i=1}^n \int_{\Lambda^n} \vert \nabla_i & F_n \vert^2 \Phi_n^2 \dd x = - \sum_{i \neq j} \int_{\{\vert x_{i}-x_{j}\vert\leq b\}} \Delta f_{ij} \frac{F_n^2}{f_{ij}} \Phi_n^2 \dd x   - \sum_{i \neq j} \int_{\Lambda^n} \nabla f_{ij} \frac{F_n^2}{f_{ij}} \cdot \nabla_i (\Phi_n^2) \dd x \nonumber \\ 
& - \sum_{\{i,j,k\}} \int_{\Lambda^n} \frac{\nabla f_{ij}}{f_{ij}} \cdot \frac{\nabla f_{ik}}{f_{ik}} F_n^2 \Phi_n^2 \dd x + \sum_{i \neq j} \int_{\Lambda^{n-1}} \int_{\{\vert x_i - x_j \vert = b\}} \partial_r f(b) F_n^2 \Phi_n^2 \dd x_i \dd \hat{x}_i, \label{eq:squared3}
\end{align}
where $ \hat{x}_i=(x_{1}, \ldots x_{i-1} , x_{i+1}, \ldots , x_n)$.
The second term in the right hand side of \eqref{eq:squared3} is precisely $-2F_n \nabla_i F_n \cdot \Phi_n \nabla_i \Phi_n$ thanks to \eqref{eq:nabF}. We use the scattering equation \eqref{eq:scatt2} to transform
\begin{equation}\label{eq:squared4}
\sum_{i \neq j} \int_{\{\vert x_{i}-x_{j}\vert\leq b\}} \Delta f_{ij} \frac{F_n^2}{f_{ij}} \Phi_n^2 \dd x =\sum_{1\leq i<j\leq n} v(x_i-x_j) F_n^2 \Phi_n^2 \dd x,
\end{equation}
in \eqref{eq:squared3} (note that there is no half factor because the sum is on $i<j$). Using \eqref{eq:squared3} and \eqref{eq:squared4} in \eqref{eq:squared} we deduce
\begin{align}
\langle \Psi_n , \mathcal H_n \Psi_n \rangle &= \sum_{i=1}^n \int_{\Lambda^n} F_n^2 \vert \nabla_i \Phi_n \vert^2 + 2 \sum_{i<j} \int_{\Lambda^{n-1}} \int_{\{\vert x_i - x_j \vert = b\}} \partial_r f(b) F_n^2 \Phi_n^2 \dd x_i \dd \hat{x}_i \nonumber \\ &\quad - \sum_{\{i,j,k\}} \int_{\Lambda^n} \frac{\nabla f_{ij}}{f_{ij}} \cdot \frac{\nabla f_{ik}}{f_{ik}} F_n^2 \Phi_n^2 \dd x.
\end{align}
In the first two terms we bound $F_n$ by $1$, and the last one we consider as a remainder. Thus,
\[ \langle \Psi_n , \mathcal H_n \Psi_n \rangle \leq \int_{\Lambda^n} \vert \nabla \Phi_n \vert^2 + \sum_{i<j} \int_{\Lambda^n} \Phi_n^2 \widetilde{v}(x_i-x_j) \dd x - \mathcal  R_n .\]
\end{proof}
We comment here how in the proof we used nowhere that $\Phi$ is a quasi-free state, therefore the lemma holds true for more general $\Phi \in \mathscr{F}_s(L^2(\Lambda_{\beta}))$.

\subsection{Number of particles in our trial state}

Now for $\Phi$ we choose the quasi-free state given by Theorem~\ref{thm.upperbound.soft.pot}, applied to the potential $\widetilde v$. We recall that $\Phi = W_{N_0} T_\nu \Omega$ is defined in \eqref{def:WT}, and $\Psi$ in \eqref{eq:defpsi}. In this section we prove the following two lemmas, giving estimates on the norm of $\Psi$ and the average number of particles in $\Psi$. The idea is to use the properties of $F_n$ to derive the bounds on $\Psi_n = F_n \Phi_n$ from the bounds on the quasi-free state $\Phi$.

\begin{lemma} \label{lemma denominator}
There is a $C>0$, independent of $v$ and $\rho$, such that
\[\Vert \Psi \Vert^2 \geq \Vert \Phi \Vert^2 \Big( 1 - C N Y^{2\beta + 2} \Big) .\]
\end{lemma}

\begin{proof}
The norm of our trial state is bounded from below by
\begin{align}
\Vert \Psi \Vert^2 &= \sum_{n \geq 0} \int_{\Lambda^n} F_n^2(x) \Phi_n^2(x) \dd x \nonumber \\ &\geq \sum_{n \geq 0} \Big( \int_{\Lambda^n} \Phi_n^2 \dd x - \sum_{1\leq i<j\leq n} \int_{\Lambda^n} (1 - f(x_i-x_j)^2) \Phi_n^2(x) \dd x \Big),  \label{eq:aPsiBound}
\end{align}
where we used the inequality 
\begin{equation}\label{eq:the.inequality}
\prod_{1 \leq i<j \leq n} f(x_i-x_j)^2 \geq 1 - \sum_{1 \leq i<j \leq n} (1- f(x_i-x_j)^2).
\end{equation}
 The second term is the $2$-body interaction potential energy of $\Phi$, thus we can write it as
\begin{align}
\sum_{n \geq 0} \sum_{1 \leq i <j \leq n} \int (1-f(x_i-x_j)^2) \Phi_n^2 \dd x &= \frac{1}{2 \vert \Lambda_{\beta} \vert} \sum_{p,q,r} \widehat{(1-f^2)}_r \langle a_{q}^* a_{p+r}^* a_{q+r} a_p \Phi, \Phi \rangle \nonumber\\
&\leq \frac{\widehat{(1-f^2)}_0}{2 \vert \Lambda_{\beta} \vert} \sum_{p,q,r}  \langle a_{q}^* a_{p+r}^* a_{q+r} a_p \Phi, \Phi \rangle \label{ine:r}.
\end{align}  
Since $\Phi = W_{N_0} T_\nu \Omega$ is a quasi-free state we can estimate this term as already done in \eqref{eq:conjW0}. We first conjugate by $W_{N_0}$ which amounts to change the $a_{0}$'s into $N_0 \leq N$. Together with Lemma~\ref{lem:ag} and Wick's theorem we deduce
\begin{align}
 \sum_{p,q,r}  \langle a_{q}^* a_{p+r}^* a_{q+r} a_p \rangle_\Phi &\leq CN^{2} +\sum_{\substack{p\neq 0,q\neq 0,r\neq 0\\ p+r\neq 0, q+r\neq 0}}  \langle a_{q}^* a_{p+r}^* a_{q+r} a_p \rangle_\Phi .
\end{align}
Then we use again Wick's Theorem to estimate the remaining sum, which is then bounded by $C N^2$ by Lemma~\ref{lem:ag}. Thus equation \eqref{ine:r} gives
\begin{equation}\label{eq.Wick1}
\sum_{n \geq 0} \sum_{1 \leq i <j \leq n} \int (1-f(x_i-x_j)^2) \Phi_n^2 \dd x \leq C \frac{N^2}{\vert \Lambda_{\beta} \vert} \int_{\Lambda} (1 - f(x)^2) \dd x \, \Vert \Phi \Vert^2 .
\end{equation}
Using $ \frac{\dd}{\dd r} [ r^2 \log \left( \frac r a \right)^2 - r^2 \log \left( \frac r a \right) + \frac{r^2}{2} ] = 2 r \log \left(\frac r a \right)^2 $ and $a\leq R\leq b$ we have that
\begin{align}
\int_\Lambda (1-f(x)^2) \dd x &= 2\pi \int_R^b \bigg( 1 - \frac{\log \left( \frac r a \right)^2}{\log \left( \frac b a \right)^2} \bigg) r \dd r + 2 \pi\int_0^R (1 - f(r)^2) r\dd r \nonumber \\ &\leq C \frac{b^2}{\log \left( \frac b a \right)} + C R^2 \leq C \rho^{-1} Y^{2 \beta + 2},\label{ine:1mf}
\end{align} 
where we used $\rho R^2 \leq Y^{2 \beta + 2}$ and $b^2 = \rho^{-1} Y^{2 \beta + 1}$.
We use this last bound in \eqref{eq.Wick1} and \eqref{eq:aPsiBound} to get
\[ \Vert \Psi \Vert^2 \geq \Vert \Phi \Vert^2 \big( 1 - C N Y^{2 \beta + 2} \big) . \]
\end{proof}

\begin{lemma}
There is a $C>0$ independent of $\rho$ and $v$ such that,
\[ \langle \Psi, \mathcal N \Psi\rangle \geq N(1- CY^2)\|\Psi\|^2 , \qquad \langle \Psi, \mathcal N^2 \Psi\rangle \leq C N^2 \|\Psi\|^2. \]
\end{lemma}

\begin{proof}
First we have by Lemma \ref{lem:numberPhi} and Lemma \ref{lemma denominator} that
\[ \langle  \Psi, \mathcal N ^2\Psi \rangle = \sum_{n \geq 0} n^{2} \int_{\Lambda^n} F_n^2 \Phi_n^2 \dd x \leq \sum_{n \geq 0} n^2 \Vert \Phi_n \Vert^2 =\langle  \Phi, \mathcal N^2 \Phi\rangle \leq C N^2 \|\Psi\|^2. \]
For the bound on $\langle \mathcal N \rangle_{\Psi}$ we use the same idea as in the proof of Lemma~\ref{lemma denominator}. From inequality \eqref{eq:the.inequality} we deduce
\begin{align} \langle  \Psi , \mathcal N\Psi \rangle &= \sum_{n \geq 0} n \int_{\Lambda^n} F_n^2 \Phi_n^2 \dd x \nonumber \\ 
&\geq \sum_{n\geq 0} n \Big( \int_{\Lambda^n} \Phi_n^2 \dd x - \sum_{i<j} \int_{\Lambda^n} (1-f(x_i-x_j)^2) \Phi_n^2\dd x \Big). \label{eq.N.lower1}
\end{align}
In the second term we recognize a number operator and a 2-particles interaction energy, which can be rewritten as
\[ \sum_{n \geq 0} n \sum_{i<j} \int_{\Lambda^n} (1-f(x_i-x_j)^2) \Phi_n^2 \dd x = \frac{1}{2L_\beta^2} \sum_{k,p,q,r \in \Lambda^*} \widehat{(1 - f^2)}_r \langle a_k^* a_k a_{p+r}^* a_q^* a_{q+r} a_p \Phi , \Phi \rangle). \]
We can compute this term using the same techniques as for \eqref{ine:r}, \textit{i.e.}, extract the $a_{0}$'s and then apply Wick's Theorem
yielding many terms of the form $A_1 A_2 A_3$ with $A_i \in \lbrace \langle a_0^\dagger a_0 \rangle_\Phi , \sum_{p \neq 0} \alpha_p , \sum_{p \neq 0} \gamma_p \rbrace$ (see \eqref{eq.Phi.apaq}). These terms are bounded by $N^3$ by Lemma~\ref{lem:ag}. Thus
\[\sum_{n \geq 0} n \sum_{i<j} \int_{\Lambda^n} (1-f(x_i-x_j)^2) \Phi_n^2 \dd x  \leq C \frac{ N^3 }{L_\beta^2} \int (1-f(x)^2 )\dd x \Vert \Phi \Vert^2. \]
Now we use the inequality \eqref{ine:1mf} to bound the right hand side of the quantity above and plug it in \eqref{eq.N.lower1} to obtain
\[ \langle \Psi, \mathcal N \Psi \rangle_{} \geq (N - C N^2 Y^{2 \beta + 2} ) \|\Psi\|^2= N (1-C Y^2)\|\Psi\|^2, \]
where in the equality used that $N=Y^{-2\beta}$. 
\end{proof}

\subsection{Remainder term}

Here we prove that the remainder term in Lemma~\ref{lemma reduction to soft} is indeed small.

\begin{lemma}\label{lemma double product}
There is a $C>0$ independent of $v$ and $\rho$ such that
\[ \vert \langle \Phi, \mathcal R \Phi \rangle \vert \leq C L_\beta^{2} \rho^{2} Y^{2 \beta + 2} \Vert \Phi \Vert^{2}. \]
\end{lemma}

\begin{proof}
The remainder term can be bounded by
\[ \vert \langle \Phi, \mathcal{R} \Phi \rangle \vert \leq \sum_{n \geq 3} \sum_{\{i,j,k\}} \int_{\Lambda^n} W(x_i-x_k) W(x_i - x_j) \Phi_n^2 \dd x  , \]
where $W(x) = \vert f(x) \nabla f(x) \vert$. This is a three-body interaction potential, which can be rewritten in second quantization as
\[ \vert \langle \Phi, \mathcal{R} \Phi \rangle \vert \leq \frac{1}{\vert \Lambda_\beta
\vert^2} \sum_{p,q,r,k,\ell \in \Lambda^*} \widehat{W}_k \widehat{W}_\ell \langle a^*_{p+\ell + k} a^*_{q-k} a^*_{r-\ell} a_r a_q a_p \rangle_{\Phi} \, \Vert \Phi \Vert^2 . \]
We can again use Wick's Theorem to estimate this part, and since Lemma~\ref{lem:ag} provides $$\sum_{p \neq 0} \alpha_p \,\,\,\text{and}\,\, \sum_{p \neq 0} \gamma_p \leq N,$$ we find
\begin{equation} \label{eq:boundonR} \vert \langle \Phi, \mathcal R \Phi \rangle \vert \leq C \frac{N^3}{\vert \Lambda_\beta \vert^2} \widehat W_0 ^2 \, \Vert \Phi \Vert^2. 
\end{equation}
Now since $f(x) = \log \left(\frac{b}{a} \right) ^{-1} \log \Big( \frac{\vert x \vert}{a} \Big)$ outside the support of $v$ and is radially increasing, we have
\begin{align*}
\widehat{W}_0 & \leq 2\pi\int_R^b \frac{\log \left( \frac r a \right)}{\log \left( \frac b a \right) ^2}  \dd r + 2\pi\int_0^R f(r) f'(r) r \dd r \\
& \leq C \frac{b}{\log \left( \frac b a \right)} + C R\Big(\frac{ \log \left( \frac R a \right)}{\log \left( \frac b a \right)}\Big)^{2} \leq C \rho^{-1/2} Y^{\beta + 1} ,
\end{align*}
where we used $\vert \log \left( \frac b a \right) \vert^{-1} \leq Y$ and $\vert \log \left( \frac R a \right) \vert \leq \vert \log \left( \frac b a \right) \vert$. Inserting this bound in \eqref{eq:boundonR} we get the result.
\end{proof}

\subsection{Conclusion : Proof of Theorem~\ref{thm.upperbound.grandcanonical}}

Using Lemmas~\ref{lemma reduction to soft},~\ref{lemma denominator} and~\ref{lemma double product} we know that our trial state $\Psi$ satisfies
\begin{equation}\label{eq.HvPsi1}
 \langle \mathcal H_v \rangle_\Psi \leq \Big( \langle \mathcal H_{\widetilde v} \rangle_\Phi  + C L_\beta^2 \rho^2 Y^{2\beta + 2} \Big) \Big( 1 + CNY^{2\beta + 2} \Big).
\end{equation}
For $\Phi$ we choose the quasi-free state given by Theorem~\ref{thm.upperbound.soft.pot} applied to the soft potential $\widetilde v$. We deduce that
\begin{align}
\frac{1}{\vert \Lambda_\beta \vert} \langle \mathcal H_{\widetilde v} \rangle_\Phi \leq  4\pi  \rho^2 \delta_0 \Big(1 + \Big(2\Gamma + \frac{1}{2} + \log(\pi) \Big) \delta_0 \Big)
+ C \rho^2 \delta_0(\widehat{\widetilde v}_0-\widehat{g}_0)+C\rho^2 \delta^{2}_0\widehat{\widetilde v}_0.
\end{align}
In Lemma~\ref{lemma.soft.pot.properties} we have bounds on $\widehat{\widetilde v}_0$ as well as $\widehat{\widetilde v}_0-\widehat{g}_0$. Therefore, remembering the choices $b=\rho^{-1/2} Y^{1/2+\beta}$ and $\beta \geq 3/2$, we get
\begin{align}
\frac{1}{\vert \Lambda_\beta \vert} \langle \mathcal H_{\widetilde v} \rangle_\Phi \leq  4\pi  \rho^2 \delta_0 \Big(1 + \Big(2\Gamma + \frac{1}{2} + \log(\pi) \Big) \delta_0 \Big)
+ \beta C \rho^2 \delta_0^3\vert \log(\delta_0)\vert+C\rho^2 \delta^{3}_0.
\label{eq. how C depends on beta}
\end{align}
We insert this into \eqref{eq.HvPsi1} together with $N = \rho L_\beta^2 = Y^{-2\beta}$ and $Y\leq 2\delta_0$, which concludes the proof of Theorem~\ref{thm.upperbound.grandcanonical}.\qed

\section{Localization to large boxes for the lower bound}\label{sec:largebox}

In this section we reduce the proof of Theorem~\ref{thm:main_lower} to an analogous statement localized to a box of size $\ell$ defined in \eqref{eq:def_ell}, namely Theorem~\ref{thm:largebox_lower}.

\subsection{Grand Canonical Ensemble}\label{subsec:GCE}

We rewrite the Hamiltonian in a grand canonical setting to approach the problem in the Fock space description. To emphasize the fact that the density parameter appears through a chemical potential in this setting, we introduce the notation $\rho_\mu >0$ as new parameter. The corresponding $Y$ will be $Y = \vert \log( \rho_\mu a^2) \vert^{-1}$ and we fix $\delta$ to be
\begin{equation}\label{eq:def_delta_mu}
\delta = \delta_{\mu}, \qquad \delta_{\mu} := \frac{1}{|\log(\rho_{\mu} a^2 |\log(\rho_{\mu} a^2)|^{-1})|}.
\end{equation}
This corresponds to normalize the scattering solution at length $\widetilde{R} = (\rho_{\mu} Y)^{-1/2}$ in \eqref{eq:delta}. With this choice we recall the definition \eqref{eq:scat_defs} of $g$.
This definition is analogous to the one of $\delta_0$ \eqref{eq:def_delta} but with $\rho_{\mu}$ in place of $\rho$.  We are going to choose, a posteriori, $\rho_{\mu} = \rho$ which implies $\delta_{\mu} = \delta_0$.

We consider the operator $\mathcal{H}_{\rho_{\mu}}$ acting on the symmetric Fock space $\mathscr{F}_s(L^2(\Omega))$ and commuting with the number operator,  whose action on the $N-$bosons space is 
\begin{align}
\mathcal{H}_{\rho_{\mu},N} &= H(N,L) - 8 \pi \delta \rho_{\mu} N = \sum_{j=1}^N -\Delta_j + \sum_{i<j} v(x_i-x_j)- 8 \pi \delta \rho_{\mu} N  \nonumber\\
&= \sum_{j=1}^N \Big(-\Delta_j -  \rho_{\mu} \int_{\mathbb{R}^2}  g(x_j-y)\,\dd y\Big) + \sum_{i<j} v(x_i-x_j).\label{eq:grandcan_hamilton}
\end{align}
We define the ground state energy density of $\mathcal{H}_{\rho_{\mu}}$:
\begin{equation}\label{def:e0}
e_0(\rho_{\mu}) := \lim_{|\Omega| \rightarrow +\infty} \frac{1}{|\Omega|} \inf_{\Psi \in \mathscr{F}_s(L^2(\Omega))\setminus \{0\}} \frac{\langle \Psi\,|\, \mathcal{H}_{\rho_{\mu}} \,|\,\Psi \rangle }{\|\Psi\|^2}.
\end{equation}
In the rest of the paper we prove the following lower bound on $e_0(\rho_\mu)$.
\begin{theorem}\label{thm:grancanonicaltherm}
There exists $C$, $\eta > 0$ such that the following holds. Let $\rho_\mu > 0$ and $v\in L^1(\Omega)$ be a positive, spherically symmetric potential with scattering length $a$ and $\supp(v) \subset B(0,R)$ such that $\Vert v \Vert_1 \leq Y^{-1/8}$ and $R \leq \rho_\mu^{-1/2}$. Then, if $\rho_\mu a^2 \leq C^{-1}$, we have, for any $\rho_\mu > 0$,
\begin{equation}
e_0(\rho_{\mu}) \geq - 4 \pi \rho_{\mu}^2  \delta \Big(1 - \Big( 2 \Gamma + \frac 1 2 + \log \pi \Big) \delta \Big) - C \rho_{\mu}^2 \delta^{2+\eta}.
\end{equation}
\end{theorem}

We now show that Theorem~\ref{thm:grancanonicaltherm} implies the main lower bound Theorem~\ref{thm:main_lower}.

\begin{proof}[Proof of Theorem~\ref{thm:main_lower}]
We start by reducing the problem to a potential which is $L^1$ and compactly supported. For a given $v$ satisfying the assumptions of Theorem~\ref{thm:main_lower}, we apply Theorem~\ref{theorem. main potential reduction} with $T = (4 \pi Y)^{-1/8}$, $R= \rho^{-1/2}$ and $\varepsilon = 1$. This provides us a potential $\widetilde v = v_{T,R,\epsilon}$ to which we can apply Theorem~\ref{thm:grancanonicaltherm}. Then for this new potential we use the ground state of $\mathcal{H}_N$ as a trial function for $\mathcal{H}_{\rho_{\mu}}$ and get
\begin{align*}
e^{\rm{2D}}(\rho , \widetilde v ) &\geq e_0(\rho_{\mu} , \widetilde v ) + 8 \pi \widetilde \delta \rho \rho_{\mu}  \\
&\geq - 4 \pi \rho_{\mu}^2  \widetilde\delta \Big(1 - \Big(2 \Gamma + \frac 1 2 + \log \pi \Big) \widetilde \delta \Big) - C \rho_{\mu}^2 \widetilde \delta^{2+\eta} + 8 \pi \widetilde \delta \rho \rho_{\mu},
\end{align*}
where $\widetilde \delta = \frac{1}{\vert \log ( \rho \widetilde a^2 \vert \log (\rho \widetilde a^2) \vert^{-1} ) \vert}$ and $\widetilde a$ is the scattering length of $\widetilde v$. Since $\widetilde v \leq v$ we have $e^{\rm{2D}}(\rho,v) \geq e^{\rm{2D}}(\rho, \widetilde v)$. Moreover, by equation \eqref{eq. scattering condtions} we can change $\widetilde \delta$ into $\delta$ up to an error of order 
\begin{equation}
\frac{1}{\log\left( \frac{R}{a}\right)^2 T} + \frac{1}{\log\left( \frac{R}{a}\right)^2} \int_{\{\vert x \vert > R\}} v(x) \log \Big( \frac{ \vert x \vert }{a} \Big)^2 \dd x \leq C \delta^{2+\min \left(\frac{1}{8}, \eta_1 \right)}.
\end{equation} 
Choosing $\rho_{\mu} = \rho$ concludes the proof.
\end{proof}
\subsection{Reduction to large boxes}\label{subsec:localization_large_boxes}

We now make use of the sliding localization technique developed in \cite{BSol} to reduce the proof of Theorem~\ref{thm:main_lower} to a localized problem in a large box $\Lambda \subset \Omega$. 
We introduce the length scale
\begin{align}\label{eq:def_ell}
\ell := K_{\ell}\, \rho_{\mu}^{-1/2} Y^{-\frac{1}{2}},
\end{align}
where $K_{\ell} \gg 1$ is a parameter fixed in Appendix~\ref{app:parameters}, and we carry out the analysis in the large box 
\begin{equation}
\Lambda := \Big[-\frac{\ell}{2}, \frac{\ell}{2}\Big]^2.
\end{equation}
For any $u \in \mathbb{R}^2$, we denote by 
\begin{equation}
\Lambda_u := \ell u + \Lambda
\end{equation}
the translated large box.
Let us introduce the localization functions: the sharp characteristic function
\begin{equation}\label{def:sharp_loc_Lambda}
\theta_u := \one_{\Lambda_u}
\end{equation}
and the regular one: let $\chi \in C_0^{M}(\mathbb{R}^2)$, for $M \in \mathbb{N}$ with $\supp\chi = [-\frac{1}{2},\frac{1}{2}]^2$ be the spherically symmetric function defined in Appendix~\ref{app:locfunction}, and
\begin{equation}
\chi_{\Lambda}(x) := \chi \Big( \frac{x}{\ell}\Big), \qquad \chi_u (x) :=  \chi_{\Lambda}(x-\ell u).
\end{equation}
The parameter $M$ is fixed in Appendix~\ref{app:parameters}. Define the following projections on $L^2(\Lambda)$,
\begin{align}\label{def:PQ}
P := \ell^{-2} | \one_{\Lambda} \rangle \langle \one_{\Lambda} |, \qquad Q  := \one - P,
\end{align}
\textit{i.e.} $P$ is the orthogonal projection in $L^2(\Lambda)$ onto the constant functions and $Q$ is the orthogonal projection to the complement.
Using these definitions, we define the following operators on $\mathscr{F}_s(L^2(\Lambda))$ through their action on any $N$-particles sector:
\begin{align}\label{eq:SF_n0ognplus}
n_0:= \sum_{j=1}^N P_j, \qquad n_{+} := \sum_{j=1}^N Q_j = N - n_0.
\end{align}
The definition is based on the idea that low energy eigenstates of the system should concentrate in the constant function. Thus, $n_0$ counts the number of particles in the condensate and $n_{+}$ the number of particles excited out of the condensate.

We start by stating the result for the kinetic energy.
 
\begin{lemma}[Kinetic energy localization]\label{lem:LocKinEnLarge} 
Let $-\Delta_u^{\mathcal{N}}$ denote the Neumann Laplacian in $\Lambda_u$ and $-\Delta$ the Laplacian on ${\mathbb R}^2$. If the regularity of $\chi$ is $M >5$ and the positive parameters $\varepsilon_{N}, \varepsilon_T, d, s, b$ are smaller than some universal constant, then for all $\ell >0$ we have  
\begin{equation}
-\Delta \geq 
  \int_{\mathbb{R}^2}  {\mathcal T}_u \,\dd u,  
\end{equation}
in terms of quadratic forms in $H^1(\mathbb{R}^2)$, where
\begin{equation}\label{eq:Texplicitexpress}
\mathcal{T}_u  := \varepsilon_{N}(-\Delta_u^{\mathcal{N}}) + (1-\varepsilon_{N})(\mathcal{T}_u^{\text{Neu,s}}+ \mathcal{T}_u^{\text{Neu,l}} + \mathcal{T}_u^{\text{gap}} + \mathcal{T}_u^{\text{kin}} ),
\end{equation}
where 
\begin{align}\label{eq:SF_DefTuLarge}
 &\mathcal{T}_u^{\text{Neu,s}} := \frac{\varepsilon_T}{2(d\ell)^2}\frac{-\Delta_u^{\mathcal{N}}}{-\Delta_u^{\mathcal{N}} + (d\ell)^{-2}}, \\
 &\mathcal{T}_u^{\text{Neu,l}} := \frac{b}{\ell^2} Q_u, \\
 &\mathcal{T}_u^{\text{gap}}:= b \frac{\varepsilon_T}{(d\ell)^2} Q_u \mathbbm{1}_{(d^{-2}\ell^{-1}, + \infty)} (\sqrt{-\Delta}) Q_u,\\
&\mathcal{T}_u^{\text{kin}} := Q_u \chi_u \left\{ (1-\varepsilon_T)\left[\sqrt{-\Delta} -\frac{1}{2s\ell}\right]_+^2 + \varepsilon_T \left[\sqrt{-\Delta} -\frac{1}{2 d s\ell}\right]_+^2 \right\} \chi_u Q_u.
\end{align}
\end{lemma}

\begin{proof}
The proof is identical to the one of \cite[Lemma 3.7]{BSol} and its adaptation to our context in \cite[Lemma 6.4]{FS2}, which are independent of dimension.
\end{proof}

\begin{remark}
The kinetic energy is composed of several terms which have to remedy some problems related to the main kinetic energy term and play the following roles:
\begin{itemize}
\item $\mathcal{T}_u^{\text{kin}}$ is the main kinetic energy term;
\item $-\Delta^{\mathcal{N}}$ is the Neumann Laplacian and compensates the loss of ellipticity at the boundary caused by the localization function $\chi$ in $\mathcal{T}^{\text{kin}}_u$;
\item $\mathcal{T}_u^{\text{Neu,s}}$ is the Neumann gap in the small box. Worth to remark is that, for large momenta, it behaves like a gap, while for small momenta its action is like a Neumann Laplacian;
\item $\mathcal{T}_u^{\text{Neu,l}}$ is the Neumann gap in the large box;
\item $\mathcal{T}_u^{\text{gap}}$ is another spectral gap which we need in order to control the number of excitations with large momenta.
\end{itemize}
\end{remark}

The localization of the potential energy relies on a direct calculation of the integral which can be found in \cite[Proposition 3.1]{BSol}. Assuming that $R \ell^{-1}$ is sufficiently small, we can introduce the following localized potentials
\begin{align}\label{eq:SF_w12u}
W(x) &:= \frac{v(x)}{\chi*\chi(x/\ell)},  &w(x,y) := \chi_{\Lambda} (x) W(x-y)\chi_{\Lambda}(y), \\
W_1(x) &:= \frac{g(x)}{\chi*\chi(x/\ell)},  &w_{1}(x,y) := \chi_{\Lambda} (x) W_1(x-y)\chi_{\Lambda}(y),  \\
W_2(x) &:= \frac{g(x) + g(x) \omega(x)}{\chi*\chi(x/\ell)},  &w_{2}(x,y) := \chi_{\Lambda} (x) W_2(x-y)\chi_{\Lambda}(y),
\end{align}
where we observe that $W, W_1, W_2$ and $w, w_1, w_2$ are localized versions of $v, g, (1+ \omega) g$, respectively, defined in  \eqref{eq:scat_defs}.

Furthermore, we introduce the translated versions for $u \in \Lambda$
\begin{equation}
w_{*,u}(x,y) = w_*(x-\ell u,y-\ell u).
\end{equation}

We are going to make use of the following approximation result. We recall the defintion of the lengthscale $\ell_{\delta}$ from \eqref{eq:Defldelta}, which, with our choice $\delta = \delta_{\mu}$ from \eqref{eq:def_delta_mu} becomes
\begin{equation}\label{eq:def_elldelta_mu}
\ell_{\delta} = \frac{e^{\Gamma}}{2} \rho_{\mu}^{-1/2} Y^{-1/2},
\end{equation}
and corresponds to the so-called healing length.

\begin{lemma}\label{lem:gomegaapprox}
There exists a universal constant $C>0$ such that, if $R \ell^{-1} < C^{-1}$, we have
\begin{itemize}
\item $W_1$ can be approximated by $g$ up to the following error
\begin{equation}\label{eq:Wg_approx}
0 \leq W_1(x) -g(x) \leq C g(x) \frac{\min\{|x|^2,R^2\}}{\ell^2},
\end{equation}
and in particular $\Vert W_1 \Vert_{L^1} \leq 8\pi \delta (1 + C R^2 \ell^{-2})$ due to \eqref{eq:IntegralScattLength}.
\item For any $f \in L^1(\mathbb{R}^2)$ such that $f(x) = f(-x)$ and $\supp f \subseteq B(0,R)$,
\begin{equation}\label{eq:fchi_approx}
\bigg\vert f * \chi_{\Lambda}(x) - \chi_{\Lambda}(x) \int_{\mathbb{R}^2} \dd x\; f(x) \bigg\vert \leq C\max_{i,j} \|\partial_i \partial_j \chi\|_{\infty} \frac{R^2}{\ell^2} \|f\|_{L^1}.
\end{equation}
\item It also holds
\begin{equation}\label{eq:gw0.approx}
\bigg\vert \frac{1}{(2\pi)^2} \int_{\mathbb{R}^2} dk\, \frac{\widehat{W}_1(k)^2 - \widehat{W}_1^2(0) \one_{\{|k|\leq \ell_{\delta}^{-1}\}}}{2k^2} - (\widehat{g \omega})_0\bigg\vert \leq C\frac{R^2}{\ell^2}\delta.
\end{equation}
\item It holds
\begin{equation}
\bigg\vert \int_{\mathbb{R}^2} \frac{(\widehat{W}_1(k) - \widehat{g}_k)^2 - (\widehat{W}_1(0) - \widehat{g}_0)^2\one_{\{|k| \leq \ell_{\delta}^{-1}\}}}{2 k^2} \dd k \bigg\vert \leq C\frac{R^4}{\ell^4} \widehat{(g \omega)}_0.
\end{equation}
\end{itemize}
\end{lemma}

\begin{proof}
For \eqref{eq:Wg_approx} we use that the support of $g$ is contained in the set $\{|x|<R\}$, therefore it is enough to give here our estimate. Using the symmetries of $\chi$, the normalization $\Vert \chi \Vert_2 = 1$ (Appendix~\ref{app:locfunction}) and a Taylor expansion we see that  
\begin{equation*}
\left|1 - \frac{1}{\chi * \chi (x/\ell)} \right| \leq  \frac{1}{|\chi * \chi (x/\ell)|} \left| \int_{\mathbb{R}^2} \chi(y) [\chi(y) - \chi(x/\ell-y)] \right|\dd y \leq C \frac{|x|^2}{\ell^2} \max_{i,j}\|\partial_i \partial_j \chi\|_{\infty},
\end{equation*}
which implies the first bound. \eqref{eq:fchi_approx} is proved similarily. For the bound \eqref{eq:gw0.approx}, by Lemma~\ref{lem:FT-w} we know that
\begin{equation}
(\widehat{g \omega})_0 = \frac{1}{(2\pi)^2} \int_{\mathbb{R}^2} \frac{\widehat{g}_k^2 - \widehat{g}_0^2 \one_{\{|k|\leq \ell_{\delta}^{-1}\}}}{2k^2}\,\dd k,
\end{equation}
and using \eqref{eq:convolutionLog} for both the expressions of $W$ and $g$ we get
\begin{align*}
&\frac{1}{(2\pi)^2}\bigg\vert\int_{\mathbb{R}^2} \, \frac{ \widehat{g}^2_k - \widehat{g}_0^2 \one_{\{|k|\leq \ell_{\delta}^{-1}\}} - \widehat{W}_1^2(k) + \widehat{W}_1^2(0)\one_{\{|k|\leq \ell_{\delta}^{-1}\}}}{2k^2} \dd k \bigg\vert \\
&\leq - C\iint |g(x) g(y) -W_1(x) W_1(y)|\log \Big( \frac{x-y}{\ell_{\delta}} \Big) \dd x\dd y \\
&\leq - \frac{C}{\ell^2} \iint |x|^2 g(x) g(y) \log \Big( \frac{x-y}{\ell_{\delta}}\Big) \dd x\dd y\\
&= \frac{C}{\ell^2} \int |x|^2 g(x) \omega(x)  \dd x \leq C\frac{R^2}{\ell^2} \delta,
\end{align*}
where we used first \eqref{eq:Wg_approx}, then the fact that in $2$ dimension the $\log$ term produces a convolution with the Green's function of the Laplacian and finally formulas \eqref{eq:ScatOmega} and \eqref{eq:IntegralScattLength} (together with the bounds $\omega \leq 1$ in the support of $g$ and $2R < \ell_{\delta}$). The last inequality has a similar proof and is omitted.
\end{proof}

We now give a result of localization to large boxes for the potential part in the Hamiltonian \eqref{eq:grandcan_hamilton}.

\begin{lemma}[Localization of the potential]\label{lem:potentialloc}
The following identity holds
\begin{multline}
-\rho_{\mu}\sum_{j=1}^N  \int_{\mathbb{R}^2} g(x_j-y) \dd y + \sum_{i<j} v(x_i-x_j) \\
= \int_{\mathbb{R}^2}  \bigg[ -\rho_{\mu} \sum_{j=1}^N  \int_{\mathbb{R}^2}  w_{1,u}(x_j,y) \dd y\, + \,\sum_{i<j} w_u(x_i,x_j) \bigg] \dd u.
\end{multline}
\end{lemma}

\begin{proof}
It is proven by direct calculation following the same lines as \cite[Proposition 3.1]{BSol}.
\end{proof}

Therefore, joining the results from Lemmas~\ref{lem:LocKinEnLarge},~\ref{lem:potentialloc} and introducing the large box Hamiltonian acting on $\mathscr{F}_s(L^2(\Lambda_u))$ as
\begin{equation}
\mathcal{H}_{\Lambda_u}(\rho_{\mu})_N := \sum_{j=1}^N \mathcal{T}^{(j)}_u   -\rho_{\mu} \sum_{j=1}^N  \int_{\mathbb{R}^2} w_{1,u}(x_j,y) \dd y \, + \,\sum_{i<j} w_u(x_i,x_j),
\end{equation}
where $\mathcal{T}^{(j)}_u$ is \eqref{eq:Texplicitexpress} for the $x_j$ variable, and the ground state energy and its density 
\begin{equation}
E_{\Lambda} (\rho_{\mu}) :=  \inf \mathrm{Spec} (\mathcal{H}_{\Lambda}(\rho_{\mu})),\qquad  e_{\Lambda} (\rho_{\mu}) := \frac{1}{\ell^2}E_{\Lambda} (\rho_{\mu}) ,
\end{equation}
we are able to prove the following. Recall that $e_0$ is defined in \eqref{def:e0}.

\begin{lemma}\label{thm:control_e_0e_L}
Under the assumptions of Lemma~\ref{lem:LocKinEnLarge},
\begin{equation}
e_0(\rho_{\mu}) \geq e_{\Lambda}(\rho_{\mu}).
\end{equation}
\end{lemma}

\begin{proof}
By direct application of Lemma~\ref{lem:LocKinEnLarge} and Lemma~\ref{lem:potentialloc} we have
\begin{equation}
\mathcal{H}_{\rho_{\mu},N}(\rho_{\mu}) \geq \int_{\ell^{-1}(\Omega + B(0,\ell/2))}\mathcal{H}_{\Lambda_u}(\rho_{\mu})_N \dd u \geq \ell^{-2}|\Omega +  B(0,\ell/2)| E_{\Lambda} (\rho_{\mu}),
\end{equation}
where the last inequality is guaranteed by the unitary equivalence $\mathcal{H}_{\Lambda_u} \cong \mathcal{H}_{\Lambda_{u'}}$ via the relation
\begin{equation}
w_{u'}(x,y) = w_u (x- \ell(u'-u), y - \ell(u'-u)).
\end{equation}
The proof is concluded taking the infimum of the spectrum of the left-hand side and dividing by $|\Omega|$ observing that in the thermodynamical limit 
\begin{equation}
\frac{\vert \Omega+ B(0,\ell/2)\vert }{|\Omega|} \xrightarrow[|\Omega| \rightarrow + \infty]{} 1.
\end{equation}
\end{proof}

Therefore, Lemma~\ref{thm:control_e_0e_L} shows that in order to prove our main result Theorem~\ref{thm:grancanonicaltherm}, it is enough to give an analogous estimate on the Hamiltonian on the large box, and it is the content of the next theorem.
\begin{theorem}\label{thm:largebox_lower}
There exist $C, \eta >0$ such that the following holds. Let $\rho_\mu >0$ and $v\in L^1(\Omega)$ be a positive, spherically symmetric potential with scattering length $a$ and $\supp(v) \subset B(0,R)$ such that $\Vert v \Vert_1 \leq  Y^{-1/8}$ and $R \leq \rho_\mu^{-1/2}$. %Assume the parameters are chosen as in Appendix~\ref{app:parameters}. 
Then, if $\rho_\mu a^2 \leq C^{-1}$, and the parameters are chosen as in Appendix~\ref{app:parameters}, we have
\begin{equation}
E_{\Lambda}(\rho_{\mu}) \geq - 4 \pi  \ell^2  \rho_{\mu}^2 \delta  \Big(1 - \Big( 2 \Gamma + \frac 1 2 + \log \pi \Big) \delta \Big) - C  \ell^2 \rho_{\mu}^2  \delta^{2 + \eta}. 
\end{equation}
\end{theorem}

The proof of Theorem~\ref{thm:largebox_lower} is given in the remaining sections of the article. 
\section{Lower bounds in position space}
\subsection{Splitting of the Potential}\label{subsec:splitpot}

By the definitions \eqref{def:PQ} of the projectors $P$ and $Q$, we see that we can split the potential in a way presented in the lemma below.
\begin{lemma}\label{lem:SF_potsplit-bigbox}
We have, recalling the definitions in \eqref{eq:SF_w12u}, that 
\begin{align} \label{eq:potsplit-bigbox}
-\rho_{\mu} \sum_{j=1}^N  \int_{\mathbb{R}^2}  w_{1}(x_j,y) \dd y + \,\frac{1}{2}\sum_{i\neq j} w(x_i,x_j) &= \sum_{j=0}^4 {\mathcal Q}_j^{\rm ren}
\end{align}
with
\begin{align}
0 \leq {\mathcal Q}_4^{\rm ren}&:=
\frac{1}{2} \sum_{i\neq j} \Big[ Q_i Q_j + \left(P_i P_j + P_i Q_j + Q_i P _j\right)\omega(x_i-x_j) \Big] w(x_i,x_j) \nonumber \\
&\,\qquad \qquad \times
\Big[ Q_j Q_i + \omega(x_i-x_j) \left(P_j P_i + P_j Q_i + Q_j P_i\right)\Big],\label{eq:SF_DefQ4}\\
{\mathcal Q}_3^{\rm ren}&:=
\sum_{i\neq j} P_i Q_j w_1(x_i,x_j) Q_i Q_j + h.c.
 \label{eq:SF_DefQ3},
 \end{align}
as well as
 \begin{align}
{\mathcal Q}_2^{\rm ren}&:=
\sum_{i\neq j} P_i Q_j w_2(x_i,x_j) Q_i P_j  + \sum_{i\neq j} P_i Q_j w_2(x_i,x_j) P_i Q_j  \\
&\quad+\frac{1}{2}\sum_{i\neq j} P_iP_j w_1(x_i,x_j) Q_i Q_j + h.c.-\rho_{\mu} \sum_{j=1}^N Q_i \int_{\mathbb{R}^2}  w_1(x_i,y) \dd y \, Q_i,
 \label{eq:SF_DefQ2}\\
{\mathcal Q}_1^{\rm ren}&:=\sum_{i,j} Q_i P_j w_2(x_i,x_j) P_i P_j - \rho_{\mu}\sum_{i=1}^N Q_i \int_{\mathbb{R}^2}  w_1(x_i,y) \dd y \, P_i + h.c.,  \label{eq:SF_DefQ1}
\intertext{and}
{\mathcal Q}_0^{\rm ren}&:= \frac{1}{2} \sum_{i \neq j} P_i P_j w_2(x_i,x_j) P_i P_j - \rho_{\mu} \sum_{j=1}^N P_j \int_{\mathbb{R}^2}   w_1(x_j,y) \dd y \, P_j. \label{eq:SF_DefQ0}
\end{align}
\end{lemma}

\begin{proof}
It follows from an elementary calculation, using that $P+Q=\one$ on $L^2(\Lambda)$ and, where needed, the identity
\begin{equation}
w_1 = w_2 -w \omega + w \omega^2.
\end{equation}
\end{proof}

We rewrite now some of the previous $Q$ terms thanks to the lemma below.

\begin{lemma}\label{lem:Qrewriting}
With the notation $\rho_0 = \frac{n_0}{\ell^2}$ we have
\begin{align}
\mathcal{Q}_0^{\rm ren} &= \frac{n_0(n_0 -1)}{2|\Lambda|} (\widehat{g}_0 + \widehat{g\omega}(0)) -\rho_{\mu} n_0 \widehat{g}(0),\\
\mathcal{Q}_1^{\rm ren} &= \Big( \frac{n_0}{|\Lambda|} - \rho_{\mu}\Big) \sum_{i=1}^N Q_i \chi_{\Lambda}(x_i) W_1 *\chi_{\Lambda}(x_i) P_i + h.c. \nonumber \\
&\quad+\frac{n_0}{|\Lambda|} \sum_{i=1}^N Q_i \chi_{\Lambda}(x_i) ((W_1 \omega) * \chi_{\Lambda})(x_i) P_i + h.c.,
\intertext{and}
\mathcal{Q}_2^{\rm ren} &\geq \sum_{i \neq j} P_i Q_j w_2(x_i,x_j) Q_i P_j + \frac{1}{2} \sum_{i \neq j} (P_i P_j w_1(x_i,x_j) Q_i Q_j + h.c.)  \nonumber \\
&\quad+ ((\rho_0 - \rho_{\mu})\widehat{W}_1(0) + \rho_0 \widehat{W_1 \omega}(0)) \sum_{j=1}^N Q_j \chi_{\Lambda}(x_j)^2 Q_j - C (\rho_{\mu} + \rho_0) \delta \Big(\frac{R}{\ell}\Big)^2 n _+.
\end{align}
\end{lemma}

\begin{proof}
The first two identities are straightforward after having observed that 
\begin{equation}\label{eq:PiPiformula}
\sum_{j=1}^N P_j w_1(x_i,x_j) P_j = \frac{1}{\ell^2}\sum_{j=1}^N P_j \int_{\Lambda}  w_1(x_i,y) \dd y= \frac{n_0}{|\Lambda|} \int_{\Lambda} w_1(x_i,y) \dd y,
\end{equation}
and
\begin{equation}\label{eq:w1integralcalc}
\int_{\Lambda}   w_1(x_i,y)dy = \chi_{\Lambda}(x_i) (W_1 * \chi_{\Lambda})(x_i).
\end{equation}
For the $\mathcal{Q}_2^{\rm ren}$ term, the only parts which require a different approach are
\begin{equation}
(\rho_0 - \rho_{\mu}) \sum_{j=1}^N Q_j \chi_{\Lambda}(x_j) W_1 * \chi_{\Lambda}(x_j) Q_j + \rho_0 \sum_{j=1}^N Q_j \chi_{\Lambda}(x_j) ((W_1\omega)* \chi_{\Lambda})(x_j) Q_j.
\end{equation} 
Using Lemma~\ref{lem:gomegaapprox} we can bound
\begin{align}
\sum_{j=1}^N Q_j \chi_{\Lambda}(x_i) W_1 * \chi_{\Lambda}(x_i) Q_j &\geq  \|W_1\|_{L^1} \sum_{j=1}^N Q_j  \chi_{\Lambda}(x_j)^2 Q_j \nonumber \\ &\quad - C\max_{i,j}\|\partial_i \partial_j \chi\|_{\infty} \frac{R^2}{\ell^2 }\|W_1\|_{L^1} \|\chi\|_{\infty} n_+.
\end{align}
Recalling that $\|W_1\|_{L^1} \leq C \delta$ (Lemma~\ref{lem:gomegaapprox}) and acting similarly for the other term, this concludes the proof.
\end{proof}

As a direct consequence of the lemma above, we can derive the following first lower bound for the large box Hamiltonian.

\begin{corollary}\label{cor:H'expression}
The following bound holds for the Hamiltonian in the large box
\begin{align}
\mathcal{H}_{\Lambda}(\rho_{\mu})\big|_N &\geq \sum_{j=1}^N \mathcal{T}^{(j)}   +  \frac{n_0(n_0 -1)}{2|\Lambda|} (\widehat{g}_0 + \widehat{g\omega}(0)) -\rho_{\mu} n_0 \widehat{g}_0 \label{eq:kin+lead} \\
&\quad+ \Big( \frac{n_0}{|\Lambda|} - \rho_{\mu}\Big) \sum_{i=1}^N Q_i \chi_{\Lambda}(x_i) W_1 *\chi_{\Lambda}(x_i) P_i + h.c. \label{eq:expresschar1} \\ 
&\quad+\frac{n_0}{|\Lambda|} \sum_{i=1}^N Q_i \chi_{\Lambda}(x_i) ((W_1 \omega) * \chi_{\Lambda})(x_i) P_i + h.c.  \label{eq:expresschar2} \\
&\quad+\sum_{i \neq j} P_i Q_j w_2(x_i,x_j) Q_i P_j + \frac{1}{2} \sum_{i \neq j} (P_i P_j w_1(x_i,x_j) Q_i Q_j + h.c.)  \label{eq:expressQ2} \\
&\quad+ ((\rho_0 - \rho_{\mu})\widehat{W}_1(0) + \rho_0 \widehat{W_1 \omega}(0)) \sum_{j=1}^N Q_j \chi_{\Lambda}(x_j)^2 Q_j  \label{eq:expressQ2char} \\
&\quad- C (\rho_{\mu} + \rho_0) \delta \Big(\frac{R}{\ell}\Big)^2 n _+ + \mathcal Q_3^{\text{ren}} + \mathcal Q_4^{\text{ren}}. \label{eq:expressQ3Q4}
\end{align}

\end{corollary}

In the lemma below we prove an estimate which is going to be useful in Section~\ref{sec:Q3loc} to localize the $\mathcal Q_3^{\text{ren}}$ term.

\begin{lemma}\label{lem:PQ'estimate}
Let $Q'$ be a possibly non self-adjoint operator on $L^2(\Lambda)$ such that $Q Q'= Q'$ and $\|Q'\| \leq 1$. Then for all $c \in (0,1)$ there is a $C >0$ such that, if $R \leq \ell$,
\begin{align*}
 \sum_{i \neq j} (P_i Q'_{j} w_1(x_i,x_j)&Q_i Q_j + h.c.) \geq   - \frac{1}{4} \mathcal Q_4^{\text{ren}} - \sum_{i \neq j} (P_i Q'_{j} w_1 \omega P_i P_j + h.c.) \\
 &- \delta n_0\Big( c K_{\ell}^{-2}  \frac{n_+}{\ell^2} +  C \frac{ K_{\ell}^{2}}{\ell^2}\sum_{j=1}^N Q_j' (Q_j')^{\dagger}  \Big).
\end{align*}
\end{lemma}

\begin{proof}
The idea is to reobtain the $Q_4$ term in the inequalities.
\begin{align}
 \sum_{i \neq j} (P_i Q'_{j} w_1 Q_i Q_j + h.c.)  &=   \sum_{i \neq j} P_i Q'_{j} w_1 \left[Q_i Q_j + \omega (P_iP_j +P_iQ_j + Q_i P_j) \right] + h.c. \nonumber\\
 &\quad -  \sum_{i \neq j} P_i Q'_{j} w_1\omega (P_iP_j +P_iQ_j + Q_i P_j) + h.c. \label{eq:Q'estim2}
\end{align}
We use Cauchy-Schwarz inequality on both the terms on the right-hand side. The first line of \eqref{eq:Q'estim2}, using that $w_1 \leq w$, is controlled by
\begin{align*}
C \sum_{i \neq j} P_i Q'_{j} w_1 (P_i Q'_{j})^{\dagger} + \frac{1}{4} \mathcal Q_4^{\text{ren}}
&= C\frac{n_0}{\ell^2} \sum_{j=1}^N Q'_{j}  \chi_{\Lambda}(x_j)(W_1 * \chi_{\Lambda})(x_j) (Q'_{j})^{\dagger} + \frac{1}{4} \mathcal Q_4^{\text{ren}} \\
&\leq C \frac{n_0}{\ell^2} \|\chi_{\Lambda}\|_{\infty}^2 \delta \sum_{j=1}^N Q_j'(Q_j')^{\dagger} + \frac{1}{4} \mathcal Q_4^{\text{ren}},
\end{align*}
where we used \eqref{eq:PiPiformula}, \eqref{eq:w1integralcalc}, the bound $\|W_1\|_{L^1} \leq C \delta( 1 + R^2 \ell^{-2} )$ and $R \leq \ell$.
For the second line of \eqref{eq:Q'estim2} we keep the $PP$ contribution and treat the other terms separately. 
They can be estimated as above. For instance,
\begin{align}
\sum_{i\neq j} ( P_i Q'_j w_1 \omega P_i Q_j  + h.c.) &\leq
\varepsilon^{-1}\sum_{i\neq j} P_i Q'_j w_1 \omega (P_i Q'_j)^\dagger + \varepsilon \sum_{i \neq j } P_i Q_j w_1 \omega P_i Q_j  \nonumber \\
&\leq C\delta \frac{n_0}{\ell^2} \Big( \varepsilon^{-1}\sum_{j=1}^N Q_j'  (Q_j')^{\dagger} + \varepsilon  n_+\Big), 
\end{align}
where we used the Cauchy-Schwarz inequality with weight $\varepsilon >0$. Choosing $\varepsilon = cC^{-1} K_{\ell}^{-2}$ with $c \in (0,1)$, we get
\begin{equation}
\sum_{i\neq j} P_i Q'_j w_1 \omega P_i Q_j \leq c^{-1} C^2\delta\frac{n_0 d^2 K_{\ell}^{2}}{(d\ell)^2} \sum_{j=1}^N Q_j'  (Q_j')^{\dagger} + c \delta n_0 K_{\ell}^{-2}\frac{n_+}{\ell^2}, 
\end{equation}
and the Lemma follows.
\end{proof}

\subsection{Localization of $3Q$ term}\label{sec:Q3loc}

In this section we show how we can restrict the action of one of the $Q$ projectors in the $\mathcal Q^{\text{ren}}_3$ term to low momenta. More precisely we define the following two sets of low and high momenta respectively,
\begin{equation}\label{eq:defPL_PH}
\mathcal{P}_L :=\{p \in \mathbb{R}^2\,|\; |p| \leq d^{-2}\ell^{-1}\}, \qquad \mathcal{P}_H :=\{p \in \mathbb{R}^2\,|\; |p| \geq K_H \ell^{-1}\}. 
\end{equation}
We choose the parameters $d$ and $K_H$ satisfying \eqref{eq:rel_d_k_H} so that the two sets are disjoint. 
We will localize the $Q$ projector using the following cutoff function,
\begin{equation}
f_L(r) := f (d^2 \ell r ), \qquad f(r) := \begin{cases}
1, &\text{if} \quad r \leq 1,\\
0, &\text{if} \quad r \geq 2,
\end{cases} 
\end{equation}
where $f \in C^{\infty}(\mathbb{R})$ is a non-increasing function. The localized projectors are
\begin{equation}
Q_L := Q f_L(\sqrt{-\Delta}), \qquad \overline{Q}_L := Q - Q_L, 
\end{equation}
and the localized version of $\mathcal Q_3^{\text{ren}}$ is
\begin{equation}\label{eq:defQ3low}
\mathcal Q^{\text{low}}_3 := \sum_{i \neq j} (P_i Q_{L,j} w_1(x_i,x_j)Q_i Q_j + h.c.).
\end{equation}
The number of high excitations, namely the number of bosons outside from the condensate and with momenta not in $\mathcal{P}_L$ is
\begin{equation}
n_+^H := \sum_{j=1}^N Q_j \, \one_{(d^{-2}\ell^{-1} , \infty)}(\sqrt{-\Delta_j}) Q_j.
\end{equation}
It is easy to see that 
\begin{equation}\label{eq:QnHestimate}
\sum_{j=1}^N \overline{Q}_{L,j}\overline{Q}_{L,j}^{\dagger} \leq n_+^H.
\end{equation}
The next lemma shows how the $\mathcal Q_3^{\text{ren}}$ term, with a small contribution from $ \mathcal Q_4^{\text{ren}}$ and the spectral gap from the kinetic energy (see \eqref{eq:Texplicitexpress}), can be approximated by $\mathcal Q_3^{\text{low}}$.

\begin{lemma}\label{lem:Q3loc}
Assume $R \leq \ell$ and the relation \eqref{eq:rel_locQ3low2} between the parameters. Then there exists $C>0$ such that, for any $n$-particles state $\Psi \in \mathscr{F}_s(L^2(\Lambda))$ with $n \leq 2 \rho_\mu \ell^2$,
\begin{equation*}
\langle \mathcal Q_3^{\text{ren}}\rangle_\Psi  + \frac{1}{4} \langle \mathcal Q_4^{\text{ren}} \rangle_\Psi + \frac{b}{100}\Big( \frac{\langle n_+ \rangle_\Psi}{\ell^2} + \varepsilon_T\frac{\langle n_+^H \rangle_\Psi}{(d\ell)^2}\Big) \geq \langle \mathcal Q_3^{\text{low}}\rangle_\Psi - C \delta \frac{n^2}{\ell^2} (d^{2M-2} + R^2 \ell^{-2}).
\end{equation*}
\end{lemma}

\begin{proof}
By definition
\begin{equation}
\mathcal Q_3^{\text{ren}} - \mathcal Q_3^{\text{low}} = \sum_{i \neq j} (P_i \overline{Q}_{L,j} w_1(x_i,x_j)Q_i Q_j + h.c.).
\end{equation}
We use now Lemma~\ref{lem:PQ'estimate} with $Q' = \overline{Q}_L$ and the estimate \eqref{eq:QnHestimate} to get 
\begin{align}
\mathcal Q_3^{\text{ren}} - \mathcal Q_3^{\text{low}} &\geq
 -\frac{1}{4} \mathcal Q_4^{\text{ren}} - \sum_{i \neq j} (P_i \overline{Q}_{L,j} w_1 \omega P_i P_j + h.c.)
 \nonumber \\
&\quad- \delta n_0\Big( c K_{\ell}^{-2}  \frac{n_+}{\ell^2} +   C \frac{d^2 K_{\ell}^{2}}{(d\ell)^2}n_+^H \Big). \label{eq:partialerrorQ3low}
\end{align}
By \eqref{eq:PiPiformula} we have 
\begin{align}
\sum_{i \neq j} &(P_i \overline{Q}_{L,j} w_1 \omega P_i P_j + h.c.)
\nonumber \\ &   
\leq \frac{n_0}{\ell^2} \bigg(\sum_{j=1}^N \overline{Q}_{L,j} \chi_{\Lambda}(x_j) \big( \Vert W_1 \omega \Vert_{L^1}  \chi_\Lambda(x_j) + \varepsilon(x_j) \big) P_j + h.c. \bigg) \label{ineq:PQPP}
\end{align} 
with $\varepsilon(x_j) = W_1 \omega * \chi_\Lambda (x_j) - \Vert W_1 \omega \Vert_{L^1} \chi_\Lambda(x_j)$. The $\varepsilon(x_j)$-term we can bound using Cauchy-Schwarz and \eqref{eq:fchi_approx},
\begin{align}
 \frac{n_0}{\ell^2} \bigg(\sum_{j=1}^N \overline{Q}_{L,j} \chi_{\Lambda}(x_j) \varepsilon(x_j)  P_j + h.c. \bigg) &\leq C \frac{n_0}{\ell^2} \sum_{j=1}^N ( \overline{Q}_{L,j} \chi_\Lambda \varepsilon \overline{Q}_{L,j}^\dagger + P_j \chi_\Lambda \varepsilon P_j ) \nonumber \\
 &\leq C \frac{n_0 R^2}{\ell^4} \delta (n_+^H + n_0). \label{eq:090822.2}
\end{align}
For the other term we take $M-1 \leq 2\widetilde{M} \leq M $ and using the notation $D_M := (\ell^{-2}-\Delta_j)^{\widetilde{M}}$, we write
\begin{equation}
 \overline{Q}_{L,j} \chi_{\Lambda}(x_j)^2 P_j + h.c. =  \overline{Q}_{L,j} D_M^{-1} [D_M\chi_{\Lambda}(x_j)^2] P_j + h.c.
\end{equation}
Therefore, by Cauchy-Schwarz inequality with weight $\varepsilon_0 >0$,
\begin{align*}
\overline{Q}_{L,j} \chi_{\Lambda}(x_j)^2 P_j + h.c. \leq \varepsilon_0 P_j + \varepsilon_0^{-1} \|D_M \chi_{\Lambda}^2\|^2_{\infty}  \overline{Q}_{L,j}  D_M^{-2} (\overline{Q}_{L,j})^{\dagger}.
\end{align*}
Now using that $\|D_M \chi^2_{\Lambda}\| \leq C \ell^{-2\widetilde{M}}$ and that $\overline{Q}_{L}$ cut momenta lower than $d^{-2}\ell^{-1}$ we obtain
\begin{align}
 \overline{Q}_{L,j} \chi_{\Lambda}(x_j)^2 P_j + h.c. &\leq \varepsilon_0 P_j + \varepsilon^{-1}_0 C  \ell^{-4\widetilde{M}} \overline{Q}_{L,j}(\ell^{-2}-\Delta_j)^{-2\widetilde{M}}  (\overline{Q}_{L,j})^{\dagger}  \nonumber\\
&\leq  \varepsilon_0 P_j + \varepsilon_0^{-1} C  d^{8 \widetilde{M}}  \overline{Q}_{L,j} (\overline{Q}_{L,j})^{\dagger}. 
\end{align} 
Therefore choosing $\varepsilon_0 =  d^{4 \widetilde{M}}$, we have
\begin{align}\label{eq:090822.3}
 \frac{n_0}{\ell^2} \bigg(\sum_{j=1}^N \overline{Q}_{L,j} \chi_{\Lambda}(x_j)^2 \Vert W_1 \omega \Vert_{L^1}   P_j + h.c. \bigg) \leq C \delta d^{2M-2} \frac{n_0 }{\ell^2}(n_+^H + n_0).
\end{align}
Inserting \eqref{eq:090822.2} and \eqref{eq:090822.3} into \eqref{ineq:PQPP} we find
\begin{equation}
\sum_{i \neq j} (P_i \overline{Q}_{L,j} w_1 \omega P_i P_j + h.c.) \leq C \delta \frac{n_0}{\ell^2} ( n_+^H + n_0) (d^{2M-2} + R^2 \ell^{-2}).
\end{equation}
We use this last bound in \eqref{eq:partialerrorQ3low} and apply it to the state $\Psi$,
\begin{align}
\langle &\mathcal Q_3^{\text{ren}}\rangle_\Psi - \langle \mathcal Q_3^{\text{low}} \rangle_\Psi \nonumber \\
&\geq
 -\frac{1}{4} \langle \mathcal Q_4^{\text{ren}} \rangle_\Psi - C \delta \frac{n}{\ell^2} ( \langle n_+^H \rangle_\Psi + n ) (d^{2M-2} + R^2 \ell^{-2}) 
 - c \delta n K_\ell^{-2} \frac{\langle n_+ \rangle_\Psi}{\ell^2} - C \delta n d^2 K_\ell^2 \frac{\langle n_+^H \rangle_\Psi}{(d\ell)^2} \nonumber \\
 & \geq  -\frac{1}{4} \langle \mathcal Q_4^{\text{ren}} \rangle_\Psi -  C \delta \frac{n^2}{\ell^2} (d^{2M-2} + R^2 \ell^{-2}) 
  - c \frac{\langle n_+ \rangle_\Psi}{\ell^2} -  C d^2 K_\ell^4 \frac{\langle n_+^H \rangle_\Psi}{(d\ell)^2},
\end{align}
where we used $n \leq 2 \rho_\mu \ell^2$ and $\ell^2 = K_\ell^2 \rho_\mu^{-1} Y^{-1}$. We conclude by choosing $c = \frac{b}{100}$ and using the relation \eqref{eq:rel_locQ3low2} between the parameters.
\end{proof}

\subsection{A priori bounds and localization of the number of excitations}\label{sec:apriori}
The purpose of this section is to get bounds on the number of excitations of the system. First of all, Theorem~\ref{thm:excitationsbound} gives a priori bounds on $n_+$. 
\begin{theorem}\label{thm:excitationsbound}
There exists a universal constant $C >0$ such that, if $\Psi \in \mathscr{F}_s(L^2(\Lambda))$ is a normalized $n-$bosons vector which satisfies 
\begin{equation}\label{eq:condensationaprioricondition}
\langle \mathcal{H}_{\Lambda}(\rho_{\mu})\rangle_{\Psi} \leq -4 \pi \rho_{\mu}^2 \ell^2 Y \left( 1 - C K_B^2 Y \vert \log Y \vert \right),
\end{equation}
with $K_B \gg 1$, and the other included parameters fixed in Appendix~\ref{app:parameters},
then
\begin{align}
\langle n_+\rangle_{\Psi} &\leq C K_B^2 K_{\ell}^2 \rho_{\mu} \ell^2   Y |\log Y|, \label{eq:condensationestimate}\\
\langle \mathcal Q_4^{\text{ren}}\rangle_{\Psi} &\leq C K_B^2 K_{\ell}^2 \rho_{\mu}^2 \ell^2 Y^2 |\log Y|,  \label{eq:estimateprioriQ4}\\
\left| \rho_{\mu} - \frac{n}{\ell^2}\right| &\leq C  K_B K_{\ell}\rho_{\mu} Y^{1/2} |\log Y|^{1/2}.\label{eq:estimatepriori.n}
\end{align}
\end{theorem}

\begin{proof}
It is proved in Appendix~\ref{app:low-smallbox}, using a second localization to "small boxes" of size $\ll \ell_\delta$.
\end{proof}

We also need to bound the number of low excitations, defined in terms of our modified kinetic energy $\mathcal T$. More precisely we define, for a certain $\widetilde{K}_H \gg 1$ fixed in Appendix~\ref{app:parameters}, the projectors
\begin{equation}\label{def:QHbar}
\overline{Q}_H = \one_{(0, \widetilde{K}_H^2 \ell^{-2})} (\mathcal{T}), \qquad Q_H = Q- \overline{Q}_H,
\end{equation}
which satisfy
\begin{equation}
P + \overline{Q}_H + Q_H = 1_{\Lambda}.
\end{equation}
We will consider the operators
\begin{equation}
n_+^L := \sum_j \overline{Q}_{H,j}, \qquad \widetilde{n}_+^H := \sum_{j} Q_{H,j},
\end{equation}
for which we prove the following result.
\begin{theorem}[Restriction on $n_+^L$]\label{thm:excitationrestriction}
There exist $C$, $\eta > 0$ such that the following holds. Let $\Psi \in \mathscr{F}_s(L^2(\Lambda))$ be a normalized $n$-particle vector which satisfies \eqref{eq:condensationaprioricondition}. Assume that the potential $v$ is such that $\Vert v \Vert_1 \leq Y^{-1/8}$. Then, for $\mathcal{M} \gg 1$ satisfying condition \eqref{eq:M_large_matrices} there exists a sequence $\{\Psi^{m}\}_{m \in \mathbb{Z}} \subseteq \mathscr{F}_s(L^2(\Lambda))$ such that $\sum_m \Vert \Psi^m \Vert^2 = 1$ and
\begin{equation}\label{eq:large_matrices_restrictioncondition}
\Psi^{m} = \one_{[0,\frac{\mathcal{M}}{2} +m]}(n_+^L) \Psi^{m},
\end{equation}
and such that the following lower bound holds true
\begin{align*}
\langle \Psi, \mathcal{H}_{\Lambda}(\rho_{\mu}) \Psi \rangle  \geq&  \sum_{2 |m|\leq \mathcal{M} }\langle \Psi^{m}, \mathcal{H}_{\Lambda}(\rho_{\mu}) \Psi^m\rangle -C \rho_{\mu}^2 \ell^2 Y^{2+\eta} \\
&  -4 \pi \rho_{\mu}^2 \ell^2 Y \Big( 1 - C K_B^2 Y \vert \log Y \vert \Big)\sum_{2|m| >\mathcal{M}} \|\Psi^m\|^2.  
\end{align*}
\end{theorem}

From this result we see that, in order to prove Theorem~\ref{thm:largebox_lower}, we only need to bound the energy of states satisfying the bound $n_+^L \leq  \mathcal M$.
In the remainder of this section, we prove Theorem~\ref{thm:excitationrestriction}.
The following lemma states that for a lower bound we can restrict to states with finitely many excitations $n_+$, up to small enough errors. The proof of this lemma is inspired by the localization of large matrices result in \cite{LS_2comp}. It is also really similar to the bounds in \cite[Proposition 21]{HST}. It could be interpreted as an analogue of the standard IMS localization formula. The error produced is written in terms of the following quantities $d_1^L$ and $d_2^L$:
\begin{align}
d_{1}^L &:= -\rho_{\mu} \sum_{i} (P_i + Q_{H,i})\int  w_1(x_i,y) \dd y \; \overline{Q}_{H,i} + h.c. \nonumber \\
&\quad +\sum_{i \neq j} (P_i + Q_{H,i})\overline{Q}_{H,j} w(x_i,x_j) \overline{Q}_{H,i} \overline{Q}_{H,j} + h.c. \nonumber \\
&\quad + \sum_{i \neq j} \overline{Q}_{H,i} (P_j + Q_{H,j}) w(x_i,x_j) (P_i + Q_{H,i})(P_j + Q_{H,j}) + h.c. \label{def.d1L} 
\intertext{and}
d_2^L &:= \sum_{i \neq j} (P_i + Q_{H,i})(P_j + Q_{H,j}) w(x_i,x_j) \overline{Q}_{H,j} \overline{Q}_{H,i} +h.c. \label{def.d2L}
\end{align}

\begin{lemma}\label{lem:localization_largeMat}
Let $\theta : \R \rightarrow [0,1]$ be any compactly supported Lipschitz function such that $\theta(s) = 1$ for $\vert s \vert < \frac 1 8$ and $\theta(s) = 0$ for $\vert s \vert > \frac 1 4$. For any $\mathcal M >0$, define $c_{\mathcal M} >0$  and $\theta_{\mathcal M}$ such that
\[ \theta_{\mathcal M}(s) = c_{\mathcal M} \theta \Big( \frac{s}{\mathcal M} \Big) , \qquad \sum_{s \in \Z} \theta_{\mathcal M}(s)^2 = 1 .\]
Then there exists a $C>0$ depending only on $\theta$ such that, for any normalized state $\Psi \in \mathscr{F}_s(L^2(\Lambda))$,
\begin{equation}\label{eq:large_matrices_decomposn+}
 \langle \Psi , \mathcal H_{\Lambda}(\rho_\mu) \Psi \rangle \geq \sum_{m \in \Z} \langle \Psi^m , \mathcal H_\Lambda(\rho_\mu) \Psi^m \rangle - \frac{C}{\mathcal{M}^2} \left( \vert \langle d_1^L \rangle_\Psi \vert + \vert \langle d_2^L \rangle_\Psi \vert \right),
 \end{equation}
where $\Psi^m = \theta_{\mathcal M}(n_+^L - m) \Psi$.
\end{lemma}

\begin{proof}
Notice that $\mathcal{H}_{\Lambda}$ only contains terms that change $n_+^L$ by $0, \pm 1$ or $\pm 2$. Therefore, we write our operator as $ \mathcal H_\Lambda (\rho_\mu) = \sum_{\vert k \vert \leq 2} \mathcal{H}_k$,
with $ \mathcal H_k n_+^L = (n_+^L + k) \mathcal H_k$. Moreover, $\mathcal H_k + \mathcal H_{-k} = d^L_k$ for $k=1,2$. We use this decomposition to estimate the localized energy,
\begin{align*}
\sum_{m \in \Z} \langle \Psi^m, \mathcal H_{\Lambda}\Psi^m \rangle &= \sum_{m \in \Z} \sum_{\vert k \vert \leq 2} \langle \theta_{\mathcal M}(n_+^L -m)  \theta_{\mathcal M}(n_+^L-m+k) \Psi, \mathcal H_k \Psi \rangle\\
&= \sum_{m, s \in \Z} \sum_{\vert k \vert \leq 2} \langle \theta_{\mathcal M}(s-m)  \theta_{\mathcal M}(s-m+k) \one_{\{n_+^L =s\}} \Psi, \mathcal{H}_k \Psi \rangle\\
&= \sum_{m,s \in \Z} \sum_{\vert k \vert \leq 2} \theta_{\mathcal M}(m) \theta_{\mathcal M}(m+k) \langle \one_{\{n_+^L = s\}} \Psi, \mathcal H_k \Psi \rangle,
\end{align*}
where in the last line we changed the index $m$ into $s-m$. We can sum on $s$ to recognize
\begin{equation}
\sum_{m \in \Z} \langle \Psi^m, \mathcal H_{\Lambda} \Psi^m \rangle = \sum_{m \in \Z} \sum_{\vert k \vert \leq 2} \theta_{\mathcal M}(m) \theta_{\mathcal M}(m+k) \langle \Psi, \mathcal H_k  \Psi \rangle.
\end{equation}
Furthermore the energy of $\Psi$ can be rewritten as
\begin{equation}
\langle \Psi, \mathcal H_{\Lambda} \Psi \rangle = \sum_{\vert k \vert \leq 2} \langle \Psi, \mathcal H_k  \Psi \rangle = \sum_{m \in \Z} \sum_{\vert k \vert \leq 2} \theta_{\mathcal M}(m)^2 \langle \Psi, \mathcal H_k \Psi \rangle,
\end{equation}
by definition of $\theta_{\mathcal M}$. Thus, the localization error is
\begin{equation}
\sum_{m \in \Z} \langle \Psi^m, \mathcal H_{\Lambda} \Psi^m \rangle - \langle \Psi, \mathcal H_{\Lambda} \Psi \rangle = \sum_{\vert k \vert \leq 2} \delta_k \langle \Psi, \mathcal H_k \Psi \rangle,
\end{equation}
with 
\begin{equation}\label{eq.def.deltak}
\delta_k = \sum_{m \in \Z} \big (  \theta_{\mathcal M} (m)  \theta_{\mathcal M}(m+k) -  \theta_{\mathcal M}(m)^2  \big ) = - \frac{1}{2} \sum_m \big ( \theta_{\mathcal M}(m) -  \theta_{\mathcal M}(m+k) \big )^2.
\end{equation}
Since $\delta_0 = 0$, $\delta_k = \delta_{-k}$ and $d_k^L = \mathcal H_k + \mathcal H_{-k}$ we find
\begin{equation}
\sum_{m \in \Z} \langle \Psi^m, \mathcal H_{\Lambda} \Psi^m \rangle - \langle \Psi, \mathcal H_{\Lambda} \Psi \rangle = \delta_1 \langle d_1^L \rangle_\Psi + \delta_2 \langle d_2^L \rangle_\Psi,
\end{equation}
and only remains to prove that $\vert \delta_k \vert \leq C \mathcal M^{-2}$. Using \eqref{eq.def.deltak} and the definition of $\theta_{\mathcal M}$,
\begin{align}
\vert \delta_k \vert = \frac{c_{\mathcal M}^2}{2} \sum_{m \in \Z} \Big[ \theta \Big( \frac{m}{\mathcal M} \Big) - \theta \Big( \frac{m+k}{\mathcal M} \Big) \Big]^2.
\end{align}
We can restrict the sum to $ m  \in \big[- \frac{ \mathcal M}{2} , \frac{ \mathcal M}{2} \big]$, since the other terms vanish due to $\theta$ being a cutoff function. This sum contains $ \mathcal M + 1$ terms which we can bound using the Lipschitz constant $L$ of $\theta$,
\begin{equation}
\vert \delta_k \vert \leq c_{\mathcal M}^2\frac{ \mathcal M + 1}{2} \frac{L^2 k^2}{\mathcal M^2} \leq \frac{2 L^2 k^2}{\mathcal M^2},
\end{equation}
where in the last inequality we used
\begin{equation}
c_{\mathcal M}^2 = \Big( \sum_{s \in \Z} \theta \Big( \frac{s}{\mathcal M} \Big)^2 \Big) ^{-1} \leq \frac{1}{ \mathcal M / 4 + 1}.
\end{equation}
\end{proof}
To estimate the error in \eqref{eq:large_matrices_decomposn+}, we need the following bounds on $d_1^L$ and $d_2^L$.
\begin{lemma}\label{lem:d1d2estimate}
Let $\widetilde{\mathcal M} >0$ and $\Psi \in \mathscr{F}_s(L^2(\Lambda))$ be a normalized $n$-bosons vector satisfying 
\[\Psi = \one_{[0,\widetilde{\mathcal{M}}]}(n_+^L)\Psi.\]
Then, assuming  the choices of parameters in Appendix~\ref{app:parameters} we have
\begin{align}\label{eq:estimateMd1d2}
&\vert \langle d_{1}^L \rangle_\Psi \vert + \vert \langle d_2^L\rangle_{\Psi} \vert  \\
&\leq  \rho_{\mu}^2 \ell^2 \|v\|_1 \Big(  \frac{\langle n_+\rangle_{\Psi}^{1/2}}{n^{1/2}} + \frac{\widetilde{\mathcal{M}}^{1/2} \langle n_+\rangle_{\Psi}^{1/2}}{n} \varepsilon_N^{-1/4} \widetilde{K}_H +   \frac{\widetilde{\mathcal{M}} \langle n_+\rangle_{\Psi}}{n^2} \varepsilon_N^{-1/2} \widetilde{K}_H^2 \Big)  +C \langle \mathcal Q_4^{\rm{ren}}\rangle_{\Psi}.\nonumber
\end{align}
\end{lemma}
\begin{proof}
We give the proof in Appendix~\ref{app:proofd1d2}.
\end{proof}
Now we can combine Lemmas~\ref{lem:localization_largeMat},~\ref{lem:d1d2estimate} and Theorem~\ref{thm:excitationsbound} to prove Theorem~\ref{thm:excitationrestriction}.
\begin{proof}[Proof of Theorem~\ref{thm:excitationrestriction}]
Given a $n$-sector state $\Psi \in L^2(\Lambda^n)$ satisfying \eqref{eq:condensationaprioricondition}, we can apply Lemma~\ref{lem:localization_largeMat} and write $\Psi^m = \theta_{\mathcal M}(n_+^L -m)\Psi$. In \eqref{eq:large_matrices_decomposn+} we split the sum into two.
The first part, for $|m| < \frac 1 2 \mathcal M$, we keep. For $|m| > \frac 1 2 \mathcal{M}$, $\Psi_m$ satisfies
\begin{equation}
\langle n_+\rangle_{\Psi^m} \geq \langle n_+^L\rangle_{\Psi^m} \geq \frac{ \mathcal{M}}{4}\| \Psi^m \|^2 ,  
\end{equation}
due to the cutoff $\theta_{\mathcal M}(n_+^L - m)$. Thanks to condition \eqref{eq:M_large_matrices} on $\mathcal M$, this is a larger bound than \eqref{eq:condensationestimate}, and thus the assumption of Theorem~\ref{thm:excitationsbound} cannot be satisfied for $\Psi^m$ and we must have the lower bound 
\begin{equation}\label{eq:PsimH.01}
\langle \Psi^m, \mathcal{H}_{\Lambda}(\rho_{\mu}) \Psi^m\rangle \geq -4 \pi \rho_{\mu}^2 \ell^2 Y \Big( 1 - C K_B^2 Y \vert \log Y \vert \Big)\|\Psi^m \|^2.
\end{equation}
We finally bound the last term in \eqref{eq:large_matrices_decomposn+}, using Lemma~\ref{lem:d1d2estimate} with $\widetilde{\mathcal M} = n$,
\begin{equation*}
\vert \langle d_{1}^L \rangle_\Psi \vert + \vert \langle d_2^L\rangle_{\Psi} \vert \leq \rho_\mu^2 \ell^2 \Vert v \Vert_1^2 \Big( \frac{1+ \varepsilon_N^{-1/4}\widetilde{K}_H}{n^{1/2}} \langle n_+ \rangle_\Psi + \frac{\varepsilon_N^{-1/2} \widetilde K_H^2}{n} \langle n_+ \rangle_\Psi \Big) + C \langle \mathcal Q_4^{\rm{ren}} \rangle_\Psi .
\end{equation*}
Now we use the condensation estimate \eqref{eq:condensationestimate} and the bound \eqref{eq:estimateprioriQ4} on $Q_4^{\text{ren}}$ to obtain 
\begin{align}\label{eq:d1d2.01}
\vert \langle d_{1}^L \rangle_\Psi \vert &+ \vert \langle d_2^L\rangle_{\Psi} \vert \nonumber \\
&\leq \rho^2_{\mu} \ell^2 \|v\|_1 \Big ( Y^{1/2} |\log Y|^{1/2} K_{\ell} K_B \widetilde{K}_H \varepsilon_N^{-1/4}+ Y |\log Y| K_{\ell}^2 K_B^2 \widetilde{K}_H^2 \varepsilon_N^{-1/2}\Big ).
\end{align}
The relation \eqref{eq:relK_B-K_H-K_N-K_L} between the parameters implies that the largest term in \eqref{eq:d1d2.01} is the first one. Using the conditions \eqref{eq:varepsilonN} and \eqref{eq:M_large_matrices} on $\varepsilon_N$ and $\mathcal{M}$ respectively, and the assumptions on $\Vert v \Vert_1$ we find
\begin{equation}\label{eq:d1d2.02}
\frac{ |\langle d_{1}^L \rangle_\Psi \vert + \vert \langle d_2^L\rangle_{\Psi} |}{\mathcal M^2} \leq \rho^2_{\mu} \ell^2 Y^{2+\eta} .
\end{equation}
Using the estimates \eqref{eq:PsimH.01} for $m > \frac 1 2 \mathcal M$ and \eqref{eq:d1d2.02} in formula \eqref{eq:large_matrices_decomposn+} we conclude the proof.
\end{proof}

\section{Lower bounds in second quantization}\label{sec:c-number}
\subsection{Second quantization formalism}

We rewrite the Hamiltonian in the second quantization formalism. Let us introduce the operators, where $\#$ can be nothing or $\dagger$ for the annihilation or creation operators on $\mathscr{F}_s(L^2(\Lambda))$, respectively,
\begin{equation}
a_0^{\#} := \frac{1}{\ell} \,a^{\#}(\theta), \qquad \text{and} \qquad [a_0, a_0^{\dagger}] = 1,
\end{equation}
being the creation and annihilation operators for bosons with zero momentum, where $\theta$ is the sharp localization function on $\Lambda$ (see \eqref{def:sharp_loc_Lambda}). For $k \in \mathbb{R}^2  \setminus\{0\}$ we also define
\begin{equation}
\widetilde{a}_k^{\#} := \frac{1}{\ell}\, a^{\#}(Q e^{ikx}\theta),
\end{equation} 
the creation and annihilation operators for bosons with non-zero momentum with $Q$ defined in \eqref{def:PQ}, and their regular analogous
\begin{equation}
a_k^{\#} :=\frac{1}{\ell} \, a^{\#}(Q e^{ikx}\chi_{\Lambda}),
\end{equation} 
where $\chi_{\Lambda}$ is the regular localization function defined in Appendix~\ref{app:locfunction}. We have the usual commutation relations, for $k,h \in \mathbb{R}^2\setminus \{0\}$
\begin{equation}\label{eq:commut_atilde}
[\widetilde a_k, \widetilde a_h ] = [a_k, a_h] = 0, \qquad \text{and} \qquad [\widetilde{a}_k,\widetilde{a}^{\dagger}_h] = \frac{1}{\ell^2}\,\langle Qe^{ikx}, Q e^{ihx} \rangle .
\end{equation}
Using that $P = \one -Q$ and $\widehat \chi_\Lambda(k) = \ell^2 \widehat \chi_\Lambda(k\ell)$,
\begin{equation}\label{eq:commut_a}
[a_k, a_h^{\dagger}] = \frac{1}{\ell^2}\, \langle Q e^{ikx}\chi_{\Lambda},\, Q e^{ihx} \chi_{\Lambda} \rangle = \widehat{\chi^2}((k-h)\ell) -\widehat{\chi}(k \ell) \overline{\widehat{\chi}}(h \ell),
\end{equation}
and
\begin{equation}\label{eq:commut_abound}
[a_k, a_h^{\dagger}] \leq 1.
\end{equation}
Let us observe, first of all, that 
\begin{equation}\label{def:n0n+}
n_0 = a^{\dagger}_0 a_0, \qquad n_+ = \frac{\ell^2}{(2\pi)^2} \int \widetilde{a}^{\dagger}_k \widetilde{a}_k \dd k.
\end{equation}
Let us introduce, for $k \in \mathbb{R}^2$, the kinetic Fourier multiplier
\begin{equation}\label{def:tauk}
\tau(k) := (1-\varepsilon_T)\Big[|k| -\frac{1}{2s\ell}\Big]_+^2 + \varepsilon_T \Big[|k| -\frac{1}{2 d s\ell}\Big]_+^2.
\end{equation}
We will need the following technical lemma to control the number operators.

\begin{lemma}\label{lem:numbercontrolhigh}
Assume the relation \eqref{eq:rel_KH_K_M} between the parameters. Let $\Psi \in \mathscr{F}_s(L^2(\Lambda))$ be a normalized state satisfying
\begin{equation}\label{eq:Psitilde_number}
\one_{[0,\mathcal{M}]} (n_+^L)\Psi= \Psi, \qquad \one_{[0, 2 \rho_{\mu}\ell^2]}(n_+)\Psi = \Psi,
\end{equation}
then the following bounds hold
\begin{align}
&\Big\langle \ell^2 \int_{\{|k| \leq 2 K_H \ell^{-1}\}} (a^{\dagger}_k a_k + \widetilde{a}^{\dagger}_k \widetilde{a}_k) \dd k \Big\rangle_{\Psi} \leq C \mathcal{M}, \label{lem81:1} \\
&\Big\langle \ell^2 \int_{\mathbb{R}^2} (a^{\dagger}_k a_k + \widetilde{a}^{\dagger}_k \widetilde{a}_k) \dd k \Big\rangle_{\Psi} \leq C \mathcal{M} + C \langle n_+^H\rangle_{\Psi}. \label{eq:high_momenta_number_control}
\end{align}
\end{lemma}

\begin{proof}
The proof is analogous for both the addends, therefore we give the proof only for the $a^{\#}_k$.
We want to compare localization in terms of kinetic energy with localization in momenta. We use \cite[Lemma 5.2]{FS2} adapted to dimension $2$:
\begin{align}
Q \chi_{\Lambda} \one_{\{|p| \leq K_H \ell^{-1}\}} \chi_{\Lambda} Q &\leq C \overline{Q}_H + C \Big( \Big( \frac{K_H}{\widetilde{K}_H}\Big)^{M}  + \varepsilon_N^{3/2}\Big),\label{eq:kinetic_momentum_comparison}\\
Q  \one_{\{|p| \leq K_H \ell^{-1}\}}  Q &\leq C \overline{Q}_H + C \Big( \Big( \frac{K_H}{\widetilde{K}_H}\Big)^{M}  + \varepsilon_N^{3/2}\Big),\label{eq:kinetic_momentum_comparison2}
\end{align}
where we recall the definition \eqref{def:QHbar} of $\overline{Q}_H$. Using \eqref{eq:kinetic_momentum_comparison} we have the following inequality in the $N-$th Fock sector
\begin{align*}
\frac{\ell^2}{(2\pi)^2}  \int_{\{|k| \leq 2 K_H \ell^{-1}\}} a^{\dagger}_k a_k \dd k \, \bigg|_N &= \sum_{j=1}^N Q_j\chi_{\Lambda}(x_j) \one_{(0,2K_H \ell^{-1}]} (\sqrt{-\Delta_j})\chi_{\Lambda}(x_j) Q_j   \\
&\leq C n_+^L + \Big( \Big(\frac{K_H}{\widetilde{K}_H}\Big)^{M} + \varepsilon_N^{3/2}\Big)n_+.
\end{align*}
Using the bounds from \eqref{eq:Psitilde_number} and the relation \eqref{eq:rel_KH_K_M} we deduce
\begin{equation}
\Big\langle \ell^2 \int_{\{|k| \leq 2 K_H \ell^{-1}\}} a^{\dagger}_k a_k \dd k \Big\rangle_{\Psi} \leq C \mathcal M + C \Big( \Big(\frac{K_H}{\widetilde{K}_H}\Big)^{M} + \varepsilon_N^{3/2}\Big)\rho_{\mu}\ell^2 \leq C \mathcal{M},
\end{equation}
thus proving \eqref{lem81:1}.
In order to obtain \eqref{eq:high_momenta_number_control} it is enough to estimate the integral on the complementary subset. We have, again on the $N-$th sector,
\begin{equation}
\ell^2 \int_{\{|k| \geq 2 K_H \ell^{-1}\}}  a^{\dagger}_k a_k  \dd k\, \bigg|_N = \sum_{j=1}^N Q_j\chi_{\Lambda}(x_j) \one_{\{|k| \geq 2K_H \ell^{-1}\}} (\sqrt{-\Delta_j})\chi_{\Lambda}(x_j) Q_j. 
\end{equation}
We insert $1 = \one_{\mathcal P_L} + \one_{\mathcal P_L^c}$ and use the Cauchy-Schwarz inequality to estimate the right-hand side,
\begin{align*}
Q \chi_{\Lambda} \one_{[2K_H \ell^{-1}, +\infty)}(\sqrt{-\Delta})\chi_{\Lambda}Q  & \leq   2 Q \one_{\mathcal{P}_L^c}(\sqrt{-\Delta}) \chi_{\Lambda}  \one_{[2K_H \ell^{-1}, +\infty)} (\sqrt{-\Delta}) \chi_{\Lambda} \one_{\mathcal{P}_L^c} (\sqrt{-\Delta}) Q.  \\
&\quad + 2 Q \one_{\mathcal{P}_L}(\sqrt{-\Delta}) \chi_{\Lambda}  \one_{[2K_H \ell^{-1}, +\infty)} (\sqrt{-\Delta}) \chi_{\Lambda} \one_{\mathcal{P}_L} (\sqrt{-\Delta}) Q.  
\end{align*}
On $\mathcal P_L^c$ we can use the bound
\begin{equation*}
Q \one_{\mathcal{P}_L^c}(\sqrt{-\Delta}) \chi_{\Lambda}  \one_{[2K_H \ell^{-1}, +\infty)} (\sqrt{-\Delta}) \chi_{\Lambda} \one_{\mathcal{P}_L^c} (\sqrt{-\Delta}) Q  \leq  \|\chi_{\Lambda}\|^2_{\infty} Q \one_{\mathcal{P}_L^c}(\sqrt{-\Delta})Q.
\end{equation*}
On $\mathcal P_L$ we bound the operator norm, multiplying and dividing by an $M$ power of the Laplacian and using that $\chi$ has $M$ bounded derivatives,
\begin{align*}
\|\one_{\mathcal{P}_L}(\sqrt{-\Delta}) \chi_{\Lambda} \one_{[2K_H \ell^{-1}, +\infty)}\|  
&\leq \| \one_{\mathcal{P}_L}(\sqrt{-\Delta}) \chi_{\Lambda} (-\Delta)^{M/2}\| \|(-\Delta)^{-M/2} \one_{[2K_H \ell^{-1}, +\infty)} \| \\
&\leq C (d^2 K_H)^{-M}.
\end{align*}
We deduce
\begin{equation}
\ell^2 \int_{\{|k| \geq 2 K_H \ell^{-1}\}} a^{\dagger}_k a_k \dd k  \leq C n_+^H + C(d^2 K_H)^{-2M} n_+,
\end{equation}
and we conclude using \eqref{eq:rel_KH_K_M} and the assumptions on $\Psi$.
\end{proof}

\subsection{Second quantized Hamiltonian}\label{subsec:sqh}

We can rewrite the $\mathcal Q_3^{\rm{low}}$ term \eqref{eq:defQ3low} in second quantized formalism
\begin{equation}
\mathcal Q_3^{\rm{low}} =  \frac{\ell^2}{(2\pi)^4}  \int_{\mathbb{R}^2 \times \mathbb{R}^2} f_L(p) \widehat{W}_1(k)  a^{\dagger}_0 \widetilde{a}_p^{\dagger} a_{p-k} a_k  \dd k \dd p + h.c.
\end{equation}
An important consideration is that we can restrict the contributions in $\mathcal Q_3^{\text{low}}$ to high momenta. This is the content of the next lemma.

\begin{lemma}[Localization of $\mathcal Q_3^{\rm{low}}$ to high momenta]\label{lem:Q3highloc}
Assume $R \leq \ell$ and the relations \eqref{cond:3Qloc}, \eqref{eq:Kl_KN_number}, \eqref{eq:rel_KH_K_M} between the parameters. If $\Psi \in \mathscr{F}_s(L^2(\Lambda))$ is a $n$-particle state satisfying \eqref{eq:condensationaprioricondition} and $\one_{[0,\mathcal{M}]} (n_+^L)\Psi= \Psi$ then we have
\begin{equation}
\langle \Psi\,|\, \mathcal Q_3^{\rm{low}}  \Psi \rangle \geq \langle \Psi \,|\, \mathcal Q_3^{\rm{high}}\Psi\rangle 
- \frac{b}{100 \ell^2} \langle n_{+} \rangle_\Psi,
\end{equation}
where
\begin{equation}
\mathcal Q^{\rm{high}}_3 = \frac{\ell^2}{(2\pi)^4}  \int_{\mathcal{P}_H \times \mathbb{R}^2} f_L(p) \widehat{W}_1(k)  a^{\dagger}_0 \widetilde{a}_p^{\dagger} a_{p-k} a_k \dd k \dd p  + h.c.,
\end{equation}
with $\mathcal{P}_H$ defined in \eqref{eq:defPL_PH}.
\end{lemma}

\begin{proof} First note that
\begin{equation}
\langle \Psi\,|\, (\mathcal Q_3^{\text{low}} -  \mathcal Q_3^{\text{high}} )\Psi\rangle = \frac{\ell^2}{(2\pi)^4}  \int_{\mathcal P^c_H \times \mathbb{R}^2} f_L(p) \widehat{W}_1(k)  \langle \Psi \,|\, a^{\dagger}_0 \widetilde{a}_p^{\dagger} a_{p-k} a_k  \Psi \rangle \dd k \dd p + h.c.
\end{equation}
For any $\varepsilon >0$, using Cauchy-Schwarz on the creation and annihilation operators,
\begin{multline}
\langle \Psi\,|\, (\mathcal Q_3^{\text{low}} -  \mathcal Q_3^{\text{high}} ) \Psi \rangle  \\
\geq  - C \delta \ell^2  \int_{\mathcal P^c_H \times \mathbb{R}^2} f_L(p) \Big(  \varepsilon \langle \Psi\,|\, \widetilde{a}^{\dagger}_p a^{\dagger}_0  a_0 \widetilde{a}_p \Psi \rangle + \varepsilon^{-1} \langle \Psi\,|\,  a^{\dagger}_k a^{\dagger}_{p-k} a_{p-k} a_k \Psi \rangle \Big) \dd k \dd p , \label{eq:original}
\end{multline}
where we used the fact that $\|\widehat{W}_1\|_{\infty} \leq \|W_1\|_1 \leq C \delta$ (from Lemma~\ref{lem:gomegaapprox}). We now use the following inequalities, obtained by Lemma~\ref{lem:numbercontrolhigh} and bounding $f_L$ by $1$,
\begin{equation}
\ell^2\int_{\mathcal{P}_H^c \times \mathbb{R}^2} \, f_L(p) \langle \Psi\,|\, \widetilde{a}^{\dagger}_p a^{\dagger}_0  a_0 \widetilde{a}_p \Psi \rangle \dd k \dd p \leq n \langle n_+ \rangle_{\Psi} \int_{\mathcal{P}_H^c} \dd k  = \frac{n \langle n_+\rangle_{\Psi}}{\ell^2} K_H^2,
\end{equation}
\begin{equation}
\ell^4\int_{\mathcal{P}_H^c \times \mathbb{R}^2} \, f_L(p)\langle \Psi\,|\,  a^{\dagger}_k a^{\dagger}_{p-k} a_{p-k} a_k \Psi \rangle \dd k \dd p  \leq C\mathcal{M} \langle n_+\rangle_{\Psi}.
\end{equation}
Therefore, applying to \eqref{eq:original} we obtain
\begin{equation}
\langle \Psi\,|\, (\mathcal Q_3^{\text{low}} -  \mathcal Q_3^{\text{high}} )\Psi \rangle \geq -C\delta \frac{\langle n_+\rangle_{\Psi}}{\ell^2} n \Big(\varepsilon K_H^2 + \varepsilon^{-1} \frac{\mathcal{M}}{n}\Big).
\end{equation}
Choosing $\varepsilon = K_H^{-1} \frac{\mathcal{M}^{1/2}}{n^{1/2}}$, we obtain
\begin{equation}
\langle \Psi\,|\, (\mathcal Q_3^{\text{low}} -  \mathcal Q_3^{\text{high}} ) \Psi \rangle \geq -C\delta \frac{\langle n_+\rangle_{\Psi}}{\ell^2} n \frac{K_H \mathcal{M}^{1/2}}{n^{1/2}}.
\end{equation}
We use Theorem~\ref{thm:excitationsbound} and \eqref{eq:Kl_KN_number} to bound $n^{1/2}$ by $2 \rho_\mu^{1/2} \ell$ and get
\begin{equation*}
\langle \Psi\,|\, (\mathcal Q_3^{\text{low}} - \mathcal Q_3^{\text{high}} ) \Psi \rangle \geq -C\delta \rho_\mu^{1/2} \ell K_H \mathcal M^{1/2} \frac{\langle n_+\rangle_{\Psi}}{\ell^2}.
\end{equation*}
By the assumption \eqref{cond:3Qloc} the error can be absorbed in a small fraction of the spectral gap.
\end{proof}

We are ready to state a bound for the second quantized Hamiltonian.

\begin{proposition}\label{propos:secondquant}
Assume $R \ll (\rho_\mu \delta)^{-1/2}$ and the relations of Appendix~\ref{app:parameters}
between the parameters. Let $\Psi$ be a normalized $n$-particle state satisfying \eqref{eq:condensationaprioricondition} and $\Psi = \one_{[0,\mathcal{M}]}(n_+^L)\Psi$. Then 
\begin{equation}\label{eq:Hlambdarhowhat}
\langle \Psi\,|\, \mathcal{H}_{\Lambda}(\rho_{\mu}) \Psi \rangle \geq \langle \Psi\,|\, \mathcal{H}_{\Lambda}^{\rm{2nd}}(\rho_{\mu}) \Psi \rangle -C \ell^2 \rho_{\mu}^2 \delta \Big( d^{2M-2} + R^2 \ell^{-2}  \Big),
\end{equation}
where
\begin{align}
\mathcal{H}_{\Lambda}^{\rm{2nd}} &:= \frac{\ell^2}{(2\pi)^2} \int_{\mathbb{R}^2} (1-\varepsilon_{N}) \tau(k) a^{\dagger}_k a_k \dd k + \frac{b}{2\ell^2} n_+ +b \frac{\varepsilon_T}{8 d^2 \ell^2 } n_+^H +b \frac{\varepsilon_T n_0 n_+^H}{16 d^2 \ell^2 (\rmu \ell^2)}  \label{eq:Tquantized}  \\
&\quad + \frac{1}{2\ell^2} a_0^{\dagger}a_0^{\dagger}a_0 a_0  (\widehat{g}_0 + \widehat{g\omega}(0)) -\rho_{\mu} a_0^{\dagger}a_0 \widehat{g}_0  \label{eq:leadquantized}\\
&\quad + \Big( \Big( \frac{1}{\ell^2} a_0^{\dagger}a_0 - \rho_{\mu}\Big) \widehat{W}_1(0) \frac{1}{(2\pi)^2} \int_{\mathbb{R}^2} \widehat{\chi}_{\Lambda}(k) a^{\dagger}_k a_0 \dd k + h.c. \Big)   \label{eq:charquantized} \\ 
&\quad +\Big(  \frac{1}{\ell^2} a_0^{\dagger}a_0 \widehat{\omega W_1}(0) \frac{1}{(2\pi)^2} \int_{\mathbb{R}^2} \widehat{\chi}_{\Lambda}(k) a^{\dagger}_k a_0 \dd k + h.c. \Big)\label{eq:char2quantized}\\
&\quad +\mathcal Q_2^{\rm{rest}} +  \mathcal Q_3^{\rm{high}}   \label{eq:Q2Q3quantized}\\
&\quad + \Big(\Big(\frac{1}{\ell^2}a^{\dagger}_0 a_0  - \rho_{\mu}\Big)\widehat{W}_1(0) + \frac{1}{\ell^2}a^{\dagger}_0 a_0 \widehat{W_1 \omega}(0)\Big)\frac{\ell^2}{(2\pi)^2} \int_{\mathbb{R}^2}  a^{\dagger}_k a_k \dd k, \label{eq:n+quantized}
\end{align}
with
\begin{align*}
\mathcal Q_2^{\rm{rest}}=  \frac{1}{(2\pi)^2} \int_{\mathbb{R}^2} (\widehat{W}_1(k) &+ \widehat{(W_1\omega)}(k)) a^{\dagger}_0 a^{\dagger}_k a_k a_0 \dd k \\
&+\frac{1}{2} \int_{\mathbb{R}^2} \widehat{W}_1(k) \Big( a^{\dagger}_0 a^{\dagger}_0 a_k a_{-k} + a^{\dagger}_k a^{\dagger}_{-k} a_0 a_0\Big) \dd k.
\end{align*}
\end{proposition}

\begin{proof}
We use the lower bound for $\mathcal{H}_{\Lambda}(\rho_{\mu})$ from Corollary~\ref{cor:H'expression}. 
First of all, in the kinetic energy expression \eqref{eq:Texplicitexpress} we remove the positive parts depending on the Neumann Laplacian, namely $\varepsilon_N(-\Delta_{\mathcal{N}})$ and $\mathcal{T}^{\text{Neu,s}}$. Using the quantization, we obtain from \eqref{eq:Texplicitexpress} the expressions in \eqref{eq:Tquantized} with the main kinetic energy term and the spectral gaps. We bounded part of the spectral gap to get the last term in \eqref{eq:Tquantized} using $n_0 \leq 2 \rho_\mu \ell^2$ (which follows from \eqref{eq:estimatepriori.n} and \eqref{eq:Kl_KN_number}). This term will be useful later (in particular in the proof of Lemma~\ref{lem:technicalest_rhofar}).

The expressions \eqref{eq:leadquantized}, \eqref{eq:charquantized}, \eqref{eq:char2quantized}, $\mathcal Q_2^{\rm{rest}}$ and \eqref{eq:n+quantized} are obtained from \eqref{eq:kin+lead}, \eqref{eq:expresschar1}, \eqref{eq:expresschar2}, \eqref{eq:expressQ2} and \eqref{eq:expressQ2char} respectively, via a straightforward application of the quantization rules. Note that in \eqref{eq:charquantized} and \eqref{eq:char2quantized} we have changed a $\widehat W_1(k)$ (resp. $\widehat{\omega W_1}(k)$) into $\widehat W_1(0)$ (resp. $\widehat{\omega W_1}(0)$). This can be justified by using \eqref{eq:fchi_approx} in \eqref{eq:expresschar1} and \eqref{eq:expresschar2}, the error being of order $R^2 \rho_\mu^2 \delta$.
We can reabsorb the term 
\begin{equation*}
- C(\rho_{\mu} + \rho_0) \delta R^2 \frac{n_+}{\ell^2},
\end{equation*} 
in a fraction of the spectral gap because $R \ll (\rho_\mu \delta)^{-1/2}$.
Let us observe that thanks to Lemma~\ref{lem:Q3loc} we can replace $\mathcal Q_3^{\rm{ren}}+ \frac{1}{4} \mathcal Q_4^{\rm{ren}}$ by $\mathcal Q_3^{\rm{low}}$ in $\mathcal{H}_{\Lambda}(\rho_{\mu})$. Part of the error is absorbed in the spectral gap, other part appears in \eqref{eq:Hlambdarhowhat}. Finally we change $\mathcal Q_3^{\rm{low}}$ into $\mathcal Q_3^{\rm{high}}$ using Lemma~\ref{lem:Q3highloc}, the error being absorbed in a fraction of the spectral gap again.
\end{proof}

\subsection{$c\,$-number substitution}\label{subsec:c-number}

In this section we show how the energy can be bounded if we minimize over a specific class of coherent states, which are eigenvectors for the annihilation operator of the condensate. In this way we can turn the action of the condensate operators in the form of multiplication per complex numbers. Let us define 
\begin{equation}
|z\rangle = e^{-\big(\frac{|z|^2}{2} + z a^{\dagger}_0\big)}\, \Omega,
\end{equation}
for any $z \in \mathbb{C}$. As anticipated, we have
\begin{equation}
a_0 |z\rangle = z\, |z \rangle .
\end{equation}
Given any state $\Psi$ we define the $z$-dependent state
\begin{equation}
\Phi(z) := \langle z\,|\, \Psi\rangle,
\end{equation}
obtained by the partial inner product in $\mathscr{F}_s(\mathrm{Ran}P)$. One can verify that these states generate the space $\mathscr{F}_s (\mathrm{Ran}Q)$. Moreover,
\begin{equation}
1 = \frac{1}{\pi} \int_{\mathbb{C}} |z\rangle \langle z| \, \dd z.
\end{equation}
We define the following $z$-dependent density,
\begin{equation}
\rho_z := \frac{|z|^2}{\ell^2}, 
\end{equation}
and $z$-dependent Hamiltonian,
\begin{align} \label{def:Kz}
\mathcal{K}(z) &= \frac{1}{2}\rho_z^2 \ell^2 (\widehat{g}_0 + \widehat{g\omega}(0)) - \rho_{\mu}\rho_z \widehat{g}_0\ell^2 \\ 
&\quad+\mathcal{K}^{\text{Bog}} + \frac{b}{2\ell^2} n_+ + \frac{\varepsilon_T b}{8 d^2 \ell^2 } n_+^H +b \frac{\varepsilon_T |z|^2 n_+^H}{16 d^2 \ell^2 (\rho_{\mu} \ell^2)} + \varepsilon_R (\rho_{\mu} - \rho_z )^2 \delta \ell^2 \label{eq:epsilonR} \\
&\quad+ (\rho_{z} - \rho_{\mu}) \widehat{W}_1(0) \frac{\ell^2}{(2\pi)^2} \int_{\mathbb{R}^2}  a^{\dagger}_k a_k \dd k  + \mathcal Q_1^{\rm{ex}}(z) + \mathcal Q_2^{\rm{ex}}(z) + \mathcal Q_3(z),
\end{align}
where $\varepsilon_R \ll 1$ is fixed in Appendix~\ref{app:parameters}, and
\begin{align*}
\mathcal{K}^{\text{Bog}} := \frac{\ell^2}{2(2\pi)^2}\int_{\mathbb{R}^2} \Big( \mathcal{A}&(k) (a^{\dagger}_k a_{k}  +a^{\dagger}_{-k} a_{-k})+  \mathcal{B}(k) (a_k a_{-k} + a^{\dagger}_k a^{\dagger}_{-k}) \\ &+\, \mathcal{C}(k) (a^{\dagger}_k + a^{\dagger}_{-k} + a_k + a_{-k}) \Big) \dd k,
\intertext{with}
\mathcal{A}(k) &:= (1-\varepsilon_N)  \tau(k) + \mathcal{B}(k), \qquad \mathcal{B}(k) := \rho_z \widehat{W}_1(k),\\
\mathcal{C}(k) &:= \frac{(\rho_z - \rho_{\mu})}{\ell^2} \widehat{W}_1(0)  \widehat{\chi}_{\Lambda}(k)  z,\\
\mathcal Q_1^{\rm{ex}}(z) &:=  \rho_z \widehat{(\omega W_1)}(0) \frac{1}{(2\pi)^2} \int_{\mathbb{R}^2} \widehat{\chi}_{\Lambda}(k) a^{\dagger}_k z  \dd k+ h.c., \\
\mathcal Q_2^{\rm{ex}}(z) &:= \frac{\ell^2}{(2 \pi)^2} \rho_z \int_{\mathbb{R}^2} (\widehat{\omega W_1}(0) + \widehat{\omega W_1}(k) ) a^{\dagger}_k a_k \dd k, \\
\mathcal Q_3(z) &:=   \frac{\ell^2}{(2\pi)^4}  \int_{\mathcal{P}_H \times \mathbb{R}^2} f_L(p) \widehat{W}_1(k)  \Big(\bar{z} \widetilde{a}_p^{\dagger} a_{p-k} a_k + h.c. \Big) \dd k \dd p.
\end{align*}
With these notations, the following theorem holds. Recall that $\mathcal{H}_{\Lambda}^{\rm{2nd}}$ is given by Proposition~\ref{propos:secondquant}.
\begin{theorem}\label{thm:cnumber}
Assume $R \leq \ell$ and \eqref{eq:Kl_KN_number}. For any normalized $n$-particle state $\Psi$ satisfying $\Psi = \one_{[0,\mathcal{M}]}(n_+^L) \Psi$ and \eqref{eq:condensationaprioricondition} we have
\begin{equation}
\langle \Psi\,|\,\mathcal{H}_{\Lambda}^{\rm{2nd}} \Psi \rangle \geq \inf_{z \in \mathbb{R}_+} \inf_{\Phi} \langle \Phi\,|\, \mathcal{K}(z)  \Phi\rangle - C \rho_{\mu} \delta (1 + \varepsilon_R K_{\ell}^4 K_B^2 |\log Y|), 
\end{equation}
where the second infimum is over all the normalized states in $\mathscr{F}(\mathrm{Ran}Q)$ such that 
\begin{equation}\label{def:statefiniteexcit}
\Phi = \one_{[0,\mathcal{M}]}(n_+^L) \Phi, \qquad \text{and} \qquad \Phi = \one_{[0,2 \rho_{\mu}\ell^2]}(n_+) \Phi.
\end{equation}
\end{theorem}

\begin{remark}
In \eqref{eq:epsilonR} we used a fraction of the spectral gap to make the second to last term appear. Since we have an a priori control only on $n_+^L$, this is going to be useful in Section~\ref{subsec:rhofar} in order to control a contribution coming from the high momenta. The last positive term of the line is needed to control a negative term in Theorem~\ref{thm:lower_rho=rho_mu} which depends in a bad way on the radius of the support of the potential.
\end{remark}

\begin{proof}
The theorem is proven via a standard technique of calculating the actions of creation and annihilation operators for the condensate on the coherent state and using its eigenvector properties, for details see \cite[Theorem 8.5]{FS2}. Practically speaking it consists in the formal substitutions 
\begin{equation}
a_0 \mapsto z, \qquad a^{\dagger}_0 \mapsto \overline{z}, \qquad a^{\dagger}_0 a_0 \mapsto |z|^2-1,
\end{equation}
and getting rid of the lower order terms in $|z|$ because they produce errors of the form
\begin{equation}
\rho_{\mu} \delta = \rho_{\mu}^2 \ell^2 \delta^2 K_{\ell}^{-2}.
\end{equation}
In order to make the last term in \eqref{eq:epsilonR} appear, we add and subtract $\varepsilon_R (\rho_{\mu} - n_0 \ell^{-2})^2 \delta \ell^2$ to $\mathcal{H}_{\Lambda}^{\rm{2nd}}$ and estimate the negative contribution, recalling the estimates in Theorem~\ref{thm:excitationrestriction} and that $n_+^2 \leq n n_+$ we get
\begin{align*}
-\varepsilon_R \Big(\rho_{\mu} - \frac{n_0}{\ell^{2}}\Big)^2 \delta \ell^2 &\geq -2 \varepsilon_R \delta \ell^{-2} ((\rho_{\mu}\ell^2 - n)^2 + n_+ n) \\
&\geq - C \varepsilon_R \frac{\delta}{\ell^2} n^2 K_B^2 Y |\log Y| K_{\ell}^2 =  - C \varepsilon_R \rho_{\mu}\delta K_B^2  |\log Y| K_{\ell}^4, 
\end{align*}
which is coherent with the error terms.
\end{proof}

\section{Lower bounds for the Hamiltonian ${\mathcal K}$}\label{sec:bogestimates}
\subsection{Estimate of \texorpdfstring{$\mathcal{K}$}{K} for \texorpdfstring{$\rho_{z}$}{rho\_z} far from \texorpdfstring{$\rho_{\mu}$}{rho\_mu}}\label{subsec:rhofar}

The purpose of this section is to show that for values of $\rho_{z}$ far from the density $\rho_{\mu}$ it is possible to prove a rough estimate on the energy and eliminate these values from the analysis.
This is the content of the proposition below.
We recall that $\mathcal{K}(z)$ is defined in \eqref{def:Kz}, and we use the notations $\varepsilon_{\mathcal M }= \frac{\mathcal M}{\rho_\mu \ell^2}$ and
\begin{equation}\label{eq:deltasdefin}
\delta_1 = \frac{\varepsilon_T^2 \varepsilon_{\mathcal M} }{d^8 K_{\ell}^4 } \Big( 1+ \frac{K_{\ell}^2}{K_{H}^2}\Big),  \quad \delta_2 =  \varepsilon_{\mathcal M}^{1/2}, \quad \delta_3 =  \delta|\log (d s K_{\ell})|  +   \frac{(d K_{\ell})^4}{\varepsilon_T^2}.
\end{equation}
\begin{proposition}\label{propos:rhofar}
Assume the relations between the parameters in Appendix~\ref{app:parameters}. There exists a $C>0$, such that if we have $\rmu a^2 \leq C^{-1}$ and 
\begin{equation}\label{eq:conditionrhofar}
|\rho_{\mu} - \rho_z| \geq C \rho_{\mu} \max\Big( (\delta_1 + \delta_2 + \delta_3)^{1/2}, \delta^{1/2}\Big),
\end{equation}
then for any state $\Phi \in \mathscr{F}(\mathrm{Ran}Q) $ satisfying \eqref{def:statefiniteexcit}, we have
\begin{equation}\label{eq:LargeEnergy}
\langle \Phi \,|\, \mathcal{K}(z) \Phi \rangle \geq -\frac{1}{2} \rho_{\mu}^2 \ell^2 \widehat{g}_0 + 8\pi\Big(\frac 1 2 +2\Gamma +\log\pi\Big) \rho_{\mu}^2 \ell^2 \delta^2 .
\end{equation}
\end{proposition}
Notice that the second order term in \eqref{eq:LargeEnergy} is larger than the one aimed for in Theorem~\ref{thm:largebox_lower}. So the statement of the proposition is that the energy is too large unless $|\rho_{\mu} - \rho_z|$ is small.
The proof of the proposition relies on the technical estimate given by the following lemma.
\begin{lemma}\label{lem:technicalest_rhofar}
Assume the relations between the parameters in Appendix~\ref{app:parameters}. For any normalized $\Phi \in \mathscr{F}(\mathrm{Ran}Q) $ such that  \eqref{def:statefiniteexcit}  holds,
\begin{align}\label{eq:controlrhofar}
\langle \Phi \,|\, \mathcal{K}(z) \Phi \rangle &\geq -\frac{\rho_{\mu}^2 \ell^2}{2} \widehat{g}_0 + \frac{\ell^2}{2}(\rho_{\mu}-\rho_z )^2 \widehat{g}_0
-  C\rho_z \rho_{\mu} \ell^2\delta\delta_1 \nonumber \\
&\quad -  C \rho_{\mu}^{1/2} (\rho_{\mu} + \rho_z)^{3/2}\ell^2 \delta\delta_2   - C\rho_z^2  \ell^2\delta \delta_3 -C \rho_{\mu} \delta^{2} K_{\ell}^{-2} (ds)^{-4}.
\end{align}
\end{lemma}
\begin{proof}[Proof of Lemma~\ref{lem:technicalest_rhofar}]
We start by estimating the $Q_1$ terms. We have for any $\varepsilon >0$
\begin{align*}
\int_{\mathbb{R}^2}   \widehat{\chi}_{\Lambda}(k) (a_k^{\dagger} z+ &a_k \bar{z}) \dd k 
\\ 
&\leq \int_{\mathbb{R}^2} |\widehat{\chi}_{\Lambda}(k)| ( \varepsilon |z|^2 +\varepsilon^{-1} a^{\dagger}_k a_k )  \dd k  \\
&
\leq C \Big(\varepsilon |z|^2 + \varepsilon^{-1} |\widehat{\chi}_\Lambda(0)|\int_{k \in \mathcal{P}_H^c}  a^{\dagger}_k a_k \dd k   + \varepsilon^{-1} \int_{k \in \mathcal{P}_H}  |\widehat{\chi}_{\Lambda}(k)| a^{\dagger}_k a_k \dd k  \Big).
\end{align*}
Considering a $\Phi$ like in the assumption we have, using $|\widehat{\chi}_{\Lambda}(0)| = \ell^2 \|\chi\|_1$ together with Lemma~\ref{lem:numbercontrolhigh},
\begin{equation}
\Big \langle |\widehat{\chi}_\Lambda(0)|\int_{k \in \mathcal{P}_H^c} a^{\dagger}_k a_k \dd k \;  + \int_{k \in \mathcal{P}_H} |\widehat{\chi}_{\Lambda}(k)| a^{\dagger}_k a_k \dd k \Big\rangle_{\Phi} \leq C \Big(\mathcal{M}+ \rho_{\mu} \ell^2  \sup_{k \in \mathcal{P}_H} (\ell^{-2} |\widehat{\chi}_{\Lambda}(k)|)\Big).
\end{equation}
Now, using \eqref{eq:fourierlocestimate} and optimizing with $\varepsilon = \sqrt{\mathcal{M}/|z|^2}$,
\begin{align}
&\Big{\langle} - \frac{\ell^2}{(2\pi)^2}\int_{\mathbb{R}^2} \,  \mathcal{C}(k) (a^{\dagger}_k+ a^{\dagger}_{-k} + a_k + a_{-k})\dd k  + \mathcal Q_1^{\text{ex}}(z) \Big{\rangle}_{\Phi} \nonumber \\ & \geq - C \delta \sqrt{\mathcal{M}} |z| (|\rho_z - \rho_{\mu}| + \rho_z) \nonumber\\
&\geq  - C \Big( \frac{\mathcal{M}}{\rho_{\mu}\ell^2}\Big)^{1/2}\rho_{\mu}^{1/2}\ell^2 \delta (\rho_{\mu} + \rho_z)^{3/2}.\label{eq:proofAtilde0}
\end{align}
For the terms that are quadratic in the field operators, we use the estimate 
\begin{equation}
\left|\left\langle \ell^2 \int_{\mathbb{R}^2}  \, \widehat{W}_1(k) a^{\dagger}_k a_k \dd k\right\rangle_{\Phi}\right| \leq C \delta (\mathcal{M} + \langle n_+^H\rangle_{\Phi}),
\end{equation}
from Lemma~\ref{lem:numbercontrolhigh} to obtain that 
\begin{multline}\label{eq:proofAtilde1}
\left\langle \mathcal Q_2^{\text{ex}} + (\rho_{z} - \rho_{\mu}) \widehat{W}_1(0) \frac{\ell^2}{(2\pi)^2} \int_{\mathbb{R}^2} \, a^{\dagger}_k a_k \dd k  + \frac{\ell^2}{(2 \pi)^2 }\int_{\mathbb{R}^2} \, \mathcal{B}_k a^{\dagger}_k a_k \dd k \right\rangle_{\Phi}  \\
\geq - C (\rho_z + \rho_{\mu}) \big( \rho_{\mu}\ell^2 \delta \varepsilon_{\mathcal M} + \delta \langle n_+^H\rangle_{\Phi} \big),
\end{multline}
where the $\mathcal{B}_k$ has been extracted from the expression of the $\mathcal{A}_k$. The first term is coherent with the error in the result and the last one can be reabsorbed in a fraction of the spectral gap because of relation \eqref{eq:rel_T2comm}.

For the remaining part of $\mathcal{A}_k$ involving $\tau_k$ we add and subtract $-\rho_z \delta \varepsilon^{-1/2} + \varepsilon \tau_k$, with $\varepsilon \geq \varepsilon_N$ and estimate
\begin{align}
(1-\varepsilon_N) \tau_k &\geq  \widetilde{\mathcal{A}}_k -\rho_z \delta \varepsilon^{-1/2} + \varepsilon \tau_k, \label{eq:Atildeestimate}
\intertext{ with } 
\widetilde{\mathcal{A}}_k &= (1-2 \varepsilon) \Big[ |k| - \frac{1}{2 d s \ell }\Big]_+^2 + \rho_z \delta \varepsilon^{-1/2}.
\end{align}
We treat the terms in \eqref{eq:Atildeestimate} separately, adding them to the remaining parts of the Hamiltonian.
The simplest one is 
\begin{equation}
-\frac{\ell^2}{(2\pi)^2}\rho_z \varepsilon^{-1/2}\delta \Big\langle \int_{\mathbb{R}^2} \, a^{\dagger}_k a_k   \dd k \Big\rangle_{\Phi} \geq - C \varepsilon^{-1/2} \rho_z \delta (\mathcal{M} + \langle n_+^H \rangle_{\Phi}),
\end{equation}
where we used Lemma~\ref{lem:numbercontrolhigh}. We use this estimate to fix the choice of $\varepsilon$ in order to absorb the last term in the fraction of the spectral gap represented by the second to last term in \eqref{eq:epsilonR}. This yields
\begin{equation}
\varepsilon = C^{-1}\varepsilon_T^{-2} (d K_{\ell})^4,
\end{equation}
for some sufficiently large constant $C$ and the relations \eqref{eq:rel_eps_M_small}, \eqref{eq:rel_T2comm} ensure that $\varepsilon_N \leq \varepsilon \ll 1$.
 For the $\widetilde{\mathcal{A}}$ term plus the $B$ terms in the Hamiltonian we use the Bogoliubov diagonalization procedure stated in Theorem~\ref{thm:bogdiag} to obtain
\begin{equation}\label{eq:tildeALHY}
\frac{\ell^2}{(2 \pi)^2} \int_{\mathbb{R}^2}  \widetilde{\mathcal{A}}_k a^{\dagger}_k a_k + \frac{\mathcal{B}_k}{2} (a^{\dagger}_k a^{\dagger}_{-k} + a_{k} a_{-k}) \,\dd k \geq  -\frac{\ell^2}{2(2 \pi)^2}  \int_{\mathbb{R}^2} \widetilde{\mathcal{A}}_k - \sqrt{\widetilde{\mathcal{A}}_k^2 - \mathcal{B}_k^2} \,\dd k,
\end{equation}
and then we use Lemma~\ref{lem:bogintestimate} and its proof choosing $K_1= \rho_z \varepsilon^{-1/2} /2$, $K_2 = 2 \rho_z$, $K = (2ds\ell)^{-1}$ and $\kappa = (1-2\varepsilon)$ to derive that 
\begin{multline}
\eqref{eq:tildeALHY} \geq -\frac{\ell^2}{2(2\pi)^2} \Big(\rho_z^2  \frac{1 + \varepsilon \delta^2}{1-2\varepsilon} \int_{\mathbb{R}^2}\dd k  \frac{\widehat{W}_1^2(k)-\widehat{W}_1^2(0)\one_{\{|k|\leq \ell_{\delta}^{-1}\}}}{2|k|^2} +  C\rho_z \varepsilon^{1/2} \delta (ds\ell)^{-2}  \\
+ \frac{C}{(1-2 \varepsilon)}\rho_z^2 \delta^2(1+ R^2 \ell_{\delta}^{-2})+\frac{C \rho_z^2}{1-2\varepsilon}\delta^2 |\log (( d s \ell)^{-1} \ell_{\delta}) |\Big).
\end{multline}
Using now Cauchy-Schwarz on the second term, Lemma~\ref{lem:gomegaapprox}, writing only the dominant terms due to the relations between the parameters and recalling the definition \eqref{eq:def_elldelta_mu} of $\ell_{\delta}$ we obtain 
\begin{equation}\label{eq:proofAtilde2}
\eqref{eq:tildeALHY} \geq -\frac{1}{2}\rho_z^2 \ell^2 \widehat{g\omega}(0) 
-C \rho_z^2\ell^2 \delta\big( \varepsilon  + \delta^2 \rho_{\mu} R^2 + \delta|\log (d s K_{\ell})|\big) - C\delta \ell^2(d s\ell)^{-4}.
\end{equation}
Due to relation \eqref{eq:relRrho_mu} the second term gives $\delta_3$, while the third one gives the last term in \eqref{eq:controlrhofar}.

We continue considering the third term in \eqref{eq:Atildeestimate} and adding it to the $\mathcal Q_3$. The latter is an integral for $k \in \mathcal{P}_H$, and dropping the part of the $\tau_k$ for $k \in \mathcal{P}_H^c$ and using that for $k \in \mathcal{P}_H $ then $\tau_k \geq |k|^2/2$, we have to estimate 
\begin{equation}\label{eq:Q3+k2_expression}
\frac{\ell^2}{(2\pi)^2} \int_{k \in \mathcal{P}_H}  \, \Big( \frac{\varepsilon}{2} k^2 a^{\dagger}_k a_k + \frac{1}{(2\pi)^2} \int \, f_L(p) \widehat{W}_1(k) (\bar{z} \widetilde{a}^{\dagger}_p a_{p-k} a_k + a^{\dagger}_k a^{\dagger}_{p-k} \widetilde{a}_p z) \Big)\dd p \dd k.
\end{equation}
We complete the square in the previous expression, introducing the operators
\begin{equation}
b_k := a_k + \frac{2}{(2\pi)^2} \int \,f_L(p) \frac{\widehat{W}_1(k)}{\varepsilon |k|^2} z a_{p-k}^{\dagger} \widetilde{a}_p \dd p,
\end{equation}
so that 
\begin{align*}
\eqref{eq:Q3+k2_expression} &= \frac{\ell^2}{(2\pi)^2} \int_{k \in \mathcal{P}_H} \, \Big( \frac{\varepsilon}{2} k^2 b^{\dagger}_k b_k - \frac{2|z|^2}{\varepsilon(2\pi)^4} \iint  f_L(p) f_L(s) \frac{\widehat{W}_1(k)^2}{k^2} \widetilde{a}_s^{\dagger} a_{s-k} a^{\dagger}_{p-k}\widetilde{a}_p\Big) \dd p \dd s \dd k \\
&\geq - \frac{2|z|^2\ell^2}{\varepsilon(2\pi)^6} \int_{k \in \mathcal{P}_H} \, \frac{\widehat{W}_1(k)^2}{k^2} \iint  f_L(p) f_L(s) \widetilde{a}_s^{\dagger}(a^{\dagger}_{p-k} a_{s-k} + [a_{s-k}, a^{\dagger}_{p-k}] )\widetilde{a}_p \dd p \dd s \dd k.
\end{align*}
For the term without commutator, estimated on a state $\Phi$ which satisfies \eqref{def:statefiniteexcit} and using Cauchy-Schwarz
\begin{equation}
\widetilde{a}_s^{\dagger}a^{\dagger}_{p-k} a_{s-k}\widetilde{a}_p \leq C (\widetilde{a}_s^{\dagger}a^{\dagger}_{p-k} a_{p-k}\widetilde{a}_s + \widetilde{a}_p^{\dagger}a^{\dagger}_{s-k} a_{s-k}\widetilde{a}_p),
\end{equation}
we have
\begin{align}
\frac{2|z|^2\ell^2}{\varepsilon(2\pi)^6}&\left\langle  \int_{k \in \mathcal{P}_H} dk\, \frac{\widehat{W}_1(k)^2}{k^2} \iint  f_L(p) f_L(s) \widetilde{a}_s^{\dagger}a^{\dagger}_{p-k} a_{p-k}\widetilde{a}_s \dd p\dd s\right\rangle_{\Phi}   \nonumber \\
&\leq C|z|^2 \varepsilon^{-1} \frac{\ell^4 \delta^2}{K_H^2} \left\langle  \int_{k \in \mathcal{P}_H}  \int   f_L(s) \widetilde{a}_s^{\dagger}a^{\dagger}_{k} a_{k}\widetilde{a}_s \dd s \dd k \right\rangle_{\Phi} \int_{p \in \mathcal{P}_L} \dd p \nonumber \\
&\leq C \varepsilon^{-1} \frac{\delta^2}{K_H^2} d^{-4} \mathcal{M} \rho_{\mu}\rho_z\ell^2, \label{eq:proofAtilde3}
\end{align}
where we used Lemma~\ref{lem:numbercontrolhigh} since the support of $f_L$ is included in the complement of $\mathcal{P}_H$, and the estimate, for $k \in \mathcal{P}_H$,
\begin{equation}
\frac{\widehat{W}_1(k)^2}{2 k^2} \leq C K_H^{-2} \delta^2 \ell^2.
\end{equation}
For the commutator part we use the estimate \eqref{eq:commut_abound}, the Cauchy-Schwarz inequality
\begin{equation}
\widetilde{a}_s^{\dagger} [a_{s-k}, a^{\dagger}_{p-k}] \widetilde{a}_p \leq C \widetilde{a}_s^{\dagger}\widetilde{a}_s + C\widetilde{a}_p^{\dagger} \widetilde{a}_p, 
\end{equation}
and Lemma~\ref{eq:gdecay} applied to $\widehat{W}_1$ instead of $\widehat{g}$ paying a small error, we get
\begin{multline}
- \frac{2|z|^2\ell^2}{\varepsilon(2\pi)^6} \left\langle \int_{k \in \mathcal{P}_H} \frac{\widehat{W}_1(k)^2}{k^2} \iint  f_L(p) f_L(s)  \widetilde{a}_s^{\dagger} [a_{s-k}, a^{\dagger}_{p-k}] \widetilde{a}_p \dd p \dd s \dd k \right\rangle_{\Phi}  \\
\geq - C\frac{|z|^2\ell^2}{\varepsilon} \delta  \left\langle  \iint  f_L(p) f_L(s) \widetilde{a}_p^{\dagger} \widetilde{a}_p \dd p\dd s \right\rangle_{\Phi}  \geq - C \varepsilon^{-1}\rho_z \delta \mathcal{M} d^{-4}, \label{eq:proofAtilde4}
\end{multline}
where in the last inequality we used Lemma~\ref{lem:numbercontrolhigh}.

Collecting formulas \eqref{eq:proofAtilde0}, \eqref{eq:proofAtilde1}, \eqref{eq:proofAtilde2}, \eqref{eq:proofAtilde3} and \eqref{eq:proofAtilde4} and observing that 
\begin{equation}\label{eq:squarecompleteg}
\frac{1}{2} \rho_z^2 \ell^2 \widehat{g}_0 -\rho_z\rho_{\mu} \ell^2 \widehat{g}_0 = \frac{1}{2} (\rho_z - \rho_{\mu})^2 \ell^2 \widehat{g}_0 - \frac{1}{2} \rho_{\mu}^2\ell^2 \widehat{g}_0,
\end{equation} 
we obtain the result. 
\end{proof}

\begin{proof}[Proof of Proposition~\ref{propos:rhofar}]
We observe that, thanks to the relations \eqref{eq:rel_d_k_H}, \eqref{eq:rel_T2comm}, \eqref{eq:rel_epsT_d_Kl}, we have $\delta_j \ll 1$ for $j =1,2,3$.
Each coefficient of the $\delta_j$ in formula \eqref{eq:controlrhofar} can be bounded by 
 \begin{equation}
 C \delta (\rho_{\mu} -\rho_z)^2\ell^2 +C \rho_{\mu}^2 \ell^2 \delta.
 \end{equation}
 Therefore, Lemma~\ref{lem:technicalest_rhofar} and $\widehat{g}_0 = 8 \pi \delta$ implies the bound 
\begin{align*}
\langle \mathcal{K}(z)\rangle_{\Phi} \geq& -\frac{1}{2}  \rho_{\mu}^2 \ell^2 \widehat{g}_0 + \frac{1}{2}  (\rho_{\mu} - \rho_z)^2 \ell^2 \widehat{g}_0( 1- C (\delta_1 + \delta_2 + \delta_3 ) )  \\
&- C \rho_{\mu}^2 \ell^2 \delta (\delta_1 + \delta_2 + \delta_3 + \delta^2 (K_\ell d s)^{-4}) \\
\geq&  -\frac{1}{2}  \rho_{\mu}^2 \ell^2 \widehat{g}_0 +  \frac{1}{4}   \ell^2 \widehat{g}_0(\rho_{\mu} - \rho_z)^2 - C \rho_{\mu}^2 \ell^2 \delta (\delta_1 + \delta_2 + \delta_3 + \delta^{2}(K_\ell d s)^{-4}).
\end{align*}
Note that $\delta^{2}(K_\ell d s)^{-4} \ll \delta$ due to \eqref{eq:K_Bsmallbox} and \eqref{eq:Kl_KN_number}. By the assumption on $(\rho_{\mu}-\rho_z)^2$ the second term is of higher order both of the $\delta_j$ errors and of the desired quantity in the statement of the Proposition.
\end{proof}
\subsection{Estimate of \texorpdfstring{$\mathcal{K}$}{K} for \texorpdfstring{$\rho_{z} \simeq \rho_{\mu}$}{rho\_z =  rho\_mu}}
\label{sec.9.2}

We study here the main case, that is when $\rho_{z}$ is close to $\rho_{\mu}$. More precisely, we consider the complementary situation to \eqref{eq:conditionrhofar}, when
\begin{equation}\label{eq:rho_mu_control_close}
|\rho_{\mu} - \rho_z| \leq K_{\ell}^{-2} \rho_{\mu},
\end{equation}
where we used that, thanks to the choices of the parameters \eqref{eq:rel_T2comm}, \eqref{eq:Kl_KN_number} and \eqref{eq:rel_epsT_d_Kl}, we have 
\begin{equation}\label{eq:boundcondition_rhodelta}
K_{\ell}^2 \max\Big( (\delta_1 + \delta_2 + \delta_3)^{1/2}, \delta^{1/2}\Big) \leq C^{-1}.
\end{equation}

Using again \eqref{eq:squarecompleteg} and reabsorbing the term $(\rho_z - \rho_{\mu}) \widehat{W}_1(0) \frac{\ell^2}{(2\pi)^2} \int a^{\dagger}_k a_k \dd k $ in part of the spectral gap of $n_+$, we have the estimate of $\mathcal{K}(z)$ from \eqref{def:Kz},
\begin{align}
\mathcal{K}(z) \geq& -\frac{1}{2} \rho_{\mu}^2 \ell^2 \widehat{g}_0 + \frac{1}{2} \rho_z^2 \ell^2 \widehat{g\omega}(0) + \frac{1}{2} (\rho_z- \rho_{\mu})^2 \ell^2 \widehat{g}_0 \nonumber\\
&+\mathcal{K}^{\text{Bog}} + \frac{b}{4\ell^2} n_+ +b \frac{\varepsilon_T}{8 d^2 \ell^2 } n_+^H +b \frac{\varepsilon_T \vert z \vert^2 n_+^H}{16 d^2 \ell^2 (\rho_{\mu} \ell^2)} +  \varepsilon_R (\rho_{\mu} - \rho_z )^2 \delta \ell^2 \nonumber \\
&+ \mathcal Q_1^{\text{ex}}(z) + \mathcal Q_2^{\text{ex}}(z) + \mathcal Q_3(z),\label{ine:Kz}
\end{align}
and in the following we want to give a lower bound for the expression above using a diagonalization method for the Bogoliubov Hamiltonian. In order to do that, let us introduce a couple of new creation and annihilation operators
\begin{equation}\label{eq:defboperator}
b_k :=\frac{1}{\sqrt{1-\alpha_k^2}}  (a_k + \alpha_k a^{\dagger}_{-k} + c_k), 
\end{equation}
where 
\begin{align*}
\alpha_k &:= \mathcal{B}(k)^{-1} \Big( \mathcal{A}(k) - \sqrt{\mathcal{A}(k)^2 - \mathcal{B}(k)^2}\Big),\\
c_k &:= \frac{2 \mathcal{C}(k)}{\mathcal{A}(k) + \mathcal{B}(k) +  \sqrt{\mathcal{A}(k)^2 - \mathcal{B}(k)^2}} \one_{\{|k| \leq \frac{1}{2}K_H \ell^{-1}\}},
\end{align*}
and the diagonalized Bogoliubov Hamiltonian
\begin{equation}\label{def:Khdiag}
\mathcal{K}^{\rm{Diag}}_H := \frac{\ell^2}{(2\pi)^2} \int_{\{ \vert k \vert \geq \frac 1 2 K_H \ell^{-1} \}} \mathcal{D}(k) b^{\dagger}_k b_k \dd k,
\end{equation}
where 
\begin{equation}
\mathcal{D}(k) := \frac{1}{2} \big( \mathcal{A}(k) + \sqrt{\mathcal{A}(k)^2 - \mathcal{B}(k)^2}\big).
\end{equation}

\begin{theorem}\label{thm:lower_rho=rho_mu}
Assume the relations between the parameters in Appendix~\ref{app:parameters}. For any state $\Phi \in \mathscr{F}_s(L^2(\Lambda))$ such that \eqref{def:statefiniteexcit} holds and $\frac{9}{10} \rho_{\mu} \leq \rho_z \leq \frac{11}{10}\rho_{\mu}$ we have
\begin{align*}
\langle \mathcal{K}^{\rm{Bog}}& \rangle_{\Phi} + \frac{1}{2} \rho_z^2 \ell^2 (\widehat{g \omega})_0 + \frac{1}{2} (\rho_z- \rho_{\mu})^2 \ell^2 \widehat{g}_0 \\
&\geq  (1 - \varepsilon_K) \left\langle \mathcal{K}^{\rm{Diag}}_H\right\rangle_{\Phi} + 4 \pi \Big( 2 \Gamma + \frac 1 2 + \log \pi \Big) \rho^2_{z} \ell^2 \delta^2 \\ &\quad - C (\rho_{\mu}-\rho_z)^2 \ell^2 \delta^2 \rho_{\mu} R^2
-C \rho_{\mu}^2 \ell^2 \delta ( K_H^{4-M}K_\ell \delta^{-1/2}) + C r(\rho_{\mu}) \ell^2,
\end{align*}
where the error term is given by 
\begin{align*}
r(\rho_{\mu}) :=  \rho^2_{\mu} \delta^2 \big( \delta \vert \log(\delta) \vert R^2 \rho_{\mu} +\delta \vert \log (\delta ) \vert + d + \varepsilon_T \vert \log \delta \vert + (sK_\ell)^{-1} +   \varepsilon_N \delta^{-1}\big). 
\end{align*}
\end{theorem}

In the proof of Theorem~\ref{thm:lower_rho=rho_mu} we are going to use the following formulas and estimates for the commutators of the operators, recalling that $\widehat{\chi}_{\Lambda}$ is even,
\begin{align}
[b_k,b_h] &= \frac{\alpha_k - \alpha_h}{\sqrt{1-\alpha_k^2}\sqrt{1-\alpha_h^2}} \Big( \widehat{(\chi^2)}((k+h)\ell) -\widehat{\chi}(k \ell) \widehat{\chi}(h \ell)\Big), \label{eq:commut_bb}\\
[b_k,b_h^{\dagger}] &= \frac{1 - \alpha_k\alpha_h}{\sqrt{1-\alpha_k^2}\sqrt{1-\alpha_h^2}} \Big( \widehat{(\chi^2)}((k-h)\ell) -\widehat{\chi}(k \ell) \widehat{\chi}(h \ell)\Big),\label{eq:commut_bbdagger}\\
[\widetilde{a}^{\dagger}_p,b^{\dagger}_k] &= \frac{\alpha_{k} }{\sqrt{1-\alpha^2_k}}[\widetilde{a}_p^{\dagger}, a_k]= \frac{\alpha_{k} }{\sqrt{1-\alpha^2_k}} \ell^{-2} \langle e^{ipx}, Q \chi_{\Lambda} e^{ikx}\rangle,\label{eq:commut_astarbstar}\\
[\widetilde{a}_p, b^{\dagger}_{-k}]&= \frac{1}{\sqrt{1-\alpha_k^2}}[\widetilde{a}_p,a^{\dagger}_{-k}]  = \frac{1}{\sqrt{1-\alpha_k^2}}(\widehat{\chi}((p+k) \ell) - \widehat{\theta}(p\ell)\widehat{\chi}(k\ell)).\label{eq:commut_abstar}
\end{align}

\begin{proof}
Let us start by showing that the contribution coming from the $\mathcal{C}(k)$ gives an error term for $|k| > \frac{1}{2} K_H \ell^{-1}$.

By Cauchy-Schwarz we have $a^{\dagger}_k + a_k \leq a^{\dagger}_k a_k +1 $ and then we recognize $n_+$ \eqref{def:n0n+},
\begin{align*}
 \frac{\ell^2}{2(2\pi)^2}&\int_{\{|k| > \frac{1}{2} K_H \ell^{-1}\}}  \mathcal{C}(k) (a^{\dagger}_k + a^{\dagger}_{-k} + a_k + a_{-k}) \, \dd k \\
 &\geq -C |\rho_{\mu}-\rho_z| |\widehat{W}_1(0)| |z| \int_{\{|k| > \frac{1}{2} K_H \ell^{-1}\}} |\widehat{\chi_{\Lambda}}(k)| (a^{\dagger}_k a_k + 1) \,\dd k \\
 &\geq - C \rho_{\mu} \delta |z| (n_+ +1) K_H^{4-M},
\end{align*}
where we use the assumption on $\rho_z$ and that by Lemma~\ref{lem:localization_properties},
\begin{equation}
\ell^{-2} \sup_{|k| > \frac{1}{2}K_H \ell^{-1}} (1+ (k\ell)^2)^2 |\widehat{\chi_{\Lambda}}(k)| \leq C K_H^{4-M}.
\end{equation}
When we apply to $\Phi$ we have $n_+ \leq 2 \rho_{\mu}\ell^2$ and
\begin{equation}
\frac{\ell^2}{2(2\pi)^2}\int_{\{|k| > \frac{1}{2} K_H \ell^{-1}\}}  \mathcal{C}(k) \langle a^{\dagger}_k + a^{\dagger}_{-k} + a_k + a_{-k} \rangle_\Phi \, \dd k \geq - C \rho_\mu^2 \ell^2 \delta (K_H^{4-M} \sqrt{\rho_\mu} \ell).
\end{equation}
Therefore
\begin{equation}
\mathcal{K}^{\text{Bog}} \geq \widetilde{\mathcal{K}}^{\text{Bog}} -C \rho_{\mu}^2 \ell^2 \delta ( K_H^{4-M} K_\ell \delta^{-1/2}),
\end{equation}
where $\widetilde{\mathcal{K}}^{\text{Bog}}$ is the same as $\mathcal{K}^{\text{Bog}}$ but with $\mathcal{C}(k)$ substituted by
\begin{equation}
\widetilde{\mathcal{C}}(k) := \mathcal{C}(k) \one_{\{|k|\leq \frac{1}{2}K_H \ell^{-1}\}}.
\end{equation}
The bound on the commutator \eqref{eq:commut_abound} allows us to use Theorem~\ref{thm:bogdiag} to diagonalize the Bogoliubov Hamiltonian
\begin{align*}
\widetilde{\mathcal{K}}^{\text{Bog}} &\geq  \widetilde{\mathcal{K}}^{\text{Diag}} - \frac{\ell^2}{2(2\pi)^2} \int_{\mathbb{R}^2} \Big( \mathcal{A}(k) - \sqrt{\mathcal{A}(k)^2 - \mathcal{B}(k)^2}\Big) \dd k \\
&\quad - (\rho_z - \rho_{\mu})^2 \widehat{W}_1(0)^2  \frac{z^2}{(2\pi)^2 \ell^2} \int_{\{|k|\leq \frac{1}{2}K_H \ell^{-1}\}}  \frac{|\widehat{\chi_{\Lambda}}(k)|^2}{\mathcal{A}(k) + \mathcal{B}(k)} \dd k,
\end{align*}
where 
\begin{align}
\widetilde{ \mathcal K}^{\text{Diag}} = \frac{\ell^2}{(2\pi)^2} \int (1-\alpha_k^2 ) \mathcal D_k b_k^\dagger b_k \dd k \geq \frac{\ell^2}{(2\pi)^2} \int_{\{ \vert k \vert > \frac 1 2 K_H \ell^{-1} \}} (1-\alpha_k^2 ) \mathcal D_k b_k^\dagger b_k \dd k.
\end{align}
Using the inequality $\vert \alpha_k \vert \leq C \rho_z \delta k^{-2} \leq C K_\ell^2 K_H^{-2}$ we find
\begin{equation}
\widetilde{ \mathcal K}^{\rm{Diag}} \geq \mathcal K^{\rm{Diag}}_H (1 - C K_\ell^4 K_H^{-4} ).
\end{equation}

The calculation of the Bogoliubov integral is given in Appendix~\ref{app:bogintegral}. Combining the results of Lemma~\ref{lem:Bogintegral_approx_tau_k}, Lemma~\ref{lem:Bogintegral_approx_W_g} and Proposition~\ref{prop:integralapprox2} and multiplying everything by $\ell^2$ we find
\begin{align*}
- \frac{\ell^2}{2(2\pi)^2} \int_{\mathbb{R}^2} &\Big( \mathcal{A}(k) - \sqrt{\mathcal{A}(k)^2 - \mathcal{B}(k)^2}\Big) \dd k+ \frac{1}{2}\widehat{g\omega}(0) \rho_z^2 \ell^2 \\
&\geq 4 \pi \Big( 2 \Gamma + \frac 1 2 + \log \pi \Big) \rho_z^2 \ell^2 \delta^2 + r(\rho_{\mu}) \ell^2, 
\end{align*}
where $r(\rho_{\mu})$ is defined in the statement of the theorem. 
For the remaining term we use the estimate 
\begin{equation}
\mathcal{A}(k) + \mathcal{B}(k) \geq 2 \rho_z \widehat{W}_1(k) \geq 2 \rho_z \widehat{W}_1(0) (1 - C\delta(kR)^2),
\end{equation}
where we used a Taylor expansion and the fact that $W_1$ is even. By this last estimate, together with Lemma~\ref{lem:localization_properties} and \eqref{eq:Wg_approx} we obtain
\begin{align*}
- (\rho_z - \rho_{\mu})^2 \widehat{W}_1^2(0) & \frac{z^2}{(2\pi)^2 \ell^2} \int_{\{|k|\leq \frac{1}{2}K_H \ell^{-1}\}}  \frac{|\widehat{\chi_{\Lambda}}(k)|^2}{\mathcal{A}(k) + \mathcal{B}(k)}
\\
&\geq - (\rho_z - \rho_{\mu})^2 \frac{\widehat{W}_1(0)}{2} \ell^2 \big( 1 + C \rho_{\mu}\delta^2 R^2 K_H^2 K_\ell^{-2} \big)\\
&\geq - (\rho_z - \rho_{\mu})^2 \frac{\widehat{g}(0)}{2} \ell^2 \big( 1 + C \rho_{\mu} \delta R^2 \big),
\end{align*}
where in the last line we used $K_H \ll \delta^{-1/2}$ from \eqref{cond:3Qloc}.
\end{proof}

\subsection{Contribution of \texorpdfstring{$\mathcal Q_3$}{Q\_3}}\label{subsec:Q3_technical}

The aim of this section is to bound the $3Q$ term from below, namely
\[ \mathcal Q_3(z) = \frac{\bar z \ell^2}{(2\pi)^4} \int_{\mathcal P_H \times \R^2} \widehat W_1(k) f_L(p) (\widetilde{a}_p^\dagger a_{p-k} a_k + \text{h.c.})  \dd k \dd p, \]
which turns out to be controlled by quadratic Hamiltonian $\mathcal{K}_H^{\rm{Diag}}$ defined in \eqref{def:Khdiag}, absorbing $\mathcal Q_2^{\rm{ex}}$ and $\mathcal Q_1^{\rm{ex}}$. More precisely we prove

\begin{theorem}\label{thm.Q3z}
Assume the relations between the parameters in Appendix~\ref{app:parameters} to be satisfied. Then there exist a universal constant $C >0$ such that for any state $\Phi$ satisfying \eqref{def:statefiniteexcit} we have
\begin{align*}
 &\left\langle (1- \varepsilon_K) \mathcal{K}_H^{\rm{Diag}} + \mathcal Q_3(z) + \mathcal Q_2^{ex} + \mathcal Q_1^{ex} + \frac{b}{100} \frac{n_+}{\ell^2} + \frac{\varepsilon_T b}{100}\frac{ n_+^{H}}{(d\ell)^{2}} \right\rangle_{\Phi} \\
 &\quad\quad\geq - C \rho_z^2 \ell^2 \delta^2 \Big( \delta K_H^{-8} K_{\ell}^{10} d^{-4}  +  \varepsilon_K^{-1} K_H^{-12} K_\ell^{10} d^{-8} +  d^{-8}K_{\ell}^2 K_H^{-4} + \varepsilon_K^{-1} K_H^{-2M-8} K_\ell^6 d^{-8} \\
 &\quad\quad\quad + \delta K_\ell^2 \vert \log \delta \vert^2 + \delta^{-1} K_\ell^2 d^{8M-2} \varepsilon_T^{-1} + \varepsilon_{\mathcal M}^{1/2} (K_\ell^4 K_H^{-4} + \delta^{-1} K_H^{-M} d^{-2}) \Big). \nonumber
\end{align*}
\end{theorem}
In order to prove this theorem, we start by rewriting $\mathcal Q_3(z)$ in terms of the $b_k$'s defined in \eqref{eq:defboperator}. Notice that $c_k = c_{p-k} = 0$ if $k \in \mathcal P_H$ and $p \in \mathcal P_L$, and
\begin{equation}
 a_k = \frac{b_k -\alpha_k b_{-k}^\dagger}{\sqrt{1- \alpha_k^2}}, \qquad a_{p-k} =  \frac{b_{p-k} -\alpha_{p-k} b_{k-p}^\dagger}{\sqrt{1- \alpha_{p-k}^2}}.
\end{equation}
Therefore,
\[a_{p-k} a_k = \frac{1}{\sqrt{1-\alpha_k^2}\sqrt{1-\alpha_{p-k}^2}} \big( b_{p-k} b_k - \alpha_k b_{p-k} b_{-k}^\dagger - \alpha_{p-k} b_{k-p}^\dagger b_k + \alpha_{p-k} \alpha_k b_{k-p}^\dagger b_{-k}^\dagger \big),\]
and $\mathcal Q_3(z) = \mathcal Q_3^{(1)} + \mathcal Q_3^{(2)} + \mathcal Q_3^{(3)} + \mathcal Q_3^{(4)}$ where
\begin{align}\label{defQ31}
\mathcal Q_3^{(1)} &= \frac{z \ell^2}{(2\pi)^4} \int_{\mathcal{P}_H \times \R^2}  \frac{f_L(p) \widehat{W}_1(k)}{\sqrt{1-\alpha_k^2}\sqrt{1-\alpha_{p-k}^2}} \big( \widetilde a_p^\dagger b_{p-k} b_k + \alpha_k \alpha_{p-k} \widetilde a_p^\dagger b_{k-p}^\dagger b_{-k}^\dagger +  h.c. \big), \\
\label{defQ32}
 \mathcal Q_3^{(2)} &= -\frac{z \ell^2}{(2\pi)^4} \int_{\mathcal{P}_H \times \R^2}   \frac{f_L(p) \widehat{W}_1(k) \alpha_k}{\sqrt{1-\alpha_k^2} \sqrt{1-\alpha_{p-k}^2}} \big( \widetilde{a}^\dagger_p b^\dagger_{-k} b_{p-k} + b_{p-k}^\dagger b_{-k} \widetilde{a}_p \big),\\
 \label{defQ33}
\mathcal Q_3^{(3)} &= -\frac{z \ell^2}{(2\pi)^4} \int_{\mathcal{P}_H \times \R^2}  \frac{f_L(p) \widehat{W}_1(k) \alpha_{p-k}}{\sqrt{1-\alpha_k^2} \sqrt{1-\alpha_{p-k}^2}} \big( \widetilde{a}^\dagger_p b^\dagger_{k-p} b_{k} + b_{k}^\dagger b_{k-p} \widetilde{a}_p \big),\\
\label{defQ34}
\mathcal Q_3^{(4)} &= - \frac{z \ell^2}{(2\pi)^4} \int_{\mathcal{P}_H \times \R^2}  \frac{f_L(p) \widehat{W}_1(k)}{\sqrt{1-\alpha_k^2} \sqrt{1-\alpha_{p-k}^2}} \alpha_k [b_{p-k} , b_{-k}^\dagger] ( \widetilde{a}_p^\dagger + \widetilde a_p) .
\end{align}

In the remaining of this section, we get lower bounds on those four terms (Lemmas~\ref{lem.Q31},~\ref{lem.Q32} and~\ref{lem.Q34} below) hence proving Theorem~\ref{thm.Q3z}.

We collect here some important technical estimates which are going to be useful in the following.

\begin{lemma}\label{lem:usefulQ3}
The following bounds hold: 
\begin{align}
&|\alpha_k| \leq  C\rho_z \delta |k|^{-2} \leq C K_{\ell}^2 K_H^{-2}, \quad \text{for } |k| \geq \frac{1}{2} K_H \ell^{-1},\label{eq:alpha_highcontrol} \\
&\mathcal{D}_k \geq \frac{1}{2}|k|^2 \geq \frac{1}{8} K_H^2 \ell^{-2}, \quad \text{for }|k| \geq \frac{1}{2} K_H \ell^{-1}, \label{eq:D_highcontrol}\\
&\bigg\vert \rho_z \widehat{(\omega W_1)}(0) -\frac{1}{(2\pi)^2}\int_{ \mathcal{P}_H} \widehat{W}_1(k) \alpha_k \dd k \bigg\vert \leq  C  \rho_z \delta^2 \vert \log \delta \vert, \label{eq:Walpha_highcontrol}\\
&\bigg\vert \widehat{(\omega W_1)}(0) -\frac{1}{(2\pi)^2}\int_{ \mathcal{P}_H} \frac{\widehat{W}_1(k)^2}{2 \mathcal{D}_k}  \dd k \bigg\vert \leq C   \delta^2 \vert \log \delta \vert, \label{eq:WD_highcontrol}
\end{align}
and
\begin{align}
\rho_z\frac{\ell^2}{(2\pi)^2}\int_{\mathbb{R}^2} \widehat{(W_1\omega)}(k) a^{\dagger}_k a_k \dd k &\geq \rho_z\widehat{(W_1\omega)}(0) \frac{\ell^2}{(2\pi)^2} \int_{\mathbb{R}^2} a^{\dagger}_k a_k \dd k -4 \rho_{z} \delta n_+^H \nonumber \\  &\quad - C \rho_z \delta d^{-2}  \frac{R}{\ell}n_+. \label{eq:Q2approx}
\end{align}
\end{lemma}

\begin{proof}
The first two inequalities are straightforward from the definitions of the terms. For the third one we split the difference in the following way,
\begin{align}
&\bigg\vert\rho_z \widehat{(\omega W_1)}(0) -\frac{1}{(2\pi)^2}\int_{k \in \mathcal{P}_H} \widehat{W}_1(k) \alpha_k \dd k\bigg\vert \nonumber \\
&\leq  C \bigg\vert \rho_z\int_{k \notin \mathcal P_H} \frac{\widehat{W}_1(k) \widehat{g}_k - \widehat{W}_1(0)\widehat{g}_0 \one_{\{|k| \leq \ell_{\delta}^{-1}\}}}{2 k^2}\dd k\bigg\vert 
+ C \left| \int_{k \in \mathcal P_H} \widehat{W}_1(k) \Big( \alpha_k - \rho_z \frac{\widehat{g}_k}{2 k^2} \Big)dk \right| 
\nonumber \\
&=: (I) + (II).
\end{align}
For the first integral we do a further splitting of the domain of integration, considering $(I) \leq (I,<) + (I,>)$ for $|k| \leq \ell_{\delta}^{-1}$ or otherwise, respectively. For $(I,<)$ we consider a Taylor expansion of the numerator and we get, recalling the symmetry of $g$ which in the integration drops the first order, 
\begin{equation}
(I,<) \leq C \rho_z R^2 \delta^2 \int_{\{|k| \leq \ell_{\delta}^{-1}\}} \leq C\rho_z R^2 \delta^2\ell^{-2}_{\delta}.
\end{equation}
For the $(I,>)$ we proceed by a direct calculation and obtain
\begin{equation}
(I,>) \leq C \rho_z \delta^2 \log(K_H \ell^{-1}\ell_{\delta}).
\end{equation}
Let us analyze the second integral. We have that $|\mathcal{B}_k / \mathcal{A}_k| \leq 1/2$ and therefore we can expand in the following way 
\begin{equation}
\widehat{W}_1(k) \alpha_k = \rho_z^{-1} \mathcal{A}_k \bigg( 1-  \sqrt{1-\frac{\mathcal{B}_k^2}{\mathcal{A}_k^2}}\bigg) \simeq \rho_z\frac{\widehat{W}_1(k)^2}{2 \mathcal{A}_k} + C \rho_z^3 \frac{\widehat{W}_1(k)^4}{\mathcal{A}_k^3}.
\end{equation}
We deduce
\begin{align*}
(II) &\leq C\bigg\vert \int_{k \in \mathcal{P}_H} \Big(\widehat{W}_1(k)  \alpha_k - \rho_z \frac{\widehat{W}_1(k)^2}{2 \mathcal{A}_k}\Big) \dd k \bigg\vert + C\rho_z\bigg\vert \int_{k \in \mathcal P_H} \widehat{W}_1(k)\bigg( \frac{\widehat{W}_1(k)}{2 \mathcal{A}_k} -  \frac{\widehat{g}_k}{2 k^2} \bigg) \dd k \bigg\vert\\
&\leq C \rho_z^3\int_{k \in \mathcal P_H} \frac{\widehat{W}_1(k)^4}{\mathcal{A}_k^3} \dd k+ C\rho_z\bigg\vert \int_{k \in \mathcal{P}_H} \widehat{W}_1(k)\bigg( \frac{\widehat{W}_1(k)}{2 \mathcal{A}_k} -  \frac{\widehat{g}_k}{2 k^2} \bigg) \dd k \bigg\vert  \\
&\leq C \rho_z^3 \ell^{4}\delta^4 K_H^{-4}  +  C\rho_z \bigg\vert \int_{k \in \mathcal P_H} \widehat{W}_1(k)^2\bigg( \frac{1}{2 \mathcal{A}_k} -\frac{1}{|k|^2} \bigg) \dd k\bigg\vert  \\ & \quad + C \rho_z \bigg\vert \int_{k \in \mathcal P_H} \bigg(  \widehat{W}_1(k)\frac{\widehat{W}_1(k)- \widehat{g}_k}{2 k^2} \bigg)\dd k \bigg\vert,
\end{align*}
where we used that $\mathcal{A}_k \geq \frac{1}{2}|k|^2$ for $k \in \mathcal P_H$. For the remaining terms, we use that in $\mathcal P_H$ we have $0 < k^2 - \tau_k \leq 2 |k| (d s\ell)^{-1}$,
\begin{align*}
C\rho_z \bigg\vert \int_{k \in \mathcal P_H} \frac{\widehat{W}_1(k)^2}{k^2}\bigg( \frac{k^2 -\mathcal{A}_k }{ \mathcal{A}_k} \bigg) \dd k \bigg\vert &\leq C \rho_z \int_{k \in \mathcal P_H} \frac{\widehat{W}_1(k)^2}{ k^2} \bigg( \frac{2|k|(ds\ell)^{-1}}{k^2} + \rho_z \frac{\widehat{W}_1(k)}{k^2}\bigg)\\
&\leq C \rho_z \delta^2 (ds)^{-1} K_H^{-1} + C\rho_z^2 \ell^2 \delta^3 K_H^{-2}.  
\end{align*}
By Cauchy-Schwarz inequality we get for the last term
\small
\begin{align*}
\rho_z \bigg\vert \int_{k \in \mathcal P_H} \bigg(  \widehat{W}_1(k)\frac{\widehat{W}_1(k)- \widehat{g}_k}{2 k^2} \bigg)\dd k \bigg\vert \leq  C \rho_z \delta \int_{k \in \mathcal P_H} \frac{\widehat{W}_1(k)^2}{2k^2} \dd k + C \rho_z \delta^{-1} \int_{k \in \mathcal P_H} \frac{(\widehat{W}_1(k) - \widehat{g}_k)^2}{2k^2}\dd k .
\end{align*}
\normalsize
We complete the domain of the integrals: by Lemma~\ref{lem:gomegaapprox} we get 
\begin{align*}
\rho_z \delta \int_{k \in \mathcal P_H} \frac{\widehat{W}_1(k)^2}{2k^2} \dd k &\leq C \rho_z \delta  \widehat{(g \omega)}_0 + C \rho_z \delta \int_{k \notin \mathcal P_H} \frac{\widehat{W}_1(k)^2 - \widehat{W}_1(0)^2 \one_{\{|k|\leq \ell_{\delta}^{-1}\}}}{2k^2} \dd k \\
&\leq C \rho_z\delta^2 + C \rho_z \delta^3  (R^2 \ell^{-2}_{\delta} + \log(K_H \ell^{-1}\ell_{\delta})),
\end{align*}
and
\begin{align*}
&\rho_z \delta^{-1} \int_{k \in \mathcal P_H} \frac{(\widehat{W}_1(k)- \widehat{g}_k)^2}{2k^2} \dd k \\
&\leq C\rho_z \delta^{-1} \frac{R^4}{\ell^4} \widehat{g\omega}(0) + C \rho_z \delta^{-1} \bigg\vert \int_{k \notin \mathcal P_H} \frac{(\widehat{W}_1(k)-\widehat{g}_k )^2 - (\widehat{W}_1(0) -\widehat{g}_0 )^2 \one_{\{|k|\leq \ell_{\delta}^{-1}\}}}{2k^2} \dd k \bigg\vert\\
&\leq C\rho_z\frac{R^4}{\ell^4} + C \rho_z \delta \Big(R^2 \ell^{-2}_{\delta} + \frac{R^2}{\ell^2}\log(K_H \ell^{-1}\ell_{\delta})\Big).
\end{align*}
We conclude the proof of \eqref{eq:Walpha_highcontrol} by collecting all the previous estimates and exploiting the relations between the parameters so that $\rho_z \delta^2 \vert \log \delta \vert$ is the dominant term.

For the inequality \eqref{eq:WD_highcontrol}, we can derive it from \eqref{eq:Walpha_highcontrol} and the control on the first term of $(II)$ above using that, for $k \in \mathcal P_H$, $\left| 1- \frac{\mathcal{A}_k}{\mathcal{D}_k} \right|\leq \frac{\mathcal{B}_k^2}{\mathcal{A}_k} \leq C \rho_{z}^2 \delta^2 |k|^{-4} $.

For the last inequality, we estimate the difference, splitting the integral for $|k|\leq \xi \ell^{-1}$ or otherwise,
\begin{align*}
\rho_z&\frac{\ell^2}{(2\pi)^2}\int_{\mathbb{R}^2} (\widehat{(W_1\omega)}(k) -\widehat{(W_1\omega)}(0))a^{\dagger}_k a_k \dd k \\
&\geq -C \rho_z \xi^2\delta \frac{R^2}{\ell^2} n_+  - \frac{2\ell^2}{(2\pi)^2} \rho_z \delta  \int_{\mathbb{R}^2} a^{\dagger}_k \one_{\{|k| \geq \xi \ell^{-1}\}} a_k \dd k
\end{align*}
where we used a Taylor expansion and estimated the integral for $|k| \leq \xi \ell^{-1}$. For the second term we exploit the second quantization in a $N-$bosons sector and we insert symmetrically the sum of projectors $1 = \one_{\{\sqrt{-\Delta} \in \mathcal P_L\}} + \one_{\{\sqrt{-\Delta} \in \mathcal P_L^c\}}$
\begin{align*}
\frac{\ell^2}{(2\pi)^2}\int_{\mathbb{R}^2} a^{\dagger}_k \one_{\{|k| \geq \xi \ell^{-1}\}} a_k \dd k\Big|_N &= \sum_{j=1}^N Q_j \chi_{\Lambda}(x_j) \one_{\{\sqrt{-\Delta_j} \geq \xi \ell^{-1}\}} \chi_{\Lambda}(x_j) Q_j \\
&\geq 2 n_+^H  + 2 \mathcal{N} n_+
\end{align*}
where we estimated by a Cauchy-Schwarz the cross terms $(\mathcal P_L, \mathcal P_L^c)$ to make them comparable to the diagonal terms and denoted by  
\begin{equation}
\mathcal{N} := \|\one_{\{\sqrt{-\Delta} \in \mathcal P_L\}} \chi_{\Lambda} (x) \one_{\{\sqrt{-\Delta} \geq \xi \ell^{-1}\}}\|^2 \leq C \xi^{-2} d^{-4}.
\end{equation}
Here we used the regularity properties of $\chi_{\Lambda}$ dividing and multiplying by $-\Delta$. We conclude optimizing $\xi$ by the choice $\xi^2 = d^{-2} \frac{\ell}{R }$.
\end{proof}

\subsubsection{Estimates on \texorpdfstring{$\mathcal Q_3^{(1)}$}{Q\_3\^1}}

The first part $\mathcal Q_3^{(1)}$ will absorb $\mathcal Q_2^{\rm{ex}}$ using $\mathcal K_H^{\rm{Diag}}$.

\begin{lemma}[Estimates on $\mathcal Q_3^{(1)}$] \label{lem.Q31}
For any state $\Phi$ satisfying \eqref{def:statefiniteexcit} we have
\begin{multline*}
\left\langle \mathcal Q_3^{(1)} + \mathcal Q_2^{ex} + \big(1- 2 \varepsilon_K \big) \mathcal K^{\rm{Diag}}_H + \frac{b}{100} \frac{n_+}{\ell^2} + \frac{b}{100}\frac{\varepsilon_T n_+^{H}}{(d\ell)^{2}} \right\rangle_{\Phi} \\
 \geq  - C \rho_z^2 \ell^2 \delta^2 \Big( \delta K_H^{-8} K_{\ell}^{10} d^{-4}  +  \varepsilon_K^{-1} K_H^{-12} K_\ell^{10} d^{-8} +  d^{-8}K_{\ell}^2 K_H^{-4} \Big).
\end{multline*}
\end{lemma}

\begin{proof}
We first reorder the creation an annihilation operators, applying a change of variables $k \mapsto -k, p \mapsto -p$ in the $\alpha$ terms,
\begin{align*}
\mathcal Q_3^{(1)} &= \frac{z \ell^2}{(2\pi)^4} \int_{\mathcal{P}_H \times \R^2} \frac{f_L(p) \widehat W_1(k)}{\sqrt{1-\alpha_k^2}\sqrt{1-\alpha_{p-k}^2}} \\
& \quad \times \big( \widetilde a_p^\dagger b_{p-k} b_k + \alpha_k \alpha_{p-k} \widetilde a_{-p}^\dagger b_{p-k}^\dagger b_{k}^\dagger + b_k^\dagger b_{p-k}^\dagger \widetilde a _p + \alpha_k \alpha_{p-k} b_k b_{p-k} \widetilde{a}_{-p} \big)  \dd k \dd p \\
&= \frac{z \ell^2}{(2\pi)^4} \int_{\mathcal{P}_H \times \R^2} \frac{f_L(p) \widehat W_1(k)}{\sqrt{1-\alpha_k^2}\sqrt{1-\alpha_{p-k}^2}} \Big( \big( \widetilde a_p^\dagger b_{p-k} + \alpha_k \alpha_{p-k} b_{p-k} \widetilde a_{-p} \big) b_k  \\ 
& \quad + b_k^\dagger \big( b_{p-k}^\dagger \widetilde a_p + \alpha_k \alpha_{p-k} \widetilde a^\dagger_{-p} b_{p-k}^\dagger \big) + \alpha_k \alpha_{p-k} \big( \left[ b_k , b_{p-k} \widetilde a_{-p} \right]  + \big[ \widetilde a _{-p}^\dagger b_{p-k}^\dagger, b_k^\dagger \big] \big) \Big)  \dd k \dd p .
\end{align*}
We can complete the square to get, for $\varepsilon_K \ll 1$ fixed in Appendix~\ref{app:parameters},
\begin{equation}\label{eq.Q31.low}
\mathcal Q_3^{(1)} + (1- 3 \varepsilon_K) \mathcal K^{\rm{Diag}}_H =  (1- 3 \varepsilon_K) \frac{\ell^2}{(2\pi)^2} \int_{\mathcal{P}_H}  \mathcal D_k c_k^\dagger c_k \dd k + \frac{\ell^2}{(2\pi)^2} \int_{\mathcal{P}_H} \big( \mathcal T_1(k) + \mathcal T_2(k) \big)\dd k, 
\end{equation}
where we mantained a small portion of $\mathcal{K}^{\rm{Diag}}_H$ in order to bound other error terms and we defined  
\begin{equation}
 c_k = b_k + \frac{z}{\mathcal D_k (1- 3\varepsilon_K) (2\pi)^2} \int \frac{f_L(p) \widehat W_1(k)}{\sqrt{1-\alpha_k^2}\sqrt{1-\alpha_{p-k}^2}} \Big( b_{p-k}^\dagger \widetilde a_p + \alpha_k \alpha_{p-k} \widetilde a_{-p}^\dagger b_{p-k}^\dagger \Big)   \dd p,
\end{equation}
\begin{equation}\label{term.T1}
\mathcal T_1(k) = \frac{z }{(2\pi)^2} \int_{ \R^2}  \frac{f_L(p) \widehat W_1(k)}{\sqrt{1-\alpha_k^2}\sqrt{1-\alpha_{p-k}^2}} \alpha_k \alpha_{p-k} \big( \big[ b_k^\dagger, \widetilde a_{-p}^\dagger b_{p-k}^\dagger \big] + h.c. \big) \dd p,
\end{equation}
\begin{align}\label{term.T2}
\mathcal T_2 (k) &= - \frac{\vert z \vert^2 \widehat W_1(k)^2}{(1-3\varepsilon_K) \mathcal D_k (1-\alpha_k^2)(2\pi)^4}  \int \frac{f_L(p) f_L(s)}{\sqrt{1-\alpha_{s-k}^2} \sqrt{1- \alpha_{p-k}^2}} \\ & \qquad \times \big( \widetilde a_{p}^\dagger b_{p-k}  + \alpha_k \alpha_{p-k} b_{p-k} \widetilde a_{-p} \big) \big( b_{s-k}^\dagger \widetilde{a}_s + \alpha_k \alpha_{s-k} \widetilde a_{-s}^\dagger b_{s-k}^\dagger \big)  \dd p \dd s.
\end{align}

$\bullet$ Let us estimate the error term $\mathcal T_1(k)$. We use 
$ \big[ b_k^\dagger, \widetilde a_{-p}^\dagger b_{p-k}^\dagger \big] = \widetilde{a}_{-p}^\dagger \big[ b_k^\dagger, b_{p-k}^\dagger \big] + \big[ b_k^\dagger, \widetilde a_{-p}^\dagger \big] b_{p-k}^\dagger $
and the Cauchy-Schwarz inequality with weights $\varepsilon_1$, $\varepsilon_2 >0$,
\begin{align*}
\mathcal T_1(k) &\geq - C z \int_{ \R^2}  \frac{f_L(p) \widehat W_1(k)}{\sqrt{1-\alpha_k^2}\sqrt{1-\alpha_{p-k}^2}} \vert \alpha_k \alpha_{p-k} \vert \\ 
& \quad \Big(  \big( \varepsilon_1 \widetilde a_{-p}^\dagger \widetilde a_{-p} + \varepsilon_1^{-1} \big) \vert [ b_k^\dagger, b_{p-k}^\dagger] \vert + \vert [b_k^\dagger, \widetilde a_{-p}^\dagger ] \vert \big( \varepsilon_2 b_{p-k}^\dagger b_{p-k}  + \varepsilon_2^{-1} \big) \Big)\dd p.
\end{align*}

By \eqref{eq:commut_astarbstar} and \eqref{eq:commut_bb} we have $ \vert [ b_k^{\dagger}, \widetilde{a}_{-p}^\dagger] \vert \leq C \vert \alpha_k \vert$ and $\vert [ b_k^{\dagger} , b_{p-k}^{\dagger} ] \vert \leq C \vert \alpha_k \vert$. Therefore using \eqref{eq:alpha_highcontrol},
\begin{align*}
\frac{\ell^2}{(2\pi)^2} \int_{\mathcal{P}_H} \mathcal T_1(k)  \dd k &\geq -C \frac{z \ell^2}{(2\pi)^4} \int_{\mathcal{P}_H \times \R^2} \frac{\vert f_L(p) \vert \rho_z^3 \delta^4 }{k^6} \\ 
& \quad \times \Big(  (\varepsilon_1 \widetilde a_{-p}^\dagger \widetilde a_{-p} + \varepsilon^{-1}_1) +  ( \varepsilon_2 b_{p-k}^\dagger b_{p-k} + \varepsilon_2^{-1} ) \Big)  \dd k \dd p .
\end{align*}
Due to the presence of the cutoff $f_L$ on low momenta and the bounds
\begin{align}
\int_{\mathcal{P}_L}  ( \varepsilon_1 \widetilde a_{-p}^\dagger \widetilde a_{-p} + \varepsilon_1^{-1} ) \dd p &\leq C \frac{\varepsilon_1 n_+}{\ell^2} + \varepsilon_1^{-1} \frac{d^{-4}}{\ell^2},\\
\int_{\mathcal{P}_L} ( \varepsilon_2 b_{k}^\dagger b_{k} + \varepsilon_2^{-1})\dd p &\leq C \frac{d^{-4}}{\ell^2} ( \varepsilon_2 b_k^\dagger b_k + \varepsilon_2^{-1}) ,
\end{align}
where we changed the $k$ variable, we find,
\begin{equation*}
\frac{\ell^2}{(2\pi)^2} \int_{\mathcal{P}_H}  \mathcal T_1(k) \dd k\geq 
 -C z \rho_z^3 \delta^4  \int_{\mathcal P_H} \frac{1}{k^6} \big(  (\varepsilon_1 n_+ + \varepsilon_1^{-1} d^{-4}) +  d^{-4} ( \varepsilon_2  b_k^\dagger b_k + \varepsilon_2^{-1} ) \big) \dd k.
\end{equation*}
We insert $\mathcal D_k \geq C^{-1} k^{2}$ in front of $b_k^\dagger b_k$ and get the bound
\begin{align*}
\frac{\ell^2}{(2\pi)^2} \int_{\mathcal{P}_H} \mathcal T_1(k) \dd k \geq & - C z \rho_z^3 \delta^4 \ell^6 K_H^{-4} \varepsilon_1 \frac{n_+}{\ell^2} - C z \varepsilon_1^{-1} d^{-4} \rho_z^3 \delta^4 \ell^4 K_H^{-4} \nonumber \\ 
& - C \varepsilon_2 \ell^6 z \rho_z^3 \delta^4 K_H^{-8} d^{-4}  \ell^2 \int_{\mathcal P_H}  \mathcal{D}_k b_k^\dagger b_k \dd k -C \varepsilon_2^{-1} \ell^4 z \rho_z^3 \delta^4 d^{-4} K_H^{-4} .
\end{align*}
One can choose $\varepsilon_1$, $\varepsilon_2$ such that the first and third terms are absorbed in the positive $\frac{b}{100} \frac{n_+}{\ell^2}$ and $\varepsilon_K \mathcal K_{H}^{\rm{Diag}}$ respectively. With this choice the second and fourth terms are errors of respective sizes
\[ C \ell^2 \rho_z^2  \delta^2 (\delta K_{\ell}^{10} K_H^{-8} d^{-4}) \quad \text{ and } \quad  C \ell^2 \rho_z^2 \delta^2 (  \delta K_{\ell}^{10} K_H^{-12} d^{-4}\varepsilon_K^{-1}) .\]

 $\bullet$ Let us now focus on the square term $\mathcal T_2 (k)$ in \eqref{term.T2}. One can write, in normal order,
\begin{align*}
\widetilde a_p^\dagger b_{p-k} + \alpha_k \alpha_{p-k} b_{p-k} \widetilde a_{-p} &= \widetilde a_p^\dagger b_{p-k} + \alpha_k \alpha_{p-k} \widetilde a_{-p} b_{p-k} + \alpha_k \alpha_{p-k}  [ b_{p-k}, \widetilde a_{-p} ] ,
\end{align*} 
and use the Cauchy-Schwarz inequality with weight $ \varepsilon_K$ on the cross terms to find
\begin{align}\label{eq.T2}
\mathcal{T}_2(k) &\geq (1+ \varepsilon_K) \mathcal T_2'(k) + (1+\varepsilon^{-1}_K) \mathcal T_2''(k),
\intertext{with}
\mathcal T_2'(k) &= -  \frac{\vert z \vert^2 \widehat W_1(k)^2 }{(1-3\varepsilon_K) \mathcal D_k (1-\alpha_k^2) (2\pi)^4} \int \frac{f_L(p) f_L(s)}{\sqrt{1-\alpha_{p-k}^2}\sqrt{1-\alpha_{s-k}^2}} \nonumber \\ 
& \qquad \times ( \widetilde a_p^\dagger + \alpha_k \alpha_{p-k} \widetilde a_{-p}) b_{p-k} b_{s-k}^\dagger (\widetilde a_s + \alpha_k \alpha_{s-k} \widetilde a _{-s}^\dagger )  \dd p \dd s,\nonumber \\
\mathcal T_2''(k) &= - \frac{\vert z \vert^2 \widehat W_1(k)^2}{(1-3\varepsilon_K) \mathcal D_k (1-\alpha_k^2) (2\pi)^4} \int  \frac{f_L(p) f_L(s)}{\sqrt{1-\alpha_{s-k}^2} \sqrt{1-\alpha_{p-k}^2}} \nonumber \\
& \qquad \times  \alpha_k^2 \alpha_{p-k} \alpha_{s-k} \vert [ b_{p-k}, \widetilde a_{-p}] \vert \vert [ \widetilde a_{-s}^\dagger, b_{s-k}^\dagger] \vert \dd p \dd s .\nonumber 
\end{align}
$\mathcal T_2''$ we can estimate (for $k \in \mathcal{P}_H$),
\begin{align*}
\frac{\ell^2}{(2\pi)^2}\int_{\mathcal{P}_H} \mathcal T_2 '' (k) \dd k&\geq - C \rho_z \ell^4 \Big( \int_{\mathcal{P}_H} \frac{\widehat W_1(k)^2}{(1-3\varepsilon_K) \mathcal D_k} \vert \alpha_k \vert^4 \dd k \Big) d^{-8} \ell^{-4} \sup \vert [ b_{p-k} , \widetilde a_{-s} ] \vert ^2 ,
\end{align*}
and by \eqref{eq:WD_highcontrol}, \eqref{eq:commut_astarbstar} and \eqref{eq:alpha_highcontrol} we get
\begin{equation}\label{eq.T2''}
(1+ \varepsilon_K^{-1}) \frac{\ell^2}{(2\pi)^2} \int_{\mathcal{P}_H} \mathcal T_2''(k) \dd k \geq -C \rho_z^2 \ell^2\delta^2( \varepsilon_K^{-1}  K_H^{-12} K_{\ell}^{10} d^{-8}).
\end{equation}
Now we use a commutator to write $\mathcal T_{2}' = \mathcal T_{2,\rm{op}}' + \mathcal T_{2,\rm{com}}'$ in normal order for the $b_k$, with
\begin{align}
\mathcal T_{2,\rm{op}}'(k) &=  - \frac{\vert z \vert^2 \widehat W_1(k)^2}{(2\pi)^4 (1-3 \varepsilon_K) \mathcal D_k (1-\alpha_k^2)} \int \frac{f_L(p) f_L(s)}{\sqrt{1-\alpha_{p-k}^2}\sqrt{1-\alpha_{s-k}^2}} \nonumber \\
& \qquad \times ( \widetilde a_p^\dagger + \alpha_k \alpha_{p-k} \widetilde a _{-p}) b_{s-k}^\dagger b_{p-k} ( \widetilde a _{s} + \alpha_k \alpha_{s-k} \widetilde{a}_{-s}^\dagger)  \dd p \dd s, \nonumber \\
\label{term.T2com}
\mathcal T_{2,\rm{com}}'(k) &=  - \frac{\vert z \vert^2 \widehat W_1(k)^2}{(2\pi)^4 (1-3\varepsilon_K) \mathcal D_k (1-\alpha_k^2)} \int  \frac{f_L(p) f_L(s)}{\sqrt{1-\alpha_{p-k}^2}\sqrt{1-\alpha_{s-k}^2}} \nonumber \\
& \qquad \times ( \widetilde a_p^\dagger + \alpha_k \alpha_{p-k} \widetilde a _{-p}) [ b_{p-k} , b_{s-k}^\dagger] ( \widetilde a _{s} + \alpha_k \alpha_{s-k} \widetilde{a}_{-s}^\dagger) \dd p \dd s .
\end{align}
$\bullet$ In order to estimate the error term $\mathcal T_{2,\rm{op}}'$, we introduce
\begin{equation}
\tau_s := \widetilde a _{s} + \alpha_k \alpha_{s-k} \widetilde{a}_{-s}^\dagger \quad \text{and} \quad \mathcal C := \sup_{p,s \in \mathcal{P}_L , k \in \mathcal{P}_H} \vert [  b_{p-k} , \tau_s ] \vert .
\end{equation}
In $\mathcal T_{2,\rm{op}}'$ we commute the $b$'s trough the $a$'s,
\begin{align*}
\tau_p^\dagger b_{s-k}^\dagger b_{p-k} \tau_s = b_{s-k}^\dagger \tau_p^\dagger \tau_s b_{p-k} + [\tau_p^\dagger, b_{s-k}^\dagger] \tau_s b_{p-k} + b_{s-k}^\dagger \tau_p^\dagger [b_{p-k}, \tau_s] +  [\tau_p^\dagger, b_{s-k}^\dagger] [b_{p-k}, \tau_s],
\end{align*}
and use the Cauchy-Schwarz inequality
\begin{align*}
\tau_p^\dagger b_{s-k}^\dagger b_{p-k} \tau_s \leq C( b_{s-k}^\dagger \tau_p^\dagger \tau_p b_{s-k} +  b_{p-k}^\dagger \tau ^\dagger_s \tau_s b_{p-k} + \mathcal C^2 ) .
\end{align*}
Inserting this in $\mathcal T_{2,\rm{op}}'$,  bounding $(1-3 \varepsilon_K)(1-\alpha_k) \geq 1/2$ and noticing that we can exchange $s$ and $p$ in the integral, we find
\begin{align*}
\mathcal T_{2,\rm{op}}'(k) \geq  -C \frac{\vert z \vert^2 \widehat W_1(k)^2}{\mathcal D_k } \int \frac{f_L(p) f_L(s)}{\sqrt{1-\alpha_{p-k}^2}\sqrt{1-\alpha_{s-k}^2}} (b_{s-k}^\dagger \tau_p^\dagger \tau_p b_{s-k} +  \mathcal C^2) \dd p \dd s.
\end{align*}
When we apply this operator to the state $\Phi$ which satisfies $\one_{[0,\mathcal M]}(n_+^L) \Phi = \Phi$ we can apply Lemma~\ref{lem:numbercontrolhigh} for the vector $b_{s-k} \Phi$ to get the estimate
\begin{align*}
 \langle \mathcal T_{2,\rm{op}}'(k) \rangle_{\Phi} \geq - C \frac{ \vert z \vert^2 \widehat W_1(k)^2}{\mathcal D_k} \Big( \ell^{-2} \mathcal M \int f_L(s) \langle b_{s-k}^\dagger b_{s-k}  \rangle_\Phi \dd s+ d^{-8} \ell^{-4} \mathcal C^2  \Big) ,
\end{align*}
and finally, using again \eqref{eq:WD_highcontrol} and \eqref{eq:D_highcontrol}, and the fact that $\mathcal{C} \leq C K_{\ell}^2 K_H^{-2}$ by \eqref{eq:commut_astarbstar} and Lemma~\ref{lem:localization_properties},
\begin{align}
\frac{\ell^2}{(2\pi)^2} \int_{\mathcal{P}_H} \langle \mathcal T_{2,\rm{op}}'(k) \rangle_{\Phi} \dd k &\geq - C \rho_z \ell^2 \delta^2 K_H^{-4} d^{-4} \mathcal M \langle \mathcal K_H^{\rm{Diag}} \rangle_{\Phi} - C \rho_z \delta d^{-8} K_{\ell}^4 K_H^{-4} \nonumber \\
&\geq-C K_{\ell}^4 K_H^{-4} d^{-4} \varepsilon_{\mathcal M} \mathcal K_H^{\rm{Diag}} - C \rho_z^2\ell^2\delta^2 d^{-8}K_{\ell}^2 K_H^{-4}. \label{eq.T2op}
\end{align}
The first part can be absorbed in the positive $\varepsilon_K \mathcal K_H^{\rm{Diag}}$, as long as the relation \eqref{eq:rel_epsK_KM2} holds, and the second part contributes to the error term.

$\bullet$ We now turn to $\mathcal T_{2,\rm{com}}'$ given in \eqref{term.T2com}. This term will absorb $\mathcal Q_2^{ex}$. We first use Lemma~\ref{lem:localization_properties}, \eqref{eq:commut_bbdagger} and \eqref{eq:alpha_highcontrol} to estimate the commutator,
\begin{align*}
\left\vert [\right.&\left.b_{p-k}, b_{s-k}^\dagger]-  \widehat{\chi^2}((p-s)\ell) \right\vert \\ &= \left\vert \alpha_{p-k} \alpha_{s-k} \widehat{\chi^2}((p-s)\ell)\right\vert + \left\vert (1- \alpha_{p-k} \alpha_{s-k}) \widehat{\chi} ((p-k)\ell) \widehat{\chi} ((s-k) \ell)  \right\vert \leq C K_{\ell}^{4}K_H^{-4},
\end{align*}
and bounding then by a Cauchy-Schwarz inequality
\begin{align*}
( \widetilde a_p^\dagger + \alpha_k \alpha_{p-k} \widetilde a _{-p})& (\widehat{(\chi^2)}((p-s)\ell) + C K_{\ell}^{4} K_{H}^{-4}) ( \widetilde a _{s} + \alpha_k \alpha_{s-k} \widetilde{a}_{-s}^\dagger) \\
&\leq  \widetilde a_p^\dagger \widehat{\chi^2}((p-s)\ell) \widetilde a_s + C (\widetilde a_{-p}^\dagger \widetilde a_{-p} + \widetilde{a}_s^\dagger \widetilde a_s ) K_\ell^4 K_H^{-4}.
\end{align*}
We get, by using Lemma~\ref{lem:numbercontrolhigh} 
\begin{align*}
\frac{\ell^2}{(2\pi)^2} \int_{\mathcal{P}_H} \mathcal T_{2,\rm{com}}'(k) \dd k &\geq 
 - \frac{\ell^2}{(2\pi)^2} \int_{\mathcal{P}_H} \frac{\vert z \vert^2 \widehat W_1(k)^2}{(2\pi)^4 (1-3\varepsilon_K) \mathcal D_k (1-\alpha_k^2)} \\
& \quad \times \int \frac{f_L(p) f_L(s)}{\sqrt{1-\alpha_{p-k}^2}\sqrt{1-\alpha_{s-k}^2}}  \widetilde a_p^{\dagger} \widehat{\chi^2}((p-s)\ell) \widetilde a_s   \dd p \dd s \dd k    \\
&\quad - C \Big( \int_{\mathcal{P}_H} \frac{|z|^2 \widehat W_1(k)^2}{\mathcal D_k}  \dd k \Big) d^{-4}  K_\ell^4 K_H^{-4} \frac{n_+}{\ell^2}.
\end{align*}
Using \eqref{eq:D_highcontrol} the last part is of order 
$K_{\ell}^6 K_H^{-4} d^{-4} \frac{n_+}{\ell^2}$
and can be absorbed in a fraction of the positive $\frac{b}{100} \frac{n_+}{\ell^2}$ by \eqref{eq:rel_T2comm}. For the first term we use the following formula valid in a Fock sector with $N$ bosons
\begin{equation}
\frac{\ell^4}{(2\pi)^4} \int f_L(p) f_L(s) \widetilde{a}_p^{\dagger}  \widehat{\chi^2}((p-s)\ell)  a_s \dd s \dd p|_N = \sum_{j=1}^N Q_{L,j}^{\dagger} \chi_\Lambda^2(x_j) Q_{L,j},
\end{equation}
to rewrite, due to \eqref{eq:rel_T2comm} and by \eqref{eq:D_highcontrol}, 
\begin{align*}
\frac{\ell^2}{(2\pi)^2} \int_{\mathcal{P}_H} \mathcal T_{2,\rm{com}}'(k) \dd k &\geq - \frac{(1+C\varepsilon_K )}{(2\pi)^{2}} \int  \frac{\rho_z \widehat W_1(k)^{2}}{\mathcal D_k (1-\alpha_k^2)} \dd k \sum_{j=1}^N Q_{L,j}^{\dagger} \chi_\Lambda^2(x_j) Q_{L,j} -\frac{b}{200}\frac{n_+}{\ell^2} \\
&\geq - (1+ C \varepsilon_K + C K_\ell^4 K_H^{-4}) \\ & \quad \times \Big( 2 \rho_z \widehat{(\omega W_1)}(0) + C \rho_z \delta^2 \vert \log \delta \vert \Big) \sum_{j=1}^N Q_{L,j}^{\dagger} \chi_\Lambda^2 Q_{L,j} -\frac{b}{200}\frac{n_+}{\ell^2}.
\end{align*} 
In this last expression we want to replace $Q_{L,j}$ by $Q_j$. Using Cauchy-Schwarz with weight $\varepsilon_0$ we find
\begin{equation}
Q_{L,j}^{\dagger} \chi_\Lambda^2 Q_{L,j} \leq (1+\varepsilon_0) Q_j \chi_\Lambda^2 Q_j + (1+\varepsilon_0^{-1}) Q_j (f_L - 1) \chi_j^2 (f_L - 1) Q_j,
\end{equation}
and since $f_L$ localizes on low momenta we can bound the second term by $n_+^H$, and $\varepsilon_0 Q_j \chi_j^2 Q_j$ by $C \varepsilon_0 n_+$,
\begin{align*}
\frac{\ell^2}{(2\pi)^2} \int_{\mathcal{P}_H} \mathcal  T_{2,\rm{com}}'(k)\dd k & \geq - 2 \rho_z \widehat{(\omega W_1)}(0) \sum_{j=1}^N \big( Q_j \chi_j^2 Q_j + C \varepsilon^{-1}_0 n_+^H + C \varepsilon_0 n_+ \big) \\
&\quad- C \Big(\rho_z \ell^2\delta^2 \vert \log \delta \vert  + \rho_z \ell^2\delta \varepsilon_K + \rho_z \ell^2 \delta K_\ell^4 K_H^{-4} + \frac{b}{200}\Big) \frac{n_+}{\ell^2}.
\end{align*}
The $n_+^H$-part can be absorbed by the positive $\frac{b}{100} \frac{\varepsilon_T n_+^H}{(d \ell)^2}$ if we choose $\varepsilon_0 \simeq \frac{\rho_z \delta d^2 \ell^2}{\varepsilon_T} \simeq \frac{d^2 K_\ell^2}{\varepsilon_T}$. With this choice the $n_+$ terms are of order
\begin{equation}
\Big( \frac{d^2 K_\ell^4}{\varepsilon_T} + \delta K_\ell^2 \vert \log \delta \vert + K_\ell^2 \varepsilon_K + K_\ell^6 K_H^{-4} + \frac{b}{200} \Big) \frac{n_+}{\ell^2} .
\end{equation}
Those terms are absorbed in a fraction of the positive $\frac{b}{100} \frac{n_+}{\ell^2}$, as long as we have the relations \eqref{eq:rel_T2comm}, \eqref{eq:rel_Kl_KH_d}, \eqref{eq:Kl_KN_number} and \eqref{eq:rel_T2comm2}.
We deduce
\begin{equation*}
\frac{\ell^2}{(2\pi)^2} \int_{\mathcal{P}_H}  \mathcal T_{2,\rm{com}}'(k)\dd k \geq - 2 \rho_z \widehat{(\omega W_1)}(0) \sum_{j=1}^N Q_j \chi_j^2 Q_j  - \frac{b}{150} \frac{n_+}{\ell^2} - \frac{b}{100}\frac{\varepsilon_T n_+^H}{(d\ell)^2}.
\end{equation*}

To compare the remaining part with $\mathcal Q_2^{\rm{ex}}$ we use \eqref{eq:Q2approx} to find
\begin{align*}
\mathcal Q_2^{\rm{ex}} &= \rho_z \frac{\ell^2}{(2\pi)^2}  \int \Big( \widehat{W\omega}(k) + \widehat{W\omega}(0) \Big) a_k^\dagger  a_k \dd k \\
&\geq 2 \rho_z \frac{\ell^2}{(2\pi)^2} \widehat{W_1 \omega}(0) \int a_k^\dagger a_k \dd k - C \rho_z \delta ( d^{-2} R \ell^{-1} n_+ + C n_+^H)\\
&= 2 \rho_z \widehat{W_1 \omega}(0) \sum_j Q_j \chi_\Lambda^2 Q_j - C \rho_z \delta d^{-2} R \ell^{-1} n_+ + C \rho_z \delta n_+^H .
\end{align*}
Using that $\rho_z \simeq \rho_{\mu}$, the remaining parts are absorbed by the spectral gaps and then we get
\begin{align*}
\frac{\ell^2}{(2\pi)^2} \int_{\mathcal{P}_H} \mathcal T_{2,\rm{com}}'(k) \dd k + \mathcal Q_2^{\rm{ex}} \geq - \frac{b}{100}\frac{n_+}{\ell^2} - \frac{b}{100} \frac{\varepsilon_T n_+^H}{(d\ell)^2}  .
\end{align*}
This last estimate, together with \eqref{eq.Q31.low}, \eqref{eq.T2}, \eqref{eq.T2''} and \eqref{eq.T2op} concludes the proof.
\end{proof}

\subsubsection{Estimates on \texorpdfstring{$\mathcal Q_3^{(2)}$}{Q\_3\^2} and \texorpdfstring{$\mathcal Q_3^{(3)}$}{Q\_3\^3}}

\begin{lemma}[Estimates on $\mathcal Q_3^{(2)}$ and $\mathcal Q_3^{(3)}$] \label{lem.Q32}
For any normalized state $\Phi$ satisfying \eqref{def:statefiniteexcit} we have
\begin{align*}
\Big\langle \mathcal Q_3^{(2)} + \mathcal Q_3^{(3)} + \frac{\varepsilon_K}{100} \mathcal K^{\rm{Diag}}_H \Big\rangle_\Phi \geq - C \rho_z^2\ell^2  \delta^2 \varepsilon_K^{-1} K_H^{-2M-8} K_\ell^6 d^{-8}.
\end{align*}
\end{lemma}

\begin{proof}
We focus on $\mathcal Q_3^{(3)}$ (the estimates on $\mathcal Q_3^{(2)}$ are similar), and decompose it into $\mathcal Q_3^{(3)} = I + II$, where
\begin{align}
I &= -  \frac{z \ell^2}{(2\pi)^4} \int_{\mathcal P_H \times \R^2}  \frac{f_L(p) \widehat W_1(k) \alpha_{p-k}}{\sqrt{1-\alpha_k^2} \sqrt{1-\alpha_{p-k}^2}} \big( b^\dagger_{k-p} \widetilde{a}^\dagger_p b_{k} + b_{k}^\dagger \widetilde{a}_p  b_{k-p} \big)  \dd k \dd p,
\intertext{and}
II &= - \frac{z \ell^2}{(2\pi)^4} \int_{\mathcal P_H \times \R^2}   \frac{f_L(p) \widehat W_1(k) \alpha_{p-k}}{\sqrt{1-\alpha_k^2} \sqrt{1-\alpha_{p-k}^2}} \big( [\widetilde a_p^\dagger , b_{k-p}^\dagger ] b_{k} + b_{k}^\dagger  [ b_{k-p},  \widetilde a_p] \big) \dd k \dd p.
\end{align}
The first part we estimate using Cauchy-Schwarz with weight $\varepsilon$, and by \eqref{eq:alpha_highcontrol}
\begin{align*}
I &\geq - \frac{z \ell^2}{(2\pi)^4} \int_{\mathcal P_H \times \R^2}  \frac{f_L(p) \widehat W_1(k) \alpha_{p-k}}{\sqrt{1-\alpha_k^2} \sqrt{1-\alpha_{p-k}^2}}  ( \varepsilon b_{k-p}^\dagger \widetilde a_p^\dagger \widetilde a_p b_{k-p} + \varepsilon^{-1} b_k^\dagger b_k )  \dd k \dd p \\
&\geq - C z \ell^2 \delta K_\ell^2 K_H^{-2} \int_{\mathcal P_H \times \R^2} f_L(p) ( \varepsilon b_{k-p}^{\dagger} \widetilde a_p^\dagger \widetilde a_p b_{k-p} + \varepsilon^{-1} b_k^\dagger b_k )  \dd k \dd p ,
\end{align*}
and using Lemma~\ref{lem:numbercontrolhigh}, 
\begin{equation}
\langle \Phi, \int f_L(p) b_{k-p}^\dagger \widetilde a_p^\dagger \widetilde a_p b_{k-p} \dd p \Phi \rangle \leq C \ell^{-2} \mathcal M \langle \Phi , b_{k}^\dagger b_{k} \Phi \rangle.
\end{equation}
We choose $\varepsilon = \sqrt{d^{-4}/ \mathcal M} $, and insert $\mathcal{D}_k \geq K_H^2 \ell^{-2}$,
\begin{align}
\langle I \rangle_\Phi &\geq - C z \delta K_\ell^2 K_H^{-2} (\varepsilon \mathcal M + \varepsilon^{-1} d^{-4} ) \int_{k \in \mathcal P_H} \langle b_k^\dagger b_k \rangle_\Phi \dd k \nonumber \\
&\geq - C (\rho_z^{1/2} \ell \delta K_\ell^2 K_H^{-4} \mathcal M^{1/2} d^{-2}) \ell^2 \int_{k \in \mathcal P_H} \mathcal D_k \langle b_k^\dagger b_k \rangle_\Phi \dd k .
\end{align}
Thanks to condition \eqref{eq:rel_epsK_KM2}, $I$ can be absorbed in the positive $\frac{\varepsilon_K}{100}\mathcal{K}^{\rm{Diag}}_H$ term. Now we return to the commutator term, which can be estimated using a Cauchy-Schwarz inequality with new weight $\varepsilon$,
\begin{align*}
II \geq - 2 \frac{ z \ell^2}{(2\pi)^4} \int_{\mathcal{P}_H \times \R^2} \vert [ b_{k-p} , \widetilde a_p ] \vert f_L(p) \vert \widehat W_1(k) \alpha_k \vert ( \varepsilon b_k^\dagger b_k + \varepsilon^{-1})  \dd k \dd p .
\end{align*}
 We use the commutator bound $\vert [b_{k-p},  \widetilde a_p] \vert \leq C \alpha_{k-p} \sup_{k \in \mathcal{P}_H} \widehat \chi (k \ell)$ from \eqref{eq:commut_astarbstar},
\small
\begin{align*}
II & \geq -C z K_\ell^2 K_H^{-2} \Big( \sup_{k \in \mathcal{P}_H} \widehat \chi (k\ell) \Big) d^{-4} \Big(  \varepsilon K_\ell^2 K_H^{-4}  \delta \ell^2 \int_{\mathcal{P}_H} \mathcal D_k b_k^\dagger b_k \dd k +\varepsilon^{-1} \int_{\mathcal{P}_H} \widehat W_1(k) \alpha_k \dd k  \Big).  
\end{align*}
\normalsize
With $\varepsilon^{-1} \simeq \varepsilon_K^{-1} z d^{-4}  K_\ell^4 K_H^{-6} \delta  \big( \sup \widehat \chi \big)  $ and our choice of parameters, the first part is absorbed in the positive $\varepsilon_K \mathcal K^{\rm{Diag}}$ term. We estimate the last part using \eqref{eq:Walpha_highcontrol} and Lemma~\ref{lem:localization_properties}, and then $II$ contributes with an error of order $\varepsilon_K^{-1} \rho_z^2 \ell^2 \delta^2 K_H^{-2M-8} K_\ell^6 d^{-8}$.
\end{proof}

\subsubsection{Estimates on \texorpdfstring{$\mathcal Q_3^{(4)}$}{Q\_3}}

First we rewrite $\mathcal Q_1^{\rm{ex}}$ as a term appearing in $\mathcal Q_3^{(4)}$.

\begin{lemma}\label{lem:Q1ex} Assume that Assumptions of Appendix~\ref{app:parameters} are satisfied. Then there exists a universal constant $C>0$ such that
\begin{align*}
\mathcal Q_1^{\rm{ex}} + \frac{b}{100} \frac{n_+}{\ell^2} + \frac{b}{100} \frac{ \varepsilon_T n_+^H}{(d\ell)^2} & \geq \frac{z \ell^2}{(2\pi)^4}\int_{\mathcal{P}_H \times \R^2} \widehat W_1(k) \alpha_k \widehat{\chi^2}(p \ell) f_L(p) (\widetilde a_p^{\dagger} + \widetilde a_p) \dd k \dd p \\
& \qquad  -C\rho^2_z\ell^2 \delta^3 K_\ell^2 \vert \log \delta \vert^2 - \rho^{2}_z \ell^2 \delta K_\ell^2 d^{8M-2} \varepsilon_T^{-1}.
\end{align*}
\end{lemma}

\begin{proof}
First we can rewrite $\mathcal Q_1^{\rm{ex}}$ in terms of the $\widetilde a_p$'s, 
\[\mathcal Q_1^{\rm{ex}} = z \rho_z \widehat{ \omega W_1}(0) \frac{\ell^2}{(2\pi)^2} \int \widehat{\chi^2}(p \ell) ( \widetilde a_p^{\dagger} + \widetilde a_p ) \dd p,\]
and then we use \eqref{eq:Walpha_highcontrol} to compare $\widehat{\omega W_1}(0)$ with an integral in $k$, and using the bound $K_{\ell}^2 K_H^{-2} \ll 1$, 
\begin{align}
\mathcal Q_1^{\rm{ex}} &\geq  \frac{z \ell^2}{(2\pi)^4} \int_{\mathcal P_H \times \R^2} \widehat{W}_1(k) \alpha_k \widehat{\chi^2}(p \ell) (\widetilde a_p^\dagger + \widetilde a_p) \dd k \dd p \nonumber \\ &\quad - C \rho_z \delta^2 \vert \log \delta \vert z \ell^2 \int_{\R^2} \widehat{\chi^2}(p \ell) (\widetilde a_p^\dagger + \widetilde a_p) \dd p . \label{eq:Q1ex.bound1}
\end{align}
The second integral can be estimated using a Cauchy-Schwarz inequality with weight $\varepsilon$,
\begin{align} \label{eq:chiapap}
\rho_z \ell^2\delta^2 z  \int_{\R^2} \widehat{\chi^2}(p \ell) (\widetilde a_p^\dagger + \widetilde a_p) \dd p &\leq \varepsilon \rho_z \ell^2 \delta^2 z \int |\widehat{\chi^2}(p \ell)| +  C  \varepsilon^{-1} \rho_z \ell^2 \delta^2 z  \int  |\widehat{\chi^2}(p \ell)| \widetilde a_p^\dagger \widetilde a_p \dd p \nonumber \\
&\leq C \varepsilon \rho_z \delta^2 z  + C \varepsilon^{-1} \rho_z \delta^2 z  n_+ .
\end{align}
where we used Lemma \eqref{lem:localization_properties}. With $\varepsilon \simeq z \delta K_{\ell}^2 \vert \log \delta \vert$, the second part is absorbed by the positive fraction of $\frac{n_+}{\ell^2}$, and the first term is of order $ \rho^2_z \ell^2\delta^3 K_\ell^2 \vert \log \delta \vert$. Hence,
\begin{equation}
 \mathcal Q_1^{\rm{ex}} \geq \frac{z \ell^2}{(2\pi)^4} \int_{\mathcal P_H \times \R^2} \widehat{W}_1(k) \alpha_k \widehat{\chi^2}(p \ell) (\widetilde a_p^\dagger + \widetilde a_p) \dd k \dd p - C  \rho^2_z\ell^2 \delta^3 K_\ell^2 \vert \log \delta \vert^2 - \frac{b}{100} \frac{n_+}{\ell^2} .
\end{equation}
Finally we want to insert the cutoff $f_L(p)$ inside the integral. The error we make is estimated similarly,
\begin{align*}
    \frac{z \ell^2}{(2\pi)^4} \int_{\mathcal P_H \times \R^2} \widehat{W}_1(k) & \alpha_k \widehat{\chi^2}(p \ell) (1-f_L(p)) (\widetilde a_p^\dagger + \widetilde a_p) \dd k \dd p \\ & \geq - C z \ell^2 \rho_z \delta \int_{\mathcal P_L^c} \widehat{\chi^2}(p \ell) (\widetilde a_p^\dagger + \widetilde a_p ) \dd p \\
    &\geq -C \varepsilon  z \rho_z \delta \ell^2 \int_{\mathcal P_L^c} |\widehat{\chi^2}(p \ell)| \dd p - C \varepsilon^{-1} z \rho_z \delta \ell^2 \int_{\mathcal P_L^c} |\widehat{\chi^2}(p \ell)| \widetilde a_p^\dagger \widetilde a_p \dd p \\
    &\geq - C \varepsilon z \rho_z \delta d^{4M-4} - C \varepsilon^{-1} z \rho_z \delta d^{4M} n_+^H ,
\end{align*}
where we used $\sup_{p \in \mathcal P_L^c} \vert \widehat{\chi^2}(p\ell) \vert \leq C d^{4M}$. With $\varepsilon \simeq z K_{\ell}^2 d^{4M+2} \varepsilon^{-1}_T$ the first part is of order $\rho^{2}_z \ell^2 \delta K_\ell^2 d^{8M-2} \varepsilon_T^{-1}  $ and the second is absorbed in a fraction of $\frac{\varepsilon_T n_+^H}{(d \ell)^2}$.
\end{proof}

Now we have all we need to estimate $\mathcal Q_3^{(4)}$.

\begin{lemma}[Estimates on $\mathcal Q_3^{(4)}$] \label{lem.Q34}
For any state $\Phi$ satisfying \eqref{def:statefiniteexcit} we have
\begin{align*}
\left\langle \mathcal Q_3^{(4)} + \mathcal Q_1^{\rm{ex}} + \frac{b}{100} \frac{n_+}{\ell^2} + \frac{b}{100} \frac{ \varepsilon_T n_+^H}{(d\ell)^2} \right\rangle_{\Phi} \geq & - C \rho_z^2\ell^2  \delta^3 K_\ell^2 \vert \log \delta \vert^2 - C \rho_z^2\ell^2  \delta K_\ell^2 d^{8M-2} \varepsilon_T^{-1} \\ & -C \rho_z^2 \ell^2 \delta^2 \varepsilon_{\mathcal M}^{1/2} (K_\ell^4 K_H^{-4} + \delta^{-1} K_H^{-M} d^{-2})  .
\end{align*}
\end{lemma}

\begin{proof}

We use the commutator formula 
$$[ b_{p-k}, b_{-k}^{\dagger}] = (1-\alpha_k \alpha_{p-k}) \Big( \widehat{\chi^2}(p\ell) - \widehat \chi (k \ell) \widehat \chi ((p-k)\ell) \Big),$$ and split into $\mathcal Q_3^{(4)} = I + II$, with
\begin{align*}
I &= - \frac{z \ell^2}{(2\pi)^4} \int_{\mathcal P_H \times \R^2} \frac{f_L(p) \widehat W_1(k)}{\sqrt{1-\alpha_k^2} \sqrt{1-\alpha_{p-k}^2}} \alpha_k (1-\alpha_k \alpha_{p-k}) \widehat{\chi^2}(p\ell) (\widetilde a_p^\dagger + \widetilde a_p)  \dd k \dd p,
\intertext{and}
II &= - \frac{z \ell^2}{(2\pi)^4} \int_{\mathcal P_H \times \R^2} \frac{f_L(p) \widehat W_1(k)}{\sqrt{1-\alpha_k^2} \sqrt{1-\alpha_{p-k}^2}} \alpha_k (1-\alpha_k \alpha_{p-k}) \widehat \chi (k\ell) \widehat \chi ((p-k)\ell) (\widetilde a_p^\dagger + \widetilde a_p)  \dd k \dd p .
\end{align*}

In $I$ we recognize the lower bound on $\mathcal Q_1^{\rm{ex}}$ given by Lemma~\ref{lem:Q1ex} with opposite sign, up to an error term:
\begin{multline*}
I + \mathcal Q_1^{\rm{ex}} + \frac{b}{100} \frac{n_+}{\ell^2} + \frac{b}{100} \frac{ \varepsilon_T n_+^H}{(d\ell)^2} + C\rho^2_z\ell^2 \delta^3 K_\ell^2 \vert \log \delta \vert^2 + \rho^{2}_z \ell^2 \delta K_\ell^2 d^{8M-2} \varepsilon_T^{-1} \\ \geq  - C z \ell^2 \int_{\mathcal P_H \times \R^2} f_L(p) \widehat W_1(k) \alpha_k^3 \widehat {\chi^2}(p\ell) (\widetilde a_p^\dagger + \widetilde a_p) \dd k \dd p .
\end{multline*}
This remaining integral can be estimated, by \eqref{eq:alpha_highcontrol}, as
\begin{align*}
&\left\vert z \ell^2 \int_{\mathcal P_H \times \R^2} f_L(p) \widehat W_1(k) \alpha_k^3 \widehat {\chi^2}(p\ell) (\widetilde a_p^\dagger + \widetilde a_p) \dd k \dd p \right \vert \\&\leq C \vert z \vert \ell^2 \rho_z^3 \delta^4 \int_{\mathcal P_H} k^{-6} \dd k \int_{\mathcal P_L} |\widehat{\chi^2}(p\ell)| (\widetilde a_p^\dagger + \widetilde a _p) \dd p,
\end{align*}
and after applying to the state $\Phi$ we use a Cauchy-Schwarz inequality with weight $\sqrt \mathcal M$,
\begin{align*}
\Big\vert z  & \ell^2 \int_{\mathcal P_H \times \R^2} f_L(p)  \widehat W_1(k)  \alpha_k^3 \widehat{\chi^2}(p\ell) \langle \widetilde a_p^\dagger + \widetilde a_p \rangle_{\Phi} \dd k \dd p \Big\vert 
 \\  & \leq C  z  \rho_z^3\ell^6  \delta^4 K_H^{-4} \Big( \sqrt{ \mathcal M} \int_{\mathcal P_L} |\widehat{\chi^2}(p\ell)| \dd p + \frac{1}{\sqrt{ \mathcal M}} \int_{\mathcal P_L} |\widehat{\chi^2}(p\ell)| \langle \widetilde a_p^\dagger \widetilde a_p \rangle_{\Phi} \dd p \Big) \\
&\leq C z \rho_z^3 \ell^4 \delta^4 K_H^{-4} \sqrt{\mathcal{M}} \leq C \rho_z^2 \ell^2 \delta^2 \big(K_\ell^4 K_H^{-4} \sqrt{\varepsilon_{\mathcal M}}\big).
\end{align*}
Finally we bound, by \eqref{eq:fourierlocestimate} and \eqref{eq:Walpha_highcontrol},
\begin{align*}
\vert \langle II \rangle_{\Phi} \vert &\leq z \ell^2 \sup_{h \in \mathcal P_H}|\widehat{\chi}(h\ell)|\int_{\mathcal P_H}  |\widehat W_1(k)| \alpha_k |\widehat \chi(k \ell)| \dd k \int_{\mathcal P_L} \langle \widetilde a_p^\dagger + \widetilde a_p \rangle_{\Phi} \dd p\\
&\leq C z \rho_z \ell^2 \delta K_H^{-M}\Big( d^2 \mathcal{M}^{1/2} \int_{\mathcal P_L} \dd p + d^{-2} \mathcal M^{-1/2} \int_{\mathcal P_L} \langle \widetilde a_p^\dagger \widetilde a_p \rangle_{\Phi} \dd p \Big),
\end{align*}
where we used a Cauchy-Schwarz inequality with weight $d^2 \sqrt{\mathcal M}$. Thus,
\begin{equation*}
\vert \langle II \rangle_\Phi \vert \leq C  \rho^2_z\ell^2 \delta K_H^{-M} d^{-2}\varepsilon_{\mathcal M}^{1/2}.
\end{equation*}
\end{proof}

\subsection{Conclusion: Proof of Theorem~\ref{thm:largebox_lower}}\label{subsec:proof_lower_final}

In Section~\ref{sec:largebox} we showed how the proof of Theorem~\ref{thm:main_lower} is reduced to the proof of Theorem~\ref{thm:largebox_lower}, which we give here.
\begin{proof}[Proof of Theorem~\ref{thm:largebox_lower}]
Recall the choices of the parameters in Appendix~\ref{app:parameters}. Let us consider a normalized $n$-particle state $\Psi \in \mathscr{F}_s(L^2(\Lambda))$ which satisfies \eqref{eq:condensationaprioricondition} for a certain large constant $C_0>0$,
\begin{equation}
\langle \mathcal{H}_{\Lambda}(\rho_{\mu})\rangle_{\Psi} \leq -4 \pi \rho_{\mu}^2 \ell^2 Y \big( 1 - C_0 K_B^2 Y \vert \log Y \vert \big).
\end{equation}
If such a state does not exists, our desired lower bound follows,
because
\begin{equation}\label{eq:simple}
-4 \pi \rho_{\mu}^2 \ell^2 Y \big( 1 - C_0 K_B^2 Y \vert \log Y \vert \big)\geq -4\pi\rho_\mu^2 \ell^2\delta\Big(1-\Big(2\Gamma +\frac 1 2 +\log \pi\Big)\delta \Big).
\end{equation}

So we can assume  the existence of $\Psi$. By Theorem~\ref{thm:excitationrestriction} there exists a sequence of $n$-particle states $\{\Psi^m\}_{m \in \mathbb Z} \subseteq \mathscr{F}_s(L^2(\Lambda))$ and $C_1, \eta_1>0$ such that 
\begin{align*}
\langle \Psi, \mathcal{H}_{\Lambda}(\rho_{\mu}) \Psi \rangle  \geq&  \sum_{2|m|\leq \mathcal{M} }\langle \Psi^{(m)}, \mathcal{H}_{\Lambda}(\rho_{\mu}) \Psi^m\rangle -C_1 \rho_{\mu}^2 \ell^2 \delta^{2+\eta_1} \\
&  -4 \pi \rho_{\mu}^2 \ell^2 Y \big( 1 - C_1 K_B^2 Y \vert \log Y \vert \big)\sum_{2|m| > \mathcal{M}} \|\Psi^m\|^2.  
\end{align*}
For $|m| \leq \frac{\mathcal M}{2}$, we have that $\Psi^m = \one_{[0,\mathcal M]}(n_+^L) \Psi^m$. If we prove the lower bound for all $\Psi^m$ such that $ |m| \leq \frac{\mathcal M}{2}$ then we would get (using \eqref{eq:simple} with $C_0$ replaced by $C_1$)
\begin{align*}
\langle \Psi, \mathcal{H}_{\Lambda}(\rho_{\mu}) \Psi \rangle  \geq&  -4\pi\rho_\mu^2 \ell^2\delta\Big(1-\Big(2\Gamma +\frac 1 2 +\log \pi\Big)\delta\Big)\sum_{m }\Vert  \Psi^{(m)}\Vert^2 -C_1 \rho_{\mu}^2 \ell^2 \delta^{2+\eta_1}, 
\end{align*}

Therefore, the theorem is proven if we derive the corresponding lower bound for any $n-$particle, normalized state $\widetilde \Psi \in \mathscr{F}_s(L^2(\Lambda))$ such that 
\begin{equation}
\widetilde \Psi = \one_{[0,\mathcal{M}]}(n_+^L) \widetilde{\Psi}.
\end{equation} 
By Proposition~\ref{propos:secondquant}, for such a state there exists a constant $C_2 >0$ such that 
\begin{equation}
\langle \widetilde{\Psi}, \mathcal{H}_{\Lambda}(\rho_{\mu}) \widetilde{\Psi}\rangle \geq \langle \widetilde{\Psi}, \mathcal{H}_{\Lambda}^{\rm{2nd}}(\rho_{\mu})\widetilde{\Psi}\rangle -C_2 \ell^2 \rho_{\mu}^2 \delta \big( d^{2M-2} + R^2 \ell^{-2}  \big),
\end{equation}
where the last term is an error term of order $\rho_{\mu}^2 \ell^2 \delta^{2+\eta_2}$, for some $\eta_2 >0$, thanks to relations \eqref{eq:rel_locQ3low} and \eqref{eq:relRrho_mu}.
Then, by Theorem~\ref{thm:cnumber}, there exists a constant $C_3>0$ such that
\begin{equation}
\langle \widetilde{\Psi},\mathcal{H}_{\Lambda}^{\rm{2nd}} \widetilde{\Psi} \rangle \geq \inf_{z \in \mathbb{R}_+} \inf_{\Phi} \langle \Phi , \mathcal{K}(z)  \Phi\rangle  - C_3 \rho_{\mu} \delta (1 + \varepsilon_R K_{\ell}^4 K_B^2 |\log Y|),
\end{equation}
where the infimum is over the $\Phi$'s which satisfy \eqref{def:statefiniteexcit}. The last term is an error term of order $\rho_{\mu}^2 \ell^2 \delta^{2+\eta_3}$ for some $\eta_3>0$, thanks to relation \eqref{eq:rel_epsR_Kl_KB}. The proof is reduced now to getting a lower bound for $\mathcal{K}(z)$.
We have two cases, according to different values of $z$:
\begin{itemize}
\item If $|\rho_z -\rho_{\mu}| \geq C \rho_{\mu} \max((\delta_1+ \delta_2+ \delta_3)^{1/2}, \delta^{1/2})$ then Proposition~\ref{propos:rhofar} implies the bound 
\begin{equation}
\langle \mathcal{K}(z)\rangle_{\Phi} \geq -\frac{1}{2} \rho_{\mu}^2 \ell^2 \widehat{g}_0 + 8\pi \Big(2\Gamma +\frac 1 2 +\log \pi\Big)\rho_{\mu}^2 \ell^2 \delta^2,
\end{equation}
and the second term is twice the LHY-term and positive, therefore there is nothing more to prove;
\item Otherwise $|\rho_z -\rho_{\mu}| \leq \rho_\mu K_{\ell}^{-2}$ (see Section~\ref{sec.9.2}). In this case we can use \eqref{ine:Kz} and Theorem~\ref{thm:lower_rho=rho_mu} to obtain $C_4,\eta_4 >0$, such that  
\begin{align}
\langle \mathcal{K}(z) \rangle_{\Phi} &\geq -\frac{1}{2} \rho_{\mu}^2 \ell^2 \widehat{g}_0+ (1-\varepsilon_K)\langle \mathcal{K}_H^{\text{Diag}}\rangle_{\Phi} +4\pi \Big(2\Gamma +\frac 1 2 +\log \pi\Big) \rho^2_{z} \ell^2 \delta^2 \nonumber \\
&\quad + \Big\langle b\frac{n_+}{4\ell^2} + b  \frac{\varepsilon_T n_+^H}{8 d^2\ell^2}  + \mathcal Q_1^{\rm{ex}}(z) + \mathcal Q_2^{\rm{ex}} + \mathcal Q_3(z)\Big\rangle_{\Phi} - C_4\rho_{\mu}^2 \ell^2\delta^{2+\eta_4},\label{eq:final_proof_1}
\end{align}
where we used that $C \rho_{\mu}^2 \ell^2 \delta ( K_H^{4-M} K_{\ell} \delta^{-1/2}) + | r(\rho_{\mu}) | \ell^2 \leq C\rho_{\mu}^2 \ell^2\delta^{2+\eta_4} $, thanks to the relations \eqref{eq:rel_s_Kl}, \eqref{eq:rel_T2comm} and that $M >4$. We conclude observing that, thanks to Theorem~\ref{thm.Q3z}, we have the existence of $C_5,\eta_5 >0$ such that  
\begin{equation}\label{eq:final_proof_2}
\Big\langle (1-\varepsilon_K)\mathcal{K}_H^{\rm{Diag}} + \mathcal Q_3(z) + \mathcal Q_2^{\rm{ex}} + \mathcal Q_1^{\rm{ex}} + \frac{b}{100} \frac{n_+}{\ell^2} + \frac{b\varepsilon_T}{100}\frac{ n_+^{H}}{(d\ell)^{2}} \Big\rangle_{\Phi} \geq - C_5  \rho_z^2\ell^2 \delta^{2 + \eta_5}, 
\end{equation}
where the error has been obtained using relations \eqref{eq:rel_Kl_KH_d}, \eqref{eq:rel_errQ3_statement}, \eqref{eq:rel_epsK_KM2}, \eqref{eq:rel_locQ3low2} and \eqref{eq:rel_KH_K_M}.
Thanks to the assumptions on $\rho_z$ and $\rho_{\mu}$, there exist $C_6,\eta_6>0$ such that 
\begin{equation}\label{eq:final_proof_3}
|\rho_z^2\ell^2 \delta^{2 } - \rho_{\mu}^2\ell^2 \delta^{2 }| \leq C_6 \rho_{\mu}^2 \ell^2 \delta^2 K_{\ell}^{-2}=  C_6 \rho_{\mu}^2 \ell^2 \delta^{2+\eta_6}, 
\end{equation}
so that, plugging \eqref{eq:final_proof_2} into \eqref{eq:final_proof_1} and substituting the $\rho_z$ by the $\rho_{\mu}$ using \eqref{eq:final_proof_3} gives the desired lower bound and the right order for the error terms.
\end{itemize}
We choose $C = \sum_{j=1}^6 C_j$ and $\eta = \min_{j=1,\ldots, 6} \eta_j$.
We conclude using that $\widehat{g}_0=8\pi\delta$ to get that
$$\inf_{z \in \mathbb{R}_+} \inf_{\Phi} \langle \Phi , \mathcal{K}(z)  \Phi\rangle\geq -4\pi\ell^2\rho_\mu^2\delta\Big(1-\Big(2\Gamma +\frac 1 2 +\log \pi\Big)\delta\Big)-C\rho_{\mu}^2 \ell^2 \delta^{2+\eta}.$$

This finishes the proof of Theorem~\ref{thm:largebox_lower}.
\end{proof}

\appendix

\section{Reduction to smaller boxes for the upper bound}\label{appendix b}

We provide here the necessary tools to go from a fixed box with compactly support potentials in the grand canonical setting, Theorem~\ref{thm.upperbound.grandcanonical}, to the thermodynamic limit with potentials allowing a tail, Theorem~\ref{thermo dynamic limit theorem}. The same techniques can be found in \cite{BCS} with only minor deviations surrounding the non-compactness of the potential.

Given a potential $v$, we define
\[e(\rho):=\lim_{L\rightarrow \infty}e_L(\rho)=\lim_{L\rightarrow \infty}\inf_{\psi\in H_0^1(\Lambda_L^{\rho L^2})}\frac{\innerp{\psi}{\mathcal{H}_v^{\rho L^2}\psi}}{L^2},\]
where the limit is taken such that $\rho L^2=N \in \mathbb{N}$ and
\begin{equation*}
\mathcal{H}_v^{N}=\sum_{i=1}^N -\Delta_i+\sum_{i<j}^N v(x_i-x_j).
\end{equation*}
We write $v=v \one_{B(0,R)}+v \one_{B(0,R)^c}=v_R+v_{tail}$ where the $v_{tail}$ will always be treated as an error term. Let $v_{R}^{per}(x)=\sum_{k\in\mathbb{Z}^2}v_{R}(x+kL)$. In order for this to be finite we understand $R$ to be smaller than $L$.  We neglect the $N$ in the hamiltonian when it is operating on the Fock space.

The result below evaluates the error when going from periodic boundary conditions to Dirichlet boundary conditions. 
\begin{lemma}\label{lem.B1}
There exists a universal $C>0$, such that given $R_0>0$ and a periodic, bosonic trial function $\Psi_L\in \mathscr{F}(\Lambda_L)$, there exists a Dirichlet trial function $\Psi_{L+2R_0}^D\in \mathscr{F}(L^2(\Lambda_{L+2R_0}))$ satisfying, for $j\in \mathbb{N}_0$,
	\begin{align}\label{moments preserved}
		\innerp{\Psi_{L+2R_0}^D}{\mathcal{N}^j\Psi_{L+2R_0}^D}&=\innerp{\Psi_L}{\mathcal{N}^j\Psi_L},
	\intertext{and} \label{energy bound}
		\innerp{\Psi_{L+2R_0}^D}{\mathcal{H}_{v_{R}}\Psi_{L+2R_0}^D}&\leq \innerp{\Psi_L}{\mathcal{H}_{v_{R}^{per}}\Psi_L}+\frac{C}{LR_0}\innerp{\Psi_L}{\mathcal{N}\Psi_L}.
	\end{align}
\end{lemma}
\begin{proof}
	The result is independent of dimension, see \cite[Lemma 2.1.3]{Robinson} or \cite[Lemma A.1]{BCS} for a proof in the 3D case.
\end{proof}
Next step is to glue the Dirichlet boxes together in order to construct a trial function on a thermodynamic box.
\begin{theorem}\label{thm.B2}
	Let $\Psi^D_{L+2R_0}\in \mathscr{F}_s(L^2(\Lambda_{L+2R_0}))$ be a trial function with Diriclet boundary conditions and extend it to $\mathbb{R}^2$ by $0$. Then for ${L}_k=k(L+2R_0+R)$, $k\in \mathbb{N}$, we define $\Psi_{{L}_k}\in \mathscr{F}_s(L^2(\Lambda_{{L}_k}))$ by
	\begin{equation}
		\Psi^{(m)}_{L_k}(x_1,\ldots,x_m)= \frac{1}{\Vert(\Psi_{L+2R_0}^{D })^{(n)}\Vert^{k^2-1}}\prod_{i=1}^{k^2}(\Psi_{L+2R_0}^{D })^{(n)}(x_{1+n(i-1)}-c_i,\ldots,x_{in}-c_i),
	\end{equation}
	if $m = n k^2$, and $\Psi_{L_k}^{(m)} = 0$ otherwise. Here $c_i$ defines an enumeration of the lattice points on $\mathbb{Z}^2(L+2R_0)$. Then $\Psi_{L_k}$ satisfies
	\begin{equation}
		\innerp{\Psi_{L_k}}{\mathcal{N}^j\Psi_{L_k}}=k^{2j}	\innerp{\Psi_{L+2R_0}^D}{\mathcal{N}^j\Psi_{L+2R_0}^D},\qquad j\in\mathbb{N}_0.
	\end{equation}
	Furthermore if $v$ satisfies the decay condition \eqref{eq:decay} of Theorem~\ref{thm:LHY_2D}, then there exits a constant $C$ only depending on $\eta_0$ and $C_0$ such that
	\begin{equation}
		\innerp{\Psi_{L_k}}{\mathcal{H}_v\Psi_{L_k}}\leq k^{2}	\innerp{\Psi_{L+2R_0}^D}{\mathcal{H}_{v_{R}}\Psi^D_{L+2R_0}}+k^2\innerp{\Psi^D_{L+2R_0}}{\mathcal{N}^2\Psi^D_{L+2R_0}}\frac{C a^{\eta_0}}{R^{2+\eta_0}}.
	\end{equation}
\end{theorem}
\begin{proof}
	The expectation of $\mathcal{N}^j$ can be computed using that
	\[\Vert \Psi_{L_k}^{(m)}\Vert^2=\begin{cases}
		\Vert \Psi_{L+2R_0}^{(n)}\Vert^2 &\text{if}\quad m=k^2n,\\
		0 &\text{otherwise}.
	\end{cases}\]
	However for the potential energy we need to estimate the interaction between the boxes and the long range interaction inside the box. We observe that
	\begin{equation}
		\begin{aligned}
			\innerp{\Psi_{L_k}}{\mathcal{H}_v\Psi_{L_k}}-k^{2}	\innerp{\Psi_{L+2R_0}^D}{\mathcal{H}_{v_{R}}\Psi_{L+2R_0}^D}&= \sum_{n\geq 0} \sum_{i<j}^{k^2n}\int\vert \Psi_{L_k}^{(k^2n)}\vert ^2v_{tail}(x_i-x_j) \dd x,
		\end{aligned}
	\end{equation}
	where we used that the kinetic energy of the two terms are equal and only the tail of the potential survives due to the corridors between the boxes. We further estimate
\begin{align}
\sum_{n\geq 0} \sum_{i<j}^{k^2n}\int\vert \Psi_{L_k}^{(k^2n)}\vert ^2v_{tail}(x_i-x_j)\dd x &\leq\sum_{n\geq 0}k^2n \sum_{j=2}^{k^2n}\int\vert \Psi_{L_k}^{(k^2n)}\vert^2v_{tail}(x_1-x_j) \dd x \nonumber\\ 
&\leq\sum_{n\geq 0}k^2n \sum_{j=2}^{k^2n}\int\vert \Psi_{L_k}^{(k^2n)}\vert^2\frac{C a^{\eta_0}}{|x_1-x_j|^{2+\eta_0}} \dd x,\label{eq:intermediate_step_A2}
\end{align}
where we used \eqref{eq:decay}.
If $s \in \mathbb{N}$ denotes the number of aligned boxes separating $x_1$ from $x_j$, then $|x_1 - x_j| \geq (s-1)L + R$ and there are $4(s+1)+1$ of such possible boxes. Summing on $s$ we get 
\begin{align*}
\eqref{eq:intermediate_step_A2}&\leq
			\sum_{n\geq 0}k^2n\sum_{s=1 }^k\frac{C_0a^{\eta_0}n(4(s+1)+1)}{((s-1)L+R)^{2+\eta_0}}\Vert \Psi_{L_k}^{(k^2n)}\Vert^2\\&\leq k^2\innerp{\Psi^D_{L+2R_0}}{\mathcal{N}^2\Psi_{L+2R_0}^D}C_0\Big(\frac{9a^{\eta_0}}{R^{2+\eta_0}}+\frac{a^{\eta_0}}{L^{2+\eta_0}}\sum_{s=1}^\infty\frac{4}{s^{1+\eta_0}}+\frac{9}{s^{2+\eta_0}}\Big).
\end{align*}
	 In fact the largest term is the contribution of $v_{tail}$ inside the box and its $8$ neighbours which here is represented by the term $\frac{9a^{\eta_0}}{R^{2+\delta}}$. 
\end{proof}
We have thus far constructed a sequence of grand canonical trial functions on larger and larger boxes, where we control the energy and the expected number of particles. The last part will be to relate this sequence to $e(\rho)$. For this we will use the continuity and convexity of $e(\rho)$ see \cite[Thm. 3.5.8 and 3.5.11]{Ruelle} together with the Legendre transformation being an involution on such functions.
\begin{theorem}\label{thm.B3}
	Let $\Psi_{L_k}\in \mathscr{F}(L^2(\Lambda_{L_k}))$ be a sequence with Dirichlet boundary conditions such that $L_k\rightarrow \infty$ as $k\rightarrow \infty$. Assume that there exist $C,c>0$ such that, for all $k\in \mathbb{N}$,
	\[\innerp{\Psi_{L_k}}{\mathcal{N}\Psi_{L_k}}\geq \rho(1+c\rho)L_k^2,\qquad \quad \innerp{\Psi_{L_k}}{\mathcal{N}^2\Psi_{L_k}}\leq C(\rho L_k^2)^2,\] then
	\[e(\rho)\leq \lim_{k\rightarrow \infty}\frac{\innerp{\Psi_{L_k}}{\mathcal{H}_v\Psi_{L_k}}}{L_k^2}.\]
\end{theorem}
\begin{proof}
We insert a chemical potential $\mu$, and find that, using the positivity of $\mathcal{H}_v$ and $\mathcal{N}$, for any $\mu\geq 0$ and $M>0$ we have
\small
\begin{equation*}
	\begin{aligned}
		&\frac{\innerp{\Psi_{L_k}}{\mathcal{H}_v\Psi_{L_k}}}{L_k^2}\\&\geq  	\frac{\innerp{\Psi_{L_k}}{(\mathcal{H}_v-\mu \mathcal{N})\chi(\mathcal{N}\leq ML_k^2)\Psi_{L_k}}}{L_k^2}+\frac{\mu}{L_k^2}\big(\innerp{\Psi_{L_k}}{\mathcal{N}\Psi_{L_k}}-\innerp{\Psi_{L_k}}{\mathcal{N}\chi(\mathcal{N}\geq ML_k^2)\Psi_{L_k}}\big)\\&\geq \sum_{m=0}^{ML_k^2}\Big(e_{L_k}\Big(\frac{m}{L_k^2}\Big)-\mu\frac{m}{L_k^2}\Big)\Vert \Psi_{L_k}^{(m)}\Vert^2+\frac{\mu}{L_k^2}\Big(\innerp{\Psi_{L_k}}{\mathcal{N}\Psi_{L_k}}-\frac{1}{ML_k ^2}\innerp{\Psi_{L_k}}{\mathcal{N}^2\Psi_{L_k}}\Big).
\end{aligned}
\end{equation*}
\normalsize
Fixing $M$ large enough in terms of $C$ and $c$ then gives
	\begin{equation}\label{eq. inequality relating to chemical potential}
		\frac{		\innerp{\Psi_{L_k}}{\mathcal{H}\Psi_{L_k}}}{L_{k}^2}\geq
		\sum_{m=0}^{ML^2}\Big(e_{L_k}\Big(\frac{m}{L_k^2}\Big)-\mu\frac{m}{L_k^2}\Big)\Vert \Psi_{L_k}^{(m)}\Vert^2+\mu\rho.
	\end{equation}
As in Theorem~\ref{thm.B2}, we glue several copies of a minimizer of $e_{L_k}(\rho)$, each copy living on a different box. We leave corridors of size $L_k^{1-\epsilon}$ between the boxes and this has the consequence of changing the density to $\rho(1+L_k^{-\epsilon})^{-2}$. Assuming further that $v$ satisfies the conditions of Theorem~\ref{thm:LHY_2D} we estimate the ignored interactions to find
\begin{equation}\label{eq. canonical, box to thermodynamic}
		e(\rho(1+L_k^{-\epsilon})^{-2})\leq  e_{L_k}(\rho)(1+L_k^{-\epsilon})^{-2}+\frac{C(\rho L_k^2)^2}{L_{k}^{4+\eta_0-(2+\eta_0)\epsilon}}.
\end{equation}
Using \eqref{eq. canonical, box to thermodynamic} in \eqref{eq. inequality relating to chemical potential} yields
	\begin{align*}
	&L_k^{-2} \innerp{\Psi_{L_k}}{\mathcal{H}\Psi_{{L}_k}}\\ &\geq  \mu\rho +(1+{L}_k^{-\epsilon})^2\sum_{m=0}^{M{L}_k^2}\Big(e\Big(\frac{m}{{L}_k^2}(1+{L}_k^{-\epsilon})^{-2}\Big)-(1+{L}_k^{-\epsilon})^{-2}\mu\frac{m}{{L}_k^2}-\frac{Cm^2}{{L}_k^{4+\eta_0-(2+\eta_0)\epsilon}}\Big)\Vert \Psi_{{L}_k}^{(m)}\Vert^2\\&\geq\mu\rho-(1+{L}_k^{-\epsilon})^{2}e^*(\mu)-C\rho^2{L}_k^{-\eta_0+(2+\eta_0)\epsilon}, 
	\end{align*}
	where $^*$ defines the Legendre transformation with respect to the interval $[0,M]$. Choosing $\epsilon>0$ small enough and letting $k$ go to infinity yields
	\begin{equation}
		\lim_{k\rightarrow \infty}\frac{		\innerp{\Psi_{{L}_k}}{\mathcal{H}\Psi_{{L}_k}}}{{L}_{k}^2}\geq\sup_{\mu\in[0,\infty)}(\mu\rho-e^*(\mu))=\sup_{\mu\in\mathbb{R}}(\mu\rho-e^*(\mu))=e(\rho),
	\end{equation}
where we used that $e^*(\mu)\geq 0$ for all $\mu\in \mathbb{R}$ and that the Legendre transformation is an involution.
\end{proof}

\section{Bogoliubov diagonalization}\label{sec:bogdiag}

\begin{theorem}\label{thm:bogdiag}
Let $a_{\pm}$ be operators on a Hilbert space satisfying $\left[ a_+, a_- \right] = 0$. For $\mathcal A > 0$, $\mathcal B \in \R$ satisfying either $\vert \mathcal B \vert < \mathcal A$ or $\mathcal B = \mathcal A$ and arbitrary $\kappa \in \C$, we have the operator identity
\begin{align*}
\mathcal{A}&(a^{\dagger}_+ a_+ + a^{\dagger}_- a_-) + \mathcal{B}(a_+^{\dagger} a^{\dagger}_- + a_+ a_-) + \kappa (a^{\dagger}_+ + a_-) + \overline{\kappa}(a_+ + a_-^{\dagger})  \\
&= (1- \alpha^2) \mathcal D ( b^\dagger_+ b_+ + b^\dagger_- b_- ) -\frac{1}{2} (\mathcal{A} - \sqrt{\mathcal{A}^2 - \mathcal{B}^2}) ([a_+,a^{\dagger}_+] + [a_-, a^{\dagger}_-]) - \frac{2|\kappa|^2}{\mathcal{A}+\mathcal{B}},
\end{align*}
where $\mathcal D = \frac{1}{2}  \big( \mathcal A + \sqrt{\mathcal A^2 - \mathcal B^2}\big)$ and
\begin{equation}
b_+ = \frac{1}{\sqrt{1 - \alpha^2}} \big( a_+ + \alpha a_-^\dagger + \bar c_0 \big), \quad b_- =  \frac{1}{\sqrt{1 - \alpha^2}} \big( a_- + \alpha a_+^\dagger + c_0 \big),
\end{equation}
with 
\begin{equation}
\alpha = \mathcal B^{-1} \big( \mathcal A - \sqrt{\mathcal A^2 - \mathcal B^2} \big), \quad c_0 = \frac{2 \bar \kappa}{\mathcal A + \mathcal B + \sqrt{\mathcal A^2 - \mathcal B^2}}.
\end{equation}
\end{theorem}

\begin{remark}
Note that the normalization of $b_\pm$ is chosen such that 
\begin{equation}
[ b_+, b_+^\dagger ] = \frac{ [ a_+ , a_+^\dagger ] - \alpha^2 [ a_- , a_-^\dagger ]}{1 - \alpha^2} , 
\end{equation}
and we recover the canonical commutation relations $[ b_+, b_+^\dagger ] = 1$ when $a_+$ and $a_-$ satisfies them as well.
\end{remark}

\begin{proof}
This follows directly from algebraic computations.
\end{proof}

\section{Calculation of the Bogoliubov integral}\label{app:bogintegral}

For functions $\alpha, \beta$, and parameter $\varepsilon \geq 0$, we define
%Let us denote by $I_{\varepsilon_N}(\alpha,\beta)$ the following integral
%
\begin{align}
	I_{\varepsilon}(\alpha,\beta) :=\frac{1}{2(2 \pi)^2} \int_{\mathbb{R}^2} &\Big( \sqrt{(1-\varepsilon)^2 \alpha^2(k) + 2 (1-\varepsilon) \rho \alpha(k) \beta(k) } - (1-\varepsilon) \alpha(k) - \rho \beta(k) \nonumber\\ & \quad +\rho^2 \frac{\widehat g^2_k - \widehat g^2_0\one_{\{|k|\leq \ell_{\delta}^{-1}\}}}{k^2}\Big) \dd k . \label{def.Ieps}
\end{align}
We recall that $\widehat{g}_0 = 8 \pi \delta$, where $\delta$ satisfies $\frac 1 2 Y \leq \delta \leq 2 Y$.
We are mainly interested into two special cases, namely $I_0(k^2,\widehat{g})$ and $I_{\varepsilon_N}(\tau, \widehat W_1)$. In this section we estimate these integrals.

\begin{lemma}\label{lem:Bogintegral_approx_tau_k}
	
	We can replace $\tau_k$ by $k^2$ up to the following error,
	
	\[ \vert I_{\varepsilon_N}(\tau, \widehat W_1) - I_{\varepsilon_N}(k^2, \widehat W_1) \vert \leq C \rho^2 \delta^2 \big( d + \varepsilon_T \vert \log Y \vert + (sK_\ell)^{-1} \big). \]
	
\end{lemma}

\begin{proof} 
	We recall the definition \eqref{def:tauk} of $\tau_k$,
	\[ \tau_k = (1-\varepsilon_T) \Big( \vert k \vert - \frac{1}{2s\ell} \Big)^2_+ + \varepsilon_T \Big( \vert k \vert - \frac{1}{2ds \ell} \Big)^2_+, \]
	from which we deduce the bounds
	\begin{align}\label{bound.tauk.1}
		\vert \tau_k - k^2 \vert &\leq 
		\begin{cases}
			\frac{1}{2s\ell} \vert k \vert + \frac{1}{2(s\ell)^2}, & \text{if} \quad \vert k \vert > \frac{1}{2ds \ell},\\
			\varepsilon_T \vert k \vert^2 + \frac{3}{2 s\ell} \vert k \vert, &\text{if} \quad \frac{1}{2 s\ell} < \vert k \vert < \frac{1}{2ds \ell}.
		\end{cases}
	\end{align}
	We write the integral as
	\begin{equation}\label{eq:Itau}
		I_{\varepsilon_N}(\tau,\widehat W_1) = \frac{1}{2(2\pi)^2} \int_{\R^2} F_k(\tau_k,\widehat W_1(k)) \dd k,
	\end{equation}
	with
	\begin{equation*}
		F_k(\tau,w) = \sqrt{(1-\varepsilon_N)^2 \tau^2 + 2 (1-\varepsilon_N)\rho w \tau} 
		- (1-\varepsilon_N)\tau - \rho w + \rho^2 \frac{\widehat g^2(k) - \widehat g^2(0) \one_{\{\vert k \vert \leq \ell_\delta^{-1}\}}}{2 k^2}.
	\end{equation*}
	We first consider separately the small $k$'s. Indeed, $\tau_k =0$ for $\vert k \vert \leq \frac{1}{2s \ell}$ and thus
	\begin{align*}
		\vert F_k(\tau_k, \widehat W_1(k)) - F_k(k^2, \widehat W_1(k)) \vert &= \Big\vert \sqrt{(1-\varepsilon_N)^2 k ^4 + 2 (1-\varepsilon_N)\rho \hat W_1(k) k^2} - (1-\varepsilon)k^2 \Big \vert \\
		&\leq C \sqrt{\rho\delta} \vert k \vert,
	\end{align*}
	(recall that $(s K_\ell)^{-1} \ll 1$) and
	\begin{equation}
		\frac{1}{2(2\pi)^2} \int_{\vert k \vert \leq (2s\ell)^{-1}} \vert F_k(\tau_k,\widehat W_1(k)) - F_k(k^2,\widehat W_1(k)) \vert \dd k \leq C \rho^2 \delta^2 (sK_\ell)^{-3}.
	\end{equation}
	The part with larger $k$ we bound using the derivatives of $F$ and deduce
	\begin{align}
		\vert I_{\varepsilon_N}(\tau, \widehat W_1) - I_{\varepsilon_N}(k^2,\widehat W_1) \vert &\leq\frac{1}{2(2\pi)^2} \int_{\vert k \vert > (2s \ell)^{-1}} \sup_{\tau \in [\tau_k, k^2] } \vert \partial_\tau F_k(\tau,\widehat W_1(k)) \vert \cdot \vert \tau_k - k^2 \vert \dd k \nonumber \\
		&\quad + C \rho^2 \delta^2 (sK_\ell)^{-3}.\label{eq:C4}
	\end{align}
	The derivative of $F$ is given by
	\begin{equation}\label{eq:Ftauw}
		\partial_\tau F(\tau,w) = \frac{(1-\varepsilon_N)^2 \tau + (1-\varepsilon_N) \rho w}{\sqrt{(1-\varepsilon_N)^2 \tau^2 + 2 (1-\varepsilon_N) \rho w \tau }} - (1- \varepsilon_N)
	\end{equation}
	and can be estimated for $\tau \in [ \tau_k, k^2 ]$ as
	\begin{equation}\label{eq:partialtau}
		\vert \partial_\tau F_k(\tau, \widehat W_1(k)) \vert \leq
		\begin{cases}
			C \frac{\sqrt{\rho \delta}}{\vert k \vert - (2s \ell)^{-1}}, &\text{if } (2s\ell)^{-1}< \vert k \vert< \sqrt{\rho \delta}, \\
			C \frac{\rho^2 \delta^2}{k^4}, &\text{if } \vert k \vert > \sqrt{\rho \delta} .
		\end{cases}
	\end{equation}
	Indeed, for $\vert k \vert < \sqrt{\rho \delta}$, we just need to bound individualy each term in \eqref{eq:Ftauw}, whereas for $\vert k \vert > \sqrt{\rho \delta}$, we have $\tau_k > C \rho \widehat W_1(k)$ and we use a Taylor expansion of the square root to get
	\begin{align*}
		\vert \partial_\tau F_k(\tau, \widehat W_1(k)) \vert \leq C \frac{\rho^2 \widehat W_1(k)^2}{(1-\varepsilon_N) \tau^2} \leq C \frac{\rho^2 \delta^2}{k^4}.
	\end{align*}
	
	Now we split the integral in \eqref{eq:C4} into 3 parts, we use \eqref{eq:partialtau}, \eqref{bound.tauk.1} to bound it and find:
	\begin{align*}
		\frac{1}{2(2\pi)^2} \int_{\{(2s\ell)^{-1} < k < \sqrt{\rho \delta}\} } \sup_{\tau \in [\tau_k, k^2] } \vert \partial_\tau F_k(\tau,\widehat W_1(k)) \vert \vert \tau_k - k^2 \vert \dd k
		&\leq C \rho^2 \delta^2 \Big( \varepsilon_T + \frac{1}{s K_\ell} \Big),\\
		\frac{1}{2(2\pi)^2} \int_{\{\sqrt{\rho \delta} < k < (2ds\ell)^{-1}\}} \sup_{\tau \in [\tau_k, k^2] } \vert \partial_\tau F_k(\tau,\widehat W_1(k)) \vert \vert \tau_k - k^2 \vert \dd k
		&\leq C \rho^2 \delta^2 \Big( \varepsilon_T \vert \log Y \vert + \frac{1}{sK_\ell} \Big), \\
		\frac{1}{2(2\pi)^2} \int_{\{k > (2ds\ell)^{-1}\}} \sup_{\tau \in [\tau_k, k^2] } \vert \partial_\tau F_k(\tau,\widehat W_1(k)) \vert \vert \tau_k - k^2 \vert \dd k
		&\leq C \rho^2 \delta^2 d.
	\end{align*}
\end{proof}

\begin{lemma}\label{lem:Bogintegral_approx_W_g}
	
	We can replace $\widehat W_1(k)$ by $\widehat g_k$ up to the following error
	
	\[ \vert I_{\varepsilon_N}(k^2, \widehat W_1) - I_{\varepsilon_N}(k^2, \widehat g) \vert \leq C \rho^2 \delta^2 K_\ell^{-1} + C \rho^2 \delta \varepsilon_N.\]
	
\end{lemma}

\begin{proof}
	Recall that $I_{\varepsilon_N}(k^2, \widehat W_1)$ is given by \eqref{eq:Itau}. We first use \eqref{eq:gw0.approx} and \eqref{eq. bound on gw_0} to replace the last part,
	\begin{equation*}
		\rho^2 \int_{\R^2} \frac{\widehat g_k^2 - \widehat g_0^2 \one_{\{\vert k \vert \leq \ell_\delta^{-1}\}}}{2 k^2} \dd k = \rho^2 \int_{\R^2} \frac{\widehat W_1(k)^2 - \widehat W_1(0)^2 \one_{\{\vert k \vert \leq \ell_\delta^{-1}\}}}{2 (1-\varepsilon_N) k^2} \dd k + \mathcal{O}( \rho^2 \delta \varepsilon_N + \rho^2 \delta^2 K_\ell^{-2}),
	\end{equation*}
	so that
	\begin{equation} \label{eq:INJN}
		I_{\varepsilon_N}(k^2, \widehat W_1) = J(\widehat W_1) + \mathcal{O}( \rho^2 \delta \varepsilon_N + \rho^2 \delta^2 K_\ell^{-2}) \quad \text{and} \quad I_{\varepsilon_N}(k^2, \widehat g) = J(\widehat g),
	\end{equation}
	with
	\begin{equation}
		J(w) = \frac{1}{2(2\pi)^2} \int_{\R^2} G_k(w_k,w_0) \dd k,
	\end{equation}
	and
	\begin{align*}
		G_k(w,w_0) &= \sqrt{(1-\varepsilon_N)^2 k^4 + 2 (1-\varepsilon_N)\rho w k^2} 
		- (1-\varepsilon_N)k^2 - \rho w + \rho^2 \frac{w^2 - w_0^2 \one_{\{\vert k \vert \leq \ell_\delta^{-1}\}}}{2 (1-\varepsilon_N) k^2}.
	\end{align*}
	Note that $G_k$ is independent of $w_0$ for $\vert k \vert > \ell_\delta^{-1}$. Then we split $J(w)$ into two parts,
	\begin{align}
		J(w) &= \frac{1}{2(2\pi)^2} \int_{\{\vert k \vert < \ell^{-1}_\delta\}} G_k(w_k,w_0) \dd k + \frac{1}{2(2\pi)^2} \int_{\{\vert k \vert > \ell^{-1}_\delta\}} G_k(w_k) \dd k \nonumber \\
		&=: J_<(w) + J_>(w).
	\end{align}
	For $k > \ell_\delta^{-1}$ we use
	\begin{equation}
		\vert J_>(\widehat W_1) - J_>(\widehat g) \vert \leq \frac{1}{2(2\pi)^2} \int_{\{\vert k \vert > \ell_\delta^{-1}\}} \sup_{w \in [\widehat g_k, \widehat W_1(k)] } \vert \partial_w G_k(w) \vert \cdot \vert \widehat W_1(k) - \widehat g_k \vert \dd k,
	\end{equation}
	with
	\begin{equation}
		\partial_w G = \frac{\rho}{\sqrt{1+ \frac{2 \rho w}{(1-\varepsilon_N) k^2}}} - \rho + \frac{\rho^2 w}{(1-\varepsilon_N)k^2} .
	\end{equation}
	We use a Taylor expansion of the square root to get
	\begin{align*}
		\vert J_>(\widehat W_1) - J_>(\widehat g) \vert & \leq C \rho^3 \int_{\vert k \vert > \ell^{-1}_\delta} \frac{\widehat g_k^2}{k^4} \vert \widehat W_1(k) - \widehat g_k \vert \dd k.
	\end{align*}
	Since $\vert \widehat W_1(k) - \widehat g_k \vert \leq C \delta^2 K_\ell^{-1}$ (by \eqref{eq:Wg_approx}) and $\int \widehat g_k^2 k^{-2} \dd k < C \delta$ (see \eqref{eq. bound on gw_0}) we deduce
	\begin{equation} \label{eq:C11}
		\vert J_>(\widehat W_1) - J_>(\widehat g) \vert \leq C \rho^2 \delta^2 K_\ell^{-1} .
	\end{equation}
	For $k < \ell_\delta^{-1}$ we start by focusing on the first part of $G_k$,
	\begin{equation}
		F_k(w) = \sqrt{(1-\varepsilon_N)^2 k^4 + 2 (1-\varepsilon_N)\rho w k^2} - (1-\varepsilon_N)k^2 - \rho w.
	\end{equation}
	Since $\vert \partial_w F_k \vert \leq C \rho$, we have
	\begin{equation}
		\Big\vert \int_{\{\vert k \vert < \ell_\delta^{-1}\}} F_k(\widehat W_1(k)) - F_k(\widehat g_k) \dd k \Big\vert \leq C \rho \int_{\{\vert k \vert < \ell_\delta^{-1}\}} \vert \widehat W_1(k) - \widehat g_k \vert \dd k \leq C \rho^2 \delta^3 K_\ell^{-1}.
	\end{equation}
	Now
	\begin{align*}
		\vert J_<(\widehat W_1) - J_<(\widehat g) \vert &\leq C \Big\vert \int_{\{\vert k \vert < \ell_\delta^{-1}\}} F_k(\widehat W_1(k)) - F_k(\widehat g_k) \dd k \Big\vert \\ & \quad + C \Big\vert \int_{\{\vert k \vert < \ell_\delta^{-1}\}} \rho^2 \frac{\widehat W_1(k)^2 - \widehat W_1(0)^2}{2(1-\varepsilon_T)k^2} \dd k \Big\vert \nonumber \\
		&\quad+ C \Big\vert \int_{\{\vert k \vert < \ell_\delta^{-1}\}} \rho^2 \frac{\widehat g_k^2 - \widehat g_0^2}{2(1-\varepsilon_T)k^2} \dd k \Big\vert \\
		&\leq C \rho^2 \delta^3 K_\ell^{-1} + C \rho^2 R^2 \delta^2 \ell_\delta^{-2} \leq C \rho^2 \delta^2 K_\ell^{-1},
	\end{align*}
	where we used \eqref{eq:332}. Combining this with \eqref{eq:INJN} and \eqref{eq:C11} the Lemma is proved.
\end{proof}

\begin{proposition}\label{prop:integralapprox2}
	There exists a universal constant $C >0$ such that, for any $\varepsilon \in [0,1)$,
	\begin{align*}
		\Big\vert I_\varepsilon (k^2, \widehat{g}_k) - 4 \pi \rho^2 \delta \Big( 1 - \frac{\delta}{Y} + \delta \log \delta + \Big(\frac 1 2 + 2 \Gamma + \log(\pi)\Big) \delta \Big) \Big\vert \\ \leq C \rho^2 \delta^3 \big( \vert \log(\delta) \vert R^2 \rho + 1 \big) + C \rho^2 \delta \varepsilon,
	\end{align*}
	where $I_\varepsilon$ is defined in \eqref{def.Ieps}. In particular when $\delta=\delta_0$ we deduce
	\begin{align*}
		\Big\vert I_\varepsilon (k^2, \widehat{g}_k) - 4 \pi \rho^2 \delta^2_0 \Big(\frac 1 2 + 2 \Gamma + \log(\pi) \Big) \Big\vert \leq C \rho^2 \delta^3 \big( \vert \log(\delta) \vert R^2 \rho + 1 \big) + C \rho^2 \delta \varepsilon.
	\end{align*}
\end{proposition}

\begin{proof}
	At first we want to replace $\widehat{g}_k$ by $\widehat{g}_0$ in the integral $I$:
	\begin{align}
		|I_\varepsilon(k^2, \widehat{g}_k)- I_\varepsilon(k^2, \widehat{g}_0)|&\leq \int_{\R^{2}} |F(k^2, \widehat{g}_k) - F(k^2, \widehat{g}_0)| \dd k\\
		& \leq \int_{\R^{2}} \sup_{g\in [\widehat{g}_k, \widehat{g}_0]}\left |\partial_g F(k^2, g)\right |\left |\widehat{g}_{k} -\widehat{g}_0 \right | \dd k =: I_{\leqslant} + I_{\geqslant},
	\end{align}
	where we split for values of $|k|$ under or above $(\rho \delta)^{1/2}$. Notice that
	\begin{equation}
		\partial_g F(k^2, \widehat g_k) =\frac{\rho k^{2}}{\sqrt{k^{4}+2\rho \widehat{g}_kk^{2} }}-\rho +\frac{\rho^{2} \widehat{g}_k}{k^{2}}.
		\label{eq:derofF}
	\end{equation}
	By a Taylor expansion we can prove that
	\begin{equation}
		I_{\leqslant} \leq
		C \int_{\{|k| \leq (\rho \delta)^{1/2}\}}\Big( R^2 (\rho \widehat{g}_0)^{1/2} k^3 + \rho \widehat{g}_0k^2 + R^2\big (\rho \widehat{g}_0\big )^{2} \Big) \dd k
		%\nonumber \\
		\leq C R^2 (\rho \delta)^3.
	\end{equation}
	In the other case we have, by Taylor expansion of the square root in \eqref{eq:derofF},
	\begin{align*}
		I_{\geqslant} &\leq C\rho\int_{\{|k| \geq (\rho \delta)^{1/2}\}} \frac{(\rho \widehat{g}_0)^2}{k^4} |\widehat{g}_{k} -\widehat{g}_0 | \dd k \\
		&\leq C \left (\rho \widehat{g}_0\right )^3 \Big( \int_{\{ (\rho \delta)^{1/2}\leq |k| \leq (\rho \delta)^{1/2} \widehat{g}_0^{-1/2}\}} \frac{R^2}{k^2} \dd k + \int_{\{|k| \geq (\rho \delta)^{1/2} \widehat{g}_0^{-1/2}\}} \frac{\dd k}{k^4} \Big) \\
		&\leq C \left (\rho \delta\right )^{3} R^2 |\log \delta | + C \left (\rho \delta\right )^{2}\delta.
	\end{align*}
		We deduce that $|I_\varepsilon(k^2, \widehat{g}_k)- I_\varepsilon(k^2, \widehat{g}_0)| \leq C \rho^2 \delta^3 (1 + R^2 \rho \log(\delta))$. Now remains to compute $I_\varepsilon(k^2, \widehat{g}_0)$. In this integral we use the new variable $q = k (\rho \widehat{g}_0)^{-\frac 1 2} (1-\varepsilon)^{\frac 1 2}$,
	\begin{equation}
		I_\varepsilon(k^2, \widehat{g}_0) = \frac{(\rho \widehat{g}_0)^2}{2(2\pi)^2(1-\varepsilon)} \int_{\R^2} \Big(\sqrt{q^4 + 2 q^2} - q^2 - 1 + \frac{\one_{\{\vert q \vert > (1-\varepsilon)^{\frac 1 2} \ell_\delta^{-1} (\rho \widehat{g}_0)^{- \frac 1 2}\}}}{2 q^2} \Big)\dd q.
	\end{equation}
	In term of $c_0 = (1-\varepsilon)^{\frac 1 2} \ell_\delta^{-1} (\rho \widehat{g}_0)^{- \frac 1 2}$, this integral is explicitly computable and equal to
	\begin{equation}
		I_\varepsilon(k^2, \widehat{g}_0) = \frac{(\rho \widehat{g}_0)^2}{4 \pi (1-\varepsilon)} \Big( \frac 1 8 - \frac{\log 2}{4} + \frac 1 2 \log(c^{-1}_0) \Big).
	\end{equation}
	With $\widehat{g}_0 = 8 \pi \delta$ and $c_0 = 2 e^{-\Gamma} e^{-\frac{1}{2 \delta}} (1-\varepsilon)^{\frac 1 2} \widehat{g}_0^{-\frac 1 2} (\rho a^2)^{-\frac 1 2}$ (see \eqref{eq:Defldelta}), we find
	\begin{equation}
		I_\varepsilon(k^2, \widehat{g}_0) = 4 \pi \rho^2 \delta^2 \Big( \frac{1}{\delta} - \frac{1}{Y} + \log \delta + \frac 1 2 + 2 \Gamma + \log(\pi) \Big) (1+ \grandO(\varepsilon)).
	\end{equation}
\end{proof}
\begin{remark}\label{remark.delta.low}
With the arbitrary parameter $\delta$ (within the range $\frac 1 2 Y \leq \delta \leq 2 Y$), one can deduce from Proposition~\ref{prop:integralapprox2} that our lower bound on the energy is
\begin{equation}
e^{\rm{2D}}(\rho) \geq 4\pi \rho^2 \delta \Big( 2 - \frac{\delta}{Y} + \delta \log \delta + \Big( \frac 1 2 + 2 \Gamma + \log (\pi) \Big) \delta \Big) - C \rho^2 Y^{2+\eta}.
\end{equation}
However, this lower bound is maximized by $\delta = Y( 1 - Y \vert \log Y \vert + o(Y \vert \log Y\vert))$, thus leading to the optimal choice $\delta = \delta_0$.
\end{remark}

We conclude this section by a general bound on Bogoliubov integrals that is used several times throughout the paper.

\begin{lemma}\label{lem:bogintestimate}
	For two functions $\mathcal{A}, \mathcal{B} : \mathbb{R}^2 \rightarrow \mathbb{R}$ such that
	\begin{equation}
		\mathcal{A}(k) \geq \kappa [|k| - K]_+^2 + 2 K_1 \delta, \qquad |\mathcal{B}(k)| \leq K_2 \delta, \qquad |\mathcal{B}(k) - \mathcal{B}(0)| \leq K_2 R^2 \delta |k|^2,
	\end{equation}
	for constants $\kappa > 0, 0 < K_2 \leq K_1, \ell_{\delta}^{-1} < K $, then there exists $C > 0$ such that
	\begin{align}
		\int_{\mathbb{R}^2} \, &\big(\mathcal{A}(k) - \sqrt{\mathcal{A}(k)^2 - \mathcal{B}(k)^2} \big)\dd k \nonumber \\
		&\leq \kappa^{-1}\int_{\mathbb{R}^2} \, \frac{\mathcal{B}^2(k) - \mathcal{B}^2(0)\one_{\{|k|\leq \ell_{\delta}^{-1}\}}}{2|k|^2}\dd k\nonumber \\
		&\quad + C \frac{K_2^2}{K_1} \delta K^2 + C \kappa^{-1} K_2^2 \delta^2 (1 + R^2 \ell^{-2}_{\delta} ) +C \kappa^{-1} K_2^2\delta^2 |\log (2K \ell_{\delta})|
		\nonumber \\
		&\quad+ C \min\Big( K_2^4 \delta^4 \kappa^{-3} K^{-4}, C \frac{K_2^2 \delta^2 \kappa^{-1}}{K_1^2} \int_{\mathbb{R}^2} \frac{\mathcal{B}(k)^2 - \mathcal{B}(0)^2 \one_{\{|k| < \ell_{\delta}^{-1}\}}}{|k|^2} \dd k \Big).
	\end{align}
\end{lemma}

\begin{proof}
	We show that the difference
	\begin{equation}
		\int_{\mathbb{R}^2} \, \Big(\big(\mathcal{A}(k) - \sqrt{\mathcal{A}(k)^2 - \mathcal{B}(k)^2} \big) - \kappa^{-1} \frac{\mathcal{B}^2(k) - \mathcal{B}^2(0)\one_{\{|k|\leq \ell_{\delta}^{-1}\}}}{2|k|^2}\Big)\dd k,
	\end{equation}
	is bounded by the desired error terms.
	
	For $ |k | \leq 2K$ we have that
	\begin{align*}
		\int_{|k|\leq 2K} \, \big(\mathcal{A}(k) - \sqrt{\mathcal{A}(k)^2 - \mathcal{B}(k)^2} \big)\dd k &\leq C \int_{\{|k|\leq 2K\}} \,\frac{\mathcal{B}^2(k)}{\mathcal{A}(k)}\dd k \\
		&\leq C \frac{K_2^2}{K_1} \delta \int_{\{|k|\leq 2K\}} \dd k\,= C \frac{K_2^2}{K_1} \delta K^2,
	\end{align*}
	while for the $\mathcal{B}$ part, using the assumption on $|\mathcal{B}(k) - \mathcal{B}(0)|$,
	\begin{equation}
		\kappa^{-1} \int_{\{|k|\leq \ell_{\delta}^{-1}\}} \frac{|\mathcal{B}^2(k) - \mathcal{B}^2(0)|}{2|k|^2}\dd k\leq C \kappa^{-1} K_2^2 R^2 \delta^2 \int_{\{|k|\leq \ell_{\delta}^{-1}\}} \dd k = C \kappa^{-1} K_2^2 R^2 \delta^2 \ell^{-2}_{\delta},
	\end{equation}
	and
	\begin{equation}
		\kappa^{-1} \int_{\{\ell^{-1}_{\delta} \leq |k| \leq 2K\}} d\, \frac{\mathcal{B}_k^2}{2|k|^2}\dd k \leq C \kappa^{-1} K_2^2\delta^2 |\log (2K \ell_{\delta})|.
	\end{equation}
		For $|k | \geq 2K$ we have, by a Taylor expansion,
	\begin{equation}
		\mathcal{A}(k) - \sqrt{\mathcal{A}(k)^2 - \mathcal{B}(k)^2} \leq \frac{1}{2} \frac{\mathcal{B}(k)^2}{\mathcal{A}(k)} + C \frac{\mathcal{B}(k)^4}{\mathcal{A}(k)^3}.
	\end{equation}
		For the first term we observe that
	\begin{equation}
		\frac{\mathcal{B}(k)^2}{\mathcal{A}(k)} \leq \kappa^{-1} \frac{\mathcal{B}(k)^2}{(|k|-K)^2} \leq \kappa^{-1} \frac{\mathcal{B}(k)^2}{|k|^2} \Big( 1+ \frac{K}{|k|}\Big),
	\end{equation}
	giving
	\begin{align*}
		\int_{\{|k| \geq 2 K\}} \, \bigg( \frac{\mathcal{B}(k)^2}{\mathcal{A}(k)} - \kappa^{-1}\frac{\mathcal{B}(k)^2}{2|k|^2}\bigg)\dd k\leq C K \kappa^{-1} \int_{\{|k| \geq 2 K\}} \frac{\mathcal{B}(k)^2}{|k|^3}\dd k \leq C \kappa^{-1} K_2^2\delta^2,
	\end{align*}
	while for the second one we can bound either
	\begin{equation}
		\int_{\{|k| \geq 2 K\}} \, \frac{\mathcal{B}(k)^4}{\mathcal{A}(k)^3}\dd k \leq C K_2^4 \delta^4 \kappa^{-3}\int_{\{|k| \geq 2 K\}} \frac{\dd k}{|k|^6} \leq C K_2^4 \delta^4 \kappa^{-3} K^{-4},
	\end{equation}
	or as in the following,
	\begin{align*}
		\int_{\{|k| \geq 2 K\}} \, \frac{\mathcal{B}(k)^4}{\mathcal{A}(k)^3}\dd k \leq \frac{K_2^2 \delta^2 \kappa^{-1}}{K_1^2} \int_{\{|k| \geq 2 K\}} \, \frac{\mathcal{B}(k)^2}{\mathcal{A}(k)} \dd k\leq C \frac{K_2^2 \delta^2 \kappa^{-1}}{K_1^2} \int_{\{|k| \geq 2 K\}} \, \frac{\mathcal{B}(k)^2}{|k|^2}\dd k,
	\end{align*}
	adding and subtracting the term $C \frac{K_2^2 \delta^2 \kappa^{-1}}{K_1^2} \int_{\{|k| <2K\}} \frac{\mathcal{B}(k)^2 - \mathcal{B}(0)^2 \one_{\{|k| < \ell_{\delta}^{-1}\}}}{|k|^2}$ we have
	\begin{equation}
		\int_{\{|k| \geq 2 K\}} \, \frac{\mathcal{B}(k)^4}{\mathcal{A}(k)^3}\dd k \leq C \frac{K_2^2 \delta^2 \kappa^{-1}}{K_1^2} \int_{\mathbb{R}^2} \frac{\mathcal{B}(k)^2 - \mathcal{B}(0)^2 \one_{\{|k| < \ell_{\delta}^{-1}\}}}{|k|^2} \dd k.
	\end{equation}
	This finishes the proof of Lemma \ref{lem:bogintestimate}.
\end{proof}

\section{A priori bounds}\label{app:low-smallbox}
In this section, we prove Theorem~\ref{thm:excitationsbound}.
We study a localized problem on a shorter length scale $d\ell$ such that 
\begin{equation}
d \ell \ll \ell_{\delta} \ll \ell.
\end{equation}
where we recall that $\ell_{\delta}$ is the healing length. We are able, in this section, to prove Bose-Einstein condensation in boxes with length scale smaller than the healing length. A key point is that, at this scale, we can use a larger Neumann gap to reabsorb the errors. We will show how the proof of Theorem~\ref{thm:excitationsbound} reduces to this localized problem.

We introduce the small box centered at $u \in \mathbb{R}^2$ to be 
\begin{equation}
B_u = \Lambda \cap \bigg\{ d\ell u + \Big[-\frac{d\ell}{2} , \frac{d\ell}{2}  \Big]^2 \bigg\}.
\end{equation}
The associated localization functions are
\begin{equation}\label{eq:Blocfunction}
\chi_{B_u}(x) := \chi \left( \frac{x}{\ell}\right) \chi \left( \frac{x}{d\ell} -u \right),
\end{equation}
where we highlight that 
\begin{equation}
\iint |\chi_{B_u}|^2 \dd x \dd u = \ell^2.
\end{equation}
In order to construct the small box Hamiltonian, we introduce the localized potentials 
\begin{align}\label{eq:SF_3.5}
W^{\rm s}(x) &:= \frac{W(x)}{\chi*\chi(x/d\ell)},  &w_{B_u}(x,y) := \chi_{B_u} (x) W^s(x-y)\chi_{B_u}(y), \\
W_1^{\rm s}(x) &:= \frac{W_1(x)}{\chi*\chi(x/d\ell)},  &w_{1,B_u}(x,y) := \chi_{B_u} (x) W_1^s(x-y)\chi_{B_u}(y),  \\
W_2^{\rm s}(x) &:= \frac{W_2(x)}{\chi*\chi(x/d\ell)},  &w_{2,B_u}(x,y) := \chi_{B_u} (x) W_2^s(x-y)\chi_{B_u}(y),
\end{align} 
where we recall that $W, W_1, W_2$ are localized versions of $v, g, (1+ \omega) g$, respectively (see formulas \eqref{eq:SF_w12u} and \eqref{eq:scat_defs}). Since  $v$ has support in $B(0,R)$, we see that $W^{\rm s}$ is well-defined as $d\ell$ is required to be larger then $R$.
Clearly $W^{\rm s}$ depends on $d\ell$ and thus $\rho_\mu$, but we will not reflect this in our notation.

Similarly to Lemma~\ref{lem:gomegaapprox}, $W_1^{\rm s}$ satisfies the following inequalities which can be proven in analogous ways considering the length scale $d \ell$ in place of $\ell$
\begin{align}
&\int W_2^{\rm s} \leq 2\int W_1^{\rm s} \leq C \delta, \label{eq:smallW_1control}\\
 &0 \leq W_1^{\rm s}(x) - g(x)\leq C g(x) \frac{|x|^2}{(d\ell)^2},\label{eq:smallW_1approx_g}\\
&\bigg| \frac{1}{(2\pi)^2} \int_{\mathbb{R}^2} \frac{\widehat{W}_1^{\rm s}(k)^2-\widehat{W}_1^{\rm s}(0)^2\one_{\{|k| \leq \ell_{\delta}^{-1}\}}}{2 k^2}\dd k  - \widehat{g\omega}(0)\bigg| \leq C \frac{R^2\delta}{(d \ell)^2}. \label{eq:smallW_approx_gomega}
\end{align}
We define furthermore, as operators on $L^2(B_u)$,
\begin{align}
P_{B_u} := \frac{1}{|B_u|} | \one_{B_u} \rangle \langle \one_{B_u}|, \qquad Q_{B_u}  := \one_{B_u} - P_{B_u},
\end{align}
\textit{i.e.} $P_{B_u}$ is the orthogonal projection in $L^2(B_u)$ onto the constant functions and $Q_{B_u}$ is the projection to the  orthogonal complement. We can therefore introduce the number operators as well
\begin{equation}
n_{B_u} := \sum_{j=1}^N \one_{B_u,j}, \quad n_{B_u,0} := \sum_{j=1}^N P_{B_u,j}, \quad n_{B_u,+} := \sum_{j=1}^N Q_{B_u,j},
\end{equation}
and the small-box kinetic energy
\begin{align}
\mathcal{T}_{B_u}  := Q_{B_u}\Big(
\chi_{B_u}
\Big[ \sqrt{-\Delta} - \frac{1}{ds\ell} \Big]^2_{+}
\chi_{B_u}
+
\frac{\varepsilon_T}{2} (1 + \pi^{-2})\frac{1}{(d\ell)^2}  \Big) Q_{B_u}.
\end{align}

We are now ready to define the localized Hamiltonian ${\mathcal H}_{B_u}$ which acts on the symmetric Fock space ${\mathcal F}_s (L^2(B_u))$. It preserves particle number and is given as
\begin{align}\label{eq:SF_Def_HB}
{\mathcal H}_{B_u}(\rmu)_{N} :=
 \sum_{i=1}^N (1-\varepsilon_N) \mathcal{T}_{B_u}^{(i)} -
\rmu \sum_{i=1}^N \int w_{1, B_u}(x_i,y)\,\dd y + \frac{1}{2}\sum_{i \neq j} w_{B_u}(x_i,x_j),
\end{align}
on the $N$-particle sector.

An adaptation to dimension $2$ of \cite[Theorem 3.10]{BSol} allows us relate $\mathcal{H}_B(\rho_{\mu})$ to the original Hamiltonian in the large box, using the condition \eqref{eq:condition_small_loc}.
This gives the lower bound
\begin{equation}
\mathcal{H}_{\Lambda}(\rho_{\mu}) \geq (1-\varepsilon_N) \frac{b}{\ell^2}\sum_{j=1}^N Q_{\Lambda,j} + \int_{\mathbb{R}^2} \, \mathcal{H}_{B_u} (\rho_{\mu})\dd u. 
\end{equation}
We would like to restrict the previous integral to boxes that are not too small. Therefore, we identify the following sets of integration, for $\xi \in [0,1]$, 
\begin{equation}
\Lambda_{\xi} := \Big\{u \in \mathbb{R}^2\,\Big|\, \vert\ell d u \vert_{\infty} - \frac{\ell}{2}(d+1) \leq  - \xi d \ell \Big\}, 
\end{equation}
underlying the property
\begin{equation}
 \Lambda_{\xi_1} \subseteq \Lambda_{\xi_2} \quad \text{if}\quad \xi_1 \geq \xi_2, 
\end{equation} 
and we observe that integration outside $\Lambda_0$ is zero because there is no more intersection between the small box and $\Lambda$. 
The following Lemma guarantees that we can restrict the integration for the potential over set $\Lambda_{1/10}$ (where $1/10$ is chosen arbitrarily) and estimate the remaining part by a frame inside $\Lambda_{1/10}$.

\begin{lemma}\label{lem:stripes_smallboxes}
For all $x \in \Lambda$ we have the estimate
\begin{multline}
-\rho_{\mu} \iint \, w_{1,{B_u}}(x,y)\dd y \dd u \\
\geq -\rho_{\mu} \int_{\Lambda_{\frac{1}{10}}} \,\int \, w_{1,{B_u}}(x,y) \dd y \dd u
-3\rho_{\mu} \int_{\Lambda_{\frac{1}{10}}\setminus\Lambda_{\frac{1}{5}}} \int \, w_{1,{B_u}}(x,y)\dd y \dd u.
\end{multline}
\end{lemma}

\begin{proof}
The proof follows the same lines as in \cite[Lemma E.1]{FS2}. We split the domain of integration $\Lambda_{1/10}$ and $\Lambda_0 - \Lambda_{1/10}$ and we estimate the integral over the latter. By the definition of $w_{1,B}$ we have simply to estimate the quantity
\begin{equation}
\int_{\Lambda_0 \setminus \Lambda_{1/10}}  \, \chi\left( \frac{x}{\ell d} - u\right) \chi\left( \frac{y}{\ell d} -u \right)\dd u.
\end{equation}
We use that $\chi$ is a product of decreasing functions in the variables $u_1, u_2$ and observe that
\begin{multline}
\max_{\frac{1}{2}(d^{-1}+1) -\frac{(\ell d)^{-1}}{10} \leq |u_1| \leq \frac{1}{2} (d^{-1}+1)}  \chi\left( \frac{x}{\ell d} - u\right) \chi\left( \frac{y}{\ell d} -u \right)   \\
\leq \min_{\frac{1}{2} (d^{-1} + 1) -\frac{2(d\ell)^{-1}}{10} \leq |u_1|\leq \frac{1}{2d} +\frac{1}{2} -\frac{(d\ell)^{-1}}{10 }} \chi\left( \frac{x}{\ell d} - u\right) \chi\left( \frac{y}{\ell d} -u \right),
\end{multline}
so that we can estimate the integral over the frame pointwise, getting a factor of $3$ due to the presence of the corners. 
\end{proof}

Thanks to Lemma~\ref{lem:stripes_smallboxes} we can write
\begin{equation}\label{eq:hamiltloclargesmallborder}
\mathcal{H}_{\Lambda}(\rho_{\mu}) \geq (1-\varepsilon_N) \frac{b}{\ell^2}\sum_{j=1}^N Q_{\Lambda,j} + \int_{\Lambda_{\frac{1}{5}}} \, \mathcal{H}_{B_u} ( \rho_{\mu}) \dd u +    \int_{\Lambda_{\frac{1}{10}}\setminus\Lambda_{\frac{1}{5}}} \, \mathcal{H}_{B_u} (4 \rho_{\mu})\dd u,
\end{equation}
where we dropped the positive part of $\mathcal{H}_{B_u}$ in $\Lambda_0\setminus\Lambda_{\frac{1}{10}}$. We are now ready to give lower bounds for kinetic and potential energies in terms of the number of particles. From this the lower bound for the small box Hamiltonian is going to follow in Corollary~\ref{cor:smalllowerbound} below. 

In order to prove Theorem~\ref{thm:excitationsbound}, we provide a lower bound on $\mathcal H_{B_u}(\rho_\mu)$. For notational simplicity we will remove the index $u$. Lemmas~\ref{lem:potentialsmalllowerbound} and~\ref{lem:kinsmallboxlower} below give first lower bounds on the potential and kinetic energy respectively.

\begin{lemma}\label{lem:potentialsmalllowerbound}
There exists a constant $C>0$ depending only on $\chi$ such that 
\begin{equation*}
-\rho_{\mu}\int \,  \sum_{j =1}^N w_{1,B}(x,y)\dd y + \frac{1}{2} \sum_{i \neq j } w_B(x_i, x_j) \geq A_0 + A_2 + \frac{1}{2} \mathcal Q_{4}^{\rm{ren, s}} - C\delta \Big( \rho_{\mu} + \frac{n_{0,B}}{|B|}\Big) n_{+,B},
\end{equation*}
with 
\begin{align*}
A_0 &:= \frac{n_{0,B}(n_{0,B} + 1)}{2|B|^2} \iint  \, w_{2,B}(x,y)\dd x \dd y \\
&\qquad- \Big(  \frac{\rho_{\mu} n_{0,B}}{|B|} + \frac{1}{4} \big( \rho_{\mu}- \frac{n_{0,B}-1}{|B|}\big)^2\Big) \iint  \, w_{1,B}(x,y)\dd x \dd y,\\
A_2 &:= \frac{1}{2} \sum_{i \neq j} P_i P_j w_{1,B} Q_i Q_j + h.c.,
\end{align*}
and $\mathcal Q_4^{\rm{ren,s}}$ is the analogue of \eqref{eq:SF_DefQ4}, but for the small box $B$.
\end{lemma}

\begin{proof}
The proof follows from an analogous potential splitting like in Lemma~\ref{lem:SF_potsplit-bigbox} and Lemma~\ref{lem:Qrewriting} and by the same lines as \cite[Lemma E.7]{FS2}. 
\end{proof}

\begin{lemma}\label{lem:kinsmallboxlower}
	For the kinetic energy on the small box in the $N$'th sector we have
\begin{align*}
Q_B&\chi_B \Big[ \sqrt{-\Delta} - \frac{1}{ds\ell} \Big]^2_{+}
\chi_B Q_B + A_2 \\
&\geq  -\frac{1}{2}\widehat{g\omega}(0) \frac{N(N+1)}{|B|^2}  \int \chi_B^2 + \mathcal{E}_2(N) + \mathcal{E}_4(N) - C \delta \frac{N +1}{|B|}  n_{+,B},
\end{align*}
where 
\begin{align}\label{eq:errors_thmsmall}
\mathcal{E}_2(N) &:= - C \delta \Big( \frac{R^2}{(d\ell)^2} + \delta |\log(d s K_{\ell})| +   \delta^2\Big) \frac{N(N+1)}{|B|^2}  \int \chi_{B}^2, \\
\mathcal{E}_4(N) &:= -C \Big( \delta^4 (ds\ell)^4 \Big( \frac{N+1}{|B|}\Big)^3 + \delta (ds\ell)^{-2}\Big) \frac{N}{|B|} \int \chi_{B}^2. 
\end{align}
\end{lemma}

\begin{proof}
Let us introduce the operators 
\begin{equation}
d^{\dagger}_p := \frac{1}{|B|^{1/2}} a^{\dagger}(Q_B \chi_B e^{-ipx}) a_0,
\end{equation}
where $a_0 = \frac{1}{\ell} a (1)$ and $a,a^{\dagger}$ are the annihilation and creation operators on $\mathscr{F}_s(L^2(\Lambda))$. Further we introduce
\begin{equation}
A_1 := \frac{\widehat{W}_1^s(0)}{(2\pi)^2}  \int_{\mathbb{R}^2} \, (d^{\dagger}_p d_p + d^{\dagger}_{-p} d_{-p})\dd p.
\end{equation}
Now using that on the $N$'th sector we have
\begin{equation}
\int \,  \Big[ \vert p \vert - \frac{1}{ds\ell} \Big]^2_{+} d^{\dagger}_p d_p \dd  p=  \frac{(n_0+1)}{|B|} \sum_{j=1}^N Q_{B,j}\chi_{B}(x_j)\Big[ \sqrt{-\Delta} - \frac{1}{ds\ell} \Big]^2_{+} \chi_B (x_j) Q_{B,j},
\end{equation}
and that $n_0 \leq N$, we get, adding $A_1$ and $A_2$ to the kinetic energy
\begin{equation}
Q_B\chi_B \Big[ \sqrt{-\Delta} - \frac{1}{ds\ell} \Big]^2_{+}
\chi_B Q_B + A_1 +A_2 \geq \frac{1}{2(2\pi)^2} \int \,h_p \dd p,
\end{equation}
where 
\begin{equation}
h_p := \mathcal{A}_p (d^{\dagger}_p d_p + d^{\dagger}_{-p} d_{-p} )  + \mathcal{B}_p (d^{\dagger}_p d^{\dagger}_{-p} + d_{-p} d_p),
\end{equation}
with 
\begin{equation}
\mathcal{A}_p := (1- \varepsilon_N) \frac{|B|}{N+1} \Big[ |p| - \frac{1}{ds\ell}\Big]^2_+  +  2 \widehat{W}_1^s(0), \qquad \mathcal{B}_p := \widehat{W}_1^s(p).
\end{equation}
The additional $A_1$ term is estimated, thanks to \eqref{eq:smallW_1control}, by 
\begin{equation}
A_1 \leq C \delta \frac{n_0+1}{|B|} n_{+,B}, 
\end{equation}
which contributes to the last term in the result of the lemma. 
By an application of Theorem~\ref{thm:bogdiag} we get the bound
\begin{equation}\label{eq:small_Bog_estimate}
\frac{1}{2(2\pi)^2} \int \,h_p\dd p  \geq -\frac{1}{2(2\pi)^2}  \frac{N}{|B|}  \int \chi_B^2\int \, (\mathcal{A}_p -\sqrt{\mathcal{A}_p^2 - \mathcal{B}^2_p})\dd p,
\end{equation}
and therefore we want to bound the latter.
We observe that, thanks to \eqref{eq:smallW_1approx_g}, $\widehat W_1^{\rm s}(0) \geq C \delta $ for a certain $C < 8\pi$. 
Choosing the parameters
\begin{equation}
K = (d s\ell)^{-1}, \quad \kappa = (1-\varepsilon_N)\frac{|B|}{N+1}, \quad K_1 =1, \quad K_2 = C,
\end{equation}
we can apply Lemma~\ref{lem:bogintestimate} to obtain 
\begin{multline*}
-\frac{1}{2(2\pi)^2} \int \, (\mathcal{A}_p -\sqrt{\mathcal{A}_p^2 - \mathcal{B}^2_p})\dd p \geq -\frac{N+1}{|B|(1-\varepsilon_N)} \widehat{g\omega}(0) - \frac{N+1}{|B|}\frac{R^2}{(d\ell)^2}\delta+ C \delta (ds\ell)^{-2}\\
  -C \frac{N+1}{|B|}   \delta^2 \Big(1 + \Big(\frac{R}{ \ell_{\delta}}\Big)^{2} \Big)  -C \frac{N+1}{|B|} \delta^2 |\log (2d s K_{\ell})| - C  \delta^4 \frac{(N+1)^3}{|B|^3} (ds\ell)^{4},
\end{multline*}
where we used that $\varepsilon_N \leq 1/2$ and \eqref{eq:smallW_approx_gomega} to approximate the leading term by $\widehat{g\omega}(0)$ getting the second term as an error. Plugging the last estimate in \eqref{eq:small_Bog_estimate} we get the result with the error terms coherent with the definitions of $\mathcal{E}_2(n)$ and $\mathcal{E}_4(n).$
\end{proof}

We will also need the following estimates.

\begin{lemma}\label{lem:gomegasmallapprox}
Let $\ell_{\text{min}}$ denote the shortest length of the box $B$, then there exists a constant $C>0$ such that 
\begin{align}\label{eq:step_techn_smallIntw:result}
&\bigg|\iint  \, w_{1,B}(x,y) \dd x \dd y-  8\pi \delta\int \chi_B^2 \bigg| \leq  C \delta \Big(\frac{R}{\ell_{\text{min}}}\Big)^2 \int \chi_B^2,
\intertext{and}
\label{eq:step_techn_smallIntw:result2}
&\iint \,  w_{2,B} (x,y)\dd x \dd y \geq \iint    w_{1,B} (x,y)\dd x \dd y + \widehat{g\omega}(0)\int \chi_B^2 - C \frac{R\delta^2}{\ell_{\min}^2} \int \chi_B^2.
\end{align}
\end{lemma}

\begin{proof}
Since $8\pi \delta= \int g$, and thanks to \eqref{eq:smallW_1approx_g}, we can write the inequality
\begin{equation}\label{eq:step_techn_smallIntw}
\bigg|\iint (W^s_1(x) -g(x) )\chi_B^2(y) \dd x \dd y  \bigg| \leq C\delta \Big(\frac{R}{\ell_{\text{min}}} \Big)^2 \int \chi_B^2,
\end{equation}
where we used $\ell_{\text{min}}\leq d\ell$.
By a Taylor expansion for the localization function and the fact that $W$ is spherically symmetric and \eqref{eq:smallW_1control}, we have, on the other hand, 
\begin{align}
\left|\iint  \, w_{1,B}(x,y)\dd x \dd y - \int \, W^s_1(x) \dd x \int \chi_B^2\right| &\leq  C R^2\|\nabla^2\chi_B\|_{\infty} \int W_1^s (x) \dd x \int \chi_B \nonumber \\
&\leq  C \Big( \frac{R}{\ell_{\text{min}}}\Big)^2  \delta  \int \chi_B^2,\label{eq:step_techn_smallIntw2}
\end{align}
and where we used that $|B|^{-1} (\int \chi_B)^2 \leq  \int \chi_B^2$ and the bound \eqref{eq:bound_loc_smallbox_M}.

Then inequality \eqref{eq:step_techn_smallIntw:result} follows by \eqref{eq:step_techn_smallIntw} and \eqref{eq:step_techn_smallIntw2}.
The inequality \eqref{eq:step_techn_smallIntw:result2} follows from a very similar argument.
\end{proof}

Combining the results of Lemma~\ref{lem:potentialsmalllowerbound} and~\ref{lem:kinsmallboxlower}, we deduce that the Hamiltonian on the small box has the following lower bound, which is coherent with the main order of the energy expansion.

\begin{theorem}\label{thm:lowerboundSmalln}
Assume the conditions from Appendix~\ref{app:parameters}, then for any box $B$ we have the following bound on the $N$'th sector
\begin{equation}
\mathcal{H}_{B} (\rho_{\mu})|_N \geq \Big(  \frac{1}{4} \Big( \rho_{\mu} - \frac{N}{|B|}\Big)^2- \frac{1}{2}\rho_{\mu}^2  \Big) \iint  w_{1,B} 
+ \frac{1}{2} \mathcal Q_4^{\rm{ren,s}} + \mathcal{E}_2(N) + \mathcal{E}_4(N),   \label{eq:lowersmallnfixed}
\end{equation}
with $\mathcal{E}_2$ and $\mathcal{E}_4$ defined in \eqref{eq:errors_thmsmall}.
\end{theorem}

\begin{proof}
The combination of Lemmas~\ref{lem:potentialsmalllowerbound},~\ref{lem:kinsmallboxlower}
gives
\begin{align*}
\mathcal{H}_B (\rho_{\mu}) &\geq
\sum_{j=1}^n Q_{B,j}\Big(
 \frac{\varepsilon_T}{2} (1 + \pi^{-2})\frac{1}{(d\ell)^2}  \Big) Q_{B,j}
+A_0  
-\frac{1}{2}\widehat{g\omega}(0) \frac{N(N+1)}{|B|^2}  \int \chi_B^2 + \frac{1}{2} \mathcal Q_4^{\rm{ren,s}}\\
&\quad+\mathcal{E}_2(N) + \mathcal{E}_4(N) - C \delta \rho_{\mu} n_{+,B}.
\end{align*}
We observe that we can choose a constant $C'>0$ such that 
\begin{equation}
\sum_{j=1}^N Q_{B,j}\Big(
\frac{\varepsilon_T}{2} \frac{1}{(d\ell)^2}  \Big) \geq C' \rho_{\mu} \delta n_{+,B},
\end{equation}
where we used \eqref{eq:rel_T2comm} and, choosing the right $C'$, we can cancel the last term with $n_{+,B}$ for a lower bound. The same can be said for the errors produced by replacing $n_0 = N- n_{+}$ by $N$.
By Lemma~\ref{lem:gomegasmallapprox} we have
\begin{align*}
A_0 &-\frac{1}{2}\widehat{g\omega}(0) \frac{N (N+1)}{|B|}  \int \chi_B^2  \\
&\geq \Big( -\frac{N^2}{2|B|^2} - \Big( \rho_{\mu}\frac{N}{|B|} + \frac{1}{4} \Big( \rho_{\mu} - \frac{N}{|B|}\Big)^2 \Big)  \Big) \iint  w_{1,B} -  C \frac{N^2}{|B|^2}\frac{R\delta^2}{\ell_{\min}^2}  \int \chi_B^2  \\
&\geq \Big(  \frac{1}{4} \Big( \rho_{\mu} - \frac{N}{|B|}\Big)^2- \frac{1}{2}\rho_{\mu}^2  \Big) \iint  w_{1,B} - C \frac{N^2}{|B|^2}\frac{R\delta^2}{\ell_{\min}^2}  \int \chi_B^2,
\end{align*}
and this gives the result since the last term can be reabsorbed in the $\mathcal{E}_2$ term.
\end{proof}

We deduce the following corollary.

\begin{corollary}\label{cor:smalllowerbound}
Assume $B$ is a small box with shortest side length $\ell_{\text{min}} \geq \frac{d\ell}{10 }$ and that the conditions of Appendix~\ref{app:parameters} hold true. Then we have the following lower bound
\begin{align*}
\mathcal{H}_{B}(\rho_{\mu}) \geq  - \frac{1}{2}\rho_{\mu}^2  \iint   \,  w_{1,B}(x,y)\dd x \dd y   -  C \rho_{\mu}^2 \delta^2 (d s K_{\ell})^{-2} \int\chi_{B}^2 -C \rho_{\mu} \delta \frac{1}{|B|} \int \chi_{B}^2. 
\end{align*}
\end{corollary}

\begin{proof}
We split the particles in $m$ subsets of $n'$ particles and a remaining group of $n''$, with $n'' < n' < n$. If we ignore the positive interactions between the subsets, and denoting by $e_B(n, \rho_{\mu})$ the ground state energy of $\mathcal{H}_B(\rho_{\mu})$ restricted to states with $n$ particles in the box $B$, then
\begin{equation}
e_B(n, \rho_{\mu}) \geq  m e_B(n', \rho_{\mu}) + e_B(n'' , \rho_{\mu}).
\end{equation}

From formula \eqref{eq:lowersmallnfixed} in Theorem~\ref{thm:lowerboundSmalln} applied for $n'$ in place of $n$ and, choosing $n' = 3 \rho_{\mu} |B|$, we see that 
the first term becomes, thanks to Lemma~\ref{lem:gomegasmallapprox}
\begin{equation}
\frac{1}{2}  \rho^2_{\mu} \iint w_{1,B}  \geq 4 \pi\rho^2_{\mu} \delta \Big( 1- C \Big(\frac{R}{\ell_{\text{min}}}\Big)^2\Big)  \int \chi_B^2.  
\end{equation} 
From the following controls on the error terms 
\begin{align}
\mathcal{E}_2(n') &\leq C \rho^2_{\mu} \delta^2  (d K_\ell)^{-2} \int \chi_{B}^2, \label{eq:technerrosmallcond1}\\
\mathcal{E}_4(n') &\leq C \rho^2_{\mu} \delta^2 (d s K_{\ell})^{-2} \int \chi_B^2, \label{eq:technerrosmallcond2}\\
C \rho_{\mu}^2\delta \Big(\frac{R}{\ell_{\text{min}}}\Big)^2\int \chi_B^2
 &\leq C\rho_{\mu}^2 \delta\frac{R^2}{(d\ell)^2}\int \chi_B^2 \leq C \rho_{\mu}^2 \delta^2 (d K_{\ell})^{-2} \int \chi_B^2, 
\end{align}
we see that the first term is the leading term of the energy. Since it is clearly positive, we obtain that with the aforementioned choice of $n'$, we have $e_{B}(n', \rho_{\mu}) \geq 0$ and, then, using the previous equality, we can state that 
\begin{equation}
e_B(n, \rho_{\mu}) \geq e_B(n'', \rho_{\mu}).
\end{equation} 

The Corollary follows using again Theorem~\ref{thm:lowerboundSmalln} with $n''$ in place of $n$ to obtain the lower bound and using \eqref{eq:technerrosmallcond1} and \eqref{eq:technerrosmallcond2} for $n''$ to control the errors, using that $s^{-1}\gg 1$ to obtain one of the error terms in the statement. A further error term of order \[C \rho_{\mu} \delta \frac{1}{|B|} \int \chi_B^2,\] is created by the substitutions of the terms $n'' \pm 1$ by $n''$.
\end{proof}

We are finally ready to use the lower bound on the small box Hamiltonian to obtain a bound on the number of excited particles in the large box, result stated in the Theorem below. By an abuse notation, from now on, the operators $n, n_+, n_0$ start again to denote the number operators in the large box.

\begin{theorem}\label{thm:small_largeHam_lowerbound}
We have the following lower bound for the large box Hamiltonian
\begin{align}
\mathcal{H}_{\Lambda}(\rho_{\mu}) \geq -4 \pi \rho_{\mu}^2 \ell^2 Y \Big(1 - \frac 1 2 Y |\log Y| \Big) + \frac{b}{2 \ell^2} n_+,\label{eq:roughlowerboundlarge}
\end{align}
and if there exists a normalized $\Psi \in \mathscr{F}_s(L^2(\Lambda))$ with $n$ particles in $\Lambda$ such that \eqref{eq:condensationaprioricondition} holds:
\begin{equation}
\langle \mathcal{H}_{\Lambda}(\rho_{\mu})\rangle_{\Psi} \leq -4 \pi \rho_{\mu}^2\ell^2 Y (1- C K_B^2 Y|\log Y|), 
\end{equation}
 then the bound \eqref{eq:condensationestimate} for the number of excitations holds:
\begin{equation}
\langle n_+\rangle_{\Psi} \leq C n K_B^2 K_{\ell}^2 Y |\log Y|.
\end{equation}
\end{theorem}

\begin{proof}
We study the integration over $\Lambda_{1/10} \setminus\Lambda_{1/5}$ from formula \eqref{eq:hamiltloclargesmallborder}. By \cite[(C.6)]{FS} we have $|\chi_{B_u}| \leq C (\ell_{\text{min}} \ell^{-1})^{M}$ and then
\begin{equation}\label{eq:chiB_border}
\int_{\Lambda_{1/10}-\Lambda_{1/5}} \int \chi_{B_u}(x)^2 \dd x \dd u \leq C \Big( \frac{\ell_{\text{min}}}{ \ell} \Big)^{2M} (\ell d )^2 d^{-2} \leq C \Big( \frac{\ell_{\text{min}} }{\ell} \Big)^{2M} \ell^2.
\end{equation}
By the joint action of Corollary~\ref{cor:smalllowerbound} and Lemma~\ref{lem:gomegasmallapprox} we get 
\begin{equation}
\mathcal{H}_{B_u} (4\rho_{\mu}) \geq -C \rho_{\mu}^2 \delta \int \chi_{B_u}^2 - C \rho_{\mu} \delta \Big( \rho_{\mu}\delta (dsK_{\ell})^{-2} + \frac{1}{|B_u|} \Big)\int \chi_{B_u}^2, 
\end{equation}
and therefore, using \eqref{eq:chiB_border} and that $|B|= d^2 \ell^2$ we have 
\begin{equation}
\int_{\Lambda_{1/10}  \setminus \Lambda_{1/5}} \mathcal{H}_{B_u} (4\rho_{\mu}) \dd u \geq - C \rho_{\mu}^2\ell^2 \delta (1 +  \delta   (dsK_{\ell})^{-2} )\Big( \frac{\ell_{\text{min}}}{\ell}\Big)^{2M} - C \rho_{\mu} \delta d^{-2},
\end{equation}
Using the definition of $\ell_{\text{min}} = d\ell/10$ and the relations between the parameters \eqref{eq:K_Bsmallbox} we get
\begin{equation}\label{eq:localdep_param_small}
\Big( \frac{\ell_{\text{min}}}{\ell}\Big)^{2M}  \leq d^{2M} \leq \delta,\qquad 
\rho_{\mu} \delta d^{-2} \leq \rho_{\mu}^2 \ell^2  \delta^2 (K_{\ell} d)^{-2} \leq \rho_{\mu}^2 \ell^2  \delta^2 K_B^2, \qquad  
\end{equation}
which makes the integral coherent with the statement of the Theorem using the expansion of $\delta$,
\begin{equation}\label{eq:delta_expansion_small}
\delta \simeq Y - Y^2 |\log Y| + \mathcal{O}(Y^3 \vert \log Y \vert^2).
\end{equation} 
For the remaining integral in formula \eqref{eq:hamiltloclargesmallborder} we use Corollary~\ref{cor:smalllowerbound} and Lemma~\ref{lem:gomegasmallapprox} to get 
\begin{align*}
\int_{\Lambda_{1/10}}& \mathcal{H}_{B_u} (\rho_{\mu}) \dd u  \\
\geq& -\int_{\Lambda_{1/10}}  \left[\iint  \dd x\dd y \, \frac{1}{2}\rho_{\mu}^2   w_{1,B_u}(x,y) + C \rho_{\mu} \delta \Big( \rho_{\mu}\delta (dsK_{\ell})^{-2} + \frac{1}{|B_u|} \Big)\int \chi_{B_u}^2\right]  \dd u\\
\geq& -4\pi \rho_{\mu}^2 \ell^2 \delta -  C \rho_{\mu}^2 \ell^2 \delta^2 K_B^2,
\end{align*}
where we used \eqref{eq:relRrho_mu}, \eqref{eq:K_Bsmallbox} and 
\begin{equation*}
\iiint \, w_{1,B_u} (x,y) \dd x \dd y\dd u = 8 \pi \delta \ell^2, \qquad \iint  \chi_{B_u}(x)^2\dd u \dd x = \ell^2.
\end{equation*}
Collecting the previous estimates, together with \eqref{eq:hamiltloclargesmallborder} and the fact that $\varepsilon_N \leq \frac{1}{2}$, we finally get \eqref{eq:roughlowerboundlarge}, using the expansion \eqref{eq:delta_expansion_small} of $\delta$.

The proof of the bound on $n_+$ is proven noting that, joining together the a priori bound \eqref{eq:condensationaprioricondition} with the obtained lower bound we get 
\begin{equation}
\frac{b}{2 \ell^2} \langle n_+ \rangle_{\Psi} \leq C K_B^2\rho_{\mu} \ell^2 Y^2 |\log Y|,
\end{equation}
and conclude recalling that $\ell = \rho_{\mu}^{-1/2}Y^{-1/2} K_{\ell}  $.
\end{proof}

We follow now a similar strategy to obtain a lower bound for the large box Hamiltonian and get an a priori bound on the number of particles and a control on $\mathcal Q_4^{\rm{ren}}$ in the large box.

\begin{corollary}
If there exists a $n$-particles state $\Psi \in \mathscr{F}_s(L^2(\Lambda))$ such that \eqref{eq:condensationaprioricondition} holds, then the a priori bounds on $n$ and $\mathcal Q_4^{\rm{ren}}$ hold:
\begin{equation}
\left| \rho_{\mu} - \frac{n}{\ell^2}\right| \leq C K_B K_{\ell} \rho_{\mu}  Y^{1/2} |\log Y|^{1/2} ,\qquad \langle \mathcal Q_4^{\rm{ren}}\rangle_{\Psi} \leq C K_B^2 K_{\ell}^2\rho_{\mu}^2 \ell^2 Y^2 |\log Y|.
\end{equation}
\end{corollary}

\begin{proof}
We observe that we have the following lower bound, reproducing analogous estimates for potential and kinetic energies from Lemmas~\ref{lem:potentialsmalllowerbound} and~\ref{lem:kinsmallboxlower} but adapted to the large box $\Lambda$ (for details, see \cite[Appendix E.2]{FS2}), where we estimate the $n_+$ contributions thanks to Theorem~\ref{thm:small_largeHam_lowerbound},
\begin{equation}\label{eq:Q_4largelowerbound}
\langle \mathcal{H}_{\Lambda} (\rho_{\mu}) \rangle_\Psi \geq \frac{1}{2} \langle \mathcal Q_4^{\rm{ren}}\rangle_\Psi - 4 \pi \rho^2_{\mu} \ell^2 \delta + 2 \pi \Big(\rho_{\mu} - \frac{n}{\ell^2} \Big)^2\ell^2 \delta  - CK_B^2 K_{\ell}^2 \rho_{\mu}^2 \ell^2\delta^2. 
\end{equation}
By the assumption, the expansion of $\delta$ in terms of $Y$ and \eqref{eq:Q_4largelowerbound} we get 
\begin{equation}
 \Big(\frac{n}{\ell^2} - \rho_{\mu} \Big)^2\ell^2 \delta + \langle \mathcal Q_4^{\rm{ren}} \rangle_\Psi \leq C K_B^2 K_{\ell}^2\rho_{\mu}^2 \ell^2 Y^2 |\log Y|,
\end{equation}
which implies the desired bounds.
\end{proof}

\section{Technical estimates for off-diagonal excitation terms}\label{app:proofd1d2}
We give here a proof of Lemma~\ref{lem:d1d2estimate}, bounding the terms $d_1^L$ and $d_2^L$ defined in \eqref{def.d1L} and \eqref{def.d2L}.
We are going to use the following dimension independent estimates which are proven in \cite[Corollary F.6]{FS2} in order to prove the technical lemma below. There exists $C>0$, such that, for any $\varphi \in \mathrm{Ran}\overline{Q}_{H}$,
\begin{equation}\label{eq:corF6}
 \|\Delta (\chi_{\Lambda}\varphi)\| \leq C \varepsilon_N^{-1/2} \frac{\widetilde{K}_H^2}{\ell^2}, \qquad \|\Delta^{\mathcal{N}} \varphi\| \leq C \varepsilon_N^{-1} \frac{\widetilde{K}_H^2}{\ell^2}.
\end{equation}

\begin{lemma}\label{lem:normQhigh}
If we assume the relations between the parameters in Appendix~\ref{app:parameters}, then there exists $C>0$ such that 
\begin{equation}
\|\overline{Q}_{H,x} w(x,y) \overline{Q}_{H,x}\| \leq C \varepsilon_N^{-1/2} \frac{\widetilde{K}_H^2}{\ell^2} \|v\|_1.
\end{equation}
\end{lemma}

\begin{proof}
The proof is an adaptation to $2$ dimensions of \cite[Lemma 5.3]{FS2}. Let $\varphi \in \mathrm{Ran}\overline{Q}_{H,x}$ with $\|\varphi\|_2 = 1$, then
\begin{equation}
\|\overline{Q}_{H,x} w(x,y) \overline{Q}_{H,x} \varphi \| \leq I_1 + I_2,
\end{equation}
where
\begin{align*}
I_1 &= \int_{\mathbb{R}^2}  \, \chi_{\Lambda}(x)^2 |\varphi(x)|^2  v(x-y)\dd x,\\
I_2 &= \int_{\mathbb{R}^2}  \, \chi_{\Lambda}(x)\, |\chi_{\Lambda}(x) - \chi_{\Lambda}(y)| \,|\varphi(x)|^2  v(x-y)\dd x.
\end{align*}
We use the technical Lemma~\ref{lem:norminftyto2} below to get 
\begin{equation}
|I_1| \leq \|\chi_{\Lambda} \varphi\|^2_{\infty}\|v\|_1 \leq C \|v\|_1 \|\chi_{\Lambda}\varphi\| \|\Delta \chi_{\Lambda} \varphi\| \leq C\frac{\widetilde{K}_H^2}{\ell^2}  \varepsilon_N^{-1/2} \|v\|_1,
\end{equation}
by \eqref{eq:corF6} and \eqref{eq:norminftyto2_R} and 
\begin{align*}
|I_2| &\leq C\frac{R}{\ell} \| \chi_{\Lambda} \varphi\|_{\infty}  \|\varphi\|_{\infty}\|v\|_1 \leq C \frac{R}{\ell}  \|\Delta (\chi_{\Lambda}\varphi) \|^{1/2}  \|\Delta^{\mathcal{N}}  \varphi \|^{1/2}_{L^2\big(\left[-\frac{L}{2},\frac{L}{2} \right]^2\big)} \|v\|_1 \\ &\leq C\varepsilon_N^{-1/2} \Big( \frac{R}{\ell}\varepsilon_N^{-1/4} \Big) \frac{\widetilde{K}_H^2}{\ell^2}\|v\|_1  \leq C \varepsilon_N^{-1/2} \frac{\widetilde{K}_H^2}{\ell^2}\|v\|_1,
\end{align*}
by a Taylor expansion for the localization function, \eqref{eq:norminftyto2_box} and \eqref{eq:norminftyto2_R} for $\varphi$ and $\chi_{\Lambda}\varphi$, respectively, \eqref{eq:corF6} and the choice of the parameters in \eqref{eq:varepsilonN}, \eqref{eq:relRrho_mu}, \eqref{eq:Kl_KN} 
and this concludes the proof.
\end{proof}

In the proof of Lemma~\ref{lem:normQhigh} we used the following result.

\begin{lemma}\label{lem:norminftyto2}
Let $-\Delta^{\mathcal{N}}$ denote the Neumann Laplacian on $[-\frac{L}{2}, \frac{L}{2}]^2$. There exists $C>0 $ such that, for all $f \in \mathcal{D}(-\Delta^{\mathcal{N}})$ such that $\int_{[-\frac{L}{2},\frac{L}{2}]^2}f(x) \dd x = 0$, we have
\begin{equation}\label{eq:norminftyto2_box}
\|f\|_{\infty} \leq C \|f\|_{L^2([-\frac{L}{2}, \frac{L}{2}]^2)}^{1/2} \|-\Delta^{\mathcal{N}}f\|_{L^2([-\frac{L}{2}, \frac{L}{2}]^2)}^{1/2}.
\end{equation}
Also, for all $f \in H^2(\mathbb{R}^2)$,
\begin{equation}\label{eq:norminftyto2_R}
\|f\|_{\infty} \leq C \|f\|^{1/2}\|\Delta f\|^{1/2}.
\end{equation}
\end{lemma}

\begin{proof}
Let us prove the last inequality, the first one being proven by an adaptation for the box. We use a scaling argument defining 
\begin{equation}
f_{\lambda}(x) := f(\lambda x), \qquad x \in \mathbb{R}^2.
\end{equation}
Given a $f \in H^2(\mathbb{R}^2)$, it is clearly possible to choose $\lambda$ such that $\|f_{\lambda}\| = \|\Delta f_{\lambda} \|$. Now, for the given $\lambda$, we have 
\begin{align*}
\|f_{\lambda}\|^2_{\infty} \leq \frac{1}{(2\pi)^2} \Big( \int_{\mathbb{R}^2} \, |\hat{f}_{\lambda}(p)|\dd p\Big)^2 \leq C \int_{\mathbb{R}^2} \, (1+|p|^4)\, |\hat{f}_{\lambda}(p)|^2\dd p = C \|\Delta f_{\lambda}\|^2,
\end{align*}
where we multiplied and divided by $(1+|p|^4)^{1/2}$ and used the Cauchy-Schwarz inequality and the choice of $\lambda$. Applied to $f_{\lambda}$ with the $\lambda$ chosen above the previous inequality becomes 
\begin{equation}
\|f\|_{\infty}^2 = \|f_{\lambda}\|_{\infty}^{2} \leq C \|\Delta f_{\lambda}\|^2 = C \|f_{\lambda}\| \|\Delta f_{\lambda}\| = C \|f\| \|\Delta f\|,
\end{equation}
where in the last equality we used the scaling properties of the dilatation in $\lambda$ w.r.t. the $L^2$ norm.
\end{proof}

\begin{proof}[Proof of Lemma~\ref{lem:d1d2estimate}]
Let $\Psi \in \mathscr{F}_s(L^2(\Lambda))$ be satisfying the assumptions of Lemma~\ref{lem:d1d2estimate}. Our goal is to prove the following estimate
\begin{align}
\langle (d_{1}^L + d_2^L)\rangle_{\widetilde{\Psi}} &\leq \rho_{\mu} \|v\|_1 \big(  \langle n_+ \rangle_{\widetilde{\Psi}}  +   n^{1/2}   \langle  n_+^L\rangle_{\widetilde{\Psi}}^{1/2} +  \langle  (n_+^L)^2\rangle_{\widetilde{\Psi}}^{1/2} \varepsilon_N^{-1/4} \widetilde{K}_H +  \langle    n_+^L \widetilde{n}_+^H\rangle_{\widetilde{\Psi}} \varepsilon_N^{-1/4} \widetilde{K}_H    \big. \nonumber\\
 &\quad+ \left. \big( \langle \widetilde{n}_+^H n_+^L \rangle_{\widetilde{\Psi}}^{1/2} \langle (n_+^L)^2 \rangle_{\widetilde{\Psi}}^{1/2}  + \langle \widetilde{n}_+^H n_+^L \rangle_{\widetilde{\Psi}}    \big)n^{-1} \varepsilon_N^{-1/2} \widetilde{K}_H^2  \right) + C \langle  Q_4^{\text{ren}}\rangle_{\widetilde{\Psi}}. \label{eq:partial_d1d2estimate}
\end{align}

We split the $d_{j}^L$ in several terms multiplying out the parentheses in \eqref{def.d1L} and \eqref{def.d2L}.
All these terms we treat individually using Cauchy-Schwarz inequalities. Similar bounds have been carried out in \cite{FS2}.
Here we just bound some representative examples to illustrate the procedure and the role played by Lemma~\ref{lem:normQhigh}.

Let us start using the Cauchy-Schwarz inequality for any $\varepsilon >0 $ to get 
\begin{equation*}
\Big|\Big\langle -\rho_{\mu} \sum_{i} P_i \int\dd y \, w_1(x_i,y)\; \overline{Q}_{H,i} + h.c. \Big\rangle_{\Psi}\Big| \leq \frac{n}{\ell^2} \|w_1\|_1 \big(\varepsilon n + \varepsilon^{-1} \langle n_+^L\rangle_{\Psi}\big), 
\end{equation*}
observing that $\|w_1\|_1 \leq C\delta$ and choosing $\varepsilon = \langle n_+^L\rangle_{\Psi}^{1/2} n^{-1/2}$, we obtain the desired quantity.

For the following term, for any $\varepsilon >0$, we use the Cauchy-Schwarz inequality and Lemma~\ref{lem:normQhigh},
\begin{align*}
\Big|\Big\langle  \sum_{i,j} P_i \overline{Q}_{H,j} w \overline{Q}_{H,i} \overline{Q}_{H,j}\Big\rangle_{\Psi}\Big| &\leq  \varepsilon \frac{n}{\ell^2} \|w\|_1\langle n^L_+ \rangle_{\Psi}  + \varepsilon^{-1} \|\overline{Q}_H w\overline{Q}_H\| \sum_{i \neq j} \langle \overline{Q}_{H,i} \overline{Q}_{H,j} \rangle_{\Psi}\\
&\leq \frac{\|w\|_1}{\ell^2}( \varepsilon n^2  + \varepsilon^{-1} \varepsilon_N^{-1/2} \widetilde{K}_H^2 \langle (n_+^L)^2 \rangle_{\Psi})
\end{align*}
where we used that $n_+^L \leq n_+$. Choosing $\varepsilon = \varepsilon_N^{-1/4} \widetilde{K}_H \langle (n_+^L)^2 \rangle_{\Psi}^{1/2} n^{-1}$, we obtain
\begin{equation}
\Big|\Big\langle  \sum_{i,j} P_i \overline{Q}_{H,j} w \overline{Q}_{H,i} \overline{Q}_{H,j}\Big\rangle_{\Psi}\Big| \leq n\frac{\langle (n_+^L)^2\rangle_{\Psi}^{1/2}}{\ell^2} \varepsilon_N^{-1/4} \widetilde{K}_H \|w\|_1.
\end{equation}

For the next term we want to apply a Cauchy-Schwarz inequality to reobtain a $Q_4^{\text{ren}}$ term. In order to do that we are going to complete the $Q_H$ to a $Q= Q_H + \overline{Q}_H$. 
\begin{align*}
\Big|\Big\langle  \sum_{i \neq j} \overline{Q}_{H,i} P_j w Q_{H,i} Q_{H,j}  + h.c.\Big\rangle_{\Psi} \Big| 
&\leq  \Big|\Big\langle \sum_{i \neq j} \overline{Q}_{H,i} P_j w Q_i Q_j \Big\rangle_{\Psi} + h.c. \Big|  \\
&\quad+  \Big|\Big\langle \sum_{i \neq j} \overline{Q}_{H,i} P_j w (Q_{H,i} \overline{Q}_j + \overline{Q}_{H,i} Q_{H,j}) \Big\rangle_{\Psi} + h.c. \Big| \\
&\quad + \Big|\Big\langle \sum_{i,j} P_i \overline{Q}_{H,j} w \overline{Q}_{H,i} \overline{Q}_{H,j} \Big\rangle_{\Psi}\Big|.
\end{align*}
The second term and the third terms can be estimated in the same manner as above, so let us focus on completing the first term in order to obtain the $Q_4$.
\begin{align}
 \Big|\Big\langle&\sum_{i \neq j}  \overline{Q}_{H,i} P_j w Q_i Q_j \Big\rangle_{\Psi} + h.c. \Big|  \\
 &\leq \Big|\Big\langle \sum_{i \neq j} \overline{Q}_{H,i} P_j w (Q_i Q_j + \omega (P_i P_j + P_i Q_j + Q_i P_j)) \Big\rangle_{\Psi} + h.c. \Big|  \label{eq:Q4recostructionQhigh}\\
 &\quad + \Big|\Big\langle \sum_{i \neq j} \overline{Q}_{H,i} P_j w  \omega ( P_i Q_j + Q_i P_j)) \Big\rangle_{\Psi} + h.c. \Big|  \\
 &\quad + \Big|\Big\langle \sum_{i \neq j} \overline{Q}_{H,i} P_j w  \omega P_i P_j) \Big\rangle_{\Psi} + h.c. \Big|.
\end{align}
The second and the third terms are treated as above, using that $0 \leq \omega \leq 1$ on the support of $w$. By a Cauchy-Schwarz inequality on the first term we get
\begin{equation}
\eqref{eq:Q4recostructionQhigh} \leq \langle \mathcal Q_4^{\text{ren}}\rangle_{\Psi} + C \frac{n}{\ell^2} \|w\|_1 \langle n_+\rangle_{\Psi}.
\end{equation}

Collecting the previous estimates including the ones not explicitly treated, we obtain \eqref{eq:partial_d1d2estimate}.

Bounding $n_+^L \leq \widetilde{\mathcal{M}}$ in \eqref{eq:partial_d1d2estimate} where it appears for higher moments than $1$, using that $\widetilde{n}_+^H \leq n $ and that $\varepsilon_N^{-1/4} \widetilde{K}_H \geq 1$ by \eqref{eq:rel_KH_KN} gives the result.
This finishes the proof of Lemma~\ref{lem:d1d2estimate}.
\end{proof}

\section{Properties of the localization function}\label{app:locfunction}
We collect here the definition and some important properties of the localization function that are used throughout the paper.

We define
\begin{equation}
\chi(x) := C_M (\zeta_1(x_1) \zeta_2(x_2))^{M+2},
\end{equation}
where
\begin{equation}
\zeta(y) := \begin{cases}
\cos(\pi y), &|y| \leq 1/2,\\
0, &|y| >1/2,
\end{cases}
\end{equation}
where $M \in \mathbb{N}$ is chosen even and large enough. The normalization constant $C_M>0$ is chosen in order to obtain $\|\chi\|_2 = 1$. We have $0 \leq \chi \in C^M(\mathbb{R}^2)$. We also define $\chi_\Lambda(x) = \chi(x/\ell)$.
 
\begin{lemma}\label{lem:localization_properties}
Let $\chi$ be the localization fuction defined above and let $M \in 2 \mathbb{N}$. Then, for all $k \in \mathbb{R}^2$, 
\begin{equation}
|\widehat{\chi}(k)| \leq \frac{C_{\chi}}{(1+ |k|^2)^{M/2}}, 
\end{equation}
where $C_{\chi} = \int |(1-\Delta)^{M/2}\chi|$.
If, furthermore, $|k| \geq \frac{1}{2} K_K \ell^{-1}$,
\begin{equation}\label{eq:fourierlocestimate}
|\widehat{\chi_{\Lambda}}(k)| = \ell^2 |\widehat{\chi}(k\ell)| \leq C \ell^2 K_{H}^{-M}.
\end{equation}
\end{lemma}
An important property for the localization function $\chi_{B_u}$, $u \in \mathbb{R}^2$, on the small boxes, namely
\begin{equation}\label{eq:chiB}
\chi_{B_u}(x) := \chi_{\Lambda}(x) \chi \Big( \frac{x}{d\ell}-u\Big),
\end{equation}
which is used in Appendix~\ref{app:low-smallbox}, is the following bound
\begin{equation}\label{eq:bound_loc_smallbox_M}
\|\nabla^2\chi_{B_u}\|_{\infty} \leq C_M \frac{1}{|B_u|\ell_{\min}^2} \int \chi_{B_u},
\end{equation}
which is taken from \cite[Appendix C]{FS}. Here it is key the fact that we do not consider a smooth function but we require $\chi$ to have a finite degree of regularity measured by the parameter $M$.

\section{Comparing Riemann sums and integrals}\label{app:sum_to_int}

We will show in this section that we could approximate integrals on $\R^2$ by Riemann sums when it was needed in \eqref{eq. sum just before the int} to prove the upper bound. Recall that the assumptions of Theorem~\ref{thm.upperbound.soft.pot} were
\begin{equation}\label{eq: assumptions on R and L}
R \leq \rho^{-1/2} Y^{1/2}, \qquad L_{\beta} = \rho^{-1/2} Y^{-\beta}.
\end{equation}

We divide $\R^2$ into small squares $\square_p$ of size $\frac{2\pi}{L}$ centered at $p \in \Lambda_L^* = \frac{2\pi}{L} \Z^2$. Then, clearly
\begin{align}
	\Big\vert \frac{4 \pi^2}{ L^2} \sum_{p \in \Lambda_L^*} f(p) - \int_{\R^2} f(k) \dd k \Big\vert
	&\leq \frac{C}{L^3}\sum_{p \in \Lambda_L^*} \sup_{\square_p} \vert \nabla f \vert. \label{ine:riesum}
\end{align}
We consider the functions present in the two sums of \eqref{eq. sum just before the int}. With $\alpha_p$ and $\gamma_p$ given in \eqref{def:alphagamma} the first term is
\begin{align}\label{eq:def_f}
	f(p)&=p^{2}+\rho_0\widehat{g}_{p}-\sqrt{p^{4}+2\rho_0\widehat{g}_{p}p^{2}}+\rho_0 \big(\widehat{v}_{p}-\widehat{g}_{p}\big)\big (\gamma_{p}+\alpha_{p}\big ) \nonumber \\
	&=p^{2}+\rho_0\widehat{g}_{p}-\sqrt{p^{4}+2\rho_0\widehat{g}_{p}p^{2}}+\rho_0 ( \widehat{v}_{p}-\widehat{g}_{p} )\Big(\frac{p^{2}}{2\sqrt{p^{4}+2\rho_0\widehat{g}_{p}p^{2}}}-\frac{1}{2}\Big),
\end{align}
and the second term 
\begin{align}\label{eq:def_d}
	d (p,r)=\widehat{v}_{r}\alpha_{p+r}\alpha_{p}.
\end{align}
We then have the following estimates
\begin{lemma}\label{lem:RiemannSums2}
Let $f,d$ be as in \eqref{eq:def_f} and \eqref{eq:def_d}. Then,
	\begin{align}
		\Big \vert \frac{1}{\vert \Lambda_{\beta} \vert}\sum_{p\in \Lambda_{\beta}^{*}}f(p)-\int_{\R^{2}}f(k) \frac{\dd k}{4\pi^2} \Big\vert &\leq C \rho^{2}Y^{1/2 + \beta}\widehat{v}_0, \label{eq:f}
		\intertext{and}
		\Big\vert \frac{1}{\vert \Lambda_{\beta} \vert^{2}}\sum_{p,r\neq 0}d(p,r)-\int_{\R^{4}}d(p,r) \frac{\dd p\dd r}{\left (4\pi^2\right )^{2}} \Big\vert &\leq C \rho^{2}Y^{1/2 + \beta}\widehat{v}_0.\label{eq:da}
	\end{align}
\end{lemma}
\begin{proof}
	In order to apply \eqref{ine:riesum}, we start by calculating the gradient
	\begin{align}
		\partial_{p}f&=2p +\rho_0 \partial_{p}\widehat{g}_{p}-2p\frac{\big( 1+\frac{\rho_0\widehat{g}_{p} }{p^{2}}+\frac{\rho_0\partial_{p}\widehat{g}_{p} }{2p}\big)}{\sqrt{1+\frac{2\rho_0\widehat{g}_{p} }{p^{2}}}}\nonumber\\
		&\quad+\rho_0 \big(\partial_{p}\widehat{v}_{p}-\partial_{p}\widehat{g}_{p}\big)\Big (\frac{p^{2}}{2\sqrt{p^{4}+2\rho_0\widehat{g}_{p}p^{2}}}-\frac{1}{2}\Big )\nonumber\\
		&\quad+\rho_0^{2} \left (\widehat{v}_{p}-\widehat{g}_{p}\right )\frac{\widehat{g}_{p} p^{3}-\frac{1}{2}\partial_{p}\widehat{g}_{p}p^{4}}{\left (p^{4}+2\rho_0\widehat{g}_{p}p^{2}\right )^{\frac{3}{2}}}\nonumber\\
		&:= A_p+B_p+C_p.
	\end{align}
	We will now systematically omit the constants and study separately the cases $p\leq \sqrt{2\rho_0 \widehat{g}_{0}}$ (case $1$ referring to $A_{p}^{1}$) and $p\geqslant \sqrt{2\rho_0 \widehat{g}_{0}}$ (case $2$ referring to $A_{p}^{2}$). We then get by elementary inequalities
\begin{align}
		 \left \vert A_p^{1}\right \vert &\leq  (\rho \widehat{g}_0)^{1/2},  
		&\left \vert A_p^{2}\right \vert& \leq   \frac{\rho^2 R \widehat{g}_0 \widehat{g}_p}{p^2}  + \frac{(\rho \widehat{g}_{p})^2}{p^3} , \qquad \quad \,\, \nonumber \\
		\left \vert B_p^{1}\right \vert &\leq  \rho R(\widehat{v}_0-\widehat{g}_0),
		&\left \vert B_p^{2}\right \vert &\leq  \frac{R\rho^{2} (\widehat{v}_0-\widehat{g}_{0})\widehat{g}_{p}}{p^{2}}, \qquad \qquad \qquad   \nonumber\\
		\left \vert C_p^{1}\right \vert &\leq  \frac{\rho^{1/2}(\widehat{v}_0-\widehat{g}_0)}{\widehat{g}_0^{1/2}}, &
		\left \vert C_p^{2}\right \vert &\leq  (\widehat{v}_0-\widehat{g_0})\Big( \frac{\rho^2 \widehat{g}_{0}}{p^{3}}+\frac{R\rho^{2} \widehat{g}_{p}}{p^{2}} \Big),\nonumber
	\end{align}
	where we used $\vert\widehat{g}_0-\widehat{g}_p\vert\leq \vert\widehat{g}_0^{3/2}\vert$ for $p\leq (\rho \widehat{g}_0)^{1/2}$.
	This way we can use inequality \eqref{ine:riesum} and the decay of $\widehat{g}_{p}$ \eqref{eq:vtildehat.bound} to get
	\begin{align}
		\frac{1}{L_{\beta}^3} \sum_{p \in \Lambda^*}\left \vert\partial_{p}f \right \vert\dd p&\leq \frac{C}{L_{\beta}} \big((\rho \widehat{g}_0)^{3/2}+ \rho^{3/2}\widehat{g}_0^{1/2}(\widehat{v}_0-\widehat{g}_0) +\rho^2\widehat{g}_0R+R\rho^2(\widehat{v}_0-\widehat{g}_0)\big)\\& \leq C \widehat v_0 \rho^2 Y^{1/2 + \beta}, \nonumber
	\end{align}
	where we used \eqref{eq: assumptions on R and L}, and $\widehat{g}_0\leq \widehat{v}_0$. We use the same method to prove \eqref{eq:da}. We have
\begin{align}\label{ine: alpha_new}
\vert\alpha_{p} \vert\leq \begin{cases}\frac{\sqrt{\rho\widehat{g}_{0} }}{p}, & \text{ for } p\leq \sqrt{\rho\widehat{g}_{0} }, \\
\frac{\rho\left \vert\widehat{g}_{p}\right \vert}{p^{2}}, & \text{ for } p\geq \sqrt{\rho\widehat{g}_{0} }.
\end{cases}
\end{align}
	We have to calculate
	\begin{equation}
		\partial_{p}\alpha_{p}=-\frac{\rho\partial_{p} \widehat{g}_{p}}{2\sqrt{p^{4}+2\rho\widehat{g}_{p}p^{2} }}-\frac{\rho \widehat{g}_{p}\left (4p^{3}+4\rho\widehat{g}_{p}p+2\rho\partial_{p}\widehat{g}_{p}p^{2}\right )}{2({p^{4}+2\rho\widehat{g}_{p}p^{2} })^{3/2}},\nonumber
	\end{equation}
yielding
\begin{align}\label{ine: dalpha_new}
\vert\partial_{p}\alpha_{p} \vert\leq 
\begin{cases}\frac{\sqrt{\rho \widehat{g}_{0}}R}{p}+\left (\rho \widehat{g}_{0}\right )^{-1/2}+\frac{ \sqrt{\rho\widehat{g}_{0} }}{p^{2}}+\frac{ \sqrt{\rho\widehat{g}_{0} }R}{p}, & \text{ for }p\leq \sqrt{\rho\widehat{g}_{0} },\\
\frac{\rho \widehat{g}_{p}R}{p^{2}}+\frac{\rho \widehat{g}_{p}}{p^{3}}+\frac{\left (\rho \widehat{g}_{p}\right )^{2}}{p^{5}}+\frac{\left (\rho \widehat{g}_{p}\right )^{2}R}{p^{4}}, & \text{ for } p\geq \sqrt{\rho\widehat{g}_{0} }.
\end{cases}
\end{align}
	The divergence in $p\to 0$ implies to remove a little box around the point $0$
	\begin{align*}
		\Big\vert\frac{16\pi^{4}}{L_{\beta}^{4}}\sum_{p,r\neq 0} d (p,r)-\int_{\R^{4}} d(p,r )\dd p\dd r\Big \vert &\leq \Big\vert\frac{16\pi^{4}}{L_{\beta}^{4}}\sum_{p,r\neq 0} d(p,r)-\int_{\big(\R^{2}\backslash[-\frac{1}{L},\frac{1}{L}]^{2}\big )^{2}} d(p,r)\dd p\dd r\Big\vert \\
		&\quad + \Big \vert\int_{\R^{2}\times[-\frac{1}{L},\frac{1}{L}]^{2}} d(p,r)\dd p\dd r\Big \vert  \\
		&\quad+ \Big \vert\int_{[-\frac{1}{L},\frac{1}{L}]^{2}\times \R^{2}} d(p,r)\dd p\dd r\Big \vert. 
	\end{align*}
	where the last two terms in the above can be bounded by $\rho^{2}Y^{1/2+\beta}\widehat{v}_{0}$. Finally a direct computation using the decay of $\widehat{g}_{p}$, the bounds \eqref{eq:vhatp}, \eqref{ine: alpha_new},  \eqref{ine: dalpha_new}, and \eqref{ine:riesum} yields
	\begin{align}
		&\Big \vert\frac{16\pi^{4}}{L_{\beta}^{4}}\sum_{p,r\neq 0} d (p,r)-\int_{\big(\R^{2}\backslash[-\frac{1}{L},\frac{1}{L}]^{2}\big )^{2}} d(p,r)\dd p\dd r\Big \vert 
		\nonumber \\
		&\leq \frac{1}{L_{\beta}^{5}}\sum_{p,r\neq 0}\sup_{\square_p\times \square_r}\left \vert\nabla_{p,r} d(p,r)\right \vert\nonumber\\
		&\leq C\frac{\widehat{v}_{0}}{L_{\beta}^{5}}\sum_{p\neq 0}\vert\partial_{p}\alpha_{p}\vert\sum_{r\neq 0}\vert\alpha_{r}\vert+C\frac{R\widehat{v}_{0}}{L_{\beta}^{5}}\Big(\sum_{r}\vert\alpha_{r}\vert\Big )^{2}\nonumber\\
		&\leq C\widehat{v}_0\rho^2Y^{1/2 + \beta},
	\end{align}
	where we used the estimates of Lemma~\ref{lem:ag}. This concludes the proof.
\end{proof}

\section{Fixing parameters for the lower bound}\label{app:parameters}
Here we collect all the relations and dependencies of the several parameters involved in the lower bound for the convenience of the reader. Furthermore, we end the section by making an explicit choice that satisfies all the relations.
Recall that we have the small parameter
\begin{equation*}
Y = Y_{\mu} = |\log(\rmu a^2)|^{-1}.
\end{equation*}
We use the following notation throughout the article
\begin{equation}\label{eq:order_comparison}
A \ll B \quad \text{if and only if there exist } C, \varepsilon > 0  \;\text{ s.t. } A \leq CY^{\varepsilon} B.
\end{equation}
In the proof of the lower bound, a number of positive parameters are needed. These are the following
\begin{equation}
d, s, \varepsilon_T, \varepsilon_K, \varepsilon_N, \varepsilon_{\mathcal{M}}
\ll
1
\ll
{\mathcal M},
K_{\ell}, K_H,\widetilde{K}_H, K_N, K_B.
\end{equation}
These will be chosen below.

Furthermore, there are length scales $\ell_{\delta}$ and $R$.
These will be chosen to satisfy
\begin{align}
&R \leq \rho_{\mu}^{-1/2}, & \text{Condition on the radius of the support }, \label{eq:relRrho_mu} \\
&\ell_{\delta} = \frac{e^{\Gamma}}{2} \rho^{-1/2}_{\mu} Y^{-1/2}, &\text{healing length condition}\label{eq:rel_healing}.
\end{align}
Some first relations between the parameters are 
\begin{align} 
&d \ll 1 \ll K_{\ell},  &\text{sep. of small and large boxes},\label{eq:d^-2<<Kl}\\
&d^{-2} \ll K_H \ll \widetilde{K}_H, &\text{sep. of low and high momenta},\label{eq:rel_d_k_H}\\
&d \ll (s K_{\ell})^{-1} \ll 1, &\text{condition for Bog. integral} \label{eq:rel_s_Kl},\\
&d^2 K_{\ell}^4 \ll \varepsilon_T  \ll  d s K_{\ell}, &\text{spectral gap condition} \label{eq:rel_T2comm},\\
&ds^{-1} \leq C, &\text{localization to small boxes}.
\end{align}

The combination of \eqref{eq:rel_d_k_H} and \eqref{eq:rel_T2comm} implies the following relations:
\begin{equation}\label{eq:rel_Kl_KH_d}
K_\ell \ll K_{\ell}^2 \ll s d^{-1} \ll d^{-1} \ll d^{-2} \ll K_H. 
\end{equation}

Defining 
\begin{equation}\label{eq:varepsilonN}
\varepsilon_N := K_N^{-1} Y, \qquad \varepsilon_{\mathcal{M}} := \frac{\mathcal{M}}{\rho_{\mu}\ell^2}, 
\end{equation}
we give the following conditions which control the magnitude of the large parameters in terms of $Y$:
\begin{align}
&(dsK_{\ell})^{-1} \ll K_B, & \text{condition errors in small box}, \label{eq:K_Bsmallbox}\\
&K_B K_{\ell} \widetilde{K}_H K_N^{1/4} \ll Y^{-1/4}, &\text{small error in large matrices}, \label{eq:relK_B-K_H-K_N-K_L} \\
&K_{\ell}^{-1} K_N^{1/4} \ll Y^{-1/2}, &\text{technical estimate in large matrices}  \label{eq:Kl_KN}, \\
&\widetilde{K}_H K_N^{-1/4} \gg Y^{1/4}, &\text{technical estimate in large matrices}, \label{eq:rel_KH_KN}\\
&K_{\ell}^2 K_H^2 \mathcal{M} \ll Y^{-1},  &\text{second localization of 3Q term}, \label{cond:3Qloc}\\
&K_B^2 K_{\ell}^2 \ll Y^{-1/4}, &\text{number for high momenta} \label{eq:Kl_KN_number},\\
&K_{\ell}^{10} K_H^{-8} d^{-4} \ll Y^{-1}, &\text{condition error in } \mathcal{T}_1. \label{eq:rel_errQ3_statement}
\end{align}
Here the magnitude of the small parameters:
\begin{align}
&\varepsilon_R \ll K_B^{-2} K_{\ell}^{-2} |\log Y|^{-1}, & \text{Condition on $\varepsilon_R$},\label{eq:rel_epsR_Kl_KB}\\
&\varepsilon_K \ll K_{\ell}^{-2}, & \text{error in } \mathcal T_{2,com}',\label{eq:rel_T2comm2}\\
&\varepsilon_K \gg K_{\ell}^4 K_H^{-4} (d^{-2} \varepsilon_{\mathcal{M}}^{1/2}+ d^{-4} \varepsilon_{\mathcal M}), &\text{condition error in } \mathcal{T}_1 \text{ and }\mathcal{T}_2, \label{eq:rel_epsK_KM2}\\
&\varepsilon_{\mathcal{M}} \ll d^{8}K_{\ell}^{4}\varepsilon_T^{-2}, &\text{condition for error } \delta_1. \label{eq:rel_epsT_d_Kl}\\
&\varepsilon_N \leq \varepsilon_T^{-2} d^4 K_\ell^4, &\text{bound from Lemma~\ref{lem:technicalest_rhofar}}. \label{eq:rel_eps_M_small}
\end{align}

We use the fundamental property of the system that the number of excitations of our state is relatively small compared to the number of particles (expressed by the condition $\varepsilon_{\mathcal{M}} \ll 1$) but still larger that a certain threshold. This property is expressed by the following condition:
\begin{equation}
\mathcal{M} \gg Y^{-7/8} |\log Y|^{1/4} K_B^{1/2} K_{\ell}^{1/2} K_N^{1/8} \widetilde{K}_H^{1/2}\|v\|_1^{1/2}. \label{eq:M_large_matrices}
\end{equation}

The following are conditions that impose constraints on the size of $M$, the degree of regularity of the localization function $\chi$:
\begin{align}
& d^{2M-2} \ll Y,   &\text{error in localization 3Q}, \label{eq:rel_locQ3low}\\
& d^2 K_{\ell}^{4} \ll \varepsilon_T , &\text{error in localization 3Q},\label{eq:rel_locQ3low2}\\
& \varepsilon_N^{3/2} + \Big( \frac{K_H}{\widetilde{K}_H}\Big)^M + (d^2K_H)^{-2M} \leq \varepsilon_{\mathcal{M}}, &\text{number for high momenta} \label{eq:rel_KH_K_M},\\
&(s^{-2} + d^{-2})(s d )^{-2} s^M \leq C,   & \text{localization to small boxes}. \label{eq:condition_small_loc}
\end{align}

A choice of parameters, non-optimal in the size of the error produced, fitting the previous conditions, is the following,
\begin{align}
 M &= 258,  &\mathcal{M} &= Y^{-\frac{31}{32}},    &\varepsilon_T &= Y^{\frac{1}{512} - \frac{1}{8192}}, \nonumber\\
K_{\ell} &= Y^{-\frac{1}{2048}},  &K_H &= Y^{-	\frac{1}{128}},  &\widetilde{K}_H &= Y^{-\frac{1}{64}},\nonumber \\ 
d &= Y^{\frac{1}{512}}, &  K_N&= Y^{-\frac{1}{512}},  &s &= Y^{\frac{1}{4096}},\nonumber\\
K_B &= Y^{-\frac{1}{512}},   &\varepsilon_K &= Y^{\frac{1}{512}}, &\varepsilon_{\mathcal{M}} &= Y^{{\frac{1}{32}} + \frac{1}{1024}}.
\end{align}
This choice is not made without any particular view towards optimality.
 
%%%%%%%%%%%%%%%%%%%%%%%%%%%%%%%%%%%%

%%%%%%%%%%%%%%%%%%%%%%%%%%%%%%%%%%%%%%%%%

\bibliographystyle{hplain}
\bibliography{2drefs}

\end{document}